\documentclass[preprint,10pt]{elsarticle}

\usepackage{amssymb}
\usepackage{amsmath}
\usepackage{amsthm}
\newtheorem{theorem}{theorem}[section]
\journal{Mathematics Journal}

\begin{document}

\begin{frontmatter}



\title{The frequency $K_i$s for symmetric traveling salesman problem}


\author{Yong Wang} %% Author name

\affiliation{organization={New Energy School, North China Electric Power University},%Department and Organization
            addressline={No.2, Beinong Road, Huilongguan, Changping}, 
            city={Beijing},
            postcode={102206}, 
            %state={},
            country={China}}

%% Abstract
\begin{abstract}
The frequency $K_i$s ($i\in[4,n]$) are studied for symmetric traveling salesman problem ($TSP$) to characterize the structure properties of the edges in optimal Hamiltonian cycle ($OHC$). For a given $K_i$ in the complete graph $K_n$, the frequency $K_i$ is computed with the set of ${{i}\choose{2}}$ optimal $i$-vertex paths with fixed endpoints (optimal $i$-vertex paths) in the $K_i$. %In the frequency $K_i$, the frequency of an edge is the number of the optimal $i$-vertex paths containing the edge in the  $K_i$. 
Given an $OHC$ edge related to $K_i$, it has certain frequency bigger than $\frac{1}{2}{{i}\choose{2}}$ in the frequency $K_i$, and that of an ordinary edge not in $OHC$ is smaller than $2(n-3)$. %On average, an $OHC$ edge in the $K_i$ has the expected frequency bigger than $\frac{i^2-4i+7}{2}$ whereas an ordinary edge has the expected frequency smaller than 2. 
Moreover, given a frequency $K_i$ containing an $OHC$ edge related to $K_n$, the frequency of the $OHC$ edge is bigger than $\frac{1}{2}{{i}\choose{2}}$ in the average case. %It implies that the average frequency of an $OHC$ edge  computed with the frequency $K_i$s is bigger than $\frac{1}{2}{{i}\choose{2}}$. 
It also found that the probability that an $OHC$ edge is contained in the optimal $i$-vertex paths increases according to $i\in [4, n]$ or keeps stable if it decreases from $i$ to $i+1\leq n$. As the frequency $K_i$s are used to compute the frequency of an edge, each $OHC$ edge reaches its own peak frequency at $i=P_0$ where $P_0=\frac{n}{2} + 2$ for even $n$ or $\frac{n+1}{2} + 1$ for odd $n$. For each ordinary edge out of $OHC$, the probability that they are contained in the optimal $i$-vertex paths decreases according to $i$, respectively, in the average case. %Moreover, the average frequency of an ordinary edge will be smaller than $\frac{1}{2}{{i}\choose{2}}$ if $i \geq 2i_d$  where $i_d$ is the smallest number meeting the condition $\frac{(n-2)(n-3) - (i_d-2)(i_d-3)}{(n-2)(n-3) - (i_d-1)(i_d-2)} \geq \sqrt{1 + \frac{2}{i_d(i_d+1)}}$ and $i_d = O(n^{\frac{4}{7}})$. 
Moreover, the probability of an ordinary edge definitely decreases if $i \geq i_d$  where $i_d = O(n^{\frac{4}{7}})$ is the smallest number meeting the inequality $\frac{(n-2)(n-3) - (i_d-2)(i_d-3)}{(n-2)(n-3) - (i_d-1)(i_d-2)} \geq \sqrt{1 + \frac{2}{i_d(i_d+1)}}$. Based on these findings, an algorithm is presented to find $OHC$ in $O(n^2i_d^42^{i_d})$ time using dynamic programming. The experiments are executed to verify these findings with various $TSP$ instances. 
%Abstract text.
\end{abstract}

\begin{graphicalabstract}
\end{graphicalabstract}

\begin{highlights}
\item The lower frequency bound for $OHC$ edges in $K_i$ is derived.
\item The lower frequency bound for $OHC$ edges in $K_n$ is derived. 
\item The frequency changes for $OHC$ and ordinary edges are demonstrated.
\item All $OHC$ edges can be found in $O(n^2i_d^42^{i_d})$ time for $i_d = O(n^{\frac{4}{7}})$. 
\item The structure properties of $OHC$ edges are characterized.
\item The experiments are executed for verifying the findings.
\end{highlights}

%% Keywords

\end{frontmatter}



\section{Introduction}
\label{sec1}
%% Labels are used to cross-reference an item using \ref command.
Traveling Salesman Problem ($TSP$) is extensively studied in combinatorial optimization, operations research, computer science and engineering \cite{DBLP:books/Gutin07}. %(Gutin and Punnen, 2007)
%In mathematics, the problem
%is stated as follows.
%$TSP$ is intuitively represented
%as a weighted complete graph $K_n$.
Given $K_n$ on $n$ vertices $\{1, \ldots, n\}$, there is certain distance $d(u,v)> 0$ for an edge $(u,v)$ where $u \neq v \in \{1, 2, \ldots, n\}$. For the symmetric $TSP$, $d(v,u)=d(u,v)$ exists for any pair of vertices $u$, $v$. A Hamiltonian cycle ($HC$) is one cycle visiting each of all vertices exactly once. The optimal Hamiltonian cycle ($OHC$) is the shortest one among all $HC$s. Finding $OHC$ is the classic $TSP$ and it is $NP$-Complete \cite{DBLP:journals/Karp75}. The number of $HC$s is $\frac{(n-1)!}{2}$ in symmetric $K_n$. As $n$ becomes big, it is impractical to find $OHC$ by the methods based on enumeration. 

The exact algorithms usually require $O(a^n)$ time for finding $OHC$ where $a>1$. For example, the dynamic programming
consumes $O(n^22^n)$ time owing to Bellman \cite{DBLP:journals/Bellman62}, and independently Held and Karp \cite{DBLP:journals/Karp62}. %(1962).
The techniques based on cutting-plane \cite{DBLP:journals/Levine00,DBLP:journals/Applegate09} %(Levine, 2000; Applegate, et al, 2009)
and branch-and-bound \cite{DBLP:journals/Carpaneto95,DBLP:journals/Klerk11} %(Carpaneto, et al, 1995; Klerk and Dobre, 2011)
are able to tackle $TSP$ with thousands of vertices. In 2006, one large Euclidean $TSP$ having 85,900 nodes was optimized  \cite{DBLP:journals/Applegate09}. %(Applegate, et al, 2009). 
In 2019, Cook \cite{DBLP:journals/Cook19} %(2019)
reported that the $OHC$ through 109,399 stars was found. However, the exact algorithms generally consume long time for resolving the $TSP$ instances of large scale.

$TSP$ is hard to resolve because $OHC$ edges do not show special distances in $K_n$. In other words, although $OHC$ edges are different from ordinary edges, it is difficult to separate them  according to the distances on edges. In the 1920s, Menger pointed out the nearest-neighbor method is not helpful for $TSP$ even though it works well to find the minimum spanning tree ($MST$) in graphs \cite{DBLP:journals/Liu08}. In contrast, many ordinary edges have some special features based on distances. Forty years ago, Jonker and Volgenant \cite{DBLP:journals/Jonker84} identified more than half ordinary edges out of $OHC$ based on 2-$opt$ $moves$. Ten years ago, Hougardy and Schroeder \cite{DBLP:journals/Hougardy14} found more ordinary edges using 3-$opt$ $moves$. Zhong \cite{DBLP:journals/Zhong18} compared the two methods when they were applied to Euclidean $TSP$, and illustrated that the method based on the 3-$opt$ $moves$ identified more ordinary edges than that according to the 2-$opt$ $moves$. % different performance of the two procedures. %procedure proposed by Hougardy and Schroeder was better than that introduced by Jonker and Volgenant. 

According to the distances on edges, the properties of $OHC$ edges are not fully  demonstrated in the above studies. In real-world applications, there are selection criteria for finding some $OHC$ edges. To reduce the search time of heuristic algorithms, a limited number of nearest neighbors (i.e., five nearest  neighbors) to each vertex are considered in the original LK algorithm \cite{DBLP:journals/Lin73}. The LK heuristic was improved to LKH-1 in the year of 2000. One of the essential improvements is the $\alpha$-measure computed for edges based on 1-tree. It was taken as the criteria for choosing a few candidate $OHC$ edges linking each vertex \cite{DBLP:journals/Helsgaun20}. The experiments illustrated that the  $\alpha$-measure was much better than the method of nearest neighbors for choosing the candidate $OHC$ edges related to each vertex, and it also played well in the new version LKH-2  \cite{DBLP:journals/Helsgaun09}. Besides the $\alpha$-measure, pseudo-backbone edges were proposed for finding the candidate $OHC$ edges contained in a set of high quality tours \cite{DBLP:journals/Jäger14}. A backbone edge is originally contained in every $OHC$ and it is relaxed to one pseudo-backbone edge. The tours of high quality can be computed with the heuristic algorithms, such as LKH. %A pseudo-backbone edge is contained in each of the high quality tours. However, it is impractical to find one or some $OHC$s using current methods, especially for large $TSP$ instances. 
It is interesting that the high quality tours contain many pseudo-backbone edges. Through contracting the pseudo-backbone edges, the size of $TSP$ is greatly reduced. Although the $\alpha$-measure and pseudo-backbone edges work well for most $TSP$ instances, they are lack of theoretical foundations for characterizing the $OHC$ edges. %In these studies, the properties of $OHC$ edges were not fully illustrated.
%Their method is exact rather than probabilistic.

In recent years, the frequencies of edges were computed with the optimal 4-vertex paths with given endpoints (optimal 4-vertex paths), or frequency $K_4$s for identifying $OHC$ edges \cite{DBLP:journals/Wang151, DBLP:journals/Wang16}. Given one $K_4$ in $K_n$, there are twelve paths visiting the four vertices contained in the $K_4$. They are the six pairs of paths each of which has the same two endpoints. The authors compared each pair of the paths with the same endpoints, and the shorter one was taken as the optimal 4-vertex path for the pair of specific endpoints, respectively. Since six pairs of vertices in one $K_4$, there are six optimal 4-vertex paths in each $K_4$. After the six optimal 4-vertex paths were obtained, they enumerated the number of the optimal 4-vertex paths containing each edge. Each number denotes the frequency of one corresponding edge. After the frequencies of all edges are figured out, the frequency $K_4$ is constructed by the edges and their frequencies, see Appendix $A$. Given an edge in $K_n$, it is contained in ${{n-2}\choose{2}}$ $K_4$s, and it is contained in the same number of frequency $K_4$s. In each frequency $K_4$, the edge has certain frequency of 1, 3, or 5. For each edge in $K_n$, the authors added the related frequencies together for computing the total frequency in $[{{n-2}\choose{2}}, 5{{n-2}\choose{i-2}}]$ and the average frequency in $[1,5]$, respectively.
%and there are total $6{{n-2}\choose{2}}$ optimal 4-vertex paths. The authors enumerated the number of optimal 4-vertex paths containing an edge according to the total optimal 4-vertex paths. This number represents the frequency of each edge in $K_n$. 
After that, it found that the average frequency of an $OHC$ edge is much higher than the average frequency of all edges and those of most ordinary edges. In theory, the average frequency of an $OHC$ edge is bigger than $3$ \cite{DBLP:journals/Wang16}. Moreover, an $OHC$ edge will be contained in more percentage of the optimal 4-vertex paths according to rising $n$ \cite{DBLP:journals/Wang19}. It implies that the average frequency of an $OHC$ edge will increase according to $n$ based on the frequency $K_4$s. For big and large $TSP$, most $OHC$ edges have the very big frequencies computed with the frequency $K_4$s.

%It found that the probability that an $OHC$ edge is contained in the optimal $i$-vertex paths keeps stable or increases according to $i$, whereas that for an ordinary edge decreases on average. Thus, the $OHC$ edges will have a higher frequency than the ordinary edges at a certain number $i=P_0$ so that they can be separated at $P_0$. The theorems illustrated that $P_0$ is smaller than $\frac{n}{2}+2$ for even $n$, or $\frac{n+1}{2} + 1$ for odd $n$. To let the theorems work, we assume each $K_i$ contains one $OHC$.
The frequency $K_4$s have been investigated, and each $OHC$ edge in $K_n$ has certain (average) frequency much higher than those of most ordinary edges based on the frequency $K_4$s. The next work is to study the frequency $K_i$s ($4<i\leq n$) for $TSP$. In one frequency $K_i$, the frequency of each edge is computed with the ${{i}\choose{2}}$ optimal $i$-vertex paths contained in one  corresponding $K_i$, see Section 2. In each  frequency $K_i$, every edge has a frequency in $[0,{{i}\choose{2}} - 1]$. An edge in $K_n$ is contained in ${{n-2}\choose{i-2}}$ frequency $K_i$s. The frequency or average frequency of each edge can be computed based on these frequency $K_i$s. Similar to the frequencies of edges computed based on the frequency $K_4$s, $OHC$ edges and ordinary edges will have much difference with respect to their frequencies or average frequencies based on frequency $K_i$s. According to the distances on edges in $K_n$, it is difficult to separate $OHC$ edges from ordinary edges. Since $OHC$ edges and ordinary edges have much difference with respect to the frequencies computed with the frequency $K_i$s, they will be separated from each other according to their frequencies.% can be separated from each other with respect to their frequencies based on frequency $K_i$s.

In this paper, we shall study the different structural properties for $OHC$ edges and ordinary edges based on the frequency $K_i$s. In the first step, the frequency $K_i$s are computed with the optimal $i$-vertex paths in the corresponding $K_i$s in $K_n$. Then, the frequency of each edge is computed with the frequency $K_i$s containing them, respectively. The $OHC$ edges and ordinary edges have different frequencies which are used to separate them from each other. As the frequency $K_i$ is computed for certain $K_i$, one will ask the difference between the frequency of an $OHC$ edge and that of an ordinary edge related to $K_i$. The second question is how big the frequency of an $OHC$ edge related to $K_n$ is when it is computed with the frequency $K_i$s containing it? Moreover, what are the frequency changes for an $OHC$ edge and ordinary edge according to $i$ as the frequency of each edge is computed with the frequency $K_i$s? %Lastly, as the $OHP$ is one optimal $n$-vertex path, what are the relationships between the $OHP$ and $OHC$ in $K_n$? 
The questions will be analyzed on the basis of the frequency $K_i$s. 
%A frequency $K_i$ is computed with the 
%${{i}\choose{2}}$ optimal $i$-vertex paths
%with given endpoints in a corresponding
%$K_i$. 
%Given all optimal $i$-vertex
%paths in $K_i$, the frequency for an edge is the 
%number of the optimal $i$-vertex paths containing
%this edge. 
%especially for large scale of $TSP$.

If each $K_i$ contains ${{i}\choose{2}}$ optimal $i$-vertex paths and one $OHC$, this paper presents the following results.  
%computed with the optimal $i$-vertex
%paths in $K_i$s.
%The difference between the $OHC$ edges
%and the other edges are better disclosed.
An $OHC$ edge in a given $K_i$ has the frequency bigger than $\frac{1}{2}{{i}\choose{2}}$ in the  frequency $K_i$, i.e., an $OHC$ edge is contained in bigger than  $\frac{1}{2}{{i}\choose{2}}$ optimal $i$-vertex paths in every $K_i$. On the other hand, an ordinary edge has the frequency smaller than $2(n-3)$ in the frequency $K_i$, and the average frequency of the ordinary edges is smaller than $\frac{i+2}{2}$ in the best average case. On average, the expected frequency of an $OHC$ edge related to $K_i$ is bigger than $\frac{i^2-4i+7}{2}$, and an ordinary edge will have the expected frequency smaller than 2. Moreover, an $OHC$ edge related to $K_n$ has the frequency bigger than $\frac{1}{2}{{i}\choose{2}}$ in a frequency $K_i$ containing it in the worst average case. As we choose $N$ frequency $K_i$s containing an edge to compute the total and average frequency, the lower frequency bound for $OHC$ edges is $\frac{N}{2}{{i}\choose{2}}$, and the average frequency is bigger than $ \frac{1}{2}{{i}\choose{2}}$. In addition, the probability that an $OHC$ edge is contained in the optimal $i$-vertex paths increases from $i$ to $i+1$ where $i\in[4,n-1]$ or keeps the nearly equal value if it decreases. As the frequency of an edge is computed with the frequency $K_i$s, each $OHC$ edge reaches its own peak frequency at $P_0 = \frac{n}{2}+2$ for even $n$ or $P_0=\frac{n+1}{2}+1$ for odd $n$, respectively. On the other hand, the probability that an ordinary edge is contained in the optimal $i$-vertex paths decreases according to $i$ in the average case. %Moreover, the average frequency of an ordinary edge will be smaller than $\frac{1}{2}{{i}\choose{2}}$ if $i\geq  2i_d$ where $i_d $ is the smallest number meeting the condition $\frac{(n-2)(n-3) - (i_d-2)(i_d-3)}{(n-2)(n-3) - (i_d-1)(i_d-2)} \geq \sqrt{1+\frac{2}{i_d(i_d+1)}}$ and $i_d = O(n^{\frac{4}{7}})$. 
Moreover, the probability of an ordinary edge definitely decreases if $i\geq  i_d$ where $i_d = O(n^{\frac{4}{7}})$ is the smallest number meeting the condition $\frac{(n-2)(n-3) - (i_d-2)(i_d-3)}{(n-2)(n-3) - (i_d-1)(i_d-2)} \geq \sqrt{1+\frac{2}{i_d(i_d+1)}}$. Based on these findings, an algorithm based on dynamic programming is presented to find all $OHC$ edges in $O(n^2i_d^42^{i_d})$ time. 
%In addition, as a graph only contains the edges in all optimal $n$-vertex paths, the number of edges will be bigger than $1.5n$ but smaller than $\frac{n(n+2)}{4}$ in the average case. 

The remainder of this paper is organized as follows. In Section 2, the optimal $i$-vertex paths in one given $K_i$ are introduced, and the frequency $K_i$ is computed with the optimal $i$-vertex paths. In Section 3, the lower frequency bound for $OHC$ edges related to $K_i$ is proven, and the upper frequency bound for ordinary edges is also derived. In Section 4, the lower frequency bound for $OHC$ edges related to $K_n$ is studied based on the  frequency $K_i$s. 
%The number of edges is also estimated for the graph only containing the edges in the optimal $n$-vertex paths. 
In section 5, the frequency and probability changes for $OHC$ edges and ordinary edges are predicted according to $i$ as they are computed based on the frequency $K_i$s. Meanwhile, the algorithm to identify $OHC$ edges is also presented. In Section 6, the experiments are executed to verify these findings with the benchmark $TSP$ instances. Finally, conclusions are drawn in the last section, and the future research is also presented. 

%Section text. See Subsection \ref{subsec1}.

%% Use \subsection commands to start a subsection.
\section{The optimal $i$-vertex paths and frequency $K_i$}
\label{sec2}

  %Given an $n$-vertex $TSP$ instance, it is generally represented as weighted complete graph $K_n$. 
For symmetric $TSP$, there are $n$ vertices and ${{n}\choose{2}}$ edges in $K_n$ $(n\geq 4)$. %Given an edge, it is contained in $OHC$ or not. 
\textit{If an edge is contained in $OHC$, it is called an $OHC$ edge. Otherwise, it is an ordinary edge.} The $TSP$ containing one $OHC$ is considered here. If some $TSP$ includes several or many $OHC$s, one can add small random distances to the distances of edges, and convert it into one modified $TSP$ having only one $OHC$. As the random distances are small enough, the original $TSP$ and modified $TSP$ will contain the same $OHC$. The $OHC$ in $K_n$ is illustrated in Figure \ref{OHC} where we assume that the subscripts of vertices are ordered according to the natural numbers $1, 2, ..., n$. The dashed lines are $OHC$ edges whereas the solid line represents an ordinary edge $(v_1,v_j)$ and $1<i<j<k<n$.

\begin{figure}
	\centering
	\includegraphics[width=2.5in,bb=0 0 400 200]{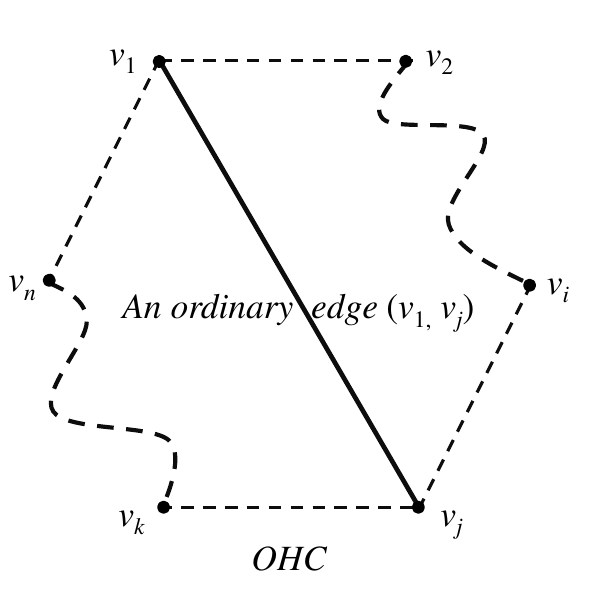}
	\caption{The illustration of $OHC$  in $K_n$.}
	\label{OHC}
\end{figure}

In $K_n$, each vertex is contained in $n-1$ edges. For example, $v_1$ is contained in the $n-1$ edges $(v_1, v_s)$ where $s\in [2,n]$ in Figure \ref{OHC}. Since there is only one $OHC$ in $K_n$, each vertex is contained in two $OHC$ edges and $n-3$ ordinary edges, respectively. For example in Figure \ref{OHC}, $v_1$ is contained in the two $OHC$ edges $(v_1,v_2)$ and $(v_1, v_n)$, and the $n-3$ ordinary edges $(v_1,v_k)$ $(k\in [3,n-1])$. Moreover, each ordinary edge is adjacent to two pairs of $OHC$ edges on both endpoints since each vertex is contained in two $OHC$ edges. For example in Figure \ref{OHC}, the ordinary edge $(v_1,v_j)$ is adjacent to the two pairs of $OHC$ edges $(v_1,v_2)$ $\&$ $(v_1,v_n)$ and $(v_j,v_i)$ $\&$ $(v_j, v_k)$. If $v_2 = v_i$ or $v_k = v_n$ for $n >4$, there are at most $n$ such ordinary edges in $K_n$. For each of the other $\frac{n(n-5)}{2}$ ordinary edges $(v_1, v_j)$, $v_2\neq v_i$ and $v_k\neq v_n$ exist. In the following, for an ordinary edge  $(v_1, v_j)$, the four vertices $v_2$, $v_i$, $v_k$ and $v_n$ contained in the two pairs of adjacent $OHC$ edges are are pairwise distinct if there are no special declarations. 

Given a set of $i$ vertices in $K_n$ where $i\in[4,n]$, the $i$ vertices are contained in one  corresponding $K_i$. As there is only one $OHC$ in the $K_i$, every vertex is contained in two $OHC$ edges and $i-3$ ordinary edges, and an ordinary edge is adjacent to two pairs of $OHC$ edges on both endpoints, respectively. Moreover, there are ${{i}\choose{2}}$ optimal $i$-vertex paths in the $K_i$. Given one $K_i$ in $K_n$, the corresponding frequency $K_i$ is computed with the optimal $i$-vertex paths ($OP^i$) in the $K_i$ where the superscript $i$ in $OP^i$ denotes the number of vertices. The optimal 4-vertex paths and frequency $K_4$s were computed in paper \cite{DBLP:journals/Wang16}, see Appendix $A$. The optimal $i$-vertex paths in the $K_i$ and the frequency $K_i$ are computed as follows. 
%It is also called frequency graph. The $OP^i$s and frequency $K_i$ are computed as follows.

\begin{description}
	\item[Optimal $i$-vertex path ($OP^i$)]: Given one $K_i$ on a set of $i$ vertices $\{v_1,v_2, \ldots, v_{i-1},v_i \}$ in $K_n$, the paths $P^i = (v_{\sigma_1}, \ldots, v_{\sigma_i})$ visit each of the $i$ vertices exactly once. Fix the two endpoints of $P^i$, such as $v_{\sigma_1}=v_1$
	and $v_{\sigma_i} =v_2$, there are $(i-2)!$ $P^i$s where $v_1$ and $v_2$ are the endpoints. The shortest one is taken as the optimal $i$-vertex path denoted as $OP^i$ for $v_1$ and $v_2$. 
\end{description}
%A $K_i$ contains $i$ vertices in $K_n$. and ${{i}\choose{2}}$ edges.  %where the superscript $i$ in $P^i$ has the same meaning as that in $OP^i$.   
In the $K_i$, there are ${{i}\choose{2}}$ pairs of vertices which are taken as the endpoints for computing the $OP^i$s. Thus, each $K_i$ contains ${{i}\choose{2}}$ $OP^i$s. As one special $K_i$ contains the equal-weight edges, it may contain more than ${{i}\choose{2}}$ $OP^i$s and one $OHC$. One can add the small random distances to the distances of edges. In this case, the modified $K_i$ will contain one $OHC$ and ${{i}\choose{2}}$ $OP^i$s. The ${{i}\choose{2}}$ $OP^i$s in the $K_i$ on vertex set $\{v_1,v_2, \ldots, v_{i-1},v_i \}$ are illustrated in Figure \ref{OPi}. Each $OP^i$ has two specified endpoints which are different from those of the other $OP^i$s. The endpoints of each $OP^i$ are actually contained in one corresponding edge in the $K_i$. For example, the first $OP^i$ in Figure \ref{OPi} has the endpoints $v_1$ and $v_2$. In the $K_i$, $v_1$ and $v_2$ are contained in the edge $(v_1,v_2)$. In this case, $(v_1,v_2)$ will not be contained in the first $OP^i$ in Figure \ref{OPi}. Thus, an edge in the $K_i$ is excluded from at least one $OP^i$. For each vertex in the $K_i$, such as $v_1$ in Figure \ref{OPi}, it is one endpoint of $i-1$ $OP^i$s. In each of these $OP^i$s, $v_1$ is contained in one edge. In each of the other ${{i-1}\choose{2}}$ $OP^i$s, $v_1$ is one intermediate vertex, and it is contained in two edges. The edges containing $v_1$ in the $OP^i$s are one subset of the edges containing $v_1$ in the $K_i$. 
%It means that only part of edges will be contained in the $OP^i$s.
It mentions that the $OHC$ in the $K_i$ contains $i$ $OP^i$s which only contain the $OHC$ edges. On the other hand, each of the other $\frac{i(i-3)}{2}$ $OP^i$s contains at least one ordinary edge, respectively. Most importantly, each $OP^i$ contains the defined set of $i-1$ edges in the $K_i$, and each $OP^i$ does not become shorter whatever the intermediate vertices are exchanged. Thus, all paths contained in an $OP^i$ are the optimal paths having a relatively smaller number of vertices. For example, each $P^k$ ($k\in [4,i]$) contained in an $OP^i$ is one $OP^k$.  

\begin{figure}
	\centering
	\includegraphics[width=3in,bb=0 0 360 250]{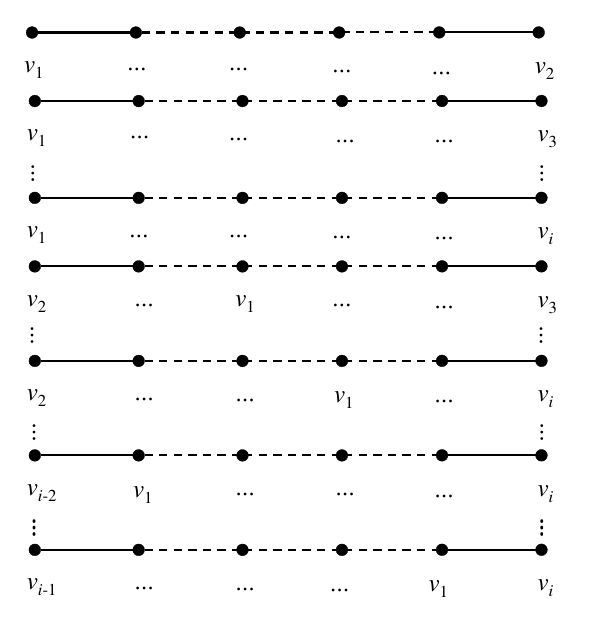}
	\caption{The ${{i}\choose{2}}$ $OP^i$s in one $K_i$.}
	\label{OPi}
\end{figure}

In $K_n$, there are ${{n}\choose{i}}$ $K_i$s, and there are total ${{i}\choose {2}}{{n}\choose {i}}$ $OP^i$s at given number $i$.  All the $P^i$s in $OHC$ and $OP^n$s of $K_n$ are $OP^i$s contained in some $K_i$s. Moreover, every path $P^k$ contained in an $OP^m$ ($k\in[4,m], m\in[k,n]$) is one $OP^k$. If one $P^k$ becomes shorter by exchanging the intermediate vertices, the $P^m$ containing it must not be optimal. It says that each path $P^k$ contained by two optimal paths $OP^i$ and $OP^j$ is also one $OP^k$ where $k\in [4, min\{i,j\}]$. Thus, the $OP^i$s and the intersecting operation defined for them form one semi-group \cite{DBLP:journals/Wang251}. The $OP^i$s containing more vertices only include the $OP^k$s containing less vertices. The $OP^k$s containing the smaller number of vertices keep the same structures (i.e., structure stability) when they are used to construct the $OP^i$s containing more vertices. 
\begin{description}
	\item[Frequency $K_i$]: After the ${{i}\choose{2}}$ $OP^i$s are computed from a given $K_i$, one can enumerate the number of $OP^i$s containing an edge in the $K_i$. This number is called the frequency of the edge. As the frequency of each edge is computed based on these $OP^i$s, the frequency $K_i$ is constructed with the edges and their frequencies.
\end{description}
In the frequency $K_i$, the vertices and edges are the same as those in the corresponding $K_i$. In contrast, the distances of edges in the $K_i$ are replaced by the corresponding frequencies of the edges, see the frequency $K_4$s in Appendix $A$. Thus, one frequency $K_i$ and the corresponding $K_i$ contain the same $OHC$ and $OP^i$s. Different from the various distances of edges in the $K_i$, each edge has a specific frequency in $\left[0,{{i}\choose{2}}-1\right]$. If an edge is not contained in any one $OP^i$, it has the frequency of zero. Since an edge in the $K_i$ is excluded from at least one $OP^i$, the maximum frequency of edges is ${{i}\choose{2}} - 1$. Moreover, an $OHC$ edge has certain  frequency much higher than that of an ordinary edge in the frequency $K_i$. There are ${{i}\choose{2}}$ $OP^i$s in the $K_i$ and each $OP^i$ contains $i-1$ edges. Since the frequency $K_i$ contains ${{i}\choose{2}}$ edges, the average frequency of all edges is $i-1$. Given a vertex in the $K_i$, there are $i-1$ $OP^i$s in which it is one endpoint. In each of these $OP^i$s, the vertex is contained in one edge. In each of the other ${{i-1}\choose{2}}$ $OP^i$s, the vertex is one intermediate vertex, and it is contained in two edges. Thus, the total frequency of the $i-1$ edges containing the vertex is $(i-1)^2$. This indicates that the average frequency of the edges containing a vertex is also $i-1$. The distances on the $i-1$ edges containing a vertex do not have such properties. 
%The frequency $K_4$s are just one special
%case of frequency $K_i$s as $i=4$.
%In real-life applications, it is hard to compute the $OP^i$s as $i$ is big. On the other hand, it is feasible to compute the frequency $K_i$s if $i$ is equal to one small number, such as $i=4$, 5 or 6, etc.  
An edge $e$ in $K_n$ is contained in ${{n-2}\choose{i-2}}$ $K_i$s, and it is contained in the same number of frequency $K_i$s. In each frequency $K_i$, $e$ has certain frequency in  $\left[0,{{i}\choose{2}}-1\right]$. The (total) frequency and average frequency of $e$ is computed with these frequency $K_i$s.
\begin{description}
	\item[The (total) frequency and average frequency of an edge in $K_n$]:  For an edge $e$ in $K_n$, it is contained in $N_i={{n-2}\choose{i-2}}$ frequency $K_i$s. Provided that the frequency of $e$ is $f_k(e) (k\in [1,N_i])$ in the $k^{th}$ frequency $K_i$, the (total) frequency of $e$ is computed as $F(e) = \sum_{k=1}^{N_i} f_k(e)$ and the expected or average frequency of $e$ is  $f(e)=\frac{F(e)}{N_i}\in \left[0,{{i}\choose{2}}-1\right]$. 
	%For an edge $e$ in $K_n$, draw $N\left(1\leq N \leq {{n-2}\choose{i-2}}\right)$ frequency $K_i$s containing $e$. Provided that the frequency of $e$ is $f_k(e) (1\leq k \leq N)$ in the $k^{th}$ frequency $K_i$, the (total) frequency of $e$ is computed as $F(e) = \sum_{k=1}^N f_k(e)$ and the averge frequency of $e$ is  $f(e)=\frac{F(e)}{N}\in [0,{{i}\choose{2}}-1]$. 
\end{description}

The frequencies or average frequencies of $OHC$ edges have much difference from those of ordinary edges as they are computed with the frequency $K_i$s where $i\in[4,n]$. The frequencies of $OHC$ edges are much bigger than those of most ordinary edges computed based on frequency $K_4$s \cite{DBLP:journals/Wang19}. The $OHC$ edges and most ordinary edges can be separated from each other according to their frequencies. Besides the frequency, the probability that an edge is contained in the $OP^i$s also illustrates the structure characteristics of the edge with respect to $OHC$. In other words, an $OHC$ edge will be contained in the $OP^i$s with the big probability, whereas most ordinary edges are contained in the $OP^i$s with a small probability. 

\begin{description}
	\item[The probability that an edge is contained in the $OP^i$s]:  For an edge $e$ in $K_n$,  the average frequency of $e$ is 
	$f(e)\in \left[0,{{i}\choose{2}}-1\right]$ based on the frequency $K_i$s containing it. The expected probability that $e$ is contained in the $OP^i$s (or an $OP^i$) is computed as $p(e\in OP^i) = \frac{f(e)}{{{i}\choose{2}}}\in [0,1)$.  
\end{description}

For convenience, the expected probability is simply called probability in this paper. In view of the definitions of average frequency and probability for an edge, one will see that they are depends on the weights of all edges in $K_n$. Once the weighted $K_n$ for a $TSP$ is given, every edge has the unique average frequency and probability. If one takes the weight of an edge as the real function defined on the edge, the average frequency or probability is a functional defined on the function space related to the edge weights. Therefore, compared with the weights of edges, the frequency or probability for an edge exhibits the even stronger capability in characterizing the graph structure in terms of $OHC$. In this case, we can build the sufficient and necessary conditions for $OHC$ edges and ordinary edges based on their frequencies and probabilities rather than the weights. Then, the rules or strategies based on these theories can be constructed to lead the algorithms to search for the $OHC$. %Since the $OHC$ edges and paths are hard to be characterized by the weights, it is lack of the rules to lead the algorithm to search the $OHC$ for $TSP$

In the following, as the $K_i$ ($i\in [4,n]$) or $K_i$s and frequency $K_i$ or frequency $K_i$s are discussed, they are totally contained in one $K_n$ for a given $TSP$ instance, and the distance of each edge is defined in the same space, such as one kind of Euclidean space, or metric space, etc. Once the $TSP$ instance is given, the $OP^i$s in each $K_i$ are determined according to the vertices and the distances on edges in the $K_i$, respectively. Meanwhile, the $OP^i$s containing each edge are also fixed according to the given $K_n$. It means that a given edge $e$ has the specific $f(e)$ and $p(e\in OP^i)$ for a given $TSP$.

\section{The lower frequency bound for $OHC$ edges and upper frequency bound for ordinary edges related to $K_i$}
\label{sec3}
We assume that each $K_i$ $(i\in [4,n])$ contains one $OHC$ and ${{i}\choose{2}}$ $OP^i$s. Given a $K_i$, each $OP^i$ in the $K_i$ is unique. It says that each $OP^i$ in the $K_i$ contains the defined set of $i-1$ edges, respectively. Some edges may be contained in many $OP^i$s, whereas some others are contained in a small number of $OP^i$s. In the worst case, the remainder edges exclude from any one $OP^i$. Thus, not all edges in the $K_i$ are contained in the $OP^i$s. The frequency of some edges will be zero in the corresponding frequency $K_i$. If the edges with the frequency of zero are neglected from the frequency $K_i$, the remained edges are contained in a graph $G_i$. The number of edges in the $G_i$ is smaller than ${{i}\choose{2}}$. %The edges in all the $OP^i$s are contained in one graph $G_i$ whose edges' number is equal to or smaller than ${{i}\choose{2}}$. %The edges in $K_i$ can be classified into $OHC$ edges and ordinary edges with respect to $OHC$. 

\begin{theorem}
	\label{th001}
	Given one $K_i$ containing ${{i}\choose{2}}$ $OP^i$s, all edges in the $OP^i$s are contained in one graph $G_i$. The minimum vertex degree of $G_i$ is 3, and the maximum vertex degree of $G_i$ is at least 4 where $i\geq 5$. 
\end{theorem}

\begin{proof}
	If $i=4$, the six frequency $K_4$s for a $K_4$ are illustrated in Figure \ref{quadrilaterals}. All edges in the $K_4$ are contained in the six $OP^4$s, and they are also contained in the $G_4$. The maximum and minimum vertex degree of $G_4$ is 3, respectively. 
	
	If $i \geq 5$, there are ${{i}\choose{2}}$ $OP^i$s in the $K_i$, and all the edges in the $OP^i$s are contained in the $G_i$. Because each $OP^i$ contains the defined set of  $i-1$ edges, no redundant ordinary edges are used to build the ${{i}\choose{2}}$ $OP^i$s. It means that the $OP^i$s  contain the least number of edges in the $K_i$. %Thus, the number of edges in the $G_i$ is as small as possible. 
	It is known that an $OP^i$ only contains the $OP^k$s for $k\in[4,i]$. As $i > 5$, part of $OP^5$s in some $K_5$s will be contained in the $OP^i$s. Moreover, the edges in these $OP^5$s are contained in the corresponding $G_5$s, respectively. %There are ${{i-2}\choose{3}}$ $K_5$s containing each edge in the $K_i$ as well as the same number of $G_5$s containing the edges in the $OP^i$s. 
	Among all the $G_5$s, there is one $G_5$ having the maximum number of edges. The $G_5$ is one sub-graph of the corresponding $K_5$. All the edges in the $G_5$ are contained in the $OP^i$s.  %Since the $OP^i$s only contain the $OP^5$s in the $G_5$, these edges are contained in some $OP^5$s in the $G_5$. 
	Due to the variation of $K_i$s, each of the ten $OP^5$s in the $K_5$ may be contained in the $OP^i$s. %It requires that the $G_5$ contains ten $OP^5$s in the $K_5$. 
	As the $OP^i$s visit the five vertices in the $G_5$ or $K_5$, they will include the proper $OP^5$s, respectively. It requires that the $G_5$ contains ten $OP^5$s.
	%Since the edges in the $G_5$s are contained in the $OP^i$s, each $G_5$ may include the $OP^5$s which are contained in the $OP^i$s. Moreover, each $OP^5$ contained in every $G_5$ may be contained in one or some $OP^i$s. 
	
	Since there is no redundant edges for building the ten $OP^5$s, the $G_5$ contains the least number of edges in the $K_5$. We assume that the maximum vertex degree of the $G_5$ is 3 while the minimum vertex degree is 2. For example in Figure \ref{G5} (a), the $G_5$ on the vertex set $\{v_1, v_2, v_3, v_4, v_5\}$ is contained in the $G_i$. In the $G_5$, $v_3$ has the minimum  degree 2, and the degree of the other vertices is 3. Provided that one $OP^i=(v_x,...,v_y)$ in the $K_i$ includes one $OP^5$ in the $G_5$ in Figure \ref{G5} (a), the $OP^i$ is illustrated on the right of the $G_5$. Moreover, this $OP^i$ just contains the $OP^5$ where $v_2$ and $v_4$ are the endpoints. %Because every $OP^5$ may be contained in the $OP^i$s,  the $OP^5$ where $v_2$ and $v_4$ are the endpoints. the $G_5$ contains ten $OP^5$s in the $K_5$.  
	As the two vertices $v_2$ and $v_4$  are taken as the endpoints of the $OP^5$, there are only two paths $(v_2,v_1,v_5,v_4)$ and $(v_2, v_3, v_4)$ composed of smaller than five vertices, respectively. Thus, the $OP^i$ does not exist due to the lack of the $OP^5$ where $v_2$ and $v_4$ are the endpoints. If the $G_5$ contains the $OP^5$ where $v_2$ and $v_4$ are the endpoints, the degree of $v_3$ must be equal to or bigger than 3. Thus, the minimum vertex degree of $G_i$ is at least three if it includes the  ${{i}\choose{2}}$ $OP^i$. The $G_4$ also conforms to this rule. 
	
	Add another edge for $v_3$ to the $G_5$ in Figure \ref{G5}, such as $(v_1,v_3)$ (or $(v_3,v_5)$), the new $G_5$ in Figure \ref{G5} (b) is built. The degree of $v_3$ becomes 3 and the maximum vertex degree of the new $G_5$ becomes 4. In the new $G_5$, there are two $P_5$s where $v_2$ and $v_4$ are the endpoints. The two $P^5$s are $(v_2,v_3,v_1,v_5,v_4)$ and $(v_2,v_5,v_1,v_3,v_4)$, and one of them will be the $OP^5$. If more edges are added to the $G_5$, the new $G_5$ will contain more $P_5$s where $v_2$ and $v_4$ are the endpoints. Among these $P^5$s, there must be one $OP^5$. %In this case, the $OP^i$ containing the $OP^5$ where $v_2$ and $v_4$ are the endpoints can be constructed. 
	Moreover, the $G_5$ contains the other $OP^5$s which can be used to construct the other $OP^i$s in the $G_i$. 
	
	The maximum vertex degree of the $G_5$ is 4, and the minimum vertex degree is 3. If $i > 5$, each of the vertices in the $G_5$ may connect to the other vertices in the $G_i$ yet not in the $G_5$. If one vertex of degree 3, such as $v_2$, $v_3$, $v_4$ or $v_5$, does not connected to the other vertices not in the $G_5$ for building the $OP^i$s, the minimum vertex degree of the $G_i$ is 3. Moreover, if $v_1$ is not connected to the other vertices not in the $G_5$ for building the $OP^i$s, $v_1$ maintains the degree 4 in the $G_i$. If $v_1$ is connected to the other vertices not in the $G_5$ for building the $OP^i$s, the maximum vertex degree of $G_i$ becomes bigger than 4.  %minimum value of the maximum vertex degree of $G_i$ will be four. %Since the maximum vertex degree of the $G_5$ is four, the maximum vertex degree of the $G_i$ will be at least four.  
	%Thus, each $G_5$ contains ten $OP^i$s in the average case.
\end{proof}

\begin{figure}
	\centering
	\includegraphics[width=3in,bb=0 0 300 260]{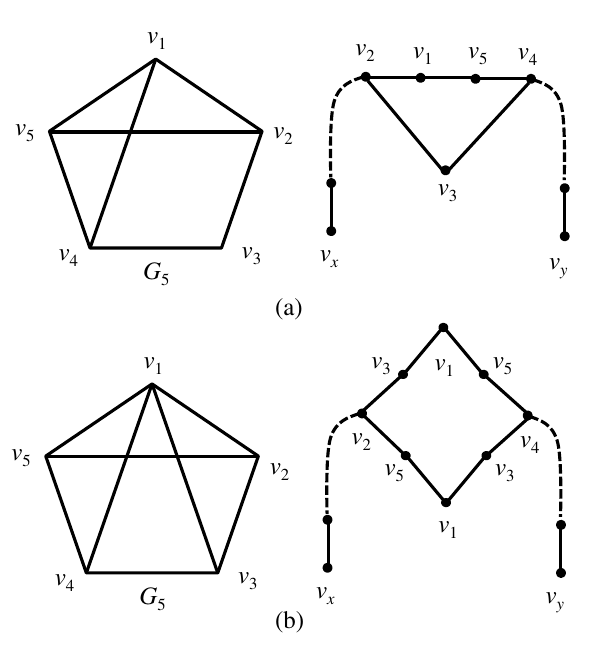}
	\caption{The two $G_5$s contained in the $G_i$ containing all edges in the ${{i}\choose{2}}$ $OP^i$s in $K_i$.}
	\label{G5}
\end{figure}

Given the ${i}\choose{2}$ $OP^i$s in certain $K_i$, the ${{i}\choose{2}}$ edges in the $K_i$ are not uniformly distributed to these $OP^i$s. Thus, the frequencies of edges are generally different in the frequency $K_i$. In general, an $OHC$ edge has the frequency much bigger than that of an ordinary edge in the average case.

\begin{theorem}
	\label{th01}
	Given a $K_i$ containing one $OHC$ and ${{i}\choose{2}}$ $OP^i$s, an $OHC$ edge is contained in more than $\frac{2(i-1)^2}{5}$ $OP^i$s, whereas an ordinary edge is contained in at most $\frac{(i-1)^2}{5(i-3)}$ $OP^i$s in the worst average case. 
\end{theorem}
%\begin{theorem}
%\end{theorem}

\begin{proof}
	Given a $K_i$, it contains ${i}\choose{2}$ $OP^i$s and one $OHC$. The $i-1$ edges containing a vertex $v$ are considered. The $i-1$ edges include two $OHC$ edges  and $i-3$ ordinary edges. Firstly, we assume that the $i-1$ edges are uniformly contained in the $OP^i$s. In this case, each edge is contained in $i-1$ $OP^i$s. For the two $OHC$ edges containing $v$, each of them is contained in $i-1$ $OP^i$s in $OHC$ because there are $i$ $OP^i$s in $OHC$. Thus, the two $OHC$ edges will never be contained in any one of the other $\frac{i(i-3)}{2}$ $OP^i$s. 
	
	Since the $i-3$ ordinary edges containing $v$ are uniformly distributed in the  $\frac{i(i-3)}{2}$ $OP^i$s not in $OHC$, each pair of two ordinary edges is contained in the same number of $OP^i$s in which $v$ is one intermediate vertex and endpoint, respectively. There are ${{i-3}\choose{2}}$ pairs of the ordinary edges containing $v$. If each pair of the ordinary edges is contained in more than one $OP^i$ where $v$ is one intermediate vertex, the total number of $OP^i$s will be bigger than ${{i}\choose{2}}$. Thus, each pair of the ordinary edges is contained in at most one such $OP^i$. As each pair of the ordinary edges is contained in one $OP^i$ in which $v$ is one intermediate vertex, there are total ${{i-3}\choose{2}}$ such $OP^i$s each of which includes two ordinary edges containing $v$.
	%  containing two edges visiting $v$ (an $OP^3$ where $v$ is the middle vertex). 
	Besides these $OP^i$s, there are $i-3$ more $OP^i$s which may contain the $i-3$ ordinary edges. In each of the $i-3$ $OP^i$s, $v$ is one endpoint, and each of the $OP^i$s includes one ordinary edge containing $v$. The $i-3$ ordinary edges are contained in at most ${{i-2}\choose{2}}$ $OP^i$s without considering the two $OHC$ edges. These $OP^i$s visit the $i-3$ ordinary edges $(i-3)^2$ times. Thus, each ordinary edge containing $v$ is contained in $i-3$ $OP^i$s rather than $i-1$ $OP^i$s on average. It indicates that each ordinary edge can not be contained in $i-1$ $OP^i$s due to the lack of the two $OHC$ edges containing $v$. 
	
	In addition, the ${{i}\choose{2}}$ $OP^i$s visiting $v$ can not be fully built if the two $OHC$ edges are not contained in some $OP^i$s not in $OHC$. Besides the $i$ $OP^i$s in $OHC$ and  ${{i-2}\choose{2}}$ $OP^i$s only including the ordinary edges, the remainder $i-3$ $OP^i$s must include the two $OHC$ edges. Moreover, $v$ is one intermediate vertex in each of the $i-3$ $OP^i$s. On average, if the two $OHC$ edges are uniformly contained in the $i-3$ $OP^i$s while the $i-3$ ordinary edges also do, each $OHC$ edge is included in at least $\frac{3i-5}{2}$ $OP^i$s in the worst average case, whereas each ordinary edge is contained in at most $i-2$ $OP^i$s. It implies that an ordinary edge will be contained in smaller than $i-1$ $OP^i$s in the best average case. 
	
	In the worst average case, an $OHC$ edge is contained in the $OP^i$s with the bigger probability than an ordinary edge.  Because each ordinary edge is not contained in $OHC$, we assume that the two $OHC$ edges and $i-3$ ordinary edges containing $v$ are uniformly contained in the $\frac{i(i-3)}{2}$ $OP^i$s excluding from $OHC$. In this case, an $OHC$ edge is contained in $2i - 4$ $OP^i$s in the worst average case, whereas an ordinary edge is contained in at most $i - 3$ $OP^i$s in the best average case. In fact, the $i-3$ ordinary edges are not contained in the $\frac{i(i-3)}{2}$ $OP^i$s with the equal probability. Otherwise, the number of $OP^i$s will be bigger than ${{i}\choose{2}}$ if $i > 4$. If $i > 4$, most ordinary edges are contained in smaller than $i-3$ $OP^i$s, and some ordinary edges are contained in none of the $OP^i$s. Since the ordinary edges with a frequency smaller than $i-3$ cannot be used to construct the $OP^i$s at the assumed positions, they must be replaced by the $OHC$ edges and some other ordinary edges. In this case, the total frequency of the $OHC$ edges will increase while that of the ordinary edges will decrease. Thus, an $OHC$ edge will be contained in greater than $2i-4$ $OP^i$s, whereas an ordinary edge will be contained in smaller than $i-3$ $OP^i$s on average. 
	
	There are $i$ $OHC$ edges and $\frac{i(i-3)}{2}$ ordinary edges in the $K_i$. Given an $OHC$ edge $e$ and ordinary edge $g$ containing $v$ in the  $K_i$, the $OP^i$s containing $e$ will be more than two times of those $OP^i$s containing $g$ because $\frac{2(i-2)}{i-3} > 2$ holds. It means that the frequency of $e$ is bigger than two times of that of $g$ in the frequency $K_i$, i.e., $f(e) > 2f(g)$. As the probability that $e$ or $g$ is contained in an $OP^i$ is considered, $p(e\in OP^i) > 2p(g\in OP^i)$ holds for $e$ and $g$ where $p(e\in OP^i)$ and $p(g\in OP^i)$ denote the probability that $e$ and $g$ are contained in an $OP^i$, respectively. %Since there are ${{i}\choose{2}}$ $OP^i$s in the $K_i$, the frequency of $e$ or $g$ can be used to compute the probability that $e$ or $g$ is contained in an $OP^i$ in the $K_i$. Thus, $p(e\in OP^i) > 2p(g\in OP^i)$ holds for $e$ and $g$. 
	This result holds for an $OHC$ edge and an associated ordinary edge in the worst average case. %and vice versa. Since the structures of the $OHC$ edges and ordinary edges containing all vertices are the same, $p(e\in OP^i) > 2p(g\in OP^i)$ holds for any one $OHC$ edge and ordinary edge appearing in the $OP^i$s in the average case if all $OHC$ edges and ordinary edges are taken as the same edge $e$ and $g$, respectively, with respect to $OHC$. %In this case, the total frequency of the $OHC$ edges will be more than two times of that of the ordinary edges based on the $OP^i$s. 
	
	Each edge is contained in $i-2$ $K_{i-1}$s each of which contains ${{i-1}\choose{2}}$ $OP^{i-1}$s. Each edge in the $K_i$ is contained in the specific number of $OP^{i-1}$s, respectively. Moreover, the $OP^i$s only contain the $OP^{i-1}$s in the $K_i$. Due to the variations of $TSP$, each $OP^i$s may contain any one $OP^{i-1}$ containing $v$. In the average case, each $OP^{i-1}$ containing $v$ is contained in the $OP^i$s with the equal probability. Since $p(e\in OP^i) > 2p(g\in OP^i)$ exists based on the frequency $K_i$, $p(e\in OP^{i-1}) > 2p(g\in OP^{i-1})$ also holds, i.e., the number of $OP^{i-1}$s containing $e$ will be more than two times of those containing $g$. If $g$ appears once in one $OP^{i-1}$, $e$ will appear more than twice in the $OP^{i-1}$s. In the same manner, $p(e\in OP^{i-k}) > 2p(g\in OP^{i-k})$ $(k\in[1,i-4])$ exists, and the number of $OP^{i-k}$s including $e$ will be more than two times of those containing $g$ until $k = i-4$. As $e$ and $g$ are contained in the $K_4$s, the number of $OP^4$s containing $e$ will be more than two times of those containing $g$. Because $p(e\in OP^4) > 2p(g\in OP^4)$, $f(e) > 2f(g)$ holds in the frequency $K_4$s containing $e$ and $g$ on average. In a frequency $K_4$, the frequency of an edge is 1, 3 or 5, see the six frequency $K_4$s in Figure \ref{quadrilaterals}. Because $f(e) > 2f(g)$, $f(g)$ = 1 exists in a frequency $K_4$ containing $e$ and $g$. Thus, $e$ will have the frequency $f(e)=3$ or 5. In this case, $e$ is the $OHC$ edge of the $K_4$ whereas $g$ is an ordinary edge. Since $g$ is arbitrary, it indicates that $e$ is the $OHC$ edge of the $K_4$s containing it. %For example in each $K_4$ in Figure \ref{quadrilaterals}, there are two $g$s among the six edges. Each $OP^4$ including one $g$ really contains two $e$s on both ends. Moreover, $g$ is contained in one $OP^4$ while each of the adjacent $e$s is contained in at least three $OP^4$s, respectively. 
	
	Not all the $OP^{i-k}$s in each $K_{i-k}$ are contained in the $OP^i$s. In fact, only part of the $OP^{i-k}$s in a small number of $K_{i-k}$s are contained in the $OP^i$s. For example, the $OHC$ in the $K_i$ contains $i$ $OP^4$s, and each of the $OP^4$s is contained in one $K_4$ where there are six $OP^4$s, respectively. In the extreme case, all the $OP^4$s in the ${{i}\choose{2}}$ $OP^i$s are different. There are at most $\frac{i(i^2-6i+11)}{2}$ $OP^4$s. However, the $K_i$ includes $6{{i}\choose{4}}$ $OP^4$s. It indicates that most $OP^4$s are not contained in the $OP^i$s. Moreover, some $OP^4$s are contained in a number of $OP^i$s. Thus, the number of different $OP^4$s in the $OP^i$s will be much smaller than $\frac{i(i^2-6i+11)}{2}$. To approve $p(e\in OP^4) > 2p(g\in OP^4)$ and $p(e\in OP^i) > 2p(g\in OP^i)$, the $OP^4$s in the $K_4$s containing $e$ will be used to construct the $OP^i$s. In most $K_4$s not containing $e$, $f(e) = 0$ exists in the corresponding frequency $K_4$s. It does not meet the condition $p(e\in OP^4) > 2p(g\in OP^4)$. Thus, the $OP^4$s in these $K_4$s not including $e$ are seldom used to build the $OP^i$s. It means that most $OP^4$s in the $OP^i$s will contain 2, 3 or at least one $e$. Every $e$ and $g$ are contained in $i-3$ $K_4$s. Given such one $K_4$, if the $OP^4$ including $g$ in the $K_4$ is contained in one $OP^i$, the $OP^4$ and other $OP^4$s containing $e$ in the $K_4$ will be contained in more than two $OP^i$s in the average case because of $p(e\in OP^i) > 2p(g\in OP^i)$. In fact, as the $OP^4$s in the $K_4$ are contained in the $OP^i$s with the equal probability, there will be more than three $OP^i$s containing $e$ because there are at least three $OP^4$s containing $e$ in the $K_4$ where there is only one $OP^4$ contains $g$. If $e$ and $g$ are contained in the same $K_{i-k}$, the $OP^{i-k}$s containing $e$ will be more than two times of those containing $g$ on average. %Thus, there will be more than two $OP^i$s containing $e$ and $f(e) > 2f(g)$ exists. %$OP^4$s including $e$ in the $K_4$ will be contained in more than two $OP^i$s because these $OP^4$s are contained in the $OP^i$s with the equal probability. Every $OP^i$ contains $i-3$ $OP^4$s, where there is one $OP^4$ in the $OHC$ in the $K_i$. 
	
	%For most $K_4$s not containing $e$, they also contain the $OP^4$s. What is the probability that they are contained in the $OP^i$s? 
	
	It mentions that there are two $OHC$ edges $e$ and $i-3$ ordinary edges $g$ containing $v$. If one $g$ is contained in an $OP^i$, each of the two $e$s will be contained in more than two $OP^i$s. Moreover, it is the $OP^4$s containing $e$ and $g$ in the corresponding $K_4$s that are also contained in the $OP^i$s. %For each of the other $g$s containing $v$ in the $OP^i$s, the $OP^4$s (or $OP^4$) including each of the two $e$s will be contained in more than two times of the $OP^i$s than the $OP^4$s containing each $g$, respectively, based on the $K_4$s containing each $g$ and the two $e$s. Because the $OP^4$s in different $K_4$s are different, the $OP^4$s containing $e$ or $g$ in different $K_4$s are different. 
	Given two $g$s, they and $e$ are contained in two $K_4$s, respectively, because $f(g) = 1$ appears once in a frequency $K_4$. The $OP^4$s in the two $K_4$s are different from each other. As the $OP^4$s including the two $g$s are contained in the $OP^i$s, the $OP^4$s including $e$ in the two $K_4$s are contained in the $OP^i$s with respect to the two $g$s, respectively. Thus, $p(e\in OP^i) > 2p(g\in OP^i)$ occurs independently for the $e$ and each $g$. It indicates that the $f(e)$ will be bigger than two times of that of all adjacent ordinary edges $g$ contained in the $OP^i$s, i.e. $f(e) > 2\sum_{k=1}^{i-3}f(g_k)$ where $f(g_k)$ is the frequency of the ordinary edge $g_k$. 
	%It means that the $OP^4$s containing $e$ in each $K_4$ For each $g$ in one $OP^i$s, 
	%Thus, the $OP^4$s containing $e$ for two $g$s contained in two $K_4$s are also different.  will also be contained in more than two times of the $OP^i$s on average. Because each ordinary edge $g$ and the two $e$s are contained in different $K_4$s where the $OP^4$s contain different vertices, the $OP^4$s contained in the $OP^i$s are also different with respect to these $K_4$s. Thus, the $OP^4$s containing $e$
	Since there are two $e$s containing $v$, the total frequency of the two $e$s containing $v$ is more than four times of that of the $i-3$ $g$s. 
	
	Given another ordinary edge $g_1$ containing $v$, $g$ will replace $g_1$ for building the $OP^i$s if $g$ has the frequency $f(g) > i-3$ whereas $g_1$ has the smaller frequency $f(g_1) < i-3$ in the frequency $K_i$. As $g$ replaces $g_1$ for building the $OP^i$s, it is the $OP^{i-1}$s containing $g$ that replace the $OP^{i-1}$s containing $g_1$ in the $OP^i$s. %Since the number of $OP^{i-1}$s containing $e$ is more than two times of those containing $g_1$, $e$ will replace $g_1$ more than two times than $g_1$ in the $OP^i$s. 
	Since $p(e\in OP^{i-1}) > 2p(g\in OP^{i-1})$ exists, the $OP^{i-1}$s containing $e$ will also replace those containing $g_1$ for building the $OP^i$s. Moreover, more than two times of the $OP^{i-1}$s containing $e$ will replace the $OP^{i-1}$s containing $g_1$ to build the $OP^i$s. Thus, $f(e) > 2f(g)$ still holds in the frequency $K_i$, and $p(e\in OP^i) > 2p(g\in OP^i)$ always holds no matter how $f(g)$ grows. As most ordinary edges are largely replaced by the $OHC$ edges and some other ordinary edges in the $OP^i$s, the frequencies of most ordinary edges will be very small, and the $OHC$ edges and some ordinary edges will have the big frequencies. Whatever the frequencies of the ordinary edges change, the total frequency is smaller than half frequency of each of the adjacent $OHC$ edges on average. 
	
	Theorem \ref{th001} says that the minimum vertex degree of $G_i$ is three where $G_i$ contains all edges in the $OP^i$s. In this case, the maximum frequency of $g$ containing $v$ will be near $\frac{(i-1)^2}{5}$, and the total frequency of the two adjacent $OHC$ edges will be bigger than $\frac{4(i-1)^2}{5}$. $\frac{(i-1)^2}{5}$ can be taken as one upper frequency bound for the ordinary edges, and $\frac{4(i-1)^2}{5}$ is one lower frequency bound for two adjacent $OHC$ edges in the frequency $K_i$. The total probability $\frac{8(i-1)}{5i}$ of two adjacent $OHC$ edges still increases according to $i$. It indicates that the total probability of two $OHC$ edges will increase according to $i$. If the $f(e) = \frac{3i-5}{2}$ and $f(g) = i-2$ is considered, $f(e) > 1.5f(g)$ can be derived. In this case, the total frequency of the two $OHC$ edges will be bigger than $\frac{3(i-1)^2}{4}$, and that of the $i-3$ ordinary edges is smaller than $\frac{(i-1)^2}{4}$.
	
	%In $K_i$, each ordinary edge is adjacent to two pairs of $OHC$ edges on both  endpoints, see Figure \ref{OHC}. If an ordinary edge is contained in one $OP^i$, each of the four adjacent $OHC$ edges will be contained in bigger than two $OP^i$s. 
	
	The ordinary edges and associated $OHC$ edges are used to construct the $OP^i$s under the frequency constraints between them. Since the frequency of each edge is enumerated from the ${{i}\choose{2}}$ $OP^i$s, the $OHC$ edges and ordinary edges also conform to the structure constraints of the $OP^i$s. Each vertex is visited once by every $OP^i$. Moreover, each vertex is contained in at most two different edges in every $OP^i$, and the $i-1$ edges in each $OP^i$ are  different from each other, respectively. %Because $p(e\in OP^i) > 2p(g\in OP^i)$ holds for every appearing $OHC$ edge and each adjacent ordinary edge in the $OP^i$s, 
	If an $OP^i$ contains one ordinary edge $g$, it will contain four $OHC$ edges on average because each vertex is contained in two $e$s and $p(e\in OP^i) > 2p(g\in OP^i)$ holds independently for each $e$ and every adjacent $g$. In this case, the total frequency of the $OHC$ edges will be more than four times of that of all the ordinary edges based on the $OP^i$s.  %Moreover, each of the two $OHC$ edges will be contained in another $OP^i$ or the other $OP^i$s at least once, respectively. 
	In the average case, an $OP^i$ contains at most $\frac{i-1}{5}$ ordinary edges and at least $\frac{4(i-1)}{5}$ $OHC$ edges. Considering the ${{i}\choose{2}}$ $OP^i$s and $i$ $OHC$ edges in the $K_i$, an $OHC$ edge will be contained in more than $\frac{2(i-1)^2}{5}$ $OP^i$s on average. On the other hand, an ordinary edge will be contained in less than $\frac{(i-1)^2}{5(i-3)}$ $OP^i$s. %Provided that the two $OHC$ edges are visted $t_1$ times and the other $i-3$ edges are visited $t_2$ times by all the $OP^i$s, respectively, $t_1 + t_2 = (i-1)^2$ holds. Once the $OHC$ edges are contained in some other $OP^i$s, the $i-3$ edges not in $OHC$ will be visited less times by the $OP^i$s. In  the worst average case, each $OHC$ edge will be contained in more number of the $OP^i$s than an edge not in $OHC$. In general, the number of $OP^i$s containing an $OHC$ edge will be much bigger than that of the $OP^i$s containing an edge not in $OHC$. For example, the six edges in a $K_4$ are not uniformly distributed in the six $OP^4$s, see the Appendix. Thus, the edges in a $K_i$ are not uniformly distributed to the $OP^i$s. 
	%for the other edges will be reduced. It says and the other edges are used to compose the $OP^i$s, 
\end{proof}	

Theorem \ref{th01} implies that an $OHC$ edge in a $K_i$ is contained in more  $OP^i$s than an ordinary edge. In a $K_i$, the $OP^i$s contain the $OP^k$s ($k\in[4,i]$) having a smaller number of vertices. Since the edges are not uniformly distributed in the $OP^k$s, the $OP^k$s  contain different number of $OHC$ edges and ordinary edges. On average, an $OHC$ edge $e$ is contained in the $OP^k$s with a bigger probability than an ordinary edge $g$ for $p(e\in OP^k) > 2p(g\in OP^k)$. Based on Theorem \ref{th01}, the $OP^k$s including more $OHC$ edges will be contained in the $OP^i$s with the bigger probability. Moreover, most $OP^k$s in the $OP^i$s are contained in the $K_k$s including the $OHC$ edges. On the other hand, the $OP^k$s including less $OHC$ edges or more ordinary edges will be seldom contained in the $OP^i$s. %

If the ordinary edges and $OHC$ edges containing each vertex are uniformly distributed in many $OP^i$s, the upper frequency bound of the ordinary edges is $i-3$, and the lower frequency bound of the $OHC$ edges is $2(i-2)$. In this case, the number of $OP^i$s will be bigger than ${{i}\choose{2}}$. In the extreme case, all $HC$s have the equal distance, and there will be $\frac{(i-1)!}{2}$ $OHC$s in the $K_i$. %Moreover, the minimum frequency for the $OHC$ edges is much bigger than the average frequency $i-1$ for all edges as $i$ is big. It says the $OHC$ edges are generally contained in a big percentage of the $OP^i$s. If the edges not in $OHC$ are averagely included in the $OP^i$s as the $OHC$ edges, the edges containing a vertex will have no difference with respect to the $OP^i$s (according to the distances of edges) in a given $K_i$. In the extreme case, there will be  $\frac{i!}{2}$ rather than ${{i}\choose{2}}$ $OP^i$s in the $K_i$. 
%Since each $K_i$ contains one $OHC$ and ${{i}\choose{2}}$ $OP^i$s, the edges must have much difference for constructing each $OP^i$. Otherwise, the number of $OP^i$s will be bigger than ${{i}\choose{2}}$ in the $K_i$. As some edges containing a vertex are uniformly distributed to two or  more $OP^i$s, there will be  contradictions. For example, the four vertices $v_1$, $v_2$, $v_3$ and $v_4$ are immediate and intermedidate vertices in two $OP^i$s. $v_1$ and $v_4$ are the endpoints of two $OP^4$s contained in the two $OP^i$s. Moreover, the edges $(v_1,v_2)$, $(v_1,v_3)$ are unifromly distributed to the two $OP^i$s as well as $(v_2,v_4)$ and $(v_3,v_4)$. In this case, the two $OP^i$s will contain the two optimal paths $OP^4_1=(v_1,v_2,v_3,v_4)$ and $OP^4_2=(v_1,v_3,v_2,v_4)$, respectively. Since $OP^4_1$ and $OP^4_2$ contain the same four vertices and have the same endpoints, there is only one $OP^4$. It is contradict that they are the $OP^4$s at the same time. 
%Since $OHC$ edges and ordinary edges are not averagely containded in the $OP^i$s, 
If there is one $OHC$ and ${{i}\choose{2}}$ $OP^i$s in the $K_i$, the $OHC$ edges and ordinary edges will not be uniformly distributed in the $OP^i$s. Moreover, the frequencies of $OHC$ edges and those of ordinary edges will have significant difference in the frequency $K_i$. For most ordinary edges, each will have certain frequency smaller than $i-3$. In the average case, as the frequency is discussed for an ordinary edge  according to the frequency $K_i$s, the frequency of smaller than $i-3$ will be considered.

%\begin{eqnarray}\label{F1}
%2f_o + (i-3)f_r = (i-1)^2
%\nonumber \\
%f_o = \frac{(i-1)^2}{2} - \frac{(i-3)f_r}{2}
%\end{eqnarray}

After the frequency $K_i$ is computed with the ${{i}\choose{2}}$ $OP^i$s, the average frequency of all edges is $i-1$. The lower frequency bound for $OHC$ edges is concerned to separate $OHC$ edges from ordinary edges. Theorem \ref{th01} has given one lower frequency bound $2(i-2)$  for $OHC$ edges. However, this frequency bound is too small comparing to the upper frequency bound $\frac{(i-1)^2}{5}$ for ordinary edges. %cannot be iteratively used to compute an even smaller frequeny for ordinary edges. 
In the next, the bigger lower frequency bound for $OHC$ edges related to $K_i$s will be proven for the worst average case. %In frequency $K_i$, the average frequency of all edges is $i-1$. The edges can be classified into two sets according to the average frequency $i-1$. 
%For an $OHC$ edge related to $K_i$, the frequency is bigger than $\frac{1}{2}{{i}\choose{2}}$ in the worst average case.

\begin{theorem}
	\label{th1}
	Given a $K_i$ $($$i\geq 4$$)$ containing ${{i}\choose{2}}$ $OP^i$s and one $OHC$, the frequency of an $OHC$ edge is bigger than $\frac{1}{2}{{i}\choose{2}}$ in the corresponding frequency $K_i$ in the worst average case. 
\end{theorem}

\begin{proof}
	%Given a $K_i$, every edge has the equal probability to be contained in an $OP^i$. 
	Let $e=(v_1,v_2)$ denote an $OHC$ edge in certain $K_i$ ($i\geq 4$), and $p_i(e\in OP^i)$ denotes the probability that $e$ is contained in an $OP^i$ in the $K_i$. 
	
	As $i = 4$, there are six frequency $K_4$s derived for the $K_4$s, see Appendix $A$. In each frequency $K_4$, the frequency of $e\in OHC$ is either 3 or 5. This indicates that $e$ is contained in at least three $OP^4$s in each $K_4$. Since there are six $OP^4$s in every $K_4$, the probability that $e$ is contained in an $OP^4$ is equal to or bigger than $\frac{1}{2}$, i.e., $p_4(e\in OP^4) \geq \frac{1}{2}$. The $K_i$ contains ${{i}\choose{2}}$ $OP^i$s. Based on the $p_i(e\in OP^i)$ for $e\in OHC$ in the $K_i$, the number of $OP^i$s containing $e$ is 
	$M_i={{i}\choose{2}}p_i(e\in OP^i)\geq M_i-\delta p_i(e\in OP^i)$ where $\delta\in [0,1]$. Given a number $i > 4$, we assume that $p_i(e\in OP^i) \geq \frac{1}{2}$ still holds. Next, we shall prove the same probability inequality for an $OHC$ edge related to one $K_{i+1}$ generated from the $K_i$ by adding one more vertex. 
	
	Add another vertex to the $K_i$, the $K_{i+1}$ is obtained as the new vertex is connected to each vertex in the $K_i$. Firstly, $e$ is considered as one $OHC$ edge in the $K_{i+1}$. The $OHC$ related to the $K_{i+1}$ contains $i+1$ $OP^{i+1}$s, and $i$ of them contain $e$. Moreover, $e$ is contained in $i-1$ $K_i$s in the $K_{i+1}$. Theorem \ref{th01} indicates that the frequency $f(e) \geq 2i-2$ exist in the  frequency $K_{i+1}$. For each of the $i$ $OP^i$s containing $e$ in $OHC$ related to the $K_{i+1}$, they are contained in two $OP^{i+1}$s in $OHC$, respectively. For each of the other $OP^i$s containing $e$ in the $OP^{i+1}$s not in $OHC$, it is only contained in one $OP^{i+1}$, respectively, and these $OP^i$s are different from each other. It is known that $e$ is contained in $2i-4$ $OP^i$s in the $K_i$ in the worst average case. Therefore, the number of $OP^i$s containing $e$ increases from the $K_i$ to the $K_{i+1}$.  %Moreover, the  frequency of $e$ rises as $O(i^2)$ according to $i$ based on equation (\ref{F1}). $e$ will be contained in more number of the $OP^i$s in the $K_{i+1}$. 
	%As the $OP^i$s including $e$ are contained in the $OP^{i+1}$s with the equal probability, each $K_i$ containing $e$ will contain the equal number of $OP^i$s containing $e$. It says the probability that $e$ is contained in the $OP^i$s in every $K_i$ will be the same. 
	In the $K_{i+1}$, $e$ is contained in $i-1$ $K_i$s in which there are $(i-1){{i}\choose{2}}$ $OP^i$s. Let $p_{i+1}(e\in OP^i)$ denote the probability that $e$ is contained in such an $OP^i$. The number of $OP^i$s containing $e$ is  $M_{i+1}=(i-1){{i}\choose{2}}p_{i+1}(e\in OP^i)$, and $M_{i+1} > M_i$ exists. 
	
	From the $K_i$ to the $K_{i+1}$, $i-2$ more $K_i$s containing $e$ are generated, and they contain $(i-2){{i}\choose{2}}$ new $OP^i$s. The $(i-2){{i}\choose{2}}$ new $OP^i$s are considered as one sample drawn from the $(i-1){{i}\choose{2}}$ $OP^i$s in the  $K_{i+1}$. Let $p_{i+1}'(e\in OP^i) > 0$ denote the probability that $e$ is contained in a new $OP^i$ for $M_{i+1} > M_i$. There will be $\Delta M=(i-2){{i}\choose{2}}p_{i+1}'(e\in OP^i)$ more $OP^i$s containing $e$.
	Because $M_i-\delta p_i(e\in OP^i)+\Delta M \leq M_{i+1}$ holds, $\left[1-\frac{\delta}{{{i}\choose{2}}}\right]p_i(e\in OP^i) + (i-2)p_{i+1}'(e\in OP^i) \leq (i-1)p_{i+1}(e\in OP^i)$ is derived. In the average case, $p_{i+1}'(e\in OP^i)$ and $p_{i+1}(e\in OP^i)$ will have no difference in values. Thus, $p_{i+1}(e\in OP^i) \geq \left[1-\frac{\delta}{{{i}\choose{2}}}\right]p_i(e\in OP^i)$ will hold since $p_{i+1}'(e\in OP^i)=p_{i+1}(e\in OP^i)$. Let $\delta$ be as small as possible, $\frac{\delta}{{{i}\choose{2}}}$ tends to zero, and $p_{i+1}(e\in OP^i) \geq p_i(e\in OP^i)$ holds. It implies that $e\in OHC$ is contained in more percentage of the $OP^i$s from the $K_i$ to the $K_{i+1}$. It means that more percentage of ordinary edges are replaced by $e$ for constructing the $OP^i$s from the $K_i$ to $K_{i+1}$. Because $p_i(e\in OP^i)\geq \frac{1}{2}$, more than half of the $OP^i$s contain $e$ with respect to the $K_i$s containing $e$. On average, $f(e) \geq \frac{1}{2}{{i}\choose{2}}$ exists in each frequency $K_i$ containing $e$, i.e., more than half of the $OP^i$s visiting $v_1$ and $v_2$ contain $e$. The inequality $p_i(e\in OP^i) \leq p_{i+1}(e\in OP^i)$ can be expanded to $p_i(e\in OP^i) \leq p_{i+k}(e\in OP^i)$ for $k\geq 1$. It implies that $e$ will be contained in more percentage of the $OP^i$s in the  $K_{i+k}$ containing more vertices if $e$ is one $OHC$ edge in the expanded graphs. The proof is the same. %Another proof method is provided for the special case of $i=4$ in paper \cite{DBLP:journals/Wang19}. 
	
	%There are $p^{i+1}(e\in OP^i)$
	
	Every $OP^{i+1}$ in the $K_{i+1}$ visits $v_1$ and $v_2$ in $e$, and each $OP^{i+1}$ contains two $OP^i$s visiting $v_1$ and/or $v_2$. If the two $OP^i$s in an $OP^{i+1}$ visit both $v_1$ and $v_2$, the probability that $e$ is not contained in the $OP^{i+1}$ is  $1-p_{i+1}(e\in OP^i)\leq \frac{1}{2}$ (i.e., it requires that one of the two $OP^i$s in the $OP^{i+1}$ does not contain $e$). Except the $OP^{i+1}$ in which $v_1$ and $v_2$ are the endpoints, each of the other $OP^{i+1}$s contains one or two $OP^i$s visiting both $v_1$ and $v_2$. Given an $OP^i$ visiting $v_1$ and $v_2$, $e$ is contained in the $OP^i$ with the probability $p_{i+1}(e\in OP^i)$. Thus, the probability that $e$ is included in each of these $OP^{i+1}$s is at least $p_{i+1}(e\in OP^i)\geq \frac{1}{2}$. 
	%Otherwise, as $k\geq 1$ vertices are added to the $K_i$ one by one while $e$ is the $OHC$ edge, $e$ will be contained in a smaller and smaller percentage of the $OP^{i+k}$s according to $k$. At last, it will have a frequency smaller than $2(i+k)-4$ in the frequency $K_{i+k}$ as $k$ is big enough. It is controdict to Theorem \ref{th01}. %and equation (\ref{F1}). 
	It mentions that each $OP^{i+1}$ is  computed with one $OP^i$ by adding one more vertex. As ${{i+1}\choose{2}}$ $OP^i$s vising $v_1$ and $v_2$ are chosen to compute the $OP^{i+1}$s, the number of the $OP^{i+1}$s containing $e$ will be bigger than  $L=\left[{{i+1}\choose{2}}-1\right]p_{i+1}(e\in OP^i)$. The probability that $e$ is contained in an $OP^{i+1}$ will be $p_{i+1}(e\in OP^{i+1}) \geq \frac{L}{{{i+1}\choose{2}}} = \left[1 -\frac{2}{i(i+1)}\right]p_ {i+1}(e\in OP^i) > \left[1 -\frac{2}{i(i+1)}\right]\left[1-\frac{2}{i(i-1)}\right]p_i(e\in OP^i)=\left(1 -\frac{4}{i^2}\right)p_i(e\in OP^i) \geq \frac{1}{2}-\frac{2}{i^2}$ where let $\delta = 1$ and $p_i(e\in OP^i) = \frac{1}{2}$ in the worst case. As $i$ is small, $p_{i+1}(e\in OP^{i+1})$ may be smaller than $\frac{1}{2}$ but bigger than 0.375 as $i = 4$. 
	
	It mentions that the $OP^{i+1}$ in which $v_1$ and $v_2$ are the endpoints is neglected as $p_{i+1}(e\in OP^{i+1})$ is derived. If this $OP^{i+1}$ is taken into account, $p_{i+1}(e\in OP^i)$ is derived as formula (\ref{F1}). In the worst case of $\delta = 1$, $p_{i+1}(e\in OP^{i+1}) > \left[1-\frac{2}{i(i-1)}\right]p_i(e\in OP^i)$ is derived for $e$. If let $\delta = 0$, $p_{i+1}(e\in OP^{i+1})\geq \frac{1}{2}$ exists. It indicates that there will be bigger than   $\frac{1}{2}{{i+1}\choose{2}}$ $OP^{i+1}$s containing $e$ in the $K_{i+1}$, and the frequency $f(e) \geq \frac{1}{2}{{i+1}\choose{2}}$ holds in the frequency $K_{i+1}$. Based on $p_{i+1}(e\in OP^{i+1}) > \left[1-\frac{2}{i(i-1)}\right]p_i(e\in OP^i)$, $p_i(e\in OP^i)$ is allowed to have the slight decrement or it rises from the $K_i$ to the $K_{i+1}$. %Thus, $e$ will be contained in more percentage of the $OP^{i+1}$s in the $K_{i+1}$. Since the probability $p_{i+k}(e\in OP^{i+k})$ increases according to a graph $K_{i+k}$ having more number of vertices as $k\geq1$, $e\in OHC$ will be contained in more percentage of the $OP^{i+k}$s as well. 
	
	\begin{equation}
		p_{i+1}(e\in OP^{i+1}) \geq \left[1-\frac{2\delta}{i(i-1)}\right]p_i(e\in OP^i),  \delta\in[0,1]
		\label{F1}
	\end{equation}
	
	If $p_i(e\in OP^i)$ increases from the $K_i$ to the $K_{i+1}$, the probability increment is limited with respect to $p_i(e\in OP^i) \leq 1$. One can let $M_i + \delta p_i(e\in OP^i) + \Delta M \geq M_{i+1}$ and $L={{i+1}\choose{2}}p_{i+1}(e\in OP^i)$. If the $OP^i$s are contained in the $OP^{i+1}$s with the equal probability, the relations between $p_{i+1}(g\in OP^{i+1})$ and $p_i(e\in OP^i)$ is given as formula (\ref{F2}).	Let $\delta = 1$, $p_{i+1}(e\in OP^{i+1}) < \left[1 + \frac{2}{i(i-1)}\right]p_i(e\in OP^i)$ is  derived as well as  $\frac{p_{i+1}(e\in OP^{i+1})}{p_i(e\in OP^i)}\in\left(1-\frac{2}{i(i-1)},1+\frac{2}{i(i-1)}\right)$. 
	%${i+1}\choose{2}$ of them will be chosen to
	% compose the $OP^{i+1}$s in $K_{i+1}$.
	\begin{equation}
		p_{i+1}(e\in OP^{i+1}) \leq \left[1 + \frac{2\delta}{i(i-1)}\right]p_i(e\in OP^i), \delta\in [0,1]
		\label{F2}
	\end{equation}

	Whether $p_i(e\in OP^i)$ decreases or increases, the decrement or increment will become smaller according to $i$. It mentions that $p_{i+1}(e\in OP^{i+1}) < \left[1+\frac{2}{i(i-1)}\right]p_i(e\in OP^i)$ also holds for an ordinary edge $g$ even if $\Delta M = 0$. In fact, the $OP^i$s will not be contained in the $OP^{i+1}$s with the equal probability. Theorem \ref{th01} indicates that an $OHC$ edge is contained in more number of $OP^i$s than an ordinary edge in the $K_i$.  In similarity, the $OP^i$s containing more $OHC$ edges or less ordinary edges in the $K_{i+1}$ will be contained in the $OP^{i+1}$s with the bigger probability. Thus, $p_{i+1}(e\in OP^{i+1}) > \left[1-\frac{2}{i(i-1)}\right]p_i(e\in OP^i)$ will hold for $e$ if $p_{i+1}(e\in OP^{i+1}) \leq p_i(e\in OP^i)$ occurs. Moreover, $p_{i+1}(e\in OP^{i+1}) \geq \left[1+\frac{2}{i(i-1)}\right]p_i(e\in OP^i)$ will usually appear if $p_{i+1}(e\in OP^{i+1}) \geq p_i(e\in OP^i)$ occurs. On the other hand, for an ordinary edge $g$ in the $K_i$ and $K_{i+1}$, if $p_{i+1}(g\in OP^{i+1}) \leq p_i(g\in OP^i)$ happens from $K_i$ to $K_{i+1}$, $p_{i+1}(g\in OP^{i+1}) \leq \left[1 - \frac{2}{i(i-1)}\right]p_i(g\in OP^i)$ will appear. If $p_{i+1}(g\in OP^{i+1}) \geq p_i(g\in OP^i)$ happens, $p_{i+1}(g\in OP^{i+1}) \leq \left[1 + \frac{2}{i(i-1)}\right]p_i(g\in OP^i)$ will appear in most cases. %For an ordinary edge $g$, 
	%In the following, it is the average case for all $OP^i$s. In the average case, an $OP^i$ will be contained in the $OP^{i+1}$s with the equal probability. However, the $OP^i$s containing the $OHC$ edges will be contained in the $OP^{i+1}$s with a bigger probability. The average cased can be used for the edges excluding from $OHC$. ***************
			
	Given an ordinary edge $g=(u_1,u_2)$ in the $K_i$ and $K_{i+1}$, the average case is  analyzed for comparison. It mentions that $g$ may have any small frequency in the frequency $K_i$ and $K_{i+1}$ based on Theorem \ref{th01}. In the best average case, $g$ is contained in $i-1$ $OP^i$s and $i$ $OP^{i+1}$s in the $K_i$ and $K_{i+1}$, respectively. In this case, $p_i(g\in OP^i) = \frac{i-1}{{{i}\choose{2}}}$ and $p_{i+1}(g\in OP^{i+1}) = \frac{i}{{{i+1}\choose{2}}}$ are computed, and $p_{i+1}(g\in OP^{i+1}) = \frac{i}{i+1}p_i(g\in OP^i)$ is derived. %If $i$ is big enough, $p_{i+1}(g\in OP^{i+1}) = \frac{i}{i+1}p_i(g\in OP^i)$ exists because $\frac{2}{(i-3)(i+1)}$ approaches zero. It indicates that $p_{i+1}(g\in OP^{i+1})$ becomes smaller according to $K_{i+1}$. 
	In fact, as $p_{i+1}(e\in OP^i) \geq p_i(e\in OP^i)$ exists, $p_{i+1}(g\in OP^i) \leq p_i(g\in OP^i)$ will occur from the $K_i$ to $K_{i+1}$. Thus, $p_{i+1}(g\in OP^{i+1})$ becomes smaller than $p_i(g\in OP^i)$ on average. Based on the frequency $K_4$s, $p_4(g\in OP^4) = \frac{1}{{{4}\choose{2}}}$ exists because $g$ is contained in one $OP^4$ and there are six $OP^4$s in each $K_4$. If we assume that $p_i(g\in OP^i) = \frac{1}{{{i}\choose{2}}}$ holds according to the $K_i$, $p_{i+1}(g\in OP^{i+1}) = \frac{i}{i+1}p_i(g\in OP^i) = \frac{1}{{{i+1}\choose{2}}}\times \frac{i}{i-1}$ is derived. If $i$ rises, $g$ maintains the frequency 1 or 2 from the frequency $K_i$ to frequency $K_{i+1}$ if it is always an ordinary edge. It implies that $g$ has the frequency 1 or 2 in the frequency $K_i$ on average. If all ordinary edges has the frequency of 1, an $OHC$ edge has the expected frequency  ${{i-1}\choose{2}} + 1$. If each ordinary edge has the frequency of 2, an $OHC$ edge has the expected frequency  $\frac{i^2-4i+7}{2}$.

	If $e\in OHC$ in the $K_i$ is substituted with another edge $e'\in OHC$ in the $K_{i+1}$, $p_{i+1}(e'\in OP^{i+1}) > \left[1 -\frac{2}{i(i-1)}\right]p_i(e\in OP^i)$ still holds. Otherwise, one can always find a sequence of $k$ $OHC$ edges $e'_k$ in the corresponding $K_{i+k}$s which are generated by adding $k\geq1$ vertices to the $K_i$ one by one.  $p_{i+k}(e'_k\in OP^{i+k})$ decreases fast according to $k$ until $p_{i+k}(e'_k\in OP^{i+k}){{i+k}\choose{2}} < 2(i+k)-4$ occurs. This result is contradict to Theorem \ref{th01} for the worst average case. Since $p_i(e\in OP^i) \geq \frac{1}{2}$ exists, $p_{i+1}(e'\in OP^{i+1})\geq \frac{1}{2}$ holds. It indicates that the frequency $f(e')\geq \frac{1}{2}{{i+1}\choose{2}}$ holds in the frequency $K_{i+1}$. In fact, either $e'$ is contained in the $K_i$ or not, $p_i(e'\in OP^i) \geq p_i(e\in OP^i)$ exists with respect to the $i-1$ $K_i$s containing $e$ and $e'$, respectively. Although $e'$ is not contained in $OHC$ in the $K_i$, $e$ will be replaced by $e'$ in $OHC$ related to the $K_{i+1}$ if $p_i(e'\in OP^i) \geq p_i(e\in OP^i)$. 
	
	Based on the proof, $f(e) \geq \frac{1}{2}{{i}\choose{2}}$ holds for an edge $e\in OHC$ in the frequency $K_i$ where  $i\in[4,n]$.  
\end{proof}	

It mentions that $p_{i+1}(e\in OP^{i+1})$ and $p_i(e\in OP^i)$ or $p_{i+1}(g\in OP^{i+1})$ and $p_i(g\in OP^i)$ may have obvious difference as $i$ is small. Once $i$ is big, $p_{i+1}(e\in OP^{i+1})\approx p_i(e\in OP^i)$ and $p_{i+1}(g\in OP^{i+1})\approx p_i(g\in OP^i)$ will appear since $1-\frac{2}{i(i-1)}$ or $1 + \frac{2}{i(i-1)}$ and $1-\frac{1}{i+1}$ tends to 1. In this case, $p_i(e\in OP^i)$ or $p_i(g\in OP^i)$ will change  smoothly and steadily from $i$ to $i+1$. %In the same manner, one can build more accurate inequality related to $p_{i+1}(g\in OP^{i+1})$ and $p_i(g\in OP^i)$ for a given edge $g$ under more known conditions. $p_{i+1}(g\in OP^{i+1})\approx p_i(g\in OP^i)$ will appear as $i$ is big enough. %and equation (\ref{F1}). 
%In addition, the proof process illustrates the probability that an $OHC$ edge is contained in the $OP^i$s increase according to $i$. 
%As $i = 4$, $p^4(e\in OP^4) \geq \frac{1}{2}$
%In each $K_i$s, it is contained in at least
%$\frac{1}{2}{{i}\choose{2}}$ $OP^i$s.
%To separate $OHC$ edges from ordinary edges, frequency $K_i$s containing more number of vertices are used to compute the frequency of edges. 
Based on Theorem \ref{th1}, an $OHC$ edge will have the frequency bigger than $\frac{1}{2}{{i}\choose{2}}$ in the frequency $K_i$. Some ordinary edges may have certain big frequency in the frequency $K_i$. However, the frequency is smaller than $\frac{1}{2}{{i}\choose{2}}$, and the expected frequency is smaller than $\frac{i+2}{2}$ in the best average case. 

\begin{theorem}
	\label{th2}
	Given a $K_i$ $(i\in[4,n])$ containing ${{i}\choose{2}}$ $OP^i$s, the frequency of an ordinary edge not in $OHC$ is smaller than  $\frac{1}{2}{{i}\choose{2}}$ in
	the corresponding frequency $K_i$, and the expected frequency is smaller than $\frac{i+2}{2}$ in the best average case. 
\end{theorem}

\begin{proof}
	Given a vertex $v$ in the frequency $K_i$, there are no four edges containing $v$ each of whose frequency is bigger than $\frac{1}{2}{{i}\choose{2}}$. Otherwise, the total frequency of the four edges is bigger than the total frequency $(i-1)^2$ of the $i-1$ edges containing $v$. 
	
	Provided that there are three edges $e_1$, $e_2$ and $e_3$ containing $v$ in the frequency $K_i$, each of their frequencies is bigger than $\frac{1}{2}{{i}\choose{2}}$. Moreover, $e_1$ and $e_2$ are the $OHC$ edges, and $e_3$ is one ordinary edge. Let $p_i(e_j\in OP^i) \geq \frac{1}{2}$ denote the probability that $e_j$ ($j$ = 1, 2, 3) is contained in an $OP^i$ in the $K_i$. When the $OP^{i-1}$s containing $v$ are used to compute the $OP^i$s by adding one more vertex, each of the $OP^{i-1}$s is contained in the $OP^i$s with the equal probability on average. If $e_1$, $e_2$ or $e_3$ is the $OHC$ edge in each  $K_{i-1}$, $p_{i-1}(e_j\in OP^{i-1}) \geq \frac{1}{2}$ exists based on Theorem \ref{th1}. If  $e_1$ or $e_2$ becomes one ordinary edge in each $K_{i-1}$ (it seldom happens in the average case, see the Theorem \ref{th3}), $e_1$ and $e_2$ are still contained in many $OP^{i-1}$s in the average case. If $p_{i-1}(e_j\in OP^{i-1})$ decreases from $i-1$ to $i$, $p_{i-1}(e_j\in OP^{i-1})\geq \frac{1}{2}$ must hold because $p_i(e_j\in OP^i)=\left[1-\frac{2\delta}{(i-1)(i-2)}\right]p_{i-1}(e_j\in OP^{i-1}) \geq \frac{1}{2}$ holds in the worst average case where $\delta\in [0,1]$. If $p_{i-1}(e_j\in OP^{i-1})$ increases from $i$ to $i+1$, $p_i(e_j\in OP^i) < \left[1 + \frac{2\delta}{(i-1)(i-2)}\right]p_{i-1}(e_j\in OP^{i-1}) < \frac{1}{2}$ appears  as $\delta$ is small enough if $p_{i-1}(e_j\in OP^{i-1}) < \frac{1}{2}$. It is contradict to the assumption. It means that $p_{i-1}(e_j\in OP^{i-1}) \geq \frac{1}{2}$ exists. %$ \frac{(i-1)(i-2)p_i(e_j\in OP^i)}{(i-1)(i-2)+2\delta}$ holds where $\delta\in [0,1]$. 
	%If $e_1$ and $e_2$ are the $OHC$ edges in a $K_{i-1}$ contained in $K_i$, $p_{i-1}(e_j\in OP^{i-1}) \geq \frac{1}{2}$ holds based on Theorem \ref{th1}. 
	If $e_3$ is still the ordinary edge in each $K_{i-1}$, the average case is considered. Since $p_i(e_3\in OP^i) \geq \frac{1}{2}$,
	$p_{i-1}(e_3\in OP^{i-1}) \geq \frac{1}{2}$ holds because $p_{i-1}(e_3\in OP^{i-1})$ becomes bigger from $i$ to $i-1$. 
	
	In the average case, each $OP^{i-k}$ ($k\in [1, i-4]$) is contained in the $OP^{i-k+1}$s with the equal probability for the three edges, respectively. 
	%If $e_1$ and $e_2$ are not the $OHC$ edges, $p_{i-k}(e_j\in OP^{i-k})$ changes as that of $e_3$. 
	It implies that $e_1$, $e_2$, and $e_3$ will be contained in bigger than half $OP^{i-k}$s in each $K_{i-k}$ containing them, respectively. As $k=i-4$, $e_1$, $e_2$, and $e_3$ are contained in one $K_4$. In the $K_4$, the three edges containing $v$ are contained in 1, 3, and 5 $OP^4$s, respectively, see Appendix $A$. There is one edge $e_j$ ($j = 1, 2$ or 3) which is contained in less than half of the $OP^4$s in the $K_4$. If the edge $e_j$ ($j = 1, 2$ or 3) is contained in exactly half $OP^4$s in each of the other $K_4$s, %smaller than $\frac{1}{2}{{i}\choose{2}}$.
	$p_4(e_j\in OP^4) < \frac{1}{2}$ exists. It is contradict that the $e_j$ is contained in more than half $OP^4$s in each $K_4$ containing it. Thus, at most two edges containing $v$ have the frequency  $f(e_j)\geq \frac{1}{2}{{i}\choose{2}}$.
	Since each $OHC$ edge has the frequency  $f(e_j) \geq \frac{1}{2}{{i}\choose{2}}$
	in the frequency $K_i$, the ordinary edge $e_3$ has the frequency
	$f(e_3) < \frac{1}{2}{{i}\choose{2}}$. %Thus, the frequency of an ordinary edge is smaller than $\frac{1}{2}{{i}\choose{2}}$. 
	
	Because $\frac{1}{2}{{i}\choose{2}} > 2i-4$ as $i\geq 7$, the average frequency of the $i-3$ ordinary edges $g$ will be smaller than $i-3$ in the  frequency $K_i$ if $i\geq 7$. Moreover, $p_i(e\in OP^i)$ for $OHC$ edges $e$ will maintain the big value or increase according to $i$ based on Theorems \ref{th01} and \ref{th1}. Meanwhile, $p_i(g\in OP^i)$ for ordinary edges decreases accordingly. Thus, $g$ will have the frequency much smaller than $i-3$ on average. As each of the two $OHC$ edges has the lowest frequency  $\frac{1}{2}{{i}\choose{2}}$,
	the average frequency of the $i-3$ ordinary edges is $\frac{1}{i-3}{{i-1}\choose{2}} \leq \frac{i+2}{2}$ if $i\geq 4$. If $g$ has the frequencies on the $i-3$ edges with the equal probability, the expected frequency of $g$ is smaller than $\frac{i+2}{2}$ in the best average case. %In the best average case, an ordinary edge has the frequency of $i-3$ in frequency $K_i$. On average, it will has any one frequency in $[0,i-3]$ in frequency $K_i$. If all integers in $[0,i-3]$ are the frequencies of ordinary edges, and an ordinary edge has a frequency in $[0,i-3]$ with the equal probability, the expected frequency of an ordinary edge will be smaller than $\frac{i-3}{2} < \frac{i+2}{2}$. 
	
	Thus, most ordinary edges will have a frequency smaller than $\frac{i+2}{2}$ in the frequency $K_i$. As $i$ becomes big, many ordinary edges have the frequency of 1 or zero in the frequency $K_i$. As they can not be used to build the $OP^i$s, these ordinary edges will be replaced by the $OHC$ edges and some other ordinary edges. Thus, the frequency of an $OHC$ edge will be bigger than $\frac{1}{2}{{i}\choose{2}}$ in the frequency $K_i$. 
\end{proof}
	Given one frequency $K_i$ ($i\in[4,n]$), the frequency of an $OHC$ edge is bigger than $\frac{1}{2}{{i}\choose{2}}$, whereas that of an ordinary edge is not. It indicates that $f(e) \geq \frac{1}{2}{{i}\choose{2}}$ is the sufficient and necessary condition for $OHC$ edges in the $K_i$. Moreover, the expected frequency for an $OHC$ edge is much higher than that for an ordinary edge in the average case. %This is the difference between $OHC$ edges and ordinary edges. 
\begin{theorem}
	\label{th22}
	Given a frequency $K_i$ $(i\in[4,n])$ containing ${{i}\choose{2}}$ $OP^i$s, the expected frequency for an $OHC$ edge is bigger than $\frac{i^2-4i+7}{2}$, whereas an ordinary edge has the expected frequency and maximum frequency smaller than 2 and $2(i-3)$, respectively in the average case. 
\end{theorem}

\begin{proof}	
	Given a frequency $K_i$, $e$ and $g$ denote an $OHC$ edge and ordinary edge, respectively. Based on Theorems \ref{th1} and \ref{th01}, $e$ has the frequency $f(e)= \frac{1}{2}{{i}\choose{2}}$ in the worst average case, and $g$ has the frequency $f(g) = i-3$ in the best average case. Because there are ${{i}\choose{2}}$ $OP^i$s in the frequency $K_i$, the probability that $e$ and $g$ is contained in an $OP^i$ is $p_i(e\in OP^i)\geq \frac{1}{2}$ and $p_i(g\in OP^i)\leq \frac{i-3}{{{i}\choose{2}}}$, respectively. Moreover, $p_i(e\in OP^i)$ increases and $p_i(g\in OP^i)$ decreases according to $i$ based on Theorem \ref{th1}. %, and $p_i(g\in OP^i) = \frac{i-1}{(i-2){{i}\choose{2}}}$ exists for an ordinary edge on average. 
	Thus, $\frac{p_i(e\in OP^i)}{p_i(g\in OP^i)} \geq \frac{i(i-1)}{4(i-3)} > \frac{i+2}{4}$ holds. It says that $f(e) > \frac{i+2}{4}f(g)$ holds in the worst average case. As the $OP^i$s include each $OP^{i-1}$ in the $K_i$ with the equal probability, $\frac{p_{i-1}(e\in OP^{i-1})}{p_{i-1}(g\in OP^{i-1})} > \frac{i+2}{4}$ holds. It says that the number of $OP^{i-1}$s containing $e$ is more than $\frac{i+2}{4}$ times of that of $OP^{i-1}$s containing $g$. As the $OP^k$s ($k\in[4,i-1]$) containing $e$ and $g$ are considered, $\frac{p_k(e\in OP^k)}{p_k(g\in OP^k)} > \frac{i+2}{4}$ also holds.  If $i=4$, $\frac{p_4(e\in OP^4)}{p_4(g\in OP^4)} > \frac{3}{2}$ holds. It is known that $p_4(e\in OP^4) > 2p_4(g\in OP^4)$ and $p_4(e\in OP^4) > 1.5p_4(g\in OP^4)$ based on Theorem \ref{th01}. Thus, $\frac{p_{i}(e\in OP^i)}{p_{i}(g\in OP^i)} > \frac{i+2}{4}$ is relaxed for small $i$. According to $i$, $e$ will be contained in more percentage of $OP^i$s than $g$ in the frequency $K_i$. 
	
	We assume that the $g$s are contained in the $OP^i$s with the equal probability as well as the $e$s do. For a vertex $v$, if one $g$ containing $v$ appears once in an $OP^i$, each of the two associated $e$s will appear at least $\frac{i+2}{4}$ times in the $OP^i$s. Based on Theorem \ref{th01}, the total frequency of the two $e$s is bigger than $\frac{i+2}{2}$ times of that of all the $g$s containing $v$. In this case, if an $OP^i$ contains one $g$, the $OP^i$ will contain more than $\frac{i+2}{2}$ $e$s in $OHC$. %since each vertex is contained in two $OHC$ edges.
	Because each $OP^i$ contains $i-1$ edges, it will contain more than $i-3$ $e$s and less than two $g$s. It is known that $OHC$ includes $i$ $OP^i$s each of which contains none of the $g$s. Thus, $e$ is contained in more than $\frac{i^2-4i+7}{2}$ $OP^i$s on average. In this case, $p_{i+1}(e\in OP^{i+1}) > p_i(e\in OP^i)$ while $p_{i+1}(e\in OP^{i+1}) \leq \left[1 + \frac{3}{i(i+1)}\right]p_i(e\in OP^i)$ hold according to $i$. The bigger the $p_i(e\in OP^i)$ is, the slower it increases from $i$ to $i+1$. The smaller the $p_i(e\in OP^i)$ is, the faster it increases from $i$ to $i+1$. For two adjacent $OHC$ edges, the probability sum of them also increases according to $i$ based on Theorem \ref{th01}. On the other hand, the expected frequency $f(g) = \frac{2(i-1)^2}{(i+3)(i-2)}$ of $g$ will tend to 2 according to $i$. If the frequency of the other ordinary edges containing $v$ is zero, the maximum frequency of $g$ will tends to $2(i-3)$. %As $i$ is big enough, the number of $g$s contained in an $OP^i$ tends to 2 and the average frequency of all ordinary edges approaches 2 as well.
	
	If each $g$ contained in the $OP^i$s has the frequency $f(g) = \frac{i+2}{2}$, there will be less than $\left[\frac{4(i-3)}{i+2}\right] < 4$ $g$s containing $v$ where $f(g) > 0$. This implies that some $g$s are not contained in the $OP^i$s if $p_i(e\in OP^i) \geq \frac{1}{2}$ and $p_i(g\in OP^i)$ approaches  $\frac{i+2}{2{{i}\choose{2}}}$. Moreover, if most $g$s has a frequency $f(g)\geq 2$, the number of $g$s with $f(g) > 0$ will be smaller than $i-3$. If one $g$ has the frequency $f(g) = 2(i-3)$, all other $g$s have the frequency $f(g) = 0$. Only if most $g$s have the  frequency $f(g)=$ 1 or 0, some $g$s contained in the $OP^i$s may have the relatively big frequency.  
\end{proof} 

%, and it is smaller than $\frac{i+2}{2}$ if $i\geq 4$, they cannot compose the ${{i}\choose{2}}$ $OP^i$s. 

%If an edge $e\in OHC$ is contained in both the $K_i$ and one $K_{i+1}$ generated by adding one vertex to the $K_i$, the inequality $p_{i+1}(e\in OP^{i+1})\geq (1 -\frac{4}{i^2})p_i(e\in OP^i)$ is derived. As  $e\in OHC$ in the $K_i$ is replaced by another edge $e'\in OHC$ in the $K_{i+1}$, $p_{i+1}(e'\in OP^{i+1}) \geq(1 -\frac{4}{i^2})p_i(e\in OP^i)$ also holds. It means $e$ or $e'$ is contained in the nearly equal percentage of the $OP^{i+1}$s as or more percentage of the $OP^{i+1}$s than that of the $OP^i$s containing $e\in OHC$ in the $K_i$. As $i$ rises, $\frac{4}{i^2}$ approaches zero and $p_i(e\in OP^i)$ will tend to  $p_{i+1}(e\in OP^{i+1})$ or $p_{i+1}(e'\in OP^{i+1})$ as well. It says $p_i(e\in OP^i)$ for $e\in OHC$ keeps stable or improved according to $i$. %As the $p_i(e\in OP^i)$ is computed based on the  $K_i$, the $p_{i+1}(e\in OP^{i+1})$ can be estimated with respect to the $K_{i+1}$. As the $p_{i+1}(e \in OP^{i+1})$ and $p_i(e\in OP^i)$ are computed with frequency $K_i$s and $K_{i+1}$s, respectively, this probability inequality can be taken as one condition to filter out the other edges not in $OHC$. 
In the worst case, $f(e)$ of an edge $e\in OHC$ may be smaller than $\frac{1}{2}{{i}\choose{2}}$ but it is close to $\frac{1}{2}{{i}\choose{2}}$ in the frequency $K_i$. Each edge is contained in ${{i-2}\choose{2}}$ $K_4$s in the $K_i$. The authors \cite{DBLP:journals/Wang16} assumed that the six frequency $K_4$s in Appendix $A$ are uniformly distributed for an edge in the $K_i$. In other words, each edge has $\frac{1}{3}{{i-2}\choose{2}}$ 1s, $\frac{1}{3}{{i-2}\choose{2}}$ 3s and $\frac{1}{3}{{i-2}\choose{2}}$ 5s with respect to the ${{i-2}\choose{2}}$ frequency $K_4$s containing it. In this case, each edge has the frequency  $F(e) = 3{{i-2}\choose{2}}$ and average frequency $f(e) = 3$ based on the frequency $K_4$s. According to $OHC$ in the $K_i$, $e$ is contained in $i-3$ frequency $K_4$s in each of which it has the frequency 3 or 5. If the frequency 1, 3, and 5 are uniformly distributed in the other frequency $K_4$s, the average frequency $f(e) > 3$ exists based on the frequency $K_4$s. In the worst case, they assumed that $e$ has no 5s, i.e., it has $\frac{2}{3}{{i-2}\choose{2}}$ 3s and $\frac{1}{3}{{i-2}\choose{2}}$ 1s in all the frequency $K_4$s containing it. In this case, $e$ has the minimum frequency   $F_{min}(e) = \frac{2\times 3+ 1\times 1}{3}\times{{i-2}\choose{2}}=\frac{7}{3}{{i-2}\choose{2}}$ based on the frequency $K_4$s. Since each $K_4$ contains six $OP^4$s, the minimum probability $p_{min}(e\in OP^4) = \frac{7}{18}$ is derived. Thus, $p_4(e\in OP^4)\geq \frac{7}{18}$ holds. Since $p_4(e\in OP^4) \geq \frac{7}{18}$, one can let $\delta = 0$ in the proof of Theorem \ref{th1} and $p_i(e\in OP^i)\geq  \frac{7}{18}$ can be proven. It means that an $OHC$ edge associated with $K_i$ has the frequency bigger than $\frac{7}{18}{{i}\choose{2}}$
in the frequency $K_i$. It is near the improved lower frequency bound $\frac{1}{2}{{i}\choose{2}}$ for the worst average case. However, this case seldom happens for big $i$ because $p_4(e\in OP^4)$ increases according to $i$ \cite{DBLP:journals/Wang19,DBLP:journals/Wang251}. Even if $p_i(g\in OP^i) = \frac{i-3}{{{i}\choose{2}}}$ appears in the best average case, $\frac{p_i(e\in OP^i)}{p_i(g\in OP^i)} \geq \frac{7(i+2)}{36} \geq \frac{7}{6}$ holds as $p_i(e\in OP^i) \geq \frac{7}{18}$ and $i\geq 4$. In this case, the total frequency of two adjacent $OHC$ edges is bigger than $\frac{7(i+2)}{18} \geq \frac{7}{3}$ times of that of the $i-3$ ordinary edges containing a vertex. In this case, two adjacent $OHC$ edges have the frequency sum bigger than $\frac{7(i-1)^2}{10}$. Based on Theorem \ref{th22}, the average frequency of the $OHC$ edges is bigger than $\frac{i^2 - 7i + 16}{2}$, and that of the ordinary edges is smaller than $5$ on average. %$\frac{7(i-1)^2}{10}$ and $\frac{3(i-1)^2}{10}$. In other cases, the frequency of an $OHC$ edge is bigger than $\frac{1}{2}{{i}\choose{2}}$ in a frequency $K_i$ containing it. 

In the frequency $K_i$, the frequency of an edge is in $\left[0, {{i}\choose{2}} - 1\right]$. %usually not averagely distributed.
Given the $i-1$ edges containing a vertex $v$, the frequency of the $j^{th}$ edge is denoted as $f(e_j)$ where $j\in [1, i-1]$. Then, the probability that the $j^{th}$ edge is contained in an $OP^i$  
is computed as $p_i(e_j\in OP^i) = \frac{2f(e_j)}{i(i-1)}$. 
Since the average frequency of all edges is $i-1$, the total probability of the $i-1$ edges is derived as $\sum_{j = 1}^{i-1}p_i(e_j\in OP^i) = \frac{2(i-1)}{i}=2-\frac{2}{i}$. The total probability of the $i-1$ edges containing $v$ is increasing according to $i$, and it tends to 2 for big $i$. As there are two $OHC$ edges each of which has the probability $p_i(e\in OP^i) \geq\frac{1}{2}$,
the total probability of the $i-3$ ordinary edges is smaller than $\frac{i-2}{i} < 1$. On average, an $OHC$ edges have the frequency bigger than $\frac{i^2-4i+7}{2}$ in the frequency $K_i$, and the total  probability of the $i-3$ edges becomes smaller than $\frac{4}{i}$.  Moreover, the average probability of the $i-3$ edges tends to zero as $i$ is big enough. In the best average case, an ordinary edge will have certain frequency smaller than $\frac{1}{i-3}{{i-1}\choose{2}}< \frac{i+2}{2}$ which is much smaller than $\frac{1}{2}{{i}\choose{2}}$ as $i\geq 5$.

As the $OHC$ edges and ordinary edges are not uniformly distributed in $OP^i$s, the frequency bounds for $OHC$ edges and ordinary edges in a frequency $K_i$ are listed in Table \ref{fbounds}. For two adjacent $OHC$ edges for the worst average case, the lower bound of the total frequency is derived in  Theorem \ref{th01}.

\begin{table}
	\begin{center}
		\caption{The frequency bounds of edges according to a frequency $K_i$ (LB = lower bound and UB = upper bound, adj. = adjacent, Avg. = average value.)}
		{\footnotesize \begin{tabular}{ p{2.5cm}  p{1cm}  p{1.5cm}  p{1.5cm}  p{0.1cm}  p{1.5cm}  p{1.6cm} }
				\hline
		& LB for & $e\in OHC$  &  &  & UB for & g$\notin OHC$   \\				
		%&  & &  &  & Ordinary &  edge & $g$  \\
		& 1 $e$ & 2 adj. $e$s & Avg. of $e$s &  & 1 $g$ & Avg. of $g$s  \\
		\hline
		Worst avg. case & $\frac{1}{2}{{i}\choose{2}}$ & $\frac{4(i-1)^2}{5}$ or $\frac{3(i-1)^2}{4}$ & $\frac{i^2-4i+7}{2}$ &  & $2(i-3)$ & 2  \\	
		Worst case & $\frac{7}{18}{{i}\choose{2}}$ & $\frac{7(i-1)^2}{10}$ & $\frac{i^2-7i+16}{2}$ &  & $5(i-3)$ & 5   \\														
		%\bottomrule
				\hline
		\end{tabular}}
		\label{fbounds}
	\end{center}
\end{table}

%As $i$ is appointed, the number of frequency $K_i$s
%is finite for a $K_i$ due to the constraints of $OP^i$s.
%Given the $i-1$ $f_j$s for the $i-1$ edges containing $v$,
%each $f_j$ can be assigned to every edge.
%If the $i-1$ $f_j$s are not equal to each other, 
%and the edges containing each vertex are assigned to 
%the same set of $f_j$s,
%the number of frequency $K_i$s will be $(i-1)!$.
%If there are $m$ ($1\leq m \leq i-1$) groups of edges and each
%edge in one group is assigned the same frequency,
%and the $m_k^{th}$ ($1\leq k \leq m$)
%group contains $i_k$ edges which can be assigned $l_k$ different $f_j$s,
%the number of frequency $K_i$s will be $\prod_{k = 1}^m(l_k)\prod_{k = 1}^m{{i-1-\sum_{f=1}^{k-1}(i_f)}\choose{i_k}}$.
%%and $i_0 = 0$.
%%As there are $k$ ($1\leq k\leq i-1$) edges have the same frequency,
%%the number of frequency $K_i$s becomes $(i-k-1)!$
%In addition, if there are $q$ sets of $f_j$s for
%the edges containing each vertex,
%the number of frequency $K_i$s will be at least $q(i-1)!$
%in case that the $f_j$s in each set are not equal to each other.

%$i$ ver each set of $i-1$ edges containing $v$ will have
%$i-1$ $f_j$s.

\section{The lower frequency bound for $OHC$ edges related to $K_n$}
\label{sec4}
%It says the frequency of most $OHC$ edge will be bigger than
%$\frac{1}{2}{{i}\choose{2}}$ and approaches the maximum frequency.
In $K_n$, an edge is contained
in ${{n-2}\choose{i-2}}$ frequency $K_i$s ($i\leq n$). It has a frequency in $\left[0,{{i}\choose{2}}-1\right]$ in every frequency $K_i$. For an $OHC$ edge related to $K_n$, the average frequency will be bigger than $\frac{1}{2}{{i}\choose{2}}$ based on the frequency $K_i$s containing it. 

\begin{theorem}
	\label{th3}
	Given a $K_i$ $(i\in[4,n])$ containing an $OHC$ edge related to $K_n$, the frequency of this edge is bigger than $\frac{1}{2}{{i}\choose{2}}$ in
	the corresponding frequency $K_i$ in the worst average case.
\end{theorem}

\begin{proof}
	
	%For an $OHC$ edge $e$ in $K_n$, it is contained in at least $n-3$ $K_4$s in each of which it is contained in at least three $OP^4$s \cite{DBLP:journals/Wang16. %\cite{DBLP:journals/Wang16}.(Wang and Remmel, 2016)
	%If $n = 4$, $e$ is contained in at least three
	%optimal 4-vertex paths in $K_4$.
	%Since a $K_4$ contains six $OP^4$s, the probability that $e$ is contained in an $OP^4$ is greater than $\frac{1}{2}$ with respect to these $K_4$s. Note the probability as $p_4(e\in OP^4) \geq \frac{1}{2}$.
	
	%Draw one of these $K_4$s as the base graph and add the other $k$ ($1\leq k \leq n-4$) vertices to this $K_4$ at random to form a bigger and bigger graph $K_{4+k}$. As $k$ rises, the probability $p_{4+k}(e\in OP^4)$ that $e$ is contained in the $OP^4$s in all $K_4$s containing $e$ does not decrease according to $k$ based on Theorem \ref{th1}, i.e. $p_{4+k}(e\in OP^4)\geq p_4(e\in OP^4) \geq \frac{1}{2}$. Otherwise, the frequency of $e$ computed with frequency $K_4$s will be smaller than $3{{n}\choose{2}}$ as $k=n-4$. It is known that $e$ has the  smallest frequency  $3{{n}\choose{2}}$ based on frequency $K_4$s for big $n$  \cite{DBLP:journals/Wang19}. Thus, $e$ will be contained in more than half of the $OP^4$s in the $K_4$s containing $e$ in $K_n$. On can also add random edges to the $n-3$ basic $K_4$s containing $e$ for producing more and more $K_4$s containing $e$, the conclusion is the same.
	Given an edge $e\in OHC$ in $K_n$, it is contained in ${{n-2}\choose{2}}$ frequency $K_4$s. In a frequency $K_4$, the frequency of $e$ is 3 or 5 based on Theorem \ref{th01}. In this case, $e$ is the $OHC$ edge in each $K_4$ containing it on average. As the average frequency of $e$ is computed with all these frequency $K_4$s, $e$ has the average frequency $f(e)>3$ in the worst average case \cite{DBLP:journals/Wang16}. %It implies that $e$ has the frequency of 3 or 5 in each frequency $K_4$ containing it on average. %Thus, the probability that $e$ is contained in the $OP^4$s in a $K_4$ containing $e$ is bigger than $\frac{}{}$. Meanwhile, $e$ is the $OHC$ edge in a $K_4$ containing it. 
	Thus, the probability $p_4(e\in OP^4)\geq \frac{1}{2}$ exists where $p_4(e\in OP^4)$ denotes the probability that $e$ is contained in an $OP^4$ in a $K_4$ containing $e$.  

	Draw one sub-graph $K_m$ $(m < n)$ containing $e$ from $K_n$. Let $p'_4(e\in OP^4)$ denote the probability that $e$ is contained in an $OP^4$ in a $K_4$ containing $e$ in the $K_m$. On average, $p'_4(e\in OP^4) \geq \frac{1}{2}$ exists. There are $M_m={{m-2}\choose{2}}{{4}\choose{2}}p'_4(e\in OP^4)$ $OP^4$s containing $e$ in the $K_m$. Take another vertex in $K_n$, and add it to the $K_m$ for constructing one second sub-graph $K_{m+1}$. Let $q_4(e\in OP^4)$ denote the probability that $e$ is contained in an $OP^4$ in a $K_4$ containing $e$ in the $K_{m+1}$, and there are $M_{m+1}={{m-1}\choose{2}}{{4}\choose{2}}q_4(e\in OP^4)$ $OP^4$s containing $e$. From the $K_m$ to the $K_{m+1}$, $m-2$ new $K_4$s are generated for $e$. Let $q'_4(e\in OP^4)$ denote the probability that $e$ is contained in an $OP^4$ in a new $K_4$ containing $e$, and there are $\Delta M=(m-2){{4}\choose{2}}q'_4(e\in OP^4)$ new $OP^4$s. As $e$ has the frequency 3 or 5 in a frequency $K_4$ for $p_4(e\in OP^4) \geq \frac{1}{2}$, $q'_4(e\in OP^4) > 0$ and $M_{m+1} \geq M_m -\delta p'_4(e\in OP^4)+ \Delta M$ holds where $\delta\in[0,1]$. View the $m-2$ new $K_4$s as one sample drawn from the $K_{m+1}$, $q'_4(e\in OP^4) = q_4(e\in OP^4)$ exists on average. In this case, $q_4(e\in OP^4) \geq \left[1-\frac{\delta}{{{m-2}\choose{2}}{{4}\choose{2}}}\right]p'_4(e\in OP^4)$ is derived. Let $\delta$ be as small as possible, $\frac{\delta}{{{m-2}\choose{2}}{{4}\choose{2}}} = 0$ while $q_4(e\in OP^4) \geq p'_4(e\in OP^4)$ holds. It indicates that $e$ will be contained in more percentage of $OP^4$s in a $K_4$ according to $m$, and $q_4(e\in OP^4)$ reaches the maximum value $p_4(e\in OP^4)$ at $n$. Based on the same proof used in Theorem \ref{th1}, $p_5(e\in OP^5) \geq q_4(e\in OP^4) \geq p'_4(e\in OP^4)$ can be derived for $e$ (i.e., let $L={{5}\choose{2}}q_4(e\in OP^4)$ in a $K_5$), where $p_5(e\in OP^5)$ is the probability that $e$ is contained in an $OP^5$ in a $K_5$ containing $e$. 
	
	If $m$ is close to $n$, for example, $m = n-1$, $p'_4(e\in OP^4) = p_4(e\in OP^4)$ holds, and $p_5(e\in OP^5)\geq p_4(e\in OP^4)$ can be derived. In the same manner, if $i\geq 5$, $p_{i+1}(e\in OP^{i+1}) \geq \left[1 - \frac{\delta}{{{m-2}\choose{i-2}}{{i}\choose{2}}}\right]p_i(e\in OP^i)$ can be derived from $i$ to $i+1$. In the worst case, let $\delta = 1$ and $i=m$, and $p_{i+1}(e\in OP^{i+1}) > \left[1 - \frac{1}{{{i}\choose{2}}}\right]p_i(e\in OP^i)$ is derived. It indicates that $p_i(e\in OP^i)$ is permitted to have a slight decrement from $i$ to $i+1$. Let $\delta$ be small enough or $n$ is big enough, $p_{i+1}(e\in OP^{i+1}) \geq p_i(e\in OP^i)$ holds. Thus, $p_n(e\in OP^n) \geq p_{n-1}(e\in OP^{n-1})\geq...\geq  p_{i+1}(e\in OP^{i+1})\geq p_i(e\in OP^i)\geq...\geq p_4(e\in OP^4)$ is derived if $\delta = 0$. It indicates that $e$ is contained in more percentage of the $OP^i$s according to $i$. More and more percentage of ordinary edges are replaced by $e$ for building the $OP^i$s containing more vertices. Since $p_4(e\in OP^4)\geq \frac{1}{2}$ exists according to the $K_4$s, $p_i(e\in OP^i) \geq \frac{1}{2}$ will hold as $i>4$.  
	It implies that $e$ is contained in more than $\frac{1}{2}{{i}\choose{2}}$
	$OP^i$s in each $K_i$ containing $e$ on average. In this case, the frequency $f(e) \geq \frac{1}{2}{{i}\choose{2}}$ exists in a frequency $K_i$ containing $e$. 
	
	In the worst case for $p'_4(e\in OP^4) = \frac{7}{18}$ according to $M_m$,  $M_{m+1}\geq M_m + \Delta M$ holds as well as $p_4(e\in OP^4)\geq p'4(e\in OP^4) = \frac{7}{18}$. Based on the same proof, $p_n(e\in OP^n) \geq p_{n-1}(e\in OP^{n-1})\geq...\geq p_i(e\in OP^i)\geq...\geq p_4(e\in OP^4) \geq \frac{7}{18}$ can be derived.

	It mentions that the probability increment cannot be too big from $i$ to $i+1$. If the $OP^i$s are uniformly contained in the $OP^{i+1}$s,  $p_{i+1}(e\in OP^{i+1}) \leq \left[1 + \frac{\delta}{{{m-2}\choose{i-2}}{{i}\choose{2}}}\right]p_i(e\in OP^i) < \left[1 + \frac{1}{{{i}\choose{2}}}\right]p_i(e\in OP^i)$ holds, where let $\delta = 1$ and  $i = m$ in the best case. Since the $OP^i$s containing more $OHC$ edges will be contained in the $OP^{i+1}$s with the bigger probability, $p_{i+1}(e\in OP^{i+1}) > \left[1 - \frac{2}{i(i-1)}\right]p_i(e\in OP^i)$ will hold for $e$ if $p_{i+1}(e\in OP^{i+1}) \leq p_i(e\in OP^i)$ occurs. Moreover, $p_{i+1}(e\in OP^{i+1}) \geq \left[1 + \frac{2}{i(i-1)}\right]p_i(e\in OP^i)$ will happen in most cases if $p_{i+1}(e\in OP^{i+1}) \geq p_i(e\in OP^i)$ appears, especially for $e$ with a relatively small $p_i(e\in OP^i)$. %Since $M_m + \delta p_4(e\in OP^4) + \Delta M \geq M_{m+1}$ exists, $q_4(e\in OP^4) \leq [1 + \frac{1}{{{m-2}\choose{2}}}]p_4(e\in OP^4)$ can be derived as $\delta = 1$. 
	
	 $p_i(e\in OP^i)$ may decrease from $i$ to $i+1$ for some $OHC$ edges in $K_n$. One can assume $p_i(e\in OP^i) = \frac{ai^2 + bi +c}{{{i}\choose{2}}}$ where $a\in (\frac{1}{4}, \frac{1}{2}]$, $b$ and $c$ are small numbers. Because $p_i(e\in OP^i) > p_{i+1}(e\in OP^{i+1})$, $(a+b)i + a + b + c>0$ is derived. It requires that $a + b > 0$ must hold for an $i$. Since $p_{i+1}(e\in OP^{i+1}) > \left[1 - \frac{2}{i(i-1)}\right]p_i(e\in OP^i)$, $(a+b)i^3 - 2(a+b-c)i^2 - (a+3b+4c)i + 2c <0$ can be derived. This inequality requires that $a + b < 0$ exists for $i$. Thus, $p_i(e\in OP^i)$ only decreases according to some small $i$s. As $i$ is big, $(a+b)i + a + b + c < 0$ will appear as well as $p_i(e\in OP^i) \leq p_{i+1}(e\in OP^{i+1})$ for $a + b < 0$. Moreover, if $p_{i+1}(e\in OP^{i+1}) > \left[1 + \frac{2}{i(i-1)}\right]p_i(e\in OP^i)$ appears, one can derive the condition $(3a+b)i^3 + 2(a+b+c)i^2 -(a-b)i + 2c < 0$. It requires $3a + b < 0$. If $a\in (\frac{1}{4}, \frac{1}{2}]$ is given, it indicates that an $e\in OHC$ with a smaller $p_i(e\in OP^i) = \frac{ai^2 + bi +c}{{{i}\choose{2}}}$ is permitted to increase faster from $i$ to $i+1$. 
	
	As $p_i(e\in OP^i) $ decreases from $i$ to $i+1$, there are several cases. The formula (\ref{F1}) guarantees that $p_i(e\in OP^i)$ to have a slight decrement from $i$ to $i+1$ and $p_i(e\in OP^i) \geq \frac{7}{18}$. If $p_i(e\in OP^i)$ always decreases from $i$ to $n$, there are two cases. If $p_n(e\in OP^n) \geq \frac{1}{2}$,  $p_i(e\in OP^i) \geq \frac{1}{2}$ exists at each number $i$.  If $p_n(e\in OP^i) < \frac{1}{2}$, this case seldom happens to $OHC$ edges, especially for big $n$. If $n$ is very small and $p_i(e\in OP^i)$ decreases from $i = 4$ to $n$, $p_n(e\in OP^n) > \frac{n+1}{2(n-1)}p_4(e\in OP^4)$ holds in theory. If $p_i(e\in OP^i)$ decreases first, and then it increases according to $i$. It decreases faster at the small numbers $i$ and increases slower at the big numbers $i$. If $p_i(e\in OP^i) < \frac{1}{2}$ never appears, $p_i(e\in OP^i)$ decreases very slow, and $e$ is contained in more than half $OP^i$s in each $K_i$ containing $e$ on average. If $p_i(e\in OP^i) < \frac{1}{2}$ appears at the small numbers $i$, since $p_i(e\in OP^i)$ decreases fast in these steps,  the decreasing steps must be limited comparing to $n$. Otherwise, $p_n(e\in OP^n) \geq \frac{1}{2}$ cannot appear or $p_n(e\in OP^i) < \frac{7}{18}$ will appear. For big $n$, since $p_4(e\in OP^4) \geq \frac{1}{2}$ and $p_i(e\in OP^i)$ may decrease according to a limited number of $i$s, the number of $p_i(e\in OP^i)$s smaller than $\frac{1}{2}$ is also limited. Moreover, even if $p_i(e\in OP^i) < \frac{1}{2}$ occurs, $p_i(e\in OP^i) \geq \frac{7}{18}$ still holds. Thus, $f(e) \geq \frac{1}{2}{{i}\choose{2}}$ occurs in nearly all frequency $K_i$s containing $e$.
	
	$p_i(e\in OP^i)$ computed based on frequency $K_i$s increases according to $i$ for big $n$. Based on the interval $\left[1 - \frac{\delta}{{{n-2}\choose{i-2}}{{i}\choose{2}}},1 + \frac{\delta}{{{n-2}\choose{i-2}}{{i}\choose{2}}}\right]$, $p_{i+1}(e\in OP^{i+1})$ and $p_i(e\in OP^i)$ will be nearly equal as $i$ is big for big $n$. 
	%Each $K_i$ contains ${{i-2}\choose{2}}$ base graphs $K_4$
	%which can be used to generate new $K_i$s containing $e$.
\end{proof}
%According to $i$, an edge will be contained in different number of $K_i$s in $K_n$. Moreover, the edges are contained in different percentages of the $OP^i$s for a specific $i$. If an edge is contained in ${{i}\choose{2}} - 1$ $OP^i$s in each $K_i$ containing the edge, the edge must be one $OHC$ edge. Thus, the average frequency of an edge computed with frequency $K_i$s or the probability that an edge is contained in the $OP^i$s can be taken as the variational function for the edge.
Given one frequency $K_i$ containing an edge $e\in OHC$ in $K_n$, the frequency  $f(e) \geq \frac{1}{2}{{i}\choose{2}}$ holds on average. According to Theorems \ref{th1} and \ref{th2}, $e$ is also the $OHC$ edge in the corresponding $K_i$. Thus, $e$ is the $OHC$ edge in every $K_i$ containing it on average. In this case, the probability inequality $p_i(e\in OP^i) \geq \frac{1}{2}$ and $p_{i+1}(e\in OP^{i+1})> \left[1 -\frac{2}{i(i-1)}\right]p_i(e\in OP^i)$ hold for $e$ as $p_{i+1}(e\in OP^{i+1})$ and $p_i(e\in OP^i)$ are computed based on frequency $K_{i+1}$s and $K_i$s, respectively. Moreover, $p_i(e\in OP^i)$ maintains the big value according to $i$. %Based on Theorem \ref{th01}, an $OHC$ edge in $K_n$ will be contained in more percentage of $OP^i$s than an ordinary edge contained in $OP^n$s but excluding from $OHC$.The probability may have a small flunction in values as it is computed with different frequency $K_i$s. However, the probability for an edge out of $OHC$ will have much difference as they are computed with different frequency $K_i$s according to $i$. 
%Since each edge is contained in ${{n-2}\choose{i-2}}$ $K_i$s in $K_n$, they are contained in the same number of frequency $K_i$s, respectively. 
%Note $p_i(e\in OP^i)$ the probability that an edge $e\in OHC$ in $K_n$ is contained in an $OP^i$ in a given $K_i$ containing $e$. $p_i(e\in OP^i)\geq \frac{1}{2}$ exists for $e\in OHC$ in the average case. Given a frequency $K_i$ containing $e$, the frequency $f(e)$ will be bigger than $\frac{1}{2}{{i}\choose{2}}$. 
As one draws $N\left(1\leq N \leq {{n-2}\choose{i-2}}\right)$
frequency $K_i$s containing $e$ to compute its total frequency $F(e)$, 
%Provided that the frequency of $e$ is $f_k(e) (1\leq k \leq N)$in the $k^{th}$ frequency $K_i$,
%$F_i(e)=\sum_{k=1}^N f_k(e)$ is computed. For $e\in OHC$ in $K_n$,
$F(e)\geq \frac{N}{2}{{i}\choose{2}}$ holds. Thus, the lower frequency bound for $OHC$ edges in $K_n$
is $\frac{N}{2}{{i}\choose{2}}$ based on frequency $K_i$s, and the average frequency $f(e) \geq \frac{1}{2}{{i}\choose{2}}$. 

$e\in OHC$ in $K_n$ has the frequency $f(e) \geq \frac{1}{2}{{i}\choose{2}}$ in each frequency $K_i$ containing it on average. There are $i$ frequencies bigger than $\frac{1}{2}{{i}\choose{2}}$ in every frequency $K_i$. As $e$ has the $i$ frequencies bigger than $\frac{1}{2}{{i}\choose{2}}$ with the equal probability, the expected frequency of $e$ will be bigger than $\frac{i^2-4i+7}{2}$ based on Theorem \ref{th22}. This indicates that the average frequency of all $OHC$ edges will be bigger than $\frac{i^2-4i+7}{2}$ as it is computed with frequency $K_i$s. Meanwhile, the average probability $p_i(e)$ of all $OHC$ edges will be bigger than $\frac{i^2-4i+7}{i(i-1)}$. It increases as $p_{i+1}(e\in OP^{i+1}) = \left[1 + \frac{3}{i(i+1)}\right]p_i(e\in OP^i) + o(i^{-2})$ from $i$ to $i+1$. For the $OHC$ edges with $f(e) > \frac{i^2 - 4i + 7}{2}$, $p_{i+1}(e\in OP^{i+1}) < \left[1 + \frac{3}{i(i+1)}\right]p_i(e\in OP^i)$ exists as $p_i(e\in OP^i)$ increases from $i$ to $i+1$. For the other $OHC$ edges with the relatively smaller average frequency $f(e) = \frac{i^2 + bi + c}{2}$ where $b < -4$ and $c$ is a small number, $p_{i}(e\in OP^i)$ will increase faster than that corresponds to $f(e) > \frac{i^2 - 4i + 7}{2}$. 

Since there are ${{n}\choose{i}}$ frequency $K_i$s in $K_n$, each frequency $K_i$ containing ${{i}\choose{2}}$ $OP^i$s, and each $OP^i$ contains $i-1$ edges, the average frequency of the $\frac{n(n-3)}{2}$ ordinary edges will be  $f(g) < \frac{(i-1)(n-1)}{n-3} - \frac{i^2-4i+7}{n-3} < i-1$. As the average frequency of $OHC$ edges is much bigger than $\frac{i^2-4i+7}{2}$,  $f(g) < i-3$ will appear according to $i$ in most cases. Meanwhile, the average probability of all ordinary edges will be $p_i(g\in OP^i) < \frac{2(n-1)}{i(n-3)} - \frac{2(i^2-4i+7)}{i(i-1)(n-3)}$ and it decreases according to $i$. Because $i\geq 4$, $p_i(g\in OP^i) < \frac{1}{2}$ exists for big $n$. It implies that $g$ is contained in smaller than $\frac{1}{2}{{i}\choose{2}}$ $OP^i$s in each $K_i$ on average. Moreover, $p_i(g\in OP^i)$ decreases sharply in proportion to the factor smaller than  $\frac{i}{i+1}$ according to $i$. One can assume $p_i(g\in OP^i) = \frac{ai^2 + bi +c}{{{i}\choose{2}}}$ and $a\in [0,\frac{1}{2}]$. If $p_{i+1}(g\in OP^{i+1}) < \frac{i}{i+1}p_i(g\in OP^i)$,  $ai^2 - ai -(a+b+c) < 0$ is derived. Thus, $a$ must be zero. It indicates that the average frequency $f(g) = bi + c$ where $b + c \geq 0$ if $p_{i+1}(g\in OP^{i+1}) < \frac{i}{i+1}p_i(g\in OP^i)$ happens from $i$ to $i+1$. In this case, more percentage of ordinary edges are replaced by $OHC$ edges for constructing the $OP^i$s containing more and more vertices. If $i$ is small, some ordinary edges $g$ may have certain average frequency $f(g) \geq \frac{1}{2}{{i}\choose{2}}$. Most of the ordinary edges are contained in the $OP^n$s not in $OHC$. However, they will have the frequency $ f(g) < 2(n-3)$ in the frequency $K_n$. It means that $p_i(g\in OP^i)$ and $f(g)$ of any $g$ increase slower or decrease faster than $p_i(e\in OP^i)$ and $f(e)$ of any $e$ according to $i\in[4,n]$ in the whole process. As $p_{i+1}(e\in OP^{i+1}) > \left[1 - \frac{2}{i(i-1)}\right]p_i(e\in OP^i)$ or $p_{i+1}(e\in OP^{i+1}) \geq \left[1 + \frac{2}{i(i-1)}\right]p_i(e\in OP^i)$ exists  from $i$ to $i+1$,  $p_{i+1}(g\in OP^{i+1}) \leq \left[1 - \frac{2}{i(i-1)}\right]p_i(g\in OP^i)$ will occur accordingly and $p_{i+1}(g\in OP^{i+1}) \leq \frac{i}{i+1}p_i(g\in OP^i)$ will appear in most cases. 

In the next, we shall analyze the frequency sum for two adjacent $OHC$ edges, such as $e_1$ and $e_2$ containing a vertex $v$, in $K_n$ based on the frequency $K_i$s. $v$ is contained in $n-1$ edges, and there are ${{n}\choose{2}}$ $OP^n$s visiting $v$. As the frequency of each of the $n-1$ edges is enumerated from all $OP^n$s, the total frequency related to the $n-1$ edges is $(n-1)^2$. We takes one $OP^4$ from $OHC$ where $v$ is one intermediate vertex. One $K_4$ is built on the four vertices in the $OP^4$. The $K_4$ contains six $OP^4$s, see Appendix $B$ from paper  \cite{DBLP:journals/Wang24}. As an intermediate vertex in the $OP^4$s, $e_1$ and $e_2$ have the frequency pair 3, 5 or 1, 5 according to the six $OP^4$s and frequency $K_4$. The frequency sum $f(e_1) + f(e_2) = $ 6 or 8 exists according to the six $OP^4$s. Moreover, 6 occurs four times and 8 occurs eight times with respect to the six $OP^4$s. Since there are 12 pairs of adjacent edges in the six $OP^4$s, the probabilities of 6 and 8 are $\frac{1}{3}$ and $\frac{2}{3}$, respectively. The expected value of $f(e_1) + f(e_2)$ is $\frac{22}{3}$, and the average frequency for one $OHC$ edge is $\frac{11}{3} > 3$. 

It is known that an $OHC$ edge in $K_n$ is also the $OHC$ edge in every  $K_4$ containing it on average. In the $n-3$ frequency $K_4$s containing $e_1$ and $e_2$, the frequency pair is 3, 5 as well as $f(e_1) + f(e_2) = $ 8 rather than 6. One sees that $\frac{4(i-1)^2}{5} < 8$ exists for $i = 4$. Theorem \ref{th1} works for $f(e_1) + f(e_2)$ if $e_1$ and $e_2$ are contained in the same frequency $K_4$s. Based on Theorem \ref{th3}, one can prove $f(e_1) + f(e_2) \geq \frac{4(i-1)^2}{5}$ as $e_1$ and $e_2$ are contained in the same frequency $K_i$s. 
%Thus, an $OHC$ edge will be contained in more than 3 $OP^4$s in each $K_4$ containing it. 
In each of the other frequency $K_4$s containing $e_1$ and $e_2$, we assume that the frequency sums 6 and 8 are distributed for them based on the six $OP^4$s in Figure \ref{optimalpaths}. In this case, the lower bound of the expected frequency sum  can be derived as $\frac{22}{3} + \frac{4}{3(n-2)}$. Based on Theorem \ref{th01}, an $OHC$ edge is contained in the $OP^i$s with the bigger probability than an ordinary edge. Thus, $f(e_1) + f(e_2) \geq \frac{22}{3} + \frac{4}{3(n-2)}$ holds. The average frequency of an $OHC$ edge is bigger than $\frac{11}{3}$. Under the constrains of the $OP^4$s, the lower bound of the average frequency for an $OHC$ edge is improved from 3 to $\frac{11}{3}$ \cite{DBLP:journals/Wang25}. Thus, the total frequency for two adjacent $OHC$ edges is bigger than $\frac{22}{3}{{n-2}\choose{2}} + \frac{2(n-3)}{3}$ based on the frequency $K_4$s. Since $p_4(e\in OP^4)$ increases according to $n$ based on Theorem \ref{th3}, the frequency sum for two adjacent $OHC$ edges also rises accordingly.   

$e_1$ and $e_2$ are contained in ${{n-3}\choose{i-3}}$ frequency $K_i$s. In each of the frequency $K_i$s, $f(e_1) + f(e_2) \geq \frac{4(i-1)^2}{5}$ holds based on Theorem \ref{th01}. In each of the other frequency $K_i$s, $f(e_1)\geq \frac{1}{2}{{i}\choose{2}}$ and $f(e_2) \geq \frac{1}{2}{{i}\choose{2}}$ hold based on Theorem \ref{th1}. In the worst case, $f(e_1) + f(e_2) = \frac{4(i-1)^2}{5}$ is used according to the ${{n-3}\choose{i-3}}$ frequency $K_i$s, and $f(e_1) = f(e_2) = \frac{1}{2}{{i}\choose{2}}$ is considered in each of the other frequency $K_i$s. The expected value of $f(e_1) + f(e_2) = {{i}\choose{2}} + \frac{(i-2)(3i^2-11i+8)}{10(n-2)}$ can be computed. The probability sum $p_i(e_1\in OP^i) + p_i(e_2\in OP^i) = 1 + \frac{(i-2)(3i^2-11i+8)}{5(n-2)i(i-1)} > 1$ can be derived. As $i$ rises if $n$ is fixed, $p_i(e_1\in OP^i) + p_i(e_2\in OP^i)$ increases accordingly even if for the worst case. It means that $p_i(e_1\in OP^i)$ and $p_i(e_2\in OP^i)$ seldom decrease from $i$ to $i+1$ at the same time. If $p_i(e_1\in OP^i)$ or $p_i(e_2\in OP^i)$ has a slight decrement from $i$ to $i+1$, the other one must have some bigger increment. If $f(e_1) = f(e_2) = \frac{7}{18}{{i}\choose{2}}$ is considered with respect to the other frequency $K_i$s, the same conclusion can be drawn. It indicates that two adjacent $OHC$ edges must be contained in more percentage of the $OP^i$s according to $i\in[4,n]$. Moreover, under the frequency constraints between $OHC$ edges and ordinary edges in frequency $K_i$s(or the inclusive restrictions among $OP^i$s ($i\in [4,n]$)), the consecutive $k\in [2,n]$ $OHC$ edges will be contained in more percentage of $OP^i$s according to $i$, see the experiments in the paper \cite{DBLP:journals/Wang25}.
%Since the total frequency of the three edges containing $v$ is 9 with respect to the frequency $K_4$, the frequency sum for the two $OHC$ edges occupies at least two thirds of the total frequency 9. $K_n$ is expanded from the $K_4$ by adding the other vertices one by one. 

Based on Theorem \ref{th3}, $e_1$ and $e_2$ will be contained in more percentage of the $OP^i$s according to $i$. Since $f(e_1) + f(e_2) \geq \frac{22}{3}$ exists based on frequency $K_4$s, $f(e_1) + f(e_2)\geq \frac{11}{9}{{i}\choose{2}}$ holds based on frequency $K_i$s. %Given one $K_i$ containing the two $OHC$ edges, the frequency sum for the two edges will be bigger than $\frac{11}{9}{{i}\choose{2}}$. 
As $i=n$, $f(e_1) + f(e_2) \geq \frac{11}{9}{{n}\choose{2}}$ holds in the frequency $K_n$. The $n-3$ ordinary edges containing $v$ have the frequency sum smaller than $\frac{(n-1)(7n-18)}{18}$. In the average case, $f(e_1) + f(e_2) \geq n^2-4n+7$ holds in the frequency $K_n$, an ordinary edge will have the frequency smaller than 2 on average. %as $n$ is big enough,?????????? and an $OHC$ edge will have the frequency approaching ${{n-1}\choose{2}} + 1$ on average. 
It implies that most ordinary edges will have the frequency of 1 and zero in the frequency $K_n$. %If each ordinary edge has the frequency of $\frac{n-3}{2} < \frac{n+2}{2}$, the number of ordinary edges containing $v$ is smaller than $[\frac{4(7n-18)(n-1)^2}{(n-3)(11n^2-19n-36)}]\approx [\frac{28}{11}]$ as $n$ is big enough. %It implies that only part of edges are contained in $OP^n$s in frequency $K_n$. %It is known that the frequency of an ordinary edge is smaller than $\frac{1}{2}{{n}\choose{2}}$ in frequency $K_n$. %As the edges not in $OHC$ have a frequency much smaller than $n$, more number of edges containing $v$ will be preserved to the graph. If some edges not in $OHC$ have a big frequency approaching $n$, the smaller number of edges will be preserved. 
%If the preserved edges containing $v$ have the equal frequency (each of the preserved edges not in $OHC$ is included in the $OP^n$s with the equal probability), the frequency of an preserved edge will be bigger  than $\frac{(n-1)^2}{3(n-3)}$. they will have the euqal frequency computed with the $OP^n$s. 
%On average, as the vertex degrees from 3 to $n-1$ are uniformly distributed to the $n$ vertices, respectively, there will be  $\frac{n(n+2)}{4}$ edges in the graph. %$\frac{(n-3)(n+2)}{4}$

The lower frequency bounds for an $OHC$ edge, two adjacent $OHC$ edges and average frequency for all $OHC$ edges are summarized in Table \ref{fbkn}. It mentions that the smallest frequency of $OHC$ edges increases according to $n$. For very small $TSP$, the smallest frequency for an $OHC$ edge may approach $\frac{1}{2}{{i}\choose{2}}$ or $\frac{7}{18}{{i}\choose{2}}$. For the big and large $TSP$, $\frac{1}{2}{{i}\choose{2}}$ will be too small as the lower frequency bound. In real-world applications, the smallest frequency sum for two adjacent $OHC$ edges is also much bigger than $\frac{11}{9}{{i}\choose{2}}$ for big and large $TSP$ \cite{DBLP:journals/Wang25}. 

\begin{table}
	\begin{center}
		\caption{The lower frequency bounds for $e\in OHC$  and upper frequency bounds for $g\notin OHC$ in $K_n$ based on frequency $K_i$s (adj. = adjacent and Avg. = average value).}
		{\footnotesize \begin{tabular}{ p{1.5cm}  p{2cm}  p{2cm}  p{2.5cm}  }
				\hline
				%& LB for & $e\in OHC$  &  &  & UB for & g$\notin OHC$   \\				
				%&  & &  &  & Ordinary &  edge & $g$  \\
				& 1 $e$ or $g$ & 2 adj. $e$s & Avg. of $e$s or $g$s    \\
				\hline
				$e\in OHC$ & $\frac{1}{2}{{i}\choose{2}}$ or $\frac{7}{18}{{i}\choose{2}}$& $\frac{11}{9}{{i}\choose{2}}$ & $\frac{i^2-4i+7}{2}$   \\	
				$g\notin OHC$ & --- & ---  & $i-1$    \\														
				%\bottomrule
				\hline
		\end{tabular}}
		\label{fbkn}
	\end{center}
\end{table}

\section{The frequency and probability changes for edges according to $i$}
\label{sec5}
Given $K_n$, the frequency and probability changes for the three types of edges, i.e., $OHC$ edges, the edges contained in $OP^n$s but not in $OHC$, and the  edges excluding from any one $OP^n$, will be analyzed according to $i$ as the frequency of each edge is computed with the frequency $K_i$s. It is known that an $OHC$ edge will have the average frequency  $f(e) \geq \frac{1}{2}{{i}\choose{2}}$ based on the  frequency $K_i$s. For the ordinary edges contained in $OP^n$s yet not in $OHC$, they may have the average frequency  $f(g) \geq \frac{1}{2}{{i}\choose{2}}$ as $i$ is much smaller than $n$. If $i$ becomes big, $f(g) < \frac{1}{2}{{i}\choose{2}}$ will appear, and $f(g) < 2(n-3)$ exists in the frequency $K_n$. For the other ordinary edges, they usually have the average frequency  $f(g) < \frac{1}{2}{{i}\choose{2}}$, and $f(g) = 0$ exists in the frequency $K_n$. The probability that an edge is contained in the $OP^i$s plays an important role to determine the  (average) frequency computed with the frequency $K_i$s. %Similarily, the probability that an edge is contained in $OP^i$s also represents the optimality of the edge with respect to $OHC$. 
As the frequency of each edge is computed with the frequency $K_i$s, the average frequency for all edges is $(i-1){{n-2}\choose{i-2}}$. Since an edge is contained in ${{n-2}\choose{i-2}}$ $K_i$s, and each $K_i$ contains ${i}\choose{2}$ $OP^i$s, the average probability that an edge is contained in an $OP^i$ is computed as $\frac{2}{i}$. This indicates that the average probability that an edge is contained in the $OP^i$s monotonously decreases according to $i$ on average. %For most ordinary edges, they will be contained in a smaller and smaller percentage of $OP^i$s according to $i$. 
However, $OHC$ edges are contained in the nearly equal or more percentage of $OP^i$s according to $i$ based on Theorems \ref{th1} and \ref{th3}. For ordinary edges out of $OHC$, the change of probability that they are contained in the $OP^i$s follows Theorem \ref{th33}.

\begin{theorem}
	\label{th33}
	%Given a vertex $v$ in $K_n$, it is contained in two $OHC$ edges and $n-3$ edges excluding from $OHC$. For the $n-3$ edges not in $OHC$, the probability that they are contained in the $OP^i$s monotonously decreases according to $i\in [4,n] $ in the average case. 
	For an ordinary edge $g\notin OHC$ in $K_n$, the probability $p_i(g\in OP^i)$ that it is contained in the $OP^i$s monotonously decreases according to $i\in [4,n] $ in the average case. Moreover,  $p_{i+1}(g\in OP^{i+1}) \leq \frac{i}{i+1}p_i(g\in OP^i)$ exists in most cases as $p_i(g\in OP^i)$ decreases from $i$ to $i+1\leq n$. 
\end{theorem}

\begin{proof}
	Without loss of generality, the $n-1$ edges containing a vertex $v$ in $K_n$ are taken as an example for analysis. The $n-1$ edges include two $OHC$ edges, for example $e_1$ and $e_2$, and $n-3$ ordinary edges $g_j$ ($j\in [1,n-3]$). The change of the probability that the $n-3$ ordinary edges are contained in the $OP^i$s will be predicted according to $i\in [4,n]$. 
	
	Each edge is contained in ${{n-2}\choose{i-2}}$ $K_i$s, and each $K_i$ includes ${{i}\choose{2}}$ $OP^i$s. Given all $K_i$s containing an edge, there are  ${{i}\choose{2}}{{n-2}\choose{i-2}}$ $OP^i$s. Let $p_i(e\in OP^i)$ and $p_i(g\in OP^i)$ denote the probability that an edge $e\in OHC$ and $g\notin OHC$ is contained in such an $OP^i$, respectively. The frequency of an edge $e$ or $g$ is computed as $F_i(e) = {{i}\choose{2}}{{n-2}\choose{i-2}}p_i(e\in OP^i)$ or $F_i(g) = {{i}\choose{2}}{{n-2}\choose{i-2}}p_i(g\in OP^i)$. For the $n-1$ edges containing $v$, the total frequency is $F_{tot} = \sum_{j=1}^{2}{{i}\choose{2}}{{n-2}\choose{i-2}}p_i(e_j\in OP^i) +  \sum_{j=1}^{n-3}{{i}\choose{2}}{{n-2}\choose{i-2}}p_i(g_j\in OP^i)$. There are ${{n-1}\choose{i-1}}$ $K_i$s containing $v$ and the associated edges. In each frequency $K_i$, the total frequency of the $i-1$ edges containing $v$ is $(i-1)^2$. Thus, the equation (\ref{F3}) is derived, and the equality $\sum_{j=1}^{n-3}p_i(g_j\in OP^i) = \frac{2(n-1)}{i}$ holds. Since the two $OHC$ edges $e_1$ and $e_2$ have the probability sum $p_i(e_1\in OP^i) + p_i(e_2\in OP^i)\geq \frac{11}{9}$, and it increases according to $i$, the $n-3$ ordinary edges have the probability sum $\sum_{j=1}^{n-3}p_i(g_j\in OP^i) \leq \frac{2(n-1)}{i} - \frac{11}{9}$. It mentions that $p_i(g_j\in OP^i) = 0$ will appear for some ordinary edges $g_j$ as $i$ is relatively big. The reason is that only $\frac{2}{i}$ of the $OP^i$s are contained in the $OP^{i+1}$s in each $K_{i+1}$. Most $OP^i$s are neglected as well as the included ordinary edges. Thus, the number of ordinary edges with the average frequency bigger than zero will be smaller than $n-3$ for big $i$. 
	\begin{eqnarray}\label{F3}
	\left[\sum_{j=1}^{2}p_i(e_j\in OP^i) + \sum_{j=1}^{n-3}p_i(g_j\in OP^i)\right]{{i}\choose{2}}{{n-2}\choose{i-2}} = (i-1)^2{{n-1}\choose{i-1}}, \nonumber \\
	\sum_{j=1}^{2}p_i(e_j\in OP^i) + \sum_{j=1}^{n-3}p_i(g_j\in OP^i) = \frac{2(n-1)}{i}, \nonumber \\   \sum_{j=1}^{n-3}p_i(g_j\in OP^i) \leq \frac{2(n-1)}{i} -  \frac{11}{9}
	\end{eqnarray}	
	As $i$ is small, the $n-3$ ordinary edges have one big probability sum. It indicates that an ordinary edge is contained in the relatively big percentage of the $OP^i$s on average. %${{n-2}\choose{i-2}}$ is much smaller than ${{n-1}\choose{i-1}}$, and each $K_i$ contains a small number of edges. %It implies that the $K_i$s containing an edge are a small part of the total $K_i$s containing $v$. Each edge has a relatively big probability with respect to the small percentage of $K_i$s containing them, respectively. The $n-3$ ordinary edges have a big probability sum as $i$ is small. 
	%As $i$ rises, ${{n-2}\choose{i-2}}$ is nearer and nearer to ${{n-1}\choose{i-1}}$. It means that the number of $K_i$s containing an edge tends to that of $K_i$s containing $v$. In this case, the $K_i$s containing an edge occupy a big percentage of the total $K_i$s containing $v$.
	As $i$ rises, the probability sum of the $n-3$ ordinary edges becomes smaller. %
	%It indicates that each of the $n-3$ ordinary edges has a relatively smaller probability with respect to the bigger percentage of $K_i$s containing them, respectively. 
	As $i=n$, $\sum_{j=1}^{n-3}p_i(g_j\in OP^i)$ reaches the minimum value $\leq \frac{7}{9} - \frac{2}{n}$. It implies that an ordinary edge is contained in the smaller percentage of the $OP^i$s according to $i$ on average. %It mentions that the $p_i(e_{n-1}\in OP^i)$ and $p_i(e_{n-2}\in OP^i)$ of the two $OHC$ edges are increasing according to $i$ based on Theorem \ref{th3}. 
	As $i = n$, $p_n(e_1\in OP^n)+p_n(e_2\in OP^n)\geq \frac{8(n-1)}{5n}$ or $\frac{2(n^2-4n+7)}{n(n-1)}$ will hold based on Theorems \ref{th01} and \ref{th22}. Thus, $\sum_{j=1}^{n-3}p_i(g_j\in OP^i)$ for the $n-3$ ordinary edges will be very small as $i$ is big.
	
	Based on Theorems \ref{th1} and \ref{th3}, $p_i(e_1\in OP^i) + p_i(e_2\in OP^i)$ increases according to $i$ in the worst case. Because $p_{i+1}(e_1\in OP^{i+1}) + p_{i+1}(e_2\in OP^{i+1}) > p_i(e_1\in OP^i) + p_i(e_2\in OP^i)$ from $i$ to $i+1$, the probability decrement $\sum_{j=1}^{n-3}p_i(g_j\in OP^i)-\sum_{j=1}^{n-3}p_{i+1}(g_j\in OP^{i+1}) = \frac{2(n-1)}{i(i+1)} + p_{i+1}(e_1\in OP^{i+1}) + p_{i+1}(e_2\in OP^{i+1}) - p_i(e_1\in OP^i)  - p_i(e_2\in OP^i) > 0$ is computed. On average, the probability decrement for an ordinary edge $g = g_j$ ($j\in [1,n-3]$) is $pd_i(g) = p_i(g\in OP^i) - p_{i+1}(g\in OP^{i+1}) = \frac{2(n-1)}{i(i+1)(n-3)} + \frac{p_{i+1}(e_1\in OP^{i+1}) + p_{i+1}(e_2\in OP^{i+1})}{n-3} - \frac{ p_i(e_1\in OP^i) + p_i(e_2\in OP^i)}{n-3} > 0$. In the average case, the probability sum  $p_i(e_1\in OP^i) + p_i(e_2\in OP^i) \geq  \frac{2(i^2-4i+7)}{i(i-1)}$ holds according to $i$. In this case, $p_i(e_1\in OP^i)+p_i(e_2\in OP^i)$ increases faster and faster according to $i$ until it is close to the biggest value. On the other hand, the probability decrement $pd_i(g)$ will become bigger according to $i$. Meanwhile, $p_i(g\in OP^i)$ will decrease faster and faster until it is near the smallest value.
	
	An ordinary edge $g=g_j$ ($j\in [1,n-3]$) and two pairs of the $OHC$ edges are adjacent on both endpoints, see Figure \ref{OHC}. $g$ and the two pairs of the $OHC$ edges are contained in $K = 2{{n-4}\choose{i-4}} - {{n-6}\choose{i-6}}$ $K_i$s. In each of the $K_i$s, $g$ will be an ordinary edge since an $OHC$ edge related to $K_n$ is the $OHC$ edge in the $K_i$s containing it based on Theorem \ref{th3}. In the best average case, $g$ will have the frequency $\frac{i+2}{2}$ in each of these frequency $K_i$s. They occupy $r = \frac{2(i-2)(i-3)}{(n-2)(n-3)} - \frac{(i-2)(i-3)(i-4)(i-5)}{(n-2)(n-3)(n-4)(n-5)}\in[0,1]$ of the frequency $K_i$s containing $g$. As $i$ is small, $g$ and the two pairs of $OHC$ edges are contained in the small percentage of the frequency $K_i$s, and $g$ is just adjacent to either one $OHC$ edge or no $OHC$ edges on both endpoints in most $K_i$s. In this case, $g$ (and some other ordinary edges) will be the $OHC$ edge  and it has a big frequency in these $K_i$s. Thus, the probability $p_i(g\in OP^i)$ will be big based on the frequency $K_i$s, and the ordinary edges have the relatively big probability sum.
	
 	As $i$ rises, $r$ increases accordingly. $g$ and the two pairs of adjacent $OHC$ edges are contained in more and more percentage of the $K_i$s. As an ordinary edge, $g$ has the small frequency in these frequency $K_i$s. Moreover, $p_i(e_1\in OP^i) + p_i(e_2\in OP^i)$ increases according to $i$. %Since each $OHC$ edge has a fequency $\geq \frac{1}{2}{{i}\choose{2}}$ in a frequency $K_i$ containing it,$e_j$ will be an ordinary edge which has a small frequency in these frequency $K_i$s. 
	Thus, $p_i(g\in OP^i)$ will become smaller, and the probability sum of the ordinary edges monotonously decreases according to $i$, see formula (\ref{F3}). If $g$ has the frequency $f_2(g) = \frac{i+2}{2}$ in each of the $K$ frequency $K_i$s, and it has the maximum frequency $f_1(g) = {{i}\choose{2}} - 1$ in each of the other frequency $K_i$s, $p_i(g\in OP^i) = 1 - \left[1 - \frac{i+4}{i(i-1)}\right]r - \frac{2}{i(i-1)}$ is formulated. If $i$ is much smaller than $[\frac{n}{2}]$, $p_i(g\in OP^i)$ increases according to $i$. Because $r$ increases according to $i$, and $\frac{i+4}{i(i-1)}$ tends to zero as $i$ becomes big, $p_i(g\in OP^i)$ will decrease finally according to $i$. Based on the formula of $p_i(g\in OP^i)$, the maximum value is obtained at $i=O(n^{\frac{4}{7}})$ much smaller than $[\frac{n}{2}]$.  After that, $p_i(g\in OP^i)$ decreases from this biggest value according to $i$ until $i=n$. Moreover, once $p_i(g\in OP^i)$ decreases, the probability decrement $pd_i(g)$ will be bigger than zero according to $i$ until $p_i(g\in OP^i)$ becomes one small value or zero. %In the best case, $p_i(g\in OP^i)$ reaches the maximum value at $i=O(n^{\frac{4}{7}})$. 
	Based on dynamic programming, the $OP^i$s are computed to see the change of $p_i(g\in OP^i)$ for ordinary edges, and an ordinary edge can be identified in $O(n^{\frac{16}{7}}2^{n^{\frac{4}{7}}})$ time for this case.
	
	For $n=1000$, the $p_i(g\in OP^i)$ and probability decrement $pd_i(g) = p_i(g\in OP^i) - p_{i+1}(g\in OP^{i+1}) >0$ is computed according to $i$ and shown in Figure \ref{pdi}. In Figure \ref{pdi}, $p_i(g\in OP^i)$ increases quickly as $i$ is small, and it reaches the biggest value close to 1 at $i=33$ which is much smaller than $[\frac{n}{2}] = 500$. In the interval $i\in[4,32]$, $pd_i(g) < 0$ appears from $i$ to $i+1$. If $i \geq 33$, $p_i(g\in OP^i)$ decreases according to $i\in [33,1000]$. In this case, $pd_i(g) > 0$ exists. One sees the interval where $p_i(g\in OP^i)$ decreases is much bigger than that where it increases. Moreover, $pd_i(g)$ increases from $i=33$ to 589, and then it decreases according to $i\in[589,1000]$ yet $pd_i(g) \geq 0$. It indicates that $p_i(g\in OP^i)$ decreases faster and faster according to $i\in[33,589]$ which occupies more than half of the total interval $[4,1000]$. Since $\frac{2p_i(g\in OP^i)}{i(i-1)}$ decreases quickly according to $i$ while $pd_i(g)$ increases in this stage, $pd_i(g) > \frac{2p_i(g\in OP^i)}{i(i-1)}$ holds from $i\geq 32$ to $i+1\leq 590$ in most cases. In the experiments, $pd_i(g) > \frac{2p_i(g\in OP^i)}{i(i-1)}$ exists if $i \geq 81$. Moreover, $pd_i(g) > \frac{p_i(g\in OP^i)}{i+1}$ happens from $i=448$ to 999. Once $p_i(g\in OP^i)$ decreases from certain $i_d$, it will decrease in proportion to a bigger factor according to $i > i_d$ until $p_{i+1}(g\in OP^{i+1}) < \frac{i}{i+1}p_i(g\in OP^i)$ appears. %We assume $p_{i+1}(g\in OP^{i+1}) < \frac{i}{i+1}p_i(g\in OP^i)$ from some $i$. Because $p_n(g\in OP^n) < \frac{2(n-3)}{{{n}\choose{2}}}$, $p_i(g\in OP^i) < \frac{4(n-3)}{i(n-1)}$ is derived. 
	Even if in the best case for $g$, the change of $p_i(g\in OP^i)$ still conforms to Theorem \ref{th3} and formula (\ref{F3}). 
	
	The average of $pd_i(g)$s from $i=33$ to 589 is 0.001022 which is bigger than $\frac{2}{45\times 44}$ at $i=45$. It indicates that the  probability decrement $pd_i(g) > \frac{2p_i(g\in OP^i)}{i(i-1)}$ happens in most cases if $pd_i(g)$ increases. At $i = 545$, the decrement of $p_i(g\in OP^i)$ is bigger than 0.5, and $p_i(g\in OP^i) \leq \frac{1}{2}$ exists as $i > 545$. From $i=590$ to 1000, $pd_i(g)$ decreases according to $i$, and the average of the $pd_i(g)$s is 0.001039 which is bigger that that in the increasing stage. Even if $pd_i(g)$ decreases from the biggest value, it is still bigger than $\frac{p_i(g\in OP^i)}{i+1}$ according to $i\in [589,1000]$ in most cases unless it approaches zero.  On average, $p_i(e_1\in OP^i)$ and $p_i(e_2\in OP^i)$ will increase faster and faster according to $i\in [33, 589]$, and the probability increment bigger than $\frac{2p_i(e\in OP^i)}{i(i-1)}$ will occur in most cases.

	\begin{figure}
		\centering
		\includegraphics[width=3in,bb=0 0 360 200]{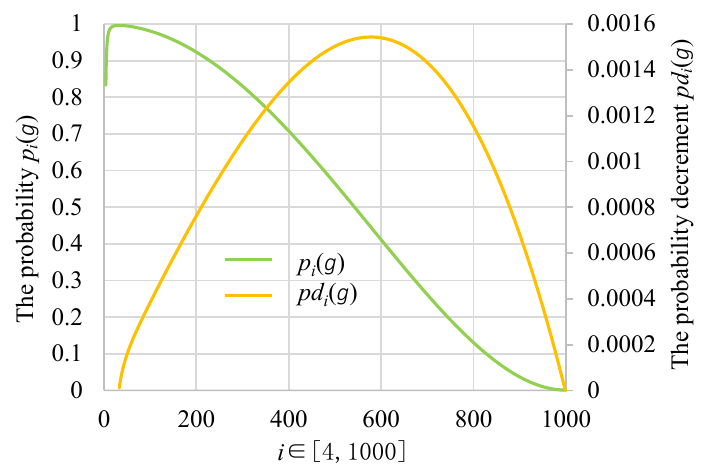}
		\caption{The changes of $p_i(g\in OP^i)=1-\left[1-\frac{i+4}{i(i-1)}\right]r - \frac{2}{i(i-1)}$  and $pd_i(g)=p_i(g\in OP^i) - p_{i+1}(g\in OP^{i+1}) > 0$ according to $i\in [4, 1000]$ for $n=1000$.} 
		\label{pdi}
	\end{figure}	
	
	In fact, $f_2(g) = \frac{i+2}{2}$ does not occur in each of the $K$ frequency $K_i$s and the maximum frequency  $f_1(g) = {{i}\choose{2}} - 1$ seldom exists in each of the other frequency $K_i$s. On average, $g$ will have the average frequency $\leq \frac{i+2}{2}$ according to the $K$ frequency $K_i$s. In each of the other frequency $K_i$s, $g$ is the $OHC$ edge in the best case. If  $g$ has the expected frequency of an $OHC$ edge according to the ${{n-2}\choose{i-2}} - K$ frequency $K_i$s, such as $f_1(g)=\frac{i^2-4i+7}{2}$, and it has the best average frequency $f_2(g) = \frac{i+2}{2}$ according to the $K$ frequency $K_i$s, the average frequency of $g$ will be $f(g) = (1-r)f_1(g) + rf_2(g) = f_1(g) - \frac{(i^2-5i+5)r}{2} < f_1(g)$ as well as $p_i(g\in OP^i) = \frac{i^2-4i+7}{i(i-1)} - \frac{(i^2-5i+5)r}{i(i-1)}$. It indicates that the average frequency and probability of $g$ increases slower than those of an $OHC$ edge according to $i$ in the average case if $f(g)$ and $p_i(g\in OP^i)$ rises. As $i$ is small, $r$ is also very small, and $p_i(g\in OP^i)$ is not affected much. In this case, $p_i(g\in OP^i)$ will increase according to $i$. If $i$ is relatively big, the probability increment will be obviously smaller than $\frac{2p_i(g\in OP^i)}{i(i-1)}$ from $i$ to $i+1$, see Theorems \ref{th1} and \ref{th3}. Moreover,  $p_i(g\in OP^i)$ will decrease according to $i$ after the biggest value. Once $p_i(g\in OP^i)$ becomes smaller according to $i$, $p_{i+1}(g\in OP^{i+1}) \leq \left[1 - \frac{2}{i(i-1)}\right]p_i(g\in OP^i)$ will exist and $p_{i+1}(g\in OP^{i+1}) \leq \frac{i}{i+1}p_i(g\in OP^i)$ will happen in most cases. 
	
	The numerical simulations illustrated that $p_i(g\in OP^i)$ decreases faster and faster according to $i$ in one big interval, see Figure \ref{pdi}. %Although $p_i(g\in OP^i)$ increases according to $i$ in some stages, it will has the smaller increment than  %smaller than that of an $OHC$ edge since it has the  average frequency smaller than $\frac{i+2}{2}$. $g$ will have the frequency $ < \frac{(i-1)(i-2)}{2} + 1$ (or $\frac{i^2-4i+7}{2}$) as $i$ is big enough. If it is an $OHC$ edge in each of the other frequency $K_i$s, $g$ will have the frequency $ < \frac{(i-1)(i-2)}{2} + 1$ (or $\frac{i^2-4i+7}{2}$) as $i$ is big enough. 
	Otherwise, if $p_i(g\in OP^i)$  increases faster or decreases slower than that of an $OHC$ edge according to $i$, it will have the frequency $f(g) > \frac{1}{2}{{n}\choose{2}}$ in the frequency $K_n$. It is contradict to that there are $n$ edges with the frequency  bigger than $\frac{1}{2}{{n}\choose{2}}$ in any frequency $K_n$, see Theorem \ref{th2}. %be the $OHC$ edge in $K_n$ or the average frequency of all $OHC$ edges will be smaller than $\frac{i^2-4i+7}{2}$. 
	If $f_1(g) < {{i}\choose{2}}  - 1$ and $f_2(g) < \frac{i+2}{2}$ exists for $g$, $p_i(g\in OP^i) < 1 - \left[1 - \frac{i+4}{i(i-1)}\right]r - \frac{2}{i(i-1)}$ holds according to $i$, and the biggest value of $p_i(g\in OP^i)$ becomes smaller than that shown in Figure \ref{pdi}. For each of the other ordinary edges, the change of $p_i(g)$ has the similar trend although they have the smaller biggest value, and decrease more quickly according to $i$ once they become smaller from $i$ to $i+1$. In fact, one can build the probability function $p_i(g\in OP^i) = \frac{ai^2 + bi +c}{i(i-1)}$ for any $g$ based on  $p_i(g\in OP^i)$, $p_{i+1}(g\in OP^{i+1})$ and $p_{i+2}(g\in OP^{i+2})$ computed with the frequency $K_i$s, $K_{i+1}$s and $K_{i+2}$s for small $i$s where $a$, $b$ and $c$ are parameters. Given an ordinary edge $g$, the actual $p_i(g\in OP^i)$ will be smaller than that computed based on the probability function as $i$ is big. 
	%On average, $p_i(e_j\in OP^i)$ for an edge $e_j$ decreases according to $i$. 
	In addition, the average probability $\bar{p}_i(g\in OP^i) \leq \frac{2(n-1)}{i(n-3)} - \frac{11}{9(n-3)}$ for the $n-3$ ordinary edges approaches zero as $i$ is big. $p_i(g\in OP^i) = 0$ will appear for most ordinary edges. 
	
	In the average case, $p_i(g\in OP^i)$ decreases in proportion to the factor $\frac{i}{i+1}$ from $i$ to $i+1$, see Theorem \ref{th1}. The equation (\ref{F3}) also indicates that the probability sum for the $n-3$ ordinary edges at $i+1$ becomes smaller than $\frac{i}{i+1}$ of that at $i$. As $p_i(g\in OP^i)$ has certain non-negligible decrement $\delta_i$, such as  $\delta_i\geq p_{i-1}(g\in OP^{i-1}) - p_i(g\in OP^i) = \frac{1}{i(i-1)}\times \frac{2(n-1)}{n-3}$, at certain number $i$, it will never have the equal increment at one bigger number  $k\leq n$. The obvious reason is that the probability sum for the ordinary edges becomes much smaller according to the bigger number $k$. The intrinsic reason is that $g$ and the two pairs of adjacent $OHC$ edges are contained in more percentage of the frequency $K_i$s where it must have the small frequency. Moreover, $p_i(e\in OP^i)$ for the adjacent $OHC$ edges and other $OHC$ edges $e\in OHC$  increases from $i$ to $k$, which also prevents $p_i(g\in OP^i)$ from increasing. Otherwise, all other ordinary edges will have the average frequency tending to zero at certain $i$ based on Theorems \ref{th2} and \ref{th22}. In this case, it is lack of ordinary edges for constructing the  ${{i}\choose{2}}$ $OP^i$s in most $K_i$s. %Provided that $e_j$ has the probability $p_i(e_j\in OP^i)$ and $p_{i+1}(e_j\in OP^{i+1})$ at $i$ and $i+1$, respectively. %Given a $K_{i+1}$ containing $e_j$, $e_j$ is contained in $(i-1){{i}\choose{2}}p_j(e_j\in OP^i)$ $OP^i$s in the $K_{i+1}$. The number of $OP^i$s visiting $v$ is $i{{i}\choose{2}}$ in the $K_{i+1}$. It mentions that the $OP^{i+1}$s in the $K_{i+1}$ only contains the $OP^i$s. It requires ${{i+1}\choose{2}}$ $OP^i$s containing $v$ to compute the $OP^{i+1}$s. As each $OP^i$ is contained in the $OP^{i+1}$s with the equal probability in the average case, the probability that an $OP^i$ is contained in the $OP^{i+1}$s is computed as $\frac{{{i+1}\choose{2}}}{i{{i}\choose{2}}} = \frac{i+1}{i(i-1)}$. As there are $(i-1){{i}\choose{2}}p_j(e_j\in OP^i)$ $OP^i$s containing $e_j$, the number of such $OP^i$s contained in the $OP^{i+1}$s is $\frac{(i-1)(i+1)p_j(e_j\in OP^i)}{2}$. Since there are ${{i+1}\choose{2}}$ $OP^{i+1}$s in the $K_{i+1}$, $p_j'(e_j\in OP^{i+1})=\frac{(i-1)p_j(e_j\in OP^i)}{i}$ is computed. 
	%Based on Theorem \ref{th1},  $p_{i+1}(e_j\in OP^{i+1}) = \frac{i}{i+1}p_i(e_j\in OP^i)$ demonstates that $p_{i+1}(e_j\in OP^{i+1})$ for $e_j$ according to a $K_{i+1}$ is $\frac{i}{i+1}$ of $p_i(e_j\in OP^i)$ accordind to a $K_i$ in the $K_{i+1}$. Thus, $p_i(e_j\in OP^i)$ for $e_j\notin OHC$ generally decreases according to $i$. 
	Theorem \ref{th1} illustrates that $p_{i+1}(g\in OP^{i+1}) < p_i(g\in OP^i)$ appears for ordinary edges  as $p_i(e\in OP^i)$ for $OHC$ edges rises from $i$ to $i+1$. Figure \ref{pdi} illustrates that the probability decrement $pd_i(g)$ becomes bigger according to $i$ in more than half of the total interval $[4,n]$. Moreover, $pd_i(g) > \frac{p_i(g\in OP^i)}{i+1}$ appears from $i$ to $i+1$ in most cases as $p_i(g\in OP^i)$ decreases. 
	
	For the ordinary edges not in $OP^n$s, $p_i(g\in OP^i)$ decreases faster than that for the edges contained in $OP^n$s. As $i$ is relatively big, $p_i(g\in OP^i) = 0$ will appear for most ordinary edges. In this case, the number of ordinary edges with $p_i(g\in OP^i) > \delta$ (here $\delta > 0$ is a small number) will be much smaller than $n-3$. Moreover, most or all of the ordinary edges with $p_i(g\in OP^i) > \delta$ are contained in the $OP^n$s. %Since the ordinary edges contained in $OP^n$s have the $p_i(g\in OP^i) > \delta$ according to $i$, 
	It means that the number of ordinary edges with $p_i(g\in OP^i) > \delta$ does not change much from $i$ to $i+1$. Based on formula (\ref{F3}), $p_{i+1}(g\in OP^{i+1})\leq \frac{i}{i+1}p_i(g\in OP^i)$ can be derived for the ordinary edges contained in the $OP^n$s. %Since most or all of the ordinary edges are contained in the $OP^n$s, the probability change for them conforms to the inequality $p_{i+1}(g\in OP^{i+1})\leq \frac{i}{i+1}p_i(g\in OP^i)$ from $i$ to $i+1$. 
	Thus, $pd_i(g) \geq \frac{p_i(g\in OP^i)}{i+1}$ will exist for them from $i$ to $i+1$. Since $p_i(e\in OP^i) \geq \frac{1}{2}$ holds for  $e\in OHC$, $g$ with $p_i(g\in OP^i) \ll 0.5$ at the small numbers $i$ is not an $OHC$ edge because $p_k(g\in OP^k) > 0.5$ will never occur at the big numbers $k$. %In addition, if $i$ is big enough, $p_{i+1}(e_j\in OP^{i+1})\approx p_i(e_j\in OP^i)$ will hold since $\frac{i}{i+1}$ tends to 1. 
	
	As $i$ is relatively big, $p_i(g_j\in OP^i)$ for each edge  remains stable from $i$ to $i+1$ whether it decreases or increases based on Theorem \ref{th3}. Given two ordinary edges in $K_n$, as an edge has the bigger probability (or frequency) than the other edge at certain big number $i\leq n$, the probability (or frequency) of the edge will be bigger than that of the other edge at another bigger number $k (i < k \leq n)$. Thus, $p_{i+1}(g\in OP^{i+1}) < \frac{i}{i+1}p_i(g\in OP^i)$ occurs in most cases for every ordinary edge according to formula (\ref{F3}). 
\end{proof}

According to $i$, $p_i(e\in OP^i)$ for an edge $e\in OHC$ keeps stable or increases, whereas $p_i(g\in OP^i)$ for an ordinary edge $g$ decreases on average. If $p_i(e\in OP^i)$ or $p_i(g\in OP^i)$ decreases according to $i$, the probability decrement according to the small numbers $i$ is bigger than that according to the big numbers $i$. In general, the probability decrement $pd_i(g) \geq \frac{p_i(g\in OP^i)}{i+1}$ happens from $i$ to $i+1$ in most cases. For most ordinary edges, $p_i(g\in OP^i)$ will decrease quickly according to $i$, and $p_i(g\in OP^i) < \frac{1}{2}$ will appear at the small numbers $i$. For $e\in OHC$, $p_i(e\in OP^i)$ will decrease much slower from $i$ to $i+1$ due to $p_{i+1}(e\in OP^{i+1}) \geq \left[1-\frac{2}{i(i-1)}\right]p_i(e\in OP^i)$. Moreover, $p_i(e\in OP^i)$ will increase from $i$ to $i+1$ and the probability increment will be bigger than $\frac{2p_i(e\in OP^i)}{i(i-1)}$ in most cases.  $p_i(e\in OP^i)\geq \frac{1}{2}$ holds according to $i$, especially as $i$ is big. Even if $p_i(g\in OP^i)$ increases from $i$ to $i+1$, the probability increment will be smaller than $\frac{2p_i(g\in OP^i)}{i(i-1)}$. Once $p_i(g\in OP^i)$ decreases from some $i$ to $i+1$, it will becomes smaller according to $i$ and never retrieves any previous bigger value. 

As the probability $p_i(e\in OP^i)$ or $p_i(g\in OP^i)$ is computed based on frequency $K_i$s, the probability change for an $OHC$ edge $e$ and ordinary edge $g$ from $i$ to $i+1$ is given in Table \ref{pchange}. As $r_p \geq 1$, it means that $p_{i+1}(e\in OP^{i+})\geq p_i(e\in OP^i)$ and $p_{i+1}(g\in OP^{i+})\geq p_i(g\in OP^i)$ from $i$ to $i+1$. In this case, $p_{i+1}(e\in OP^{i+1})$ is bigger or smaller than $\left[1+\frac{2}{i(i-1)}\right]p_i(e\in OP^i)$ whereas $p_{i+1}(g\in OP^{i+1}) < \left[1+\frac{2}{i(i-1)}\right]p_i(g\in OP^i)$ exists. If $r_p < 1$, it means that $p_{i+1}(e\in OP^{i+}) < p_i(e\in OP^i)$ and $p_{i+1}(g\in OP^{i+})< p_i(g\in OP^i)$ from $i$ to $i+1$. $p_{i+1}(e\in OP^{i+1}) > \left[1-\frac{2}{i(i-1)}\right]p_i(e\in OP^i)$ holds whereas $p_{i+1}(g\in OP^{i+1}) \leq \left[1-\frac{2}{i(i-1)}\right]p_i(e\in OP^i)$ or $\frac{i}{i+1}p_i(g\in OP^i)$ exists.
\begin{table}
	\begin{center}
		\caption{The changes of $p_i$ from $i$ to $i+1 \leq n$ for $e\in OHC$ and $g\notin OHC$.}
		{\footnotesize \begin{tabular}{ p{1.5cm}  p{3.5cm}  p{3.5cm}  }
				\hline
				%& LB for & $e\in OHC$  &  &  & UB for & g$\notin OHC$   \\				
				%&  & &  &  & Ordinary &  edge & $g$  \\
				&  $r_p = \frac{p_{i+1}}{p_i} \geq 1$ & $r_p = \frac{p_{i+1}}{p_i} \leq 1$   \\
				\hline
				$e\in OHC$ & $r_p \geq$ or $< 1 + \frac{2}{i(i-1)}$ & $r_p \geq 1 - \frac{2}{i(i-1)}$  \\	
				$g\notin OHC$ & $r_p < 1 + \frac{2}{i(i-1)}$ & $r_p \leq 1 - \frac{2}{i(i-1)}$ or $\frac{i}{i+1}$   \\														
				%\bottomrule
				\hline
		\end{tabular}}
		\label{pchange}
	\end{center}
\end{table}

The average frequency of an edge relies on the probability that it is contained in the $OP^i$s. The total frequency is up to the probability and the number of $K_i$s containing it. According to $i\in[4,n]$, each edge $e_j$ $(j\in \left[1,{{n}\choose{2}}\right])$ will have one corresponding peak (average) frequency computed with the frequency $K_i$s at certain number $P_j \leq n$. If $i > P_j$, the frequency of $e_j$ will decrease according to $i$. As the frequency of each edge is computed with the frequency $K_i$s, the frequency changes for different types of edges are analyzed as follows. %Given two edges $e_1$ and $e_2$, they will have the peak frequencies at $P_1$ and $P_2$, respectively. If $P_1 < P_2$, the following Lemma \ref{th32} is given. 
%\begin{proof}
%An edge is contained in ${{n-2}\choose{i-2}}$ $K_i$s in $K_n$ for $i\geq 4$. Each $K_i$ contains ${{i}\choose{2}}$ $OP^i$s. 

At certain number $i\geq 4$, an edge is contained in ${{n-2}\choose{i-2}}$ $K_i$s, and there are total ${{i}\choose{2}}{{n-2}\choose{i-2}}$ $OP^i$s. As $i = P_0= \frac{n}{2}+2$ for even $n$ or $P_0 = \frac{n+1}{2} + 1$ for odd $n$, the number of the $OP^i$s reaches the maximum value. From $i=4$ to $P_0$, the number of $OP^i$s in the $K_i$s increases according to $i$. If an edge is contained in the $OP^i$s with the same probability according to $i$, the frequency of the edge increases in proportion to the factor $\frac{(i+1)(n-i)}{(i-1)^2} > 1$. It indicates that the frequency of the edge increases exponentially according to $i (\leq P_0)$. For an $OHC$ edge $e$, $p_i(e\in OP^i)$ keeps nearly equal or increases from $i$ to $i+1$. Thus, the frequency of each $OHC$ edge increases exponentially according to $i \leq P_0$, and each $OHC$ edge reaches its own peak frequency at $P_0$, respectively. 

As $i\geq P_0$, the number of $OP^i$s in  the $K_i$s containing an edge decreases in proportion to the factor  $\frac{(i+1)(n-i)}{(i-1)^2} < 1$ from $i$ to $i+1$. If an edge maintains the equal probability to be contained in the $OP^i$s, the frequency of the edge will decrease exponentially according to $i$. As $i = n-2$, the number of $K_{n-2}$s is equal to that of $K_4$s, i.e.,  ${{n-2}\choose{2}}={{n-2}\choose{n-4}}$. %The frequency of the edge computed with  frequency $K_{n-2}$s is bigger than that computed based on frequency $K_4$s since there are ${{n-2}\choose{2}}$ $OP^{n-2}$s in each $K_{n-2}$. 
Thus, the frequency of the edge computed with the frequency $K_{n-2}$s is  $\frac{(n-2)(n-3)}{12}$ times of that computed with the frequency $K_4$s. As $i = n-1$ and $n$, the number of $OP^i$s in the $K_i$s containing the edge decreases in proportion to the factors $\frac{2(n-1)}{(n-3)^2}$ and $\frac{n}{(n-2)^2}$, respectively. Finally, the frequency of the edge computed with all $OP^n$s becomes $\frac{n(n-1)}{6(n-2)(n-3)}$ times of that computed with the frequency $K_4$s. For big $n$, the frequency of the edge in the frequency $K_n$ will be $\frac{1}{6}$ times of that computed with the frequency $K_4$s. If the frequency of an edge in the frequency $K_n$ is much bigger than $\frac{1}{6}$ times of that computed with the frequency $K_4$s, the edge must be contained in more percentage of $OP^i$s as $i > 4$. Theorems \ref{th1} and \ref{th3} states that $OHC$ edges have such property. An $OHC$ edge has the lowest frequency $3{{n-2}\choose{2}}$ based on the frequency $K_4$s. As it maintains the same probability according to $i$, the frequency of the edge in the frequency $K_n$ is  $\frac{n(n-1)}{6(n-2)(n-3)}\times 3{{n-2}\choose{2}} =  \frac{1}{2}{{n}\choose{2}}$. It is the lower frequency bound for $OHC$ edges in the frequency $K_n$. 

For most ordinary edges $g$ excluding from any one $OP^n$, $p_i(g\in OP^i)$ will decrease in proportion to a factor smaller than $\frac{i}{i+1}$ from $i$ to $i+1 < P_0$, and the frequency $F(g)$ increases much slower than those of $OHC$ edges and the ordinary edges contained in the $OP^n$s. Once $\frac{p_{i+1}(g\in OP^{i+1})}{p_i(g\in OP^i)} < \frac{(i-1)^2}{(i+1)(n-i)}$ appears at certain number $i = P' < P_0$, $F(g)$ reaches the peak value at $P'$, and it becomes smaller after $i > P'$. Meanwhile, $p_i(g\in OP^i)$ will become smaller based on Theorem \ref{th33}, and $F(g)$ will decrease exponentially according to $i$ until it reaches zero. For the ordinary edges $g$ contained in the $OP^n$s, $p_i(g\in OP^i)$ may increase as $i$ is small, $F(g)$ will increase according to the small numbers $i$ since $\frac{(i+1)(n-i)}{(i-1)^2} \gg 1$ exists. As $i$ becomes big, $g$ will be contained in the big percentage of $K_i$s each of which includes one or two pairs of the adjacent $OHC$ edges, respectively. In this case, $g$ has the small frequency in each of these frequency $K_i$s, and $p_i(g\in OP^i)$ will become smaller according to $i$ until $\frac{p_{i+1}(g\in OP^{i+1})}{p_i(g\in OP^i)} < \frac{(i-1)^2}{(i+1)(n-i)}$ occurs at certain number $P \leq P_0$. $g$ has the peak frequency at $P$, and $F(g)$ decreases if $i>P$. 

%%For most ordinary edges in $K_n$, each has its peak frequency before $P_0$ as the frequency is computed with frequency $K_i$s containing them, respectively. As an ordinary edge $g$ has its peak frequency at a certain number $i < P_0$, it implies that  $p_i(g\in OP^i)$ has a non-neglible decrement from $i$ to $i+1$. From $i$ to $i+1$, the number of $OP^i$s in $K_i$s containing $g$ is still rising exponentially according to the factor $r_i=\frac{(i+1)(n-i)}{(i-1)^2} > 1$. However, the frequency $F_i(g)$ decreases from $i$ to $i+1$. It indicates that $p_{i+1}(g\in OP^{i+1})$ decreases exponentially according to a smaller factor $d_p =\frac{p_{i+1}(g\in OP^{i+1})}{p_i(g\in OP^i)} < 1/r_p$. From $i+1$, the percentage of $OP^i$s containing $g$ becomes much smaller. %For an ordinary edge not in any one $OP^n$, it will be contained in an even smaller percentage of $OP^i$s from $i+1$. 
%%In addition, as some other edges, such as $OHC$ edges, are contained in more and more percentage of the following $OP^i$s, this edge will be contained in smaller and smaller percentage of the $OP^i$s. As $i$ big enough, the frequency of the edge will decreases exponentially to a small value or zero. 

The frequency change for an edge depends on the probability that it is contained in the  $OP^i$s according to $i$. However, the probability change for an edge is different from that of frequency. If the probability does not decrease for an edge, the frequency will never decrease from $i=4$ to $P_0$. If the frequency of an edge decreases at certain number $P' < P_0$ or $P < P_0$, the probability must decrease before $P'$ or $P$, and finally has the non-negligible decrement at $P'$ or $P$. Thus, the probability decreases before the frequency for ordinary edges. 
Some ordinary edges $g$ in $OP^n$s  may have the big $p_i(g\in OP^i)$ as $i$ is small. However, $p_i(g\in OP^i)$ will decrease at some number and it has the non-negligible decrement according to a factor smaller than $\frac{(i-1)^2}{(i+1)(n-i)} < \frac{i}{i+1}$ at certain number $i=P$. After that, it becomes smaller according to $i > P$. As $p_i(g\in OP^i)$ decreases, $F(g)$ increases very slow or becomes smaller according to $i$. If each $g$ has the average frequency $f(g) < \frac{1}{2}{{i}\choose{2}}$ at certain number $i$, $OHC$ edges and ordinary edges can be separated from each other at this number.  %Finally, the frequency of $g$ becomes smaller than $\frac{1}{2}{{n}\choose{2}}$ in frequency $K_n$. %Based on the above analysis, the three types of edges in $K_n$ will have their own biggest frequencies and probabilities at different numbers in $[4,n]$, respectively.  

\begin{theorem}
	\label{th30}
	If the average frequency and probability of an ordinary edge in $K_n$ are computed with the frequency $K_i$s $(i\in[4,n])$, the probability will decrease at $i_d = O(n^{\frac{4}{7}})$ meeting the inequality $\frac{(n-2)(n-3) - (i_d-2)(i_d-3)}{(n-2)(n-3) - (i_d-1)(i_d-2)} \geq \sqrt{1+\frac{2}{i_d(i_d+1)}}$, and the average frequency will be smaller than $\frac{1}{2}{{i}\choose{2}}$ if $i \geq 2i_d$.  
	
	%%Given an $OHC$ edge in $K_n$, its peak frequency is reached at $P_0=\frac{n}{2}+2$ for even $n$ or $P_0=\frac{n+1}{2} + 1$ for odd $n$, and its peak probability is obtained at $n$ as $n$ is big enough. For an ordinary edge, its  average frequency will be smaller than $\frac{1}{2}{{i}\choose{2}}$ if $i \geq [0.3660n + 1.5849]$.  
	
	%Given an $OHC$ edge in $K_n$, its peak frequency is reached at $P_0=\frac{n}{2}+2$ for even $n$ or $P_0=\frac{n+1}{2} + 1$ for odd $n$, and its peak probability is obtained at $n$ as $n$ is big enough. For an ordinary edge contained in a certain $OP^n$, its peak frequency is reached at $i\leq \lceil0.2929n + 1.2675\rceil$,   while its peak probability is obtained before $\lceil0.2929n + 1.2675\rceil$. %For an ordinary edge excluding from any one $OP^n$, its peak frequency and peak probability are obtained before $P_0$, respectively. %For each of the other edges, the peak probability will be computed before $P_0$. 
\end{theorem}

\begin{proof}
	In $K_n$, the number of the $K_i$s containing an edge is  ${{n-2}\choose{i-2}}$, and there are ${{i}\choose{2}}{{n-2}\choose{i-2}}$ $OP^i$s in these $K_i$s. %At $P_0= \frac{n+1}{2}+1$ for odd $n$ or $\frac{n}{2} + 2$ for even $n$, the number of $OP^i$s reaches the maximum value. At this point, an edge may be contained in the maximum number of $OP^i$s, and have the peak frequency computed with frequency $K_i$s. 
	The $n-1$ edges containing a vertex $v$ is considered here. They include two $OHC$ edges, some ordinary edges contained in the $OP^n$s, and the other ordinary edges excluding from any one $OP^n$. There are ${{n-1}\choose{i-1}}$ $K_i$s containing $v$ as well as the same number of frequency $K_i$s. In each frequency $K_i$, the total frequency of the $i-1$ edges containing $v$ is $(i-1)^2$. As the frequency of each edge is computed with the frequency $K_i$s containing them, respectively, the total frequency of the $n-1$ edges is $(i-1)^2{{n-1}\choose{i-1}}$. The three types of edges will share the total frequency, and they will have different frequencies and probabilities according to $i$. 

	Firstly, an $OHC$ edge $e$ is contained in the nearly equal or more percentage of $OP^i$s from $i$ to $i+1\leq n$ based on Theorems \ref{th1} and \ref{th3}. %Given one frequency $K_i$ containing $e$, the frequency $f(e)\geq \frac{1}{2}{{i}\choose{2}}$ exists on average.
	The average frequency  $f(e) \geq \frac{1}{2}{{i}\choose{2}}$ exists as it is computed based on the frequency $K_i$s. In some worst cases, if $p_{i+1}(e\in OP^{i+1}) \leq p_i(e\in OP^i)$ occurs from $i$ to $i+1$ based on the frequency $K_{i+1}$s and $K_i$s, the probability inequality $p_{i+1}(e\in OP^{i+1}) \geq \left[1 - \frac{2}{i(i-1)}\right]p_i(e\in OP^i)$ holds. Moreover, $p_{i+1}(e\in OP^{i+1}) \geq \left[1 + \frac{2}{i(i-1)}\right]p_i(e\in OP^i)$ will appear from $i$ to $i+1$ in most cases if $p_{i+1}(e\in OP^{i+1}) \geq p_i(e\in OP^i)$. If $p_i(e\in OP^i)$ is close to 1, $p_{i+1}(e\in OP^{i+1}) \leq \left[1 + \frac{2}{i(i-1)}\right]p_i(e\in OP^i)$ will exist. %Based on Theorem \ref{th3}, $p_{i+1}(e\in OP^{i+1})\geq [1 -\frac{2}{i(i-1)}]p_i(e\in OP^i)$ holds, where $p_{i+1}(e\in OP^{i+1})$ and $p_i(e\in OP^i)$ are the probabilities that $e$ is contained in $OP^{i+1}$s and $OP^i$s, respectively. As $i$ is big enough, $p_{i+1}(e\in OP^{i+1})\approx p_i(e\in OP^i)$ exists for $p_{i+1}(e\in OP^{i+1})\leq [1 + \frac{2}{i(i-1)}]p_i(e\in OP^i)  \leq 1$. Either $p_i(e\in OP^i)$ decreases or increases according to $i$, the decreasing or increasing rate is very small from $i$ to $i+1$. It indicates that $p_i(e\in OP^i)$ changes smoothly and steadly according to $i$. Let $F_{i+1}(e)$ and $F_i(e)$ denote the frequencies computed with frequency $K_{i+1}$s and $K_i$s, respectively. $F_{i+1}(e) = {{i+1}\choose{2}}{{n-2}\choose{i-1}}p_{i+1}(e\in OP^{i+1})$ and $F_i(e)={{i}\choose{2}}{{n-2}\choose{i-2}}p_i(e\in OP^i)$ are computed. According to $p_{i+1}(e\in OP^{i+1})\geq [1 -\frac{2}{i(i-1)}]p_i(e\in OP^i)$,  $F_{i+1}(e)\geq [1+\frac{2}{i-1}-\frac{2(i+1)}{i(i-1)^2}]\frac{n-i}{i-1}F_i(e)$ is derived. If $i$ is much smaller than $P_0$, $[1+\frac{2}{i-1}-\frac{2(i+1)}{i(i-1)^2}]\frac{n-i}{i-1} >> 1$ exists so $F_{i+1}(e)>F_i(e)$ holds. Else if $i$ approaches $P_0$ for large $n$, $p_{i+1}(e\in OP^{i+1})\approx p_i(e\in OP^i)$ appears and $F_{i+1}(e) \geq F_i(e)$ still holds since the number of $OP^i$s is increasing from $i$ to $i+1$. Thus, the peak frequency is reached at $P_0=\frac{n}{2}+2$ or $\frac{n+1}{2}+1$. 

	For the ordinary edges $g$ contained in the $OP^n$s, they may have certain big average frequency $f(g)> \frac{1}{2}{{i}\choose{2}}$ from $i=4$ to $P_0$. Moreover, $f(g)$ also increases from $i$ to $i+1$ as $i$ is small for big $n$. An ordinary edge $g$ and the two $OHC$ edges containing $v$ are contained in ${{n-4}\choose{i-4}}$ frequency $K_i$s, see Figure \ref{OHC}. In each of the frequency $K_i$s, $g$ is an ordinary edge and it has the frequency $f_2(g)<i-3$ on average. Because $g$ contains two vertices, $g$ is contained in $K = 2{{n-4}\choose{i-4}} - {{n-6}\choose{i-6}}$ frequency $K_i$s where $f_2(g)<i-3$ exists. Since an edge is contained in ${{n-2}\choose{i-2}}$ frequency $K_i$s, the $K$ frequency $K_i$s occupy $ r = \frac{2(i-2)(i-3)}{(n-2)(n-3)} - \frac{(i-2)(i-3)(i-4)(i-5)}{(n-2)(n-3)(n-4)(n-5)}$ of the total frequency $K_i$s containing $g$. As $i$ is much smaller than $n$, $\frac{(i-2)(i-3)}{(n-2)(n-3)} \approx \frac{(i-4)(i-5)}{(n-4)(n-5)}$ exists. Note $\epsilon = \frac{(i-2)(i-3)}{(n-2)(n-3)} \approx \frac{(i-4)(i-5)}{(n-4)(n-5)}$, $r = 2\epsilon - \epsilon^2$ and $1-r = (1-\epsilon)^2$ is formulated. Because $i\geq 4$, $\epsilon \in [\frac{2}{(n-2)(n-3)}, 1]$ exists. 
	
	For $n=1000$, the changes of $r$ and $1-r$ according to $\epsilon$ are illustrated in Figure \ref{r1-r} where $\epsilon \in [2\times 10^{-6},1]$. It mentions that  $\epsilon$ has the positive correlation with $i$ due to $\epsilon = \frac{(i-2)(i-3)}{(n-2)(n-3)}$. If  $i$ is small, $\epsilon$ and $r$ approaches zero whereas $1-r$ is close to 1. It means that $g$ is not adjacent to the two pairs of $OHC$ edges on either endpoint in most frequency $K_i$s. In this case, $g$ may be one $OHC$ edge in most of the $K_i$s and will have the frequency $f_1(g) \geq \frac{1}{2}{{i}\choose{2}}$ in the corresponding frequency $K_i$s. Thus, the average frequency  $f(g) \geq \frac{1}{2}{{i}\choose{2}}$ exists. As  $i$ rises, $\epsilon$ and $r$ tends to 1 step by step whereas $1-r$ tends to zero accordingly. If $r = 1-r$ or $r = \frac{1}{2}$, $\epsilon = 1-\frac{\sqrt{2}}{2}$ is computed. At the number $i$, the number of the frequency $K_i$s where $f_2(g)<i-3$ is at least $\frac{1}{2}{{n-2}\choose{i-2}}$.  %Since $\epsilon$ has the positive correction with $i$ due to $\epsilon = \frac{(i-2)(i-3)}{(n-2)(n-3)}$,
	As $r$ increases and $1-r$ decreases according to $i$, $g$ and the  two pairs of adjacent $OHC$ edges are contained in more percentage of the frequency $K_i$s. Meanwhile, $g$ will have a small frequency in more percent of the frequency $K_i$s. Thus, $f(g)$ will decrease according to $i$. 
	
\begin{figure}
	\centering
	\includegraphics[width=3in,bb=0 0 360 200]{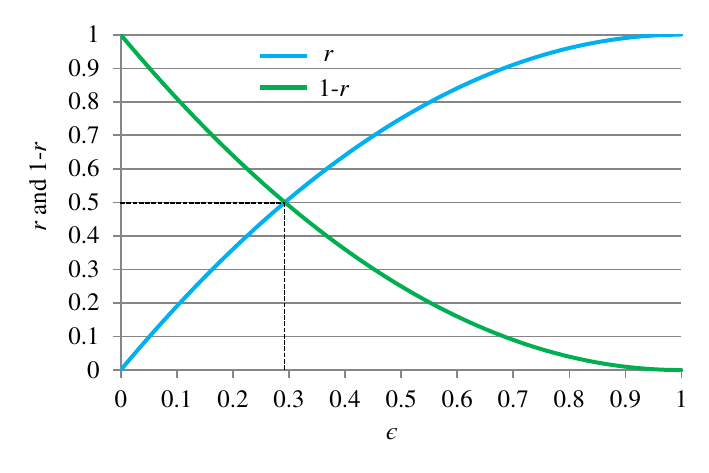}
	\caption{The changes of $r$ and $1-r$ according to $\epsilon \in [2\times 10^{-6}, 1]$ for $n=1000$.} 
	\label{r1-r}
\end{figure}
	
	In the extreme case, $f_2(g) = i-3$ is considered in each of the $K$ frequency $K_i$s, and  the maximum frequency $f_1(g) = {{i}\choose{2}} - 1$ is used in each of the other frequency $K_i$s. In this case, $g$ will have the average frequency $f(g) = {{i}\choose{2}} - \left[{{i}\choose{2}} - i +2\right]r - 1$. If $f(g)\leq \frac{1}{2}{{i}\choose{2}}$ appears, $2\epsilon -\epsilon^2 \geq \frac{\frac{1}{2}{{i}\choose{2}}-1}{{{i}\choose{2}}-i+2}$ can be derived as well as $\epsilon \geq 1 - \sqrt{\frac{1}{2}}\times \sqrt{\frac{{{i}\choose{2}} - 2i + 6}{{{i}\choose{2}} - i + 2}}$. As $i$ is relatively big, $\epsilon \geq 1 - \sqrt{\frac{1}{2}} = 0.2929$ is computed. Because $\epsilon = \frac{(i-2)(i-3)}{(n-2)(n-3)}$, $i\geq [\sqrt{\epsilon}(n-2.5) + 2.5] = [0.5412n + 1.1470]$ is derived. Since $i$ begins from 4 rather than zero, $i \geq [0.5412n + 5.1470]$ is used. From $i = [0.5412n + 5.1470]$, $g$ has the average frequency  $f(g)\leq \frac{1}{2}{{i}\choose{2}}$, $OHC$ edges and ordinary edges can be separated from each other according to their average frequencies. 
	
	As $g$ has the maximum frequency $f_1(g) = {{i}\choose{2}} - 1$ in each of the ${{n-2}\choose{i-2}} - K$ frequency $K_i$s, $p_i(g\in OP^i)$ increases in the fastest way to the biggest value according to $i$.  Because $p_i(g\in OP^i)$ increases and $p_i(g\in OP^i)\gg \frac{1}{2}$ exists as $i$ is small, $p_i(g\in OP^i)$ must decrease before certain $i \ll [0.5412n + 5.1470]$.  %Moreover, the number $i$ is much smaller than $[0.5412n + 5.1470]$. 
	As $p_i(g\in OP^i)$ decreases according to $i$, it never becomes bigger based on Theorem \ref{th33}. %Although $p_1(g\in OP^i)$ increases according to $i$, $p_i(g\in OP^i)$ will becomes smaller because $\epsilon = \frac{(i-2)(i-3)}{(n-2)(n-3)}$ increases according to $i$. Thus, $p_i(g\in OP^i)$ will decrease from certain $i\ll[0.5412n+5.1470]$. 
	As $n = 1000$, the numerical simulations illustrated that $p_i(g\in OP^i)$ decreases from $i = 33$ which is much smaller than $[0.5412n + 5.1470]=546$, see Figure \ref{pdi}. If one assumes that $g$ has the frequency $f_1(g) < {{i}\choose{2}} - 1$ in each of the ${{n-2}\choose{i-2}} - K$ frequency $K_i$s, the maximum value of $p_i(g\in OP^i)$ will be obtained later yet it is still much smaller than $[0.5412n + 5.1470]$. For example if $f_1(g) = \frac{i^2-4i+7}{2}$ as $n=1000$, the numerical simulations illustrated that $p_i(g\in OP^i)$ reaches the maximum value $0.952687$ at $i=92$, and then it decreases according to $i > 92$. 
	
	If $g$ has the probability $p_1(g\in OP^i) > \frac{1}{2}$ based on the ${{n-2}\choose{i-2}} - K$ frequency $K_i$s and $p_2(g\in OP^i) \leq \frac{2(i-3)}{i(i-1)}$ according to the $K$ frequency $K_i$s, $p_i(g\in OP^i) = (1-r)p_1(g\in OP^i) + rp_2(g\in OP^i) = p_1(g\in OP^i) - [p_1(g\in OP^i) - p_2(g\in OP^i)]r$ can be formulated. If $p_i(g\in OP^i)$ increases according to $i$ like an $OHC$ edge, $p_i(g\in OP^i) > \frac{1}{2}$ must hold as well as $p_1(g\in OP^i) > \frac{1}{2}$ increases at the same time. Since $p_2(g\in OP^i) \leq \frac{2(i-3)}{i(i-1)} \ll \frac{1}{2}$ as $i$ is relatively big, $p_1(g\in OP^i)$ is much bigger than $p_2(g\in OP^i)$, and $rp_2(g\in OP^i)$ tends to zero if $i$ is much smaller than $n$. In this case, $p_i(g\in OP^i) \approx (1-r)p_1(g\in OP^i) = (1-\epsilon)^2p_1(g\in OP^i)$ holds where $\epsilon = \frac{(i-2)(i-3)}{(n-2)(n-3)}$.
	
	Provided that $p_1(g\in OP^i) = \frac{ai^2+bi+c}{{{i}\choose{2}}}$, $a\in(\frac{1}{4}, \frac{1}{2}]$, $b\leq 0$ and $c$ are coefficients. Because $p_1(g\in OP^i)$ increases according to $i$, $p_1(g\in OP^{i+1}) > p_1(g\in OP^i)$ holds and $(a+b)(i+1) +2c < 0$ is derived. Thus, $b < -a$ holds and $b < -\frac{1}{2}$. If $i = 4$, $p_1(g\in OP^4) = \frac{16a + 4b + c}{6} > \frac{1}{2}$ has the minimum value before it decreases.  Because $p_1(g\in OP^4) = \frac{16a + 4b + c}{6} < 1$, $3 - 16a < 4b + c < 6 - 16a$ is derived. As $i = 5$, $p_1(g\in OP^5) = \frac{25a + 5b +c}{10}\in (\frac{1}{2}, 1)$ reaches the second minimum value before $p_1(g\in OP^i)$ decreases. Since $p_1(g\in OP^5) > p_1(g\in OP^4)$, $4b + c < \frac{-5a + 3b}{2}$ is derived. Considering $4b + c > 3 - 16a$, $b > 2- 9a$ is derived. As $a = \frac{1}{2}$, $b > -\frac{5}{2}$ and $b\in (-\frac{5}{2}, -\frac{1}{2})$ exists. Since $p_1(g\in OP^4) < 1$, $b < \frac{6-16a-c}{4}$ holds. Because $b > 2 - 9a$, $c < 20a - 2$ is derived. Since $c > 3 - 16a - 4b$ and $b > 2 - 9a$, $c > 20a - 5$ is derived. Thus, $20a - 5 < c < 20a - 2$ and $c\in (0,8)$ is derived for $a\in (\frac{1}{4}, \frac{1}{2}]$. %Because $b < -a$, $c > 3 - 12a$ is derived and $c > -3$. Moreover, $\lvert c\rvert$ is a small number here. Otherwise, $p_1(g\in OP^i) > 1$ or $p_1(g\in OP^i) < \frac{1}{2}$ will appear as $i$ is small. 
	The values of $a$, $b$ and $c$ are not accurate because the lower probability bound $\frac{1}{2}$ for $OHC$ edges is used here. As $i$ is small, $p_1(g\in OP^i)$ will be much bigger than $\frac{1}{2}$ for the ordinary edges contained in $OP^n$. If $p_1(g\in OP^i)$ rises to 1 according to $i\ll P_0$, $a$ is close to or equal to $\frac{1}{2}$, $b$ will be near $-\frac{1}{2}$ and $c$ will be close to zero.  
	
	Since $p_1(g\in OP^i)$ is near 1 as $i$ is small, $a = \frac{1}{2}$ is considered. %If $p_1(g\in OP^i)$ increases according to $i\in $, $p_1(g\in OP^4) < p_1(g\in OP^i)$ holds, i.e., $\frac{8 + 4b + c}{6} < \frac{\frac{1}{2}i^2 + bi +c}{{{i}\choose{2}}}$. One can derive the inequality $4b + c < \frac{12b - 2i}{i+3}$. As $i$ is much bigger than $b$, the upper bound of $4b + c$ tends to $-2$.  Because $4b + c > 3 - 16a = -5$ as $i = 4$, $a = \frac{1}{2}$ and $p_1(g\in OP^4) = \frac{1}{2}$, $\frac{12b-8}{7} > -5$ holds and $b > -\frac{9}{4}$ as well. 
	In the worst average case, $f(e) > 1.5f(g)$ is considered for each edge $e\in OHC$ and $g\notin OHC$ in a frequency $K_i$ based on Theorem \ref{th01}. As $g$ is an $OHC$ edge in each of the ${{n-2}\choose{i-2}} - K$ frequency $K_i$s, $p_1(g\in OP^i) > \frac{3(i-1)^2}{8{{i}\choose{2}}} \approx \frac{3}{4}$ is derived for big $i$. Based on Theorem \ref{th2}, $p_1(g\in OP^4) > \frac{3}{4}$ also holds as $i = 4$. Thus, $p_1(g\in OP^4) = \frac{16a + 4b + c}{6} > 0.75$ rather than $\frac{1}{2}$. Since $4b + c < \frac{-5a + 3b}{2}$, $b > 3- 9a$ is derived as well as $b\in (-\frac{3}{2}, -\frac{1}{2})$ and $c\in (0, 4)$. %Since $p_1(g\in OP^i) < 1$, $b \leq -\frac{1}{2}$ must be guaranteed if $a = \frac{1}{2}$. 
	
	As $n$ is fixed,  $p_i(g\in OP^i)$ and $p_{i+1}(g\in OP^i)$ are given as formula (\ref{F4}). As $p_{i+1}(g\in OP^{i+1}) \leq p_i(g\in OP^i)$ occurs, the formula (\ref{F5}) is derived as $i$ is much bigger than $a$, $b$ and $c$. The left term in formula (\ref{F5}) becomes bigger according $i$ whereas the right term becomes smaller simultaneously. At certain $i = i_d\ll P_0$, the inequality will hold. It means that $p_i(g\in OP^i)$ decreases from $i_d$.

	\begin{eqnarray}\label{F4}
	p_i(g\in OP^i) = \left[1 - \frac{(i-2)(i-3)}{(n-2)(n-3)}\right]^2\times \frac{ai^2+bi+c}{{{i}\choose{2}}}, \nonumber \\
	p_{i+1}(g\in OP^{i+1}) = \left[1 - \frac{(i-1)(i-2)}{(n-2)(n-3)}\right]^2\times \frac{a(i+1)^2+b(i+1)+c}{{{i+1}\choose{2}}}
\end{eqnarray}

	\begin{eqnarray}\label{F5}
	%\left[\frac{(n-2)(n-3) - (i-2)(i-3)}{(n-2)(n-3) - (i-1)(i-2)}\right]^2 \geq  \frac{a(i+1)^2+b(i+1)+c}{ai^2+bi+c}\times \frac{i-1}{i+1} \approx 1+\frac{2}{i(i+1)}, \nonumber \\
	\left[\frac{1 - \frac{(i-2)(i-3)}{(n-2)(n-3)}}{1 - \frac{(i-1)(i-2)}{(n-2)(n-3)}}\right]^2 \geq  \frac{a(i+1)^2+b(i+1)+c}{ai^2+bi+c}\times \frac{i-1}{i+1} \approx 1+\frac{2}{i(i+1)}, \nonumber \\
	\frac{(n-2)(n-3) - (i-2)(i-3)}{(n-2)(n-3) - (i-1)(i-2)} \geq \sqrt{1+\frac{2}{i(i+1)}} 
\end{eqnarray}		
	
	One can use the numerical method to find the smallest number $i_d$ meeting the inequality (\ref{F5}) for given $n$. At the number  $i_d\ll P_0$, $p_i(g\in OP^i)$ reaches the biggest value, and it will decrease according to $i > i_d$. In the derivation of formula (\ref{F5}), the small residual $-\frac{(3a +b)i^2 + (a+3b+2c)i + 2c}{i(i+1)(ai^2 + bi +c)} < 0$ is neglected. If this residual is taken into account, the first number $i_d$ to meet the inequality (\ref{F5}) will become smaller. For example, the small residual is in $\left(-\frac{2}{i(i-1)}, 0\right)$ if $a = \frac{1}{2}$, $b\in\left(-\frac{3}{2}, -\frac{1}{2}\right)$ and $i$ is much bigger than $c$. 
	
	As $g$ is the $OHC$ edge in each of the ${{n-2}\choose{i-2}} - K$  frequency $K_i$s, $p_1(g\in OP^{i+1}) \leq \left[1 + \frac{2}{i(i+1)}\right]p_1(g\in OP^i)$ holds based on Theorem \ref{th3}. One can derive the inequality $(3a + b)i^3 + 2(a + b + c)i^2 - (a - b)i + 2c > 0$. It indicates that $3a + b > 0$ because $i\leq O(n^\frac{4}{7})$ is arbitrary. If $p_{i+1}(g\in OP^{i+1}) \leq p_i(g\in OP^i)$ occurs, the formula (\ref{F5}) can also be derived. In this case, the smallest $i_d$ has the biggest value for given $n$. As $i$ approaches $O(n^{\frac{4}{7}})$, $p_1(g\in OP^i)\approx p_1(g\in OP^{i+1})$ will hold for big $n$ whether $p_1(g\in OP^i)$ decreases or increases. 
	
	As $p_i(g\in OP^i)$ decreases from $i_d = O(n^{\frac{4}{7}})$, $p_{i+1}(g\in OP^{i+1}) \leq \left[1 - \frac{2}{i(i-1)}\right]p_i(g\in OP^i)$ will occur at some $i > i_d$. If the small residual $-\frac{(3a +b)i^2 + (a+3b+2c)i + 2c}{i(i+1)(ai^2 + bi +c)} < 0$ is not considered, one can derive the inequality $\frac{(n-2)(n-3) - (i-2)(i-3)}{(n-2)(n-3) - (i-1)(i-2)} \geq \sqrt{1 + \frac{4i}{i^3 - 3i -2}}$. The smallest vale $i = O(n^{\frac{2}{3}})$ and it has been proven by numerical simulation. In fact, if $p_i(g\in OP^i)$ decreases, it will decrease fast and $p_1(g\in OP^i)$ becomes smaller simultaneously since $\frac{(i-2)(i-3)}{(n-2)(n-3)} \approx \frac{(i-1)(i-2)}{(n-2)(n-3)}$ as $i\ll n$. It implies that the coefficient $a$ also decreases according to $i > i_d$. Thus, $p_{i+1}(g\in OP^{i+1}) \leq \left[1 - \frac{2}{i(i-1)}\right]p_i(g\in OP^i)$ will appear at $i$ much smaller than $O(n^{\frac{2}{3}})$.
		
	Since $c$ is small, $p_1(g\in OP^i)$ is not affected by $c$ for big $i$. In this case, we can assume $c = 0$ and $p_1(g\in OP^i) = \frac{2(ai^2+bi)}{i(i-1)}$ for simplification. If $p_i(g\in OP^i)$ increases according to $i$, $p_1(g\in OP^i)$ will tend to 1 in this stage. The coefficient $a = \frac{1}{2}$ and $\lvert{b}\rvert < \frac{3}{2}$ exist. If $i = 4$, $p_1(g\in OP^i) = \frac{2(16a + 4b)}{12} > \frac{1}{2}$ also holds. If $a = \frac{1}{2}$, $b > -\frac{5}{4}$ is derived. In this case, $p_i(g\in OP^i)$ will increase according to $i$ in the maximum interval $[4, i_d]$ for given $n$. Thus, $b\in \left[-\frac{5}{4}, -\frac{1}{2}\right]$ exists for $g$ contained in the $OP^n$s. In this case, the small residual is in $\left[-\frac{2}{i(i+1)}, -\frac{1}{2i(i+1)}\right]$ and $\frac{(n-2)(n-3) - (i-2)(i-3)}{(n-2)(n-3) - (i-1)(i-2)}\geq \sqrt{1 + \frac{1.5}{i(i+1)}}$ holds. It indicates that $p_i(g\in OP^i)$ will decrease at the even smaller number $i_d\ll P_0$. Since $p_i(g\in OP^i)$ rises to $1$ at $i_d \ll P_0$, $b$ will be close to $-\frac{1}{2}$. For example, if $a = \frac{1}{2}$, $b = -\frac{5}{4}$ and $c = 0$, $p_1(g\in OP^i) = \frac{1}{2}$ is computed as $i = 4$. This probability is too small for $g$ contained in the $OP^n$s. %, especially for big and large $TSP$. 

	The smallest $i_d$ computed based on formula (\ref{F5}) is lined according to $n\in [1E3, 1E7]$  and shown in Figure \ref{minid}. The function $4n^{\frac{4}{7}}$ is also given for comparison. One sees that $i_d < 4n^{\frac{4}{7}}$ is guaranteed according to $n$. It means that such an ordinary edge can be identified in $O(n^{\frac{16}{7}}2^{n^{\frac{4}{7}}})$ time based on dynamic programming. For example if $n=1000$ and 10000, the numerical simulations illustrate that $p_i(g\in OP^i)$ reaches the biggest value at $i_d = 80$ and 369, respectively. As $i\geq 80$ and $369$, $p_i(g\in OP^i)$ will decrease according to $i$, respectively. If we assume that $f_1(g) = \frac{i^2-4i+7}{2}$ and $f_2(g) = i-3$ for $n = 1000$, $p_i(g\in OP^i)$ decreases from $i = 92 > 80$. It means that $\frac{i^2 - 4i + 7}{2}$ is too small as the average frequency $f_1(g)$ if $p_i(g\in OP^i)$ increases according to $i$. Thus, the ordinary edges contained in the $OP^n$s have a big average frequency as $i$ is small. Only if $i > i_d$, $p_i(g\in OP^i)$ begins decreasing according to $i$. Meanwhile, $f(g)$  increases slower than that of the $OHC$ edges or decreases until $f(g) < \frac{1}{2}{{i}\choose{2}}$ appears.

	\begin{figure}
		\centering
		\includegraphics[width=3.5in,bb=0 0 360 200]{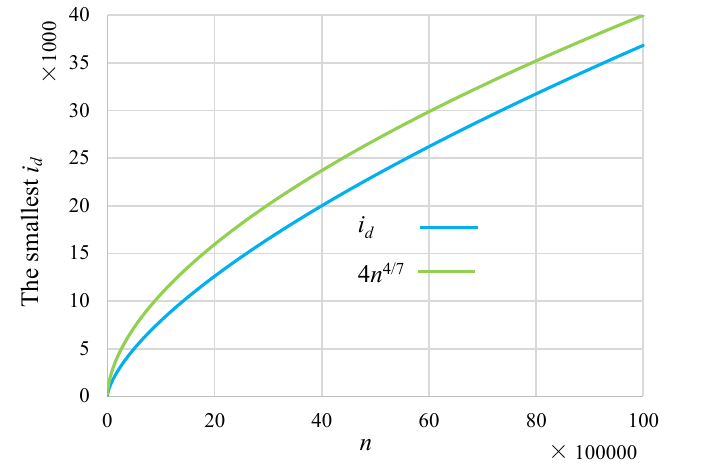}
		\caption{The smallest $i_d$ meeting the inequality (\ref{F5}) according to $n\in [1E3, 1E7]$.} 
		\label{minid}
	\end{figure}

	%As $n$ is very large, the numerical value of the left term in formula (\ref{F5}) will be 1 according to $i$ in a big interval even if it is still rising. However, the right term $\sqrt{1 + \frac{2}{i(i+1)}} = 1$ will be computed although it is decreasing according to $i > 14069878$. In this situation, the double precision should be improved in the computation.  
	
	If $p_i(g\in OP^i)$ decreases from $i$ to $i+1$, $p_{i+1}(g\in OP^{i+1}) \leq \left[1-\frac{2}{i(i-1)}\right]p_i(g\in OP^i)$ will appear and $p_{i+1}(g\in OP^{i+1}) \leq \frac{i}{i+1}p_i(g\in OP^i)$ occurs in most cases until $p_i(g\in OP^i)$ reaches the smallest value. It mentions that the smallest $i_d$ computed with formula (\ref{F5}) has the biggest value for $b\in \left[-\frac{5}{4}, -\frac{1}{2}\right]$. $p_i(g\in OP^i)$ will decrease at some $i < i_d$. As $i = i_d$, $p_i(g\in OP^i)$ has passed the biggest value and the probability decrement $pd_i(g)$ has become bigger than $\frac{p_i(g\in OP^i)}{i+1}$ in most cases, see Theorem \ref{th33}. Provided that $p_{i+1}(g\in OP^i) \leq \frac{i}{i+1}p_i(g\in OP^i)$ occurs from $i = i_d$ computed based on formula (\ref{F5}), $p_i(g\in OP^i) \leq \frac{1}{2}$  and $f(g) \leq \frac{1}{2}{{i}\choose{2}}$ will appear if $i \geq 2i_d$. At this time, the  ordinary edges can also be identified according to their average frequencies. 
	
	In Figure \ref{OHC}, if $v_2 = v_i$ or $v_k = v_n$, there are $n$ such ordinary edges $g = (v_1,v_j)$. For such an ordinary edge, $K = 2{{n-2}\choose{i-2}} - {{n-5}\choose{i-5}}$ exists. Using the same method, one can derive the other inequality for computing the smallest $i_d$. This number is near but bigger than that computed based on formula (\ref{F5}) for given $n$. It indicates that $p_i(g\in OP^i)$ will reach the biggest value at the bigger number $i_d$ for these ordinary edges. Thus, the $n$ ordinary edges are a little more difficult to be identified than the other ordinary edges in theory. 
	
	For the ordinary edges excluding from any one $OP^n$, they will have the smaller average frequency than the ordinary edges contained in the $OP^n$s  based on Theorem \ref{th33}. Thus, they will have the average frequency $f(g) \leq \frac{1}{2}{{i}\choose{2}}$ at the even smaller number $i < 2i_d$, and the probability $p_i(g\in OP^i)$ will decrease before the corresponding $i < i_d$, respectively. 
\end{proof}

Based on Theorem \ref{th33}, the percentage of the $OP^i$s containing $g\notin OHC$ will become smaller according to $i$ in the average case. In each of the $K$ frequency $K_i$s, $g$ has the frequency $f_2(g)<\frac{i+2}{2}$ on average. Except the $K$ frequency $K_i$s containing $g$, $g$ will not have the maximum frequency $f_1(g) = {{i}\choose{2}} - 1$ in each of the other frequency $K_i$s. These frequency $K_i$s include the frequency $K_i$s containing $g$ and none of the four adjacent $OHC$ edges whose number is $J = {{n-6}\choose{i-2}}$, and the  remainder frequency $K_i$s containing $g$ and one $OHC$ edge adjacent to $g$ on either endpoint or $g$ and one adjacent $OHC$ edge on each endpoint whose number is $L = {{n-2}\choose{i-2}} - J - K$. As $n=1000$, the percents of the $J$, $K$, $L$ and $J + L$ frequency $K_i$s are computed according to $i\in [4,1000]$ and the percent changes are shown in Figure \ref{JKL}, respectively. The percents of $J$ and $J + L$ frequency $K_i$s are always decreasing according to $i$ because that of $K$ frequency $K_i$s is always increasing. Moreover, the percent of $J$ frequency $K_i$s decreases faster than that of the $J+L$ frequency $K_i$s. The percent of $L$ frequency $K_i$s increases first, and it finally decreases according to $i$ after the biggest value. As the percent of the $L$ frequency $K_i$s decreases according to $i$, it is much bigger and decreases slower than that of the $J$ frequency $K_i$s. The intersection point between the percent curves related to $J$ and $K$, and that between the percent curves related to $J$ and $L$ will be computed for predicting the changes of the average frequency and probability of $g$ according to $i$. The two points will also be compared with $i_d$ computed with formula (\ref{F5}). 

\begin{figure}
	\centering
	\includegraphics[width=3in,bb=0 0 360 200]{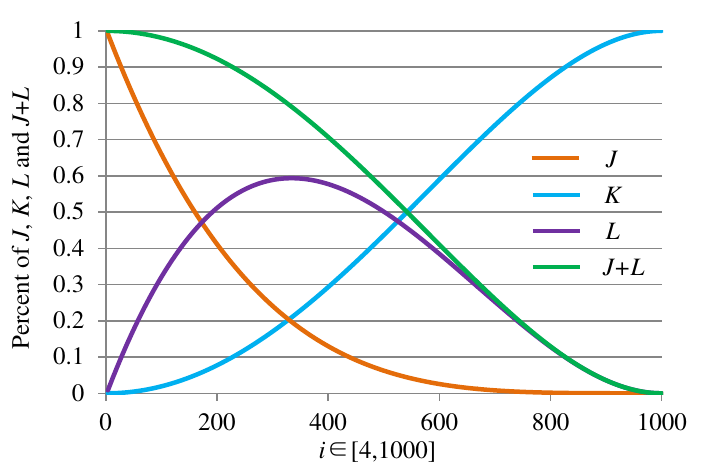}
	\caption{The percents of $J$, $K$, $L$ and $J+L$ frequency $K_i$s according to $i\in [4, 1000]$ for $n=1000$.} 
	\label{JKL}
\end{figure}	

In each of the $J$ frequency $K_i$s, $g$ is the $OHC$ edge in the ideal case. Based on Theorem \ref{th22}, we assume the average frequency $f_1(g) \geq \frac{i^2-4i+7}{2}$ at each $i$. %Based on Theorem \ref{th22}, the frequency of each ordinary edge in each frequency $K_i$ is bigger than 1 on average. Moreover, each ordinary edge will have the frequency tending to 2 according to $i$ in the average case. 
In this case, $f_1(g)$ increase according to $i$ based on the $J$ frequency $K_i$s. However, the percent of $J$ decreases fast according to $i$. Plus the $K$ frequency $K_i$s where $f_2(g) < \frac{i+2}{2}$, the average frequency $f(g) < \frac{i^2-4i+7}{2}$ will appear as $i$ becomes big. %Moreover, $\frac{K}{{{n-2}\choose{i-2}}}$ increases according to $i$  whereas $\frac{J}{{{n-2}\choose{i-2}}}$ always decreases. Although $f(g)$ and $p_i(g\in OP^i)$ increases  according to $i$ as $i$ is small, they will decrease finally according to $i$ if the percent of $K$ is big enough. 
For example, $K > J$ will appear if $i > [0.3236n+4]$ for big $n$. In this case, $f(g) < \frac{1}{2}{{i}\choose{2}}$ appears according to the $J+K$ frequency $K_i$s. Moreover, $p_i(g\in OP^i)$ will decrease before $i = [0.3236n+4]$ since $p_i(g\in OP^i)$ is much bigger than $\frac{1}{2}$ as $i$ is small. If $g$ has the probability $p_3(g\in OP^i)$ according to the $J$ frequency $K_i$s and $p_2(g\in OP^i)$ according to the $K$ frequency $K_i$s, $p_i(g\in OP^i) = \frac{K}{J+K}\times p_2(g\in OP^i) + \frac{J}{J+K}\times p_3(g\in OP^i)$ is derived based on the $J + K$ frequency $K_i$s. If the $L$ frequency $K_i$s are not considered, $p_i(g\in OP^i)$ will decrease before the number $i_d$ computed with formula (\ref{F5}) because $\frac{K}{J+K} > r$ holds according to $i$. %For $n=1000$, if $f_1(g) = {{i}\choose{2}} - 1$ and $f_2(g) = i-3$, the numerical simulations illustrated that $p_i(g\in OP^i)$ reaches the biggest value 0.99614 at $i=32$, and it becomes smaller if $i > 32$ which is much smaller than $i_d = 80$ computed with formula (\ref{F5}). 

In each of the remainder $L={{n-2}\choose{i-2}} - J - K$ frequency $K_i$s, $g$ is adjacent to one $OHC$ edge on either endpoint or one $OHC$ edge on each endpoint, respectively. Given any $g$ and $e\in OHC$, it is known that $p_i(e\in OP^i)$ increases faster than $p_i(g\in OP^i)$  according to $i\in [4,n]$ in the average case, see Theorem \ref{th3}. %As $p_i(e\in OP^i)$ increases from $i$ to $i+1$, $p_{i+1}(e\in OP^{i+1}) \geq [1+\frac{2}{i(i-1)}]p_i(e\in OP^i)$ happens in most cases until $p_i(e\in OP^i)$ approaches 1. 
Among the $L$ frequency $K_i$s, there are $L'={{n-3}\choose{i-3}} - {{n-4}\choose{i-4}} - {{n-5}\choose{i-5}} + {{n-6}\choose{i-6}}$ frequency $K_i$s containing $g$ and one adjacent $OHC$ edge $e$. %In Figure \ref{OHC}, $g$ and one adjacent $OHC$ edge $e$ are contained in $L={{n-3}\choose{i-3}} - {{n-4}\choose{i-4}} - {{n-5}\choose{i-5}} + {{n-6}\choose{i-6}}$ frequency $K_i$s where $g$ may be one $OHC$ edge. 
The $L'$ frequency $K_i$s containing $e$ and $g$ are taken as one sample drawn from the frequency $K_i$s containing $e$. In the average case, $p_i(e\in OP^i)$ and $f(e)$ will increase faster or decrease slower than $p_i(g\in OP^i)$ and $f(g)$, respectively, based on the $L'$ frequency $K_i$s. If $p_i(e\in OP^i)$ decreases according to $i$ as $i$ is small, it will reach the smallest value at certain small $i$, and then it will increase according to $i$ after the smallest value based on Theorem \ref{th3}. Moreover, $p_i(e\in OP^i)$ will increase to one big value within a limited number of steps because $p_{i+1}(e\in OP^{i+1}) \geq \left[1+\frac{2}{i(i-1)}\right]p_i(e\in OP^i)$ usually occurs if $p_i(e\in OP^i)$ increases from the smallest value. If $p_i(e\in OP^i)$ is close to 1, $p_{i+1}(e\in OP^{i+1}) \leq \left[1+\frac{2}{i(i-1)}\right]p_i(e\in OP^i)$ will hold from $i$ to $i+1$. Based on Theorem \ref{th22}, $e$ will have the expected frequency $f(e)>\frac{i^2-4i+7}{2}$ according to $i$ based on these frequency $K_i$s. Moreover, in each of the frequency $K_i$s, the frequency of each ordinary edge will be bigger than 1 and they will approach 2 according to $i$, respectively. In this case, $f(g) < \frac{i^2-4i+7}{2}$ will appear based on the $L'$ frequency $K_i$s. Even if $g$ is an $OHC$ edge in most of the $L$ $K_i$s, $g$ will have an average frequency $f(g) < \frac{i^2-4i+7}{2}$ according to these frequency $K_i$s. $p_i(g\in OP^i)$ may increase according to $i$ based on the $L$ frequency $K_i$s at first. Because $p_i(e\in OP^i)$ for each $e$ adjacent to $g$ increases according to $i$, $p_i(g\in OP^i)$ will reach the biggest value at certain small $i$ and it will decrease finally according to $i$. Once $p_i(g\in OP^i)$ decreases from $i$ to $i+1$, it never becomes bigger according to $i$ with respect to the $L$ frequency $K_i$s based on Theorem \ref{th33}. Moreover,   $p_{i+1}(g\in OP^{i+1}) <  \frac{i}{i+1}p_i(g\in OP^i)$ will  appear in most cases. 

Based on the $J + L = {{n-2}\choose{i-2}} - K$ frequency $K_i$s, $p_i(g\in OP^i)$ will also decrease as $i$ is relatively big. If $g$ has the probability $p_3(g\in OP^i)$ based on the $J$ frequency $K_i$s, and it has the probability $q_3(g\in OP^i)$ based on the $L$ frequency $K_i$s, the $p_i(g\in OP^i)$ will be analyzed according to $i$ based on the $J + L$ frequency $K_i$s. As $i$ is small, we assume that $p_3(g\in OP^i) \gg \frac{1}{2}$ and $q_3(g\in OP^i) > \frac{1}{2}$, and they increase according to $i$ in the ideal case until they are near 1 in case that $p_3(g\in OP^i) \geq q_3(g\in OP^i)$. Provided that $q_3(g\in OP^i)$ reaches the maximum value at certain $i < P_0$ and $q_3(g\in OP^i) = p_3(g\in OP^i)$ appears, $p_3(g\in OP^i)$ and $q_3(g\in OP^i)$ will have different changes according to the following $i$s with respect to the $J$ and $L$ frequency $K_i$s, respectively. In the following, $p_3(g\in OP^i)$ is still increasing according to $i$ in the best case. However, $q_3(g\in OP^i)$ will decrease simultaneously. %while $q_3(g\in OP^i)$ is generally smaller than $\frac{i^2-4i+7}{i(i-1)}$. 
As $p_3(g\in OP^i)$ increases from $i$ to $i+1$, the probability increment will be smaller than $\frac{2p_3(g\in OP^i)}{i(i-1)}$ based on Theorem \ref{th3}. On the other hand, the probability decrement for $q_3(g\in OP^i)$ will be bigger than $\frac{2q_3(g\in OP^i)}{i(i-1)}$ in most cases. Let $r_1 = \frac{J}{J+L}\neq 0$ and $r_2 = \frac{L}{J+L}\neq 0$, and $p_i(g\in OP^i) = r_1p_3(g\in OP^i) + r_2q_3(g\in OP^i)$ is formulated. The percent of $J$ is decreasing according to $i$ while that of $L$ increases before $i = [0.3236n+4]$ and decreases after $i = [0.3236n+4]$. At $i=[0.3236n+4]$, the percent of $J$ and $K$ frequency $K_i$s is approximately $20\%$, respectively and that of $L$ frequency $K_i$s is approximately $60\%$ in Figure \ref{JKL}. If $g$ has the frequency $f_1(g) = \frac{i^2}{2}+bi+c$ in each of the $J+L$ frequency $K_i$s and $f_2(g) = i-3$ in each of the $K$ frequency $K_i$s, the average frequency $f(g) < 0.4i^2 + (0.8b+0.2)i+0.8c - 0.6$ and $p_i(g\in OP^i) \approx 0.8$ as $i\geq[0.3236n+4]$. Based on all frequency $K_i$s containing $g$, $p_i(g\in OP^i)$ decreases before $i=[0.3236n+4]$ because the maximum value is near 1. Thus, $i_d < [0.3236n + 4]$ exists for $g\notin OHC$. 

In $p_i(g\in OP^i) = r_1p_3(g\in OP^i) + r_2q_3(g\in OP^i)$, whether the percent of $L$ increases or decreases according to $i$, $r_2$ increases faster or decreases slower than $r_1$. $p_i(g\in OP^i)$ will become bigger according to $i$ if $r_1 > r_2$ in the ideal case. Otherwise, $p_i(g\in OP^i)$ will decrease according to $i$ once $r_1 \leq r_2$ while $q_3(g\in OP^i)$  decreases according to $i$. If $r_1 = r_2$, one can derive that $i \approx [(3-2\sqrt{2})(n-2.5) + 2.5] = [0.1716n + 2.0711]$. $i = [0.1716n + 6.0711]$ is considered as $i$ starts from 4. The percentage of $K$ frequency $K_i$ is approximately equal to $5.89\%$. Since $i_d$ computed with formula (\ref{F5}) is much smaller than $[0.1716n + 6.0711]$, it means that $p_i(g\in OP^i)$ tends to 1 if it is still rising from $i_d$ to $i=[0.1716n + 6.0711]$. However, the percent of $J+L$ frequency $K_i$s is smaller than 0.9411 at $i = [0.1716n + 6.0711]$. As the $K$ frequency $K_i$s are also considered, the average frequency $f(g) = 0.9411(\frac{i^2}{2} + bi +c) + 0.0589(i-3) < 0.9411\times \frac{i^2}{2}$ is derived. It means that $p_i(g\in OP^i)$ decreases before $i = [0.1716n+6.0711]$. %If $i > [0.1716n + 6.0711]$ while $r_1 < r_2$, $q_3(g\in OP^i)$ will become smaller according to $i$ for big $n$, and $p_i(g\in OP^i)$ will decrease based on the $J + L$ frequency $K_i$s as well. 
Moreover, $q_3(g\in OP^i)$ also decreases  before $i=[0.1716n+6.0711]$.  The reason is that $p_i(g\in OP^i)$  decreases because the probability decrement for $q_3(g\in OP^i)$ is bigger than the probability increment for $p_3(g\in OP^i)$ from $i$ to $i+1$. 

Since $p_i(g\in OP^i)$ decreases according to $i$ based on the $J+L = {{n-2}\choose{i-2}} - K$ frequency $K_i$s before $i=[0.1716n + 2.0711]$, the probability decrement $pd_i(g)$ will reach the biggest value later than that shown in Figure \ref{pdi}. Moreover, it will keep a big value according to $i$ until $p_i(g\in OP^i)$ tends to the smallest value. One can assume $p_i(g\in OP^i) = \frac{ai^2 + bi +c}{{{i}\choose{2}}} > 0$ based on the $J+L$ frequency $K_i$s. If $p_i(g\in OP^i) > p_{i+1}(g\in OP^{i+1})$ happens, $(a+b)(i+1) + 2c > 0$ is derived. It means that $b > -a$ holds for any $i$. If $a = \frac{1}{2}$, $p_i(g\in OP^i) > 1$ will appear for big $i$ and small $c$. Thus, if $p_i(g\in OP^i)$ decreases from $i$ to $i+1$, $a < \frac{1}{2}$ must hold and it becomes smaller according to $i$. That's why $pd_i(g)$ becomes bigger according to $i$ or keeps a big value for big $i$. 

As $i > i_d$, $p_i(g\in OP^i) = 0$ exists for most ordinary edges. For the ordinary edges contained in the $OP^n$s, $p_{i+1}(g\in OP^{i+1})\leq \frac{i}{i+1}p_i(g\in OP^i)$ holds in most cases based on formula (\ref{F5}). As $p_i(g\in OP^i)$s for all the $n-3$ ordinary edges containing a vertex decrease from $i_d \ll [0.1716n + 6.0711]$, the $p_i(e\in OP^i)$s for the two adjacent $OHC$ edges definitely increase, see formula (\ref{F3}). %Based on the $J+L$ frequency $K_i$s, $p_i(g\in OP^i)$ will decrease if $i\geq[0.1716n + 6.0711]$. Moreover, the probability decrement increases according to $i$ in a big interval until $p_i(g\in OP^i)$ becomes one small value, see Figure \ref{pdi}. If the biggest value of $p_i(g\in OP^i)$ is close to 1, the decrement of $p_i(g\in OP^i)$ will be bigger than $\frac{2[k-0.1776n-6.0711]}{[0.1716n+6.0711][0.1716n-5.0711]}$ from $i=[0.1716n+6.0711]$ to $k \in[i,n]$. If the $K$ frequency $K_i$s are also considered, $p_i(g\in OP^i)$ will decrease faster according to $i\in[0.1716n+6.0711]$.%In the ideal case, the probability decrement is e

\begin{theorem}
	\label{th5}
	$TSP$ on $K_n$ can be resolved in $O(n^2i_d^42^{i_d})$ time based on the frequency $K_i$s and dynamic programming where $i_d = O(n^{\frac{4}{7}})$ is the smallest number meeting the inequality (\ref{F5}). 
\end{theorem}

\begin{proof}
	Based on Theorems \ref{th3} and \ref{th30}, the probabilities $p_i(e\in OP^i)$ for an $OHC$ edge and $p_i(g\in OP^i)$ for an ordinary edge have different evolutionary trends according to $i$. Most importantly, $p_i(g\in OP^i)$ definitely decreases if $i\geq i_d = O(n^{\frac{4}{7}})$, and $p_{i+1}(g\in OP^{i+1}) < \left[1 - \frac{2}{i(i-1)}\right]p_i(g\in OP^i)$ will happen in most cases once $p_i(g\in OP^i)$ becomes smaller from $i$ to $i+1$. To distinguish between OHC edges and ordinary edges, it is necessary to compute the $p_i(e\in OP^i)$ and $p_i(g\in OP^i)$ at $i = i_d$ to see the probability change for each edge. The dynamic programming is used to compute the $OP^i$s. Given a $K_{i_d}$, it consumes $O(i_d^22^{i_d})$ time to compute an $OP^{i_d}$. Based on the binomial distribution, one can choose a constant number of $K_{i_d}$s for each edge to compute the probability for a given $TSP$ instance on $K_n$. 
	%Based on Theorems \ref{th3} and \ref{th30}, an $OHC$ edge has the average frequency bigger than $\frac{1}{2}{{i}\choose{2}}$ according to $i\in [4,n]$ for big $i$, whereas the average frequency of an ordinary edge will be smaller than $\frac{1}{2}{{i}\choose{2}}$ if $i \geq 2i_d$ computed with formula (\ref{F5}). 
	%%Moreover, the frequency for an ordinary edge decreases quickly after the peak frequency. 
	%We assume that $n$ is not infinite. To identify the $OHC$ edges, the $OP^i$s at $i=2i_d$ for each edge should be computed for computing the average frequency. Given one $K_i$ containing $2i_d$ vertices, it consumes $O(i_d^22^{2i_d})$ time to compute an $OP^{2i_d}$ using dynamic programming. According to the binomial distribution, one can choose a constant number of $K_{2i_d}$s for each edge to compute the average frequency for a given $TSP$ instance on $K_n$. 
	%As a constant number of frequency $K_{0.3660n}$s are chosen to compute the average frequency for each edge, 
	The computation time is $O(n^2i_d^42^{i_d})$ since each $K_{i_d}$ contains ${{i_d}\choose{2}}$ $OP^{i_d}$s, and $K_n$ contains ${{n}\choose{2}}$ edges. All $OHC$ edges will be found in this time. 
	%As $i > \frac{n}{2} + 1$ or $\frac{n+1}{2}$, the number of $OP^i$s in the $K_i$s containing each edge decreases according to $i$. As the other edges out of $OP^n$s are included in the smaller and smaller percentage of the $OP^i$s, the edges in the $OP^n$s will be included in more and more percentage of the $OP^i$s.
	%%According to Theorem \ref{th30}, an $OHC$ edge has its peak probability at $n$ as $n$ is big enough, whereas an ordinary edge obtains its peak probability before $P_0$. Moreover, the probability for an ordinary edge decreases quickly after the peak value. To identify the $OHC$ edges, the $OP^i$s at $i=P_0$ for each edge should be computed for computing the probability. Given $K_i$ containing $P_0$ vertices, it consumes $O(n^22^{0.5n})$ time to compute an $OP^{0.5n}$ using dynamic programming. According to the  binomial distribution, one can choose a constant number of $K_{0.5n}$s to compute the  probability for each edge. As a constant number of frequency $K_{0.5n}$s are chosen to compute the  probability for each edge, the computation time is $O(n^62^{0.5n})$ since each $K_{0.5n}$ contains ${{0.5n}\choose{2}}$ $OP^{0.5n}$s, and ${{n}\choose{2}}$ edges in $K_n$ are considered. All $OHC$ edges will be found in this time. 
\end{proof}

In applications, the probabilities $p_i(e\in OP^i)$ and $p_i(g\in OP^i)$ are useful to identify $OHC$ edges and ordinary edges. According to $i$, $p_i(e\in OP^i)$ for an $OHC$ edge will increase or keep the nearly equal value. For an ordinary edge, $p_i(g\in OP^i)$s will increase very slowly or decrease according to $i < i_d$, and it must decrease according to $i > i_d$ computed with formula (\ref{F5}). %If no $OHC$ edges decrease from $i = i_d$, the $OHC$ will be found. %If the initial $p_4(g\in OP^4)$ is known, the actual number $i$ from which $p_i(g\in OP^i)$ decreases can be computed with numerical methods for given $n$. Moreover, each probability of the ordinary edges has the non-negligible decrement at certain number $i<P_0$. Given one $TSP$ instance, if $p_i(g\in OP^i)$s for all ordinary edges decreases before certain constant number $c$, it is necessary to compute the frequency $K_i$s at $i=c$ for all the edges. In this case, $OHC$ will be found in $O(n^2c^42^{c})$ time. 

\section{Experiments and analysis}
In this section, the experiments are executed for the constructed and real-world $TSP$ instances to verify the above findings. Firstly, the small $TSP$ instances are constructed and the frequency $K_n$s are computed for showing the lower frequency bound for $OHC$ edges and upper frequency bound for ordinary edges where $n \in[4,14]$. The frequencies and probabilities of edges are computed with frequency $K_i$s, and the frequency and probability changes for edges will also be illustrated  according to $i\in[4,n]$ for the small $TSP$ instances. Secondly, the experiments for the real-world $TSP$ instances are executed to show the lower frequency and probability bounds for $OHC$ edges, and the frequency and probability changes for $OHC$ edges and all edges according to $i$ as the frequency of each edge is computed with the frequency $K_i$s where $i\in[4,8]$. Given certain $K_i$ ($i\in[4,n]$) in $K_n$, the dynamic programming is used to compute the ${{i}\choose{2}}$ $OP^i$s and frequency $K_i$. The total computation time is $O(i^42^i)$. As the frequency of each edge in $K_n$ is computed with $N$ frequency $K_i$s at given number $i$, it requires $O(n^2i^42^iN)$ time where $N$ is the number of the selected $K_i$s for an edge. In the experiments, the program to compute the $OP^i$s, frequency and probability of each edge is encoded in C++ language and run on one laptop with CPU 2.3GHz and inner memory 4.00GB. 

The probability $p_i(e\in OP^i)$ for an edge $e\in OHC$ or $OHC$ edges and $p_i(g\in OP^i)$ for an ordinary edge $g\notin OHC$ or ordinary edges will be used with high frequency in the experiments. To simplify the writing, $p_i(e\in OP^i)$ is replaced by $p_i(e)$ and $p_i(g\in OP^i)$ is replaced by $p_i(g)$, respectively. 
\subsection{The experiments for small $TSP$ instances $(n\in[4,14])$}
We construct 11 small $TSP$ instances to verify the lower frequency bound for $OHC$ edges and upper frequency bound for ordinary edges. As $n\in[4,8]$, each edge in $K_n$ has one random distance in $(0,10]$. If $n\in[9,14]$, the $K_n$s are constructed based on Oliver30 which is one Euclidean $TSP$. Each $K_n$ contains the first $n$ vertices in Oliver30 and the vertex set is $\{0,1,...,n-1\}$. Due to the time complexity to compute the $OP^n$s, 
the biggest $TSP$ instance has 14 vertices in the experiments. Each small $TSP$ instance has one $OHC$ and at least ${{n}\choose{2}}$ $OP^n$s. If the number of $OP^n$s is bigger than ${{n}\choose{2}}$, it will be difficult to choose the right $OP^n$s for computing the big frequency for some $OHC$ edges. We sometimes meet the $K_n$ having a lot of equal-weight edges. In this case, there may be many $OHC$s and $OP^i$s with the same endpoints in the $K_i$s. If many $OP^i$s not including the $OHC$ edges are used, the frequencies of edges computed with these $OP^i$s may not work well for separating $OHC$ edges from ordinary edges. In general, if there are one $OHC$ and ${{i}\choose{2}}$ $OP^i$s in most $K_i$s, the frequencies of edges computed with the frequency $K_i$s will work well for separating $OHC$ edges from most  ordinary edges. For each of the small $K_n$s, we first compute the $OP^n$s ($4\leq n < 14$) and $OHC$. Then, the ${{n}\choose{i}}{{i}\choose{2}}$ $OP^i$s for given $i$ are used to compute the  frequency of each edge in the corresponding $K_n$. Based on the known $OHC$, the lower frequency bound for $OHC$ edges and upper frequency bound for ordinary edges will be compared. For general $TSP$, the experimental results are enough to show the different structure properties denoted by frequencies between the $OHC$ edges and ordinary edges. 

\subsubsection{The frequency bounds for $OHC$ and ordinary edges in frequency $K_n$}
Given $K_n$, the ${{n}\choose{2}}$ $OP^n$s are computed and the frequency of each edge is computed with the $OP^n$s. The frequencies of edges in these frequency $K_n$s are listed in Table \ref{Examp} where $n$ changes from 4 to 14. In each column, the frequencies of the edges denoted by $f_k$ ($k\in[1,{{n}\choose{2}}]$) in each frequency $K_n$ are ordered from big to small values. They form one monotone decreasing frequency sequence $\left(f_1, f_2,\cdots, f_{{n}\choose{2}}\right)$, and the frequency of zero is not shown. In addition, the smallest frequency of the $OHC$ edges in each frequency $K_n$ is denoted with the boldface number. In the last three rows, $f_{lb} = \frac{1}{2}{{n}\choose{2}}$ denotes the lower frequency bound for $OHC$ edges in $K_n$, $N_f$ denotes the number of edges with the frequency bigger than zero and $N_{tot} = {{n}\choose{2}}$ denotes the total number of edges in $K_n$.

 %The experimental results for the 11 small $TSP$ instances approximately verify the $p(e\in OHC)\to 1$ according to $n$.  

\begin{table}
	\begin{center}
		\caption{The frequencies($>0$) of edges computed with $OP^n$s in $K_n$.}
		{\footnotesize \begin{tabular}{ p{0.5cm}  p{0.5cm}  p{0.5cm}  p{0.5cm}  p{0.5cm}  p{0.5cm}  p{0.5cm}  p{0.5cm}  p{0.5cm}  p{0.5cm}  p{0.5cm}  p{0.5cm} }
				\hline
		$f_k$ &  &  &  &  &  & $n$ &  &  &  &  &   \\
		& 4 & 5 & 6 & 7 & 8 & 9 & 10 & 11 & 12 & 13 & 14  \\
				\hline
		$f_1$ & 5 & 9 & 14 & 20 & 27 & 35 & 44 & 54 & 65 & 77 & 90  \\
		$f_2$ & 5 & 8 & 13 & 18 & 27 & 35 & 44 & 53 & 65 & 77 & 90  \\
		$f_3$ & 3 & 7 & 12 & 16 & 25 & 33 & 42 & 47 & 63 & 74 & 86  \\
		$f_4$ & \textbf{3} & 6 & 10 & 15 & 19 & 30 & 40 & 45 & 57 & 69 & 82  \\
		$f_5$ & 1 & \textbf{5} & 8 & 13 & 18 & 27 & 35 & 44 & 56 & 68 & 81  \\
		$f_6$ & 1 & 2 & \textbf{6} & 12 & 18 & 24 & 35 & 44 & 54 & 66 & 78  \\
		$f_7$ &  & 1 & 3 & \textbf{12} & 17 & 22 & 31 & 43 & 53 & 65 & 77  \\
		$f_8$ &  & 1 & 3 & 8 & \textbf{15} & 22 & 29 & 40 & 50 & 64 & 76  \\
		$f_9$ &  & 1 & 3 & 3 & 6 & \textbf{16} & 27 & 40 & 50 & 61 & 73  \\
		$f_{10}$ &  &  & 1 & 3 & 5 & 10 & \textbf{20} & 36 & 43 & 52 & 72  \\
		$f_{11}$ &  &  & 1 & 2 & 5 & 6 & 12 & \textbf{24} & 41 & 51 & 62  \\
		$f_{12}$ &  &  & 1 & 2 & 4 & 5 & 6 & 14 & \textbf{32} & 47 & 60  \\
		$f_{13}$ &  &  &  & 1 & 3 & 4 & 6 & 10 & 16 & \textbf{41} & 53  \\
		$f_{14}$ &  &  &  & 1 & 3 & 4 & 4 & 9 & 11 & 18 & \textbf{51}  \\
		$f_{15}$ &  &  &  &  & 2 & 4 & 4 & 7 & 9 & 13 & 21  \\
		$f_{16}$ &  &  &  &  & 1 & 3 & 4 & 7 & 7 & 12 & 20  \\
		$f_{17}$ &  &  &  &  & 1 & 3 & 4 & 6 & 6 & 10 & 14  \\
		$f_{18}$ &  &  &  &  &  & 2 & 4 & 6 & 6 & 9 & 13  \\
		$f_{19}$ &  &  &  &  &  & 1 & 4 & 6 & 6 & 8 & 10  \\
		$f_{20}$ &  &  &  &  &  & 1 & 2 & 2 & 6 & 6 & 9  \\
		$f_{21}$ &  &  &  &  &  & 1 & 2 & 2 & 4 & 6 & 6  \\
		$f_{22}$ &  &  &  &  &  &  & 2 & 2 & 4 & 6 & 6  \\
		$f_{23}$ &  &  &  &  &  &  & 2 & 2 & 3 & 5 & 6  \\
		$f_{24}$ &  &  &  &  &  &  & 1 & 2 & 3 & 4 & 6  \\
		$f_{25}$ &  &  &  &  &  &  & 1 & 2 & 2 & 4 & 6  \\
		$f_{26}$ &  &  &  &  &  &  &  & 1 & 2 & 4 & 5  \\
		$f_{27}$ &  &  &  &  &  &  &  & 1 & 2 & 3 & 5  \\
		$f_{28}$ &  &  &  &  &  &  &  & 1 & 2 & 3 & 4  \\
		$f_{29}$ &  &  &  &  &  &  &  &  & 2 & 2 & 4  \\
		$f_{30}$ &  &  &  &  &  &  &  &  & 2 & 2 & 3  \\
		$f_{31}$ &  &  &  &  &  &  &  &  & 1 & 2 & 2  \\
		$f_{32}$ &  &  &  &  &  &  &  &  & 1 & 2 & 2  \\
		$f_{33}$ &  &  &  &  &  &  &  &  & 1 & 1 & 2  \\
		$f_{34}$ &  &  &  &  &  &  &  &  & 1 & 1 & 2  \\
		$f_{35}$ &  &  &  &  &  &  &  &  &  & 1 & 2  \\
		$f_{36}$ &  &  &  &  &  &  &  &  &  & 1 & 1  \\
		$f_{37}$ &  &  &  &  &  &  &  &  &  & 1 & 1  \\
		$f_{38}$ &  &  &  &  &  &  &  &  &  &  & 1  \\
		$f_{39}$ &  &  &  &  &  &  &  &  &  &  & 1  \\		
		$f_{lb}$ & 3 & 5 & 7.5 & 10.5 & 14 & 18 & 22.5 & 27.5 & 33 & 39 & 45.5 \\	
		$N_f$ & 6 & 9 & 12 & 14 & 17 & 21 & 25 & 28 & 34 & 37 & 39  \\
		$N_{tot}$ & 6 & 10 & 15 & 21 & 28 & 36 & 45 & 55 & 66 & 78 & 91  \\
				
				%\bottomrule
				\hline
		\end{tabular}}
		\label{Examp}
	\end{center}
\end{table}

Based on Theorem \ref{th1}, an $OHC$ edge in $K_n$ has the lower frequency bound $f_{lb} = \frac{1}{2}{{n}\choose{2}}$ in the worst average case, and $\frac{7}{18}{{n}\choose{2}}$ in the worst case. Moreover, $f_{lb}$ will increase according to $n$. If $n$ ($n\geq 5 $) is small, the smallest frequency of the $OHC$ edges may be smaller than $\frac{1}{2}{{n}\choose{2}}$ but must be bigger than $\frac{7}{18}{{n}\choose{2}}$. For the ordinary edges, the upper frequency bound  is smaller than $2(n-3)$ based on Theorem \ref{th22}. In Table \ref{Examp}, the top $n$ values of $f_k$ in each column belong to the $OHC$ edges in each $K_n$. In view of the datum, one firstly observes that: The $OHC$ edges have the top $n$ frequencies much bigger than those of the ordinary edges in each frequency $K_n$. 
%(2) The minimum frequency of the $OHC$ edge is bigger than $2n-4$ in most cases. Moreover, it becomes much bigger than $2n-4$ according to $n$ and approaches the maximum frequency $\frac{(n+1)(n-2)}{2}$.
The frequencies of the ordinary edges are very small, and most of them are equal to zero as $n$ is relatively big, such as $n\geq 10$. It says that every $OHC$ edge is included in many of the ${{n}\choose{2}}$ $OP^n$s. On the other hand, an ordinary edge is included in a small number of the $OP^n$s so the frequency is very small and even equals zero. 

Secondly, the maximum frequency of the $OHC$ edges is ${{n}\choose{2}} - 1$. Only one $OP^n$ does not include such an $OHC$ edge where the two vertices  are the endpoints of the $OP^n$.  Moreover, the smallest frequency of the $OHC$ edges in each column is bigger than $\frac{7}{18}{{n}\choose{2}}$, and most of them are bigger than $f_{lb} = \frac{1}{2}{{n}\choose{2}}$. If the smallest frequency of the $OHC$ edges is smaller than $f_{lb}$, the number of such $OHC$ edges is only one in these frequency $K_n$s. Except the smallest frequency, the frequencies of the other $OHC$ edges are much bigger than $f_{lb}$, such as $n=$  9, 10, 11, 12, 13 and 14. It indicates that Theorem \ref{th1} works well for the constructed $K_n$s containing random distances and small real-life $TSP$ instances. We investigated the edge with the smallest frequency for $n=9\sim 14$. It is always the edge $(2,6)$ which is not the $OHC$ edge in Oliver30 $(K_{30})$. We also investigated the frequencies of the other $OHC$ edges in these frequency graphs. It found that the edges with the frequency bigger than $\frac{i^2-4i+7}{2}$ are mostly the $OHC$ edges in Oliver30. On the other hand, the edges with the frequency smaller than $\frac{i^2-4i+7}{2}$ are mostly not the $OHC$ edges in Oliver30. It indicates that although some ordinary edges in $K_n$ are the $OHC$ edges in some sub-graphs $K_i$s $(i< n)$, the frequency of such ordinary edges will be much smaller than the lower bound for the average frequency of all $OHC$ edges, i.e., $\frac{i^2-4i+7}{2}$, in these frequency $K_i$s. Thus, the average frequency and probability of ordinary edges are smaller than those of most $OHC$ edges based on frequency $K_i$s. %Theorem \ref{th30} is partially verified by the experiments. 
In the frequency $K_n$s if $n > 4$, the frequencies of some ordinary edges are equal to zero. As $n$ becomes big, the number of such ordinary edges rises accordingly. Since these ordinary edges are not used to build the $OP^n$s, they will be replaced by the $OHC$ edges and the ordinary edges with $f(g) > 0$ for building the $OP^n$s. Because the $OHC$ edges have the frequencies much bigger than the  ordinary edges with $f(g) > 0$ in the frequency $K_n$s, the ordinary edges  with $f(g) = 0$ are largely replaced by the $OHC$ edges to build the $OP^n$s. 

Thirdly, the biggest frequency of the ordinary edges is smaller than $2(i-3)$ in each frequency $K_n$. For some instances, the biggest frequency of the ordinary edges is close to $f_{lb}$ as $n$ is very small, such as $n$= 4, 5, 6. As $n$ rises, $f_{lb}$ becomes much bigger than the biggest frequency of the ordinary edges, and most of the ordinary edges with $f(g) > 0$ have the frequency  $f(g) < \frac{n+2}{2}$. Moreover, the average frequency of the ordinary edges with $f(g) > 0$ is also smaller than $\frac{n+2}{2}$ in each frequency $K_n$. %Except the ordinary edges excluding from any one $OP^n$, most remained ordinary edges in each frequency $K_n$ have the frequency smaller than $\frac{n+2}{2}$. 
It implies that most ordinary edges with $f(g) > 0$ are also replaced by the $OHC$ edges to build the $OP^n$s. Thus, the frequencies of the $OHC$ edges are much bigger than those of  ordinary edges. Theorem \ref{th2} works well for the $K_n$s containing random distances and small real-life $TSP$ instances. It indicates that the $OHC$ edges and ordinary edges have different structure properties. The different structure properties of them are fully characterized by the frequencies of edges computed with the $OP^n$s. Thus, the $OHC$ edges and ordinary edges can be separated from each other based on their frequencies in the frequency $K_n$. 

%The zero frequency is not shown in Table \ref{Examp}. We count the number of edges with frequency above zero and note the number as $N_f$. The total number of edges is noted as $N_w = {{n}\choose{2}}$. 

Based on the datum in Table \ref{Examp}, the average frequency of the $OHC$ edges and that of the ordinary edges are computed according to each frequency $K_n$ and given in Table \ref{averf}. $F_{tot}$, $F_{ohc}$ and $F_{ord}$ denote the total frequency of all edges, $OHC$ edges and ordinary edges in each frequency $K_n$, respectively. $f_{ohc}$ and $f_{ord}$ are the average frequency of the $n$ $OHC$ edges and $\frac{n(n-3)}{2}$ ordinary edges, respectively. $f_{oavg} = \frac{n^2-4n+7}{2}$ denotes the lower bound of the average frequency of all $OHC$ edges in each frequency $K_n$. $r_f = \frac{f_{ohc}}{f_{ord}}$ is the ratio between $f_{ohc}$ and $f_{ord}$. Firstly, although both $F_{ohc}$ and $F_{ord}$ increase according to $n$, $F_{ohc}$ approaches $F_{tot}$ and it is much bigger than $F_{ord}$. Meanwhile, $f_{ohc}$ is much bigger than $f_{ord}$, see the value of $r_f$. It means that an $OHC$ edge is contained in much bigger number of the $OP^n$s than an ordinary edge in $K_n$ on average. In addition, $f_{ohc} > f_{oavg}$ holds according to $n\in [4,14]$. Although the smallest frequency of the $OHC$ edges approaches $f_{lb} = \frac{1}{2}{{i}\choose{2}}$, the average frequency of all $OHC$ edges is closer to ${{i}\choose{2}} - 1$. It implies that the number of the $OHC$ edges having the  frequency near $f_{lb}$ is limited. Otherwise, the average frequency of all $OHC$ edges will be close to $f_{lb}$ rather than bigger than $f_{oavg}$. It indicates most $OHC$ edges will have the frequency much bigger than $f_{lb}$ in the frequency $K_n$. On average, an $OHC$ edge is contained in more than $f_{oavg} = \frac{n^2-4n+7}{2}$ $OP^n$s in each frequency $K_n$. On the other hand, $f_{ord}\in[1,2]$ implies that an ordinary edge is contained in less than two $OP^n$s in each of the frequency $K_n$s. Moreover, the average frequency of the ordinary edges with $f(g) > 0$ in each frequency $K_n$ is smaller than $\frac{i+2}{2}$. In addition, $f_{ord}$ approaches 2 according to $n$, which means that an ordinary edge is also included in more number of $OP^n$s according to $n$ on average. Theorems \ref{th1}, \ref{th2} and \ref{th22} are totally approved by the experimental results.% frequencies of $OHC$ edges in Table \ref{Examp} also 
%Among the ${{n}\choose{2}}$ $$
%$n$ of them are included in the $OHC$.

\begin{table}
	\begin{center}
		\caption{The average frequency of $OHC$ edges and that of ordinary edges according to the frequency $K_n$s in Table \ref{Examp}.}
		{\footnotesize \begin{tabular}{ p{0.5cm}  p{0.5cm}  p{0.5cm}  p{0.5cm}  p{0.5cm}  p{0.5cm}  p{0.5cm}  p{0.5cm}  p{0.5cm}  p{0.5cm}  p{0.5cm}  p{0.5cm} }
		\hline
		$i$ & 4 & 5 & 6 & 7 & 8 & 9 & 10 & 11 & 12 & 13 & 14  \\
		\hline
		$F_{tot}$ & 18 & 40 & 75 & 126 & 196 & 288 & 405 & 550 & 726 & 936 & 1183  \\		
		$F_{ohc}$ & 16 & 35 & 63 & 106 & 166 & 244 & 347 & 470 & 629 & 812 & 1031  \\	
		$F_{ord}$ & 2 & 5 & 12 & 20 & 30 & 44 & 58 & 80 & 97 & 124 &  152   \\			
		$f_{ohc}$ & 4 & 7 & 10.5 & 15.14 & 20.75 & 27.11 & 34.7 & 42.73 & 52.42 & 62.46 & 73.64   \\
		$f_{ord}$ & 1 & 1 & 1.33 & 1.43 & 1.5 & 1.63 & 1.66 & 1.82 & 1.80 & 1.91 & 1.97     \\  		
		$f_{oavg}$ & 3.5 & 6 & 9.5 & 14 & 19.5 & 26 & 33.5 & 42 & 51.5 & 62 & 73.5   \\	
		$r_f$ & 4 & 7 & 7.88 & 10.6 & 13.83 & 16.64 & 20.94 & 23.5 & 29.18 & 32.74 &  37.30  \\															
		%\bottomrule
		\hline
		\end{tabular}}
		\label{averf}
	\end{center}
\end{table}

To compare the frequencies of $OHC$ edges and ordinary edges in local graph structures, the total frequency of the two $OHC$ edges and that of the ordinary edges containing each vertex are computed for the 11 $TSP$ instances in Table \ref{Examp}, and the results are illustrated in Table \ref{twotimes}. In the first column of Table \ref{twotimes}, the vertices in each $K_n$ are denoted by natural numbers $0, 1, 2, ..., n$. In each of the following columns, the total frequency of the two $OHC$ edges containing each vertex in one frequency $K_n$ is given before that of the associated ordinary edges, respectively. There is one slash between the two total frequencies for each vertex. In the last three rows, the three lower frequency bounds for two adjacent $OHC$ edges $f_{lb_1} = \frac{4(n-1)^2}{5}$, $f_{lb_2} = \frac{3(n-1)^2}{4}$ and $f_{lb_3} = \frac{7(n-1)^2}{10}$ are computed according to Table \ref{fbounds} for comparisons. 
\begin{table}
	\begin{center}
		\caption{The total frequency of the two $OHC$ edges and that of the ordinary edges containing each vertex in the frequency $K_n$s ($n\in[4,14]$).}
		{\footnotesize \begin{tabular}{ p{0.5cm}  p{0.4cm}  p{0.4cm}  p{0.4cm}  p{0.4cm}  p{0.6cm}  p{0.6cm}  p{0.6cm}  p{0.6cm}  p{0.7cm}  p{0.8cm}  p{0.8cm} }
		\hline
		Vertex &  &  &  &  &  & $n$ &  &  &  &  &   \\
		No.& 4 & 5 & 6 & 7 & 8 & 9 & 10 & 11 & 12 & 13 & 14  \\
		\hline
		0 & 8/1 & 14/2 & 21/4 & 27/9 & 44/5 & 59/5 & 75/6 & 94/6 & 115/6 & 138/6 & 163/6  \\
		1 & 8/1 & 13/3 & 23/2 & 28/8 & 37/12 & 62/2 & 79/2 & 98/2 & 119/2 & 142/2 & 167/2  \\
		2 & 8/1 & 15/1 & 22/3 & 27/9 & 40/9 & 43/21 & 55/26 & 68/32 & 86/35 & 106/38 & 128/41  \\	
		3 & 8/1 & 15/1 & 18/7 & 32/4 & 45/4 & 46/18 & 60/21 & 76/24 & 93/28 & 113/31 & 135/34  \\
		4 &  & 13/3 & 22/3 & 33/3 & 44/5 & 57/7 & 71/10 & 91/9 & 107/14 & 120/24 & 142/27  \\
		5 &  &  & 20/5 & 31/5 & 45/4 & 57/7 & 73/8 & 80/20 & 100/21 & 121/23 & 144/25  \\	
		6 &  &  &  & 34/2 & 35/14 & 49/15 & 62/19 & 77/23 & 95/26 & 115/29 & 137/32  \\
		7 &  &  &  &  & 42/7 & 63/1 & 77/4 & 93/7 & 116/5 & 138/6 & 162/7  \\
		8 &  &  &  &  &  & 52/12 & 75/6 & 88/12 & 109/12 & 132/12 & 157/12  \\	
		9 &  &  &  &  &  &  & 67/14 & 85/15 & 97/24 & 115/29 & 134/32  \\
		10 &  &  &  &  &  &  &  & 87/13 & 106/15 & 113/31 & 126/43  \\
		11 &  &  &  &  &  &  &  &  & 115/16 & 128/16 & 138/31  \\		
		12 &  &  &  &  &  &  &  &  &  & 143/1 & 168/1  \\		
		13 &  &  &  &  &  &  &  &  &  &  &  162/7 \\	
		$f_{lb_1}$ & 7.2 & 12.8 & 20 & 28.8 & 39.2 & 51.2 & 64.8 & 80 & 96.8 & 115.2 &  135.2 \\ 		
		$f_{lb_2}$ & 6.75 & 12 & 18.75 & 27 & 36.75 & 48 & 60.75 & 75 & 90.75 & 108 &  126.75 \\ 		
		$f_{lb_3}$ & 6.3 & 11.2 & 17.5 & 25.2 & 34.3 & 44.8 & 56.7 & 70 & 84.7 & 100.8 &  118.3 \\ 																				
		%\bottomrule
		\hline
		\end{tabular}}
		\label{twotimes}
	\end{center}
\end{table}

For most vertices in each frequency $K_n$, the total frequency of two adjacent $OHC$ edges is bigger than $f_{lb_1}$ and $f_{lb_2}$, and it is bigger than $f_{lb_3}$ for nearly all vertices in these frequency $K_n$s. Only for the vertex 2 in $K_9$ and $K_{10}$, the total frequency of the two adjacent $OHC$ edges is slightly smaller than $f_{lb_3}$. %The total frequency of the two $OHC$ edges is bigger than four times of that of the associated ordinary edges. 
As $n=$ 4 and 5, the total frequency of the two $OHC$ edges is bigger than four times of that of the associated ordinary edges for every vertex. As $n$ becomes big, the total frequency of the two $OHC$ edges is smaller than four times of that of the associated ordinary edges for a small number of vertices, but it is bigger than two or three times of that of the associated ordinary edges for these vertices. Moreover, the total frequency of the two $OHC$ edges becomes much bigger than that of the associated ordinary edges for most vertices according to $n$. It says that the percentage that the $OHC$ edges are contained in the $OP^n$s increases quickly according to $n$. On the other hand, the ordinary edges will be contained in the smaller percentage of $OP^n$s according to $n$. The experimental results illustrated that not only all $OHC$ edges and ordinary edges in $K_n$, but also the two $OHC$ edges and ordinary edges containing each vertex conform to Theorems \ref{th01}, \ref{th1}. 

\subsubsection{The frequency change for the replaced $OHC$ edges from frequency $K_n$ to $K_{n+1}$ for the small $TSP$ instances}
In Table \ref{Examp} from $n$=9 to 14, the previous graph $K_n$ is contained in the next bigger graph $K_{n+1}$. As the vertex $n$ is added to $K_n$ to construct the $K_{n+1}$, some $OHC$ edges in the $K_n$ will be replaced by the other new edges for building the $OHC$ in the $K_{n+1}$. If the $OHC$ edges in the $K_n$ becomes ordinary edges in the $K_{n+1}$, it will have the frequency smaller than $f_{lb}$ in the frequency $K_{n+1}$. The frequencies of the replaced $OHC$ edges and new emerging $OHC$ edges in the frequency $K_n$ and $K_{n+1}$ will be investigated according to $n$=9$\sim$14 ($n-1\in [8,13]$). The experimental results are given in Table \ref{repf}. In Table \ref{repf}, the replaced $OHC$ edges and new $OHC$ edges are denoted with vertices, i.e., $(u,v)$ where $u\neq v$ and $u, v\in [0,13]$. 

\begin{table}
	\begin{center}
		\caption{The frequencies of the replaced $OHC$ edges and new $OHC$ edges from the frequency $K_n$ to the frequency $K_{n+1}$.}
		{\footnotesize \begin{tabular}{ p{0.4cm}  p{0.4cm}  p{0.4cm}  p{0.4cm}  p{0.4cm}  p{0.4cm}  p{0.4cm}  p{0.6cm}  p{0.6cm}  p{0.6cm}  p{0.6cm}  p{0.6cm} p{0.6cm}}
				\hline
		$n$ &  &  &  &  &  & edges &  &  &  &  &   \\
		& (4,8) & (4,9) & (9,8) & (4,10) & (10,9) & (4,11) & (11,10) & (11,12) & (12,10) & (12,13) & (13,10) & $f_{lb}$  \\				
				\hline
		9 & 22 &  &  &  &  &   &  &  &  &  &  & 18  \\		
		10 & 2 & 27 & 40 &  &  &  &  &  &  &  &  & 22.5  \\	
		11 & 0 & 0 & 45 & 47 & 40 &  &  &  &  &  &  & 27.5  \\		
		12 & 0 & 1 & 56 & 7 & 41 & 50 & 65 &  &  &  &  &  33 \\	
		13 & 0 & 1 & 68 & 13 & 47 & 51 & 10 & 77 & 66 &  &  &  39 \\		
		14 & 0 & 1 & 81 & 14 & 53 & 60 & 21 & 78 & 1 & 90 & 72 &  45.5 \\					 				
				%\bottomrule
				\hline
		\end{tabular}}
		\label{repf}
	\end{center}
\end{table}

For example, $(4,8)$ is the $OHC$ edge in the $K_9$. As the vertex 9 is added to the $K_9$, $(4,8)$ is broken while two new edges $(4,9)$ and $(9,8)$ are used to construct the $OHC$ in the $K_{10}$. As $(4,8)$ is the $OHC$ edge in the $K_9$, it has the frequency of 22 which is bigger than $f_{lb}$. However, the frequency of $(4,8)$ becomes 2 in the frequency $K_{10}$ where it becomes one ordinary edge. The frequency of (4,8) drops sharply from the frequency $K_9$ to frequency $K_{10}$. In the following frequency $K_n$s ($n\in [11,14]$), it is an ordinary edge and the frequency becomes zero. 

On the other hand, the two new emerging $OHC$ edges $(4,9)$  and $(9,8)$ in the $K_{10}$ have the frequencies of 27 and 40 in the frequency $K_{10}$, respectively. Both of the frequencies are bigger than $f_{lb}$ in the frequency $K_{10}$. One can observe the frequencies of the other replaced $OHC$ edges and new emerging $OHC$ edges in the frequency $K_n$s and $K_{n+1}$s. The replaced $OHC$ edges have the frequency much smaller than $f_{lb}$ in the frequency $K_{n+1}$s once they become ordinary edges, whereas the new emerging $OHC$ edges have the frequency much bigger than $f_{lb}$ in the frequency $K_{n+1}$s. For the $OHC$ edges contained in the $K_n$ and $K_{n+1}$, they preserve the frequency bigger than $f_{lb}$ in the frequency $K_n$ and $K_{n+1}$, such as the edges (9,8) from the frequency $K_{10}$ to the frequency $K_{14}$, (10,9) from the frequency $K_{11}$ to the frequency $K_{14}$, etc. Theorems \ref{th1} and \ref{th2} are verified with these experimental results. 

In addition, among the illustrated $OHC$ edges in these $K_n$s in Table \ref{repf}, $(8,9)$, $(9,10)$, $(10,11)$, $(11,12)$ and $(12,13)$ are also the $OHC$ edges of Oliver30. $(8,9)$, $(9,10)$ and $(11,12)$ maintain the big frequencies in these frequency $K_n$s according to $n$. It indicates that an $OHC$ edge in $K_n$ is generally the $OHC$ edge in the $K_i$s containing it. For the edge $(10,11)$, although it is an ordinary edge in the  $K_{13}$ and $K_{14}$, the frequency is near or bigger than 13 and 14 in the frequency $K_{13}$ and $K_{14}$. It means that the frequency of an $OHC$ edge in $K_n$ will not be too small (for example, 0, 1, 2, etc.) even if it is an ordinary edge in some frequency $K_i$s containing it. In most of the other frequency $K_i$s, it will have the even bigger frequency to approve the average frequency to be bigger than $f_{lb}$. 
%old edges in $K_i$
%new edges in $K_{i+1}$
%具有最小频率的边不一定是同一条边

\subsubsection{The lower frequency bounds for $OHC$ edges computed based on frequency $K_i$s for small $TSP$ instances}
In the next, we shall verify the lower frequency bound for $OHC$ edges computed with the frequency $K_i$s according to $i\in[4,n]$. The experiments for the small $TSP$ instances containing $n=9\sim 14$ vertices in Table \ref{Examp} are executed to verify the lower frequency bound of $OHC$ edges computed with the frequency $K_i$s. Firstly, the frequency $K_i$s containing every edge are computed for given $K_n$. Then, the average frequency of each edge is computed with these frequency $K_i$s, respectively. At last, the edges in each frequency $K_n$ are ordered according to their average frequencies from big to small values. Based on the known $OHC$, the frequencies of the $OHC$ edges are compared with the $f_{lb}$, and the number of the edges with the average frequency bigger than that of some specified $OHC$ edge is also computed. The experimental results are given in Tables \ref{smallf}$\sim$\ref{naminf}.

In Table \ref{smallf}, the smallest average frequency of the  $OHC$ edges in each $K_n$ is shown according to $i$ from 4 to $n$. In the last row, the lower frequency bound $f_{lb}$ is given for comparisons. As $n$ is small, such as $n$=9, 10, the smallest average frequency of the $OHC$ edges is smaller than but close to $f_{lb}$. As $n$ rises, the smallest average frequency of the $OHC$ edges computed with the frequency $K_i$s becomes bigger accordingly. For example, as the frequency of each edge is computed with the frequency $K_4$s, the smallest frequency of the $OHC$ edges becomes from 2.90 to 3.42 as $n$ rises from 9 to 14. As $n$=9 and 10, the smallest average frequency of the $OHC$ edges is smaller than $f_{lb}$, and it becomes bigger than $f_{lb}$ as $n$=11, 12, 13 and 14. As the frequency of an edge is computed with the frequency $K_i$s containing more vertices, the same results are obtained. It indicates that the average frequency of an $OHC$ edge rises according to $n$ as it is computed with the frequency $K_i$s for given $i$ on average. It means that the $OHC$ edges will be contained in more and more percentage of $OP^i$s according to $n$.   If $K_n$ always contains an $OHC$ edge and $n$ grows, the average frequency of the $OHC$ edge computed with the frequency $K_i$s will increase according to $n$. The previous part of Theorem \ref{th3} is approved. 

Moreover, the smallest average frequency of the $OHC$ edges increases very fast according to $i$ as $n$ is fixed. For example $n$=13 and 14, the smallest average frequency of the  $OHC$ edges is smaller than $f_{lb}$ as $i$ is small. However, the smallest average frequency becomes bigger than $f_{lb}$ as $i$ rises to one relatively big number, such as $i\geq 11$. It implies that the smallest average frequency of the $OHC$ edges increases faster than $f_{lb}=\frac{1}{2}{{i}\choose{2}}$ according  to $i$. If the smallest average frequency of the $OHC$ edges is smaller than $f_{lb}$ as $i$ is small, it can be predicted that it will be bigger than $f_{lb}$ at certain number $i$, and it will always be bigger than $f_{lb}$ after this number $i$. If we compare $p_{i+1}(e)$ to $p_i(e)$, we will find $p_{i+1}(e) > \left[1 + \frac{2}{i(i+1)}\right]p_i(e)$ exists in most cases for the $OHC$ edge having the smallest average frequency for so small $TSP$ instances. For example, $p_{i+1}(e) > \left[1 + \frac{2}{i(i+1)}\right]p_i(e)$ holds from $i = 7$ for $K_9$ and $K_{10}$, from $i = 9$ for $K_{13}$ and $K_{14}$. It indicates that the $OHC$ edges will be contained in more and more percentage of $OP^i$s according to $i$ as $n$ is fixed. Theorem \ref{th3} is totally verified by the experiments. In real-world applications, the frequencies of edges computed with the frequency $K_i$s containing more vertices are better to characterize $OHC$ edges. On the other hand, it consumes more time to compute the $OP^i$s and frequency $K_i$s.  

\begin{table}
	\begin{center}
		\caption{The smallest average frequency of $OHC$ edges computed with frequency $K_i$s.}
		{\footnotesize \begin{tabular}{ p{0.4cm}  p{0.4cm}  p{0.4cm}  p{0.4cm}  p{0.4cm}  p{0.6cm}  p{0.6cm}  p{0.6cm}  p{0.6cm}  p{0.6cm}  p{0.6cm}  p{0.6cm} }
		\hline
		$n$ &  &  &  &  &  & $i$ &  &  &  &  &   \\
		 & 4 & 5 & 6 & 7 & 8 & 9 & 10 & 11 & 12 & 13 & 14  \\				
		\hline
		9 & 2.90 & 3.97 & 5.49 & 8.05 & 11.86 & 16 &  &  &  &  &   \\		
		10 & 2.93 & 4.05 & 5.64 & 8.13 & 11.55 & 15.63 & 20 &  &  &  &   \\	
		11 & 3.11 & 4.46 & 6.33 & 8.88 & 12.08 & 15.86 & 20 & 24 &  &  &   \\		
		12 & 3.22 & 4.73 & 6.68 & 9.22 & 12.43 & 16.43 & 21.18 & 26.4 & 32 &  &   \\	
		13 & 3.36 & 5.04 & 7.13 & 9.75 & 12.97 & 16.97 & 21.87 & 27.6 & 34 & 41 &   \\		
		14 & 3.42 & 5.22 & 7.38 & 9.96 & 13.15 & 17.14 & 22.11 & 28.18 & 34.94 & 42.58 & 51 \\	 	
		$f_{lb}$ & 3 & 5 & 7.5 & 10.5 & 14 & 18 & 22.5 & 27.5 & 33 & 39 & 45.5 \\	 																	
		%\bottomrule
		\hline
		\end{tabular}}
		\label{smallf}
	\end{center}
\end{table}

%仅有少量OHC边的频率较小，大部分OHC边频率很大
%有些一般边的频率大于OHC边的频率
%OHC边和一般边的Pi变化过程
%i=0.366n-4
\subsubsection{The number of ordinary edges with frequency bigger than the smallest frequency of $OHC$ edges for small $TSP$ instances}
In the frequency $K_n$, the frequency of each ordinary edge is smaller than that of an $OHC$ edge. Moreover, the number of the $OHC$ edges having the frequency near the lower frequency bound $f_{lb}$ will be very small. In addition, several smallest frequencies of the $OHC$ edges may be close for some intricate $TSP$. Since the average frequency considering all $OHC$ edges is bigger than $f_{oavg}$, the frequencies of the $OHC$ edges will rise quickly from $f_{lb}$ to $f_{oavg}$, see the experimental results in Table \ref{Examp}. The results will be different for the $OHC$ edges in $K_n$ as the frequency of each edge is computed with the frequency $K_i$s if $i\ll n$. In this case, some ordinary edges may have the frequencies bigger than those of some $OHC$ edges (see Theorems \ref{th33} and \ref{th30}). To illustrate the difference between the frequencies of the $OHC$ edges and those of ordinary edges computed with the frequency $K_i$s, Table \ref{noe} illustrates the number of edges with the frequency bigger than the first and third smallest frequency of the $OHC$ edges, and Table \ref{naminf} gives the number of such ordinary edges. 

As $i$ is small, such as $i$=4, there are many ordinary edges having the frequencies bigger than the first smallest frequency of the $OHC$ edges. The reason is that these ordinary edges are contained in the big percentage of the frequency $K_4$s without containing one or two pairs of the adjacent $OHC$ edges in $K_n$. In this case, these ordinary edges are the $OHC$ edges in most of the frequency $K_4$s and have the big frequency of 3 or 5. Thus, these ordinary edges have the average frequencies bigger than the first smallest frequency of the $OHC$ edges at $i=4$. If $i$ becomes big, the number of such ordinary edges decreases quickly. For example at $i$=7, there are only a few such ordinary edges for each $K_n$ in Table \ref{naminf}. If $i$ is big, an ordinary edge is contained in the big percentage of the frequency $K_i$s where the ordinary edge is adjacent to one or two pairs of the $OHC$ edges in $K_n$. At this time, the ordinary edges will not be the $OHC$ edges and they have the frequency smaller than $i-3$ in these frequency $K_i$s. In this case, the average frequency of most ordinary edges becomes small at the big number $i$. On the other hand, an $OHC$ edge is contained in the bigger or nearly equal percentage of $OP^i$s, and the average frequency increases quickly to the even bigger value. Thus, the number of the ordinary edges with the frequency bigger than the smallest frequency of the $OHC$ edges becomes much smaller at the big number $i$. Theorem \ref{th33} is approved by these experimental results.

In addition, the number of total edges or ordinary edges with the average frequency bigger than the first smallest frequency of the $OHC$ edges is small comparing to the total number of edges $N_{tot}$ or ordinary edges $N_{ord}$. It says that most ordinary edges having the frequency smaller than the smallest frequency of the $OHC$ edges. Moreover, the number of such ordinary edges becomes smaller according to $i$ until it equals zero as $i$ becomes relatively big. It indicates that the average frequencies of the $OHC$ edges rise faster than those of the ordinary edges according to $i\in[4,n]$ as well as the probabilities that they are contained in the $OP^i$s. Thus, the frequencies of the ordinary edges must be smaller than that of an $OHC$ edge at certain $i$. It mentions that although some ordinary edges have the  frequencies bigger than the smallest frequency of the $OHC$ edges as $i$ is small, their frequencies are smaller than those of most other $OHC$ edges.  

As the third smallest frequency of the $OHC$ edges are used to filter out the ordinary edges according to their average frequencies, the number of the preserved ordinary edges becomes even smaller. As $i$ is small, for example $i$=4, the number of the preserved ordinary edges according to the third smallest frequency of the $OHC$ edges becomes much smaller than that according to the first smallest frequency. It indicates that the third smallest frequency is much bigger than the first smallest frequency of the $OHC$ edges. Moreover, the third smallest frequency rises faster than that of most ordinary edges according to $i$. For example as $i$=6, there are no ordinary edges having the frequency bigger than the third smallest frequency for $n$=9, 10 and 11, and there are only 1, 3 and 4 ordinary edges having the  frequency bigger than the third smallest frequency for $n$=12, 13 and 14. When $i$=8, there are only two ordinary edges having the average frequency bigger than the third smallest frequency for $n$=14. For the other instances where $n\in[9,13]$ if $i\geq 8$, the average frequency of no ordinary edges is bigger than the third smallest frequency of the $OHC$ edges. 

\begin{table}
	\begin{center}
		\caption{The number of edges with the frequency bigger than the first/third smallest frequency of the $OHC$ edges.}
		{\footnotesize \begin{tabular}{ p{0.4cm}  p{0.6cm}  p{0.6cm}  p{0.6cm}  p{0.6cm}  p{0.4cm}  p{0.4cm}  p{0.4cm}  p{0.4cm}  p{0.4cm}  p{0.4cm}  p{0.4cm}  p{0.4cm} }
		\hline
		$n$	&  &  &  &  &  & $i$ &  &  &  &  &  & $N_{tot}$ \\
		& 4 & 5 & 6 & 7 & 8 & 9 & 10 & 11 & 12 & 13 & 14 &  \\				
		\hline
		9 & 16/7 & 14/6 & 12/6 & 10/6 & 8/6 & 8/6 &  &  &  &  &  &  36 \\		
		10 & 20/8 & 18/8 & 16/7 & 11/7 & 9/7 & 9/7 & 9/7 &  &  &  &  &  45 \\	
		11 & 22/9 & 19/8 & 18/8 & 13/8 & 10/8 & 10/8 & 10/8 & 10/8 &  &  &  & 55 \\		
		12 & 23/14 & 19/12 & 18/10 & 16/9 & 12/9 & 11/9 & 11/9 & 11/9 & 11/9 &  &  &  66 \\	
		13 & 25/17 & 22/16 & 20/13 & 18/10 & 15/10 & 12/10 & 12/10 & 12/10 & 12/10 & 12/10 &  &  78 \\	
		14 & 27/20 & 25/20 & 23/15 & 22/13 & 16/11 & 15/11 & 13/11 & 13/11 & 13/11 & 13/11 & 13/11 &  91  \\					
		%\bottomrule
		\hline
		\end{tabular}}
		\label{noe}
	\end{center}
\end{table}

\begin{table}
	\begin{center}
		\caption{The number of ordinary edges with the frequency bigger than the first/third smallest frequency of the $OHC$ edges.}
		{\footnotesize \begin{tabular}{ p{0.4cm}  p{0.6cm}  p{0.6cm}  p{0.6cm}  p{0.6cm}  p{0.4cm}  p{0.4cm}  p{0.4cm}  p{0.4cm}  p{0.4cm}  p{0.4cm}  p{0.4cm}  p{0.4cm} }
		\hline
		$n$	&  &  &  &  &  & $i$ &  &  &  &  &  & $N_{ord}$ \\
		 & 4 & 5 & 6 & 7 & 8 & 9 & 10 & 11 & 12 & 13 & 14 &  \\				
		\hline
		9 & 8/1 & 6/0 & 4/0 & 2/0 & 0/0 & 0/0 &  &  &  &  &  &  27 \\		
		10 & 11/2 & 9/1 & 7/0 & 2/0 & 0/0 & 0/0 & 0/0 &  &  &  &  &  35 \\	
		11 & 12/1 & 9/0 & 8/0 & 3/0 & 0/0 & 0/0 & 0/0 & 0/0 &  &  &  & 44 \\		
		12 & 12/5 & 8/3 & 7/1 & 5/0 & 1/0 & 0/0 & 0/0 & 0/0 & 0/0 &  &  &  54 \\	
		13 & 13/7 & 10/6 & 8/3 & 6/0 & 3/0 & 0/0 & 0/0 & 0/0 & 0/0 & 0/0 &  &  65 \\	
		14 & 14/9 & 12/9 & 10/4 & 9/2 & 3/0 & 2/0 & 0/0 & 0/0 & 0/0 & 0/0 & 0/0 &  77  \\					
		%\bottomrule
		\hline
		\end{tabular}}
		\label{naminf}
	\end{center}
\end{table}

%frequency and probability change 
%两个发现:(1)  i=4时，OHC占的比例随n减小，所以n很大时i=4效果不好
%(2)pi(g)平均和最大值随i下降，没有上升的趋势,虽然p_min(e)有下降，但下降幅度很小，且有限次下降，最后随i有上升趋势，最终大于p_max(g)
The experimental results demonstrate that there are only a small number of $OHC$ edges having the average frequency near the lower frequency bound $f_{lb}$ or the smallest frequency of the $OHC$ edges, and the frequencies of the other $OHC$ edges are much bigger than $f_{lb}$ or the smallest frequency of the $OHC$ edges.  Although there are many ordinary edges having the average frequency bigger than $f_{lb}$ or the smallest frequency of the $OHC$ edges as $i$ is small, the number of the ordinary edges having the average frequency bigger than the third smallest frequency of the $OHC$ edges is very small. One can use certain frequency threshold bigger than $f_{lb}$ to filter out the ordinary edges,  much more ordinary edges will be abandoned, whereas only a small number of $OHC$ edges are  neglected. Moreover, the frequencies of the $OHC$ edges rises faster than those of most ordinary edges according to $i$. As the frequency of each edge is computed with the frequency $K_i$s containing more vertices, one can use the even bigger frequency threshold to filter out more ordinary edges while more $OHC$ edges will be preserved. At $i$=6, if the third smallest frequency of the $OHC$ edges are used to eliminate the edges with the smaller frequency, 7, 8, 9 $OHC$ edges are preserved  and all ordinary edges are eliminated for $n$=9, 10 and 11 in Table \ref{noe}. At $i$=7, 7, 8, 9, 10, 11 and 12 $OHC$ edges are preserved for $n\in[9,14]$ according to the third smallest frequency of the $OHC$ edges. All the ordinary edges are eliminated for $n$=9, 10, 11, 12 and 13, and only 2 ordinary edges are preserved for $n$= 14. %according to the third smallest frequency of the $OHC$ edges. %and most of the preserved edges are $OHC$ edges.

\subsubsection{The frequency and probability changes for $OHC$ and ordinary edges according to $i$ for small $TSP$ instances}
In the following, the frequency and probability changes for $OHC$ edges and ordinary edges in $K_n$ will be illustrated according to $i$. The experiments are executed for the six small $TSP$ instances where $n\in [9,14]$ in Table \ref{Examp}. Given $K_n$, the frequency of an edge is computed with the frequency $K_i$s from $i$=4 to $n$, respectively. The experimental results are given in Table \ref{chgf}. In Table \ref{chgf}, the total frequency of all edges is denoted by $F_{tot}$, and $F_{ohc}$ denotes the total frequency of all $OHC$ edges for $i\in[4,n]$. $p_{ohc} = \frac{F_{ohc}}{F_{tot}}\times 100\%$ represents the percent of $F_{ohc}$. $p_i(e)$ and $p_i(g)$ denote the average probability of the $n$ $OHC$ edges and $\frac{n(n-3)}{2}$ ordinary edges, respectively. $p_i(e)$ or $p_i(g)$ also represents the probability that an $OHC$ edge $e$ or an ordinary edge $g$ is contained in the $OP^i$s based on the $K_i$s containing $e$ or $g$ in the average case. In addition, $p_{min}(e)$ denotes the minimum probability of the $OHC$ edges, and $p_{max}(g)$ denotes the maximum probability of the ordinary edges according to $i$.

Firstly, $F_{ohc}$ approaches $F_{tot}$ according to $i$. One also observes that $p_{ohc}$ becomes bigger according to $i$. Moreover, $p_i(e)$ also increases simultaneously and it is bigger than $\frac{i^2-4i+7}{i(i-1)}$ according to $i$. It indicates that the $OHC$ edges are  contained in more and more percentage of $OP^i$s according to $i$, and an $OHC$ edge is contained in more than $f_{oavg} = \frac{i^2-4i+7}{2}$ $OP^i$s based on each $K_i$ containing it on average. Since $p_i(e) > \frac{i^2 - 4i + 7}{i(i-1)}$, $p_{i+1}(e) < \left[1 + \frac{3}{i(i+1)}\right]p_i(e)$ exists from $i$ to $i+1$. On the other hand, $p_i(g)$ decreases according to $i$, and it is smaller than $\frac{2(n-1)}{i(n-3)} - \frac{2(i^2-4i+7)}{i(i-1)(n-3)}$ according to $i$. Even if at $i$=4, $p_i(g) < \frac{1}{2}$ exists for these small $TSP$ instances. It implies that an ordinary edge is contained in smaller than $\frac{1}{2}{{i}\choose{2}}$ $OP^i$s in a $K_i$ containing it on average. Moreover, $p_i(g)$ decreases sharply in proportion to certain factor smaller than $\frac{i}{i+1}$ from $i$ to $i+1$, i.e., $p_{i+1}(g) < \frac{i}{i+1}p_i(g)$. It means that the ordinary edges are contained in smaller and smaller percentage of $OP^i$s according to $i$. In addition, $p_i(e)$ rises according to $n$ for given $i$. It indicates that the $OHC$ edges are contained in more and more percentage of $OP^i$s according to $n$.  Since $p_{ohc}(e)$ rises according to $i$, $p_{ohc}(e)$ and $p_i(e)$ reach the maximum value at $i=n$, respectively. As $n$ is big, $p_{ohc}(e)$ and $p_i(e)$ will approach 1 at $i=n$, and $p_i(g)$ will tend to zero. Theorem \ref{th3} is approved with the experimental results. 
%Moreover, $p_i(e)>\frac{11}{18}>\frac{1}{2}$ holds according to $i$ and $n$. It means that most $OHC$ edges are contained in much more than half $OP^i$s in a $K_i$ containing them, respectively. 

It mentions that $p_{ohc}$ is much smaller than 1 if $i$ is small, such as $i=4$. It means that the ordinary edges occupy the big percentage of the edges in the $OP^i$s containing a small number of vertices. Moreover, $p_{ohc}$ decreases according to $n$ for given $i$. It indicates that all  ordinary edges occupy more percentage of the edges in the $OP^i$s according to $n$ for given $i$. In this case, there will be many ordinary edges having the frequency bigger than $f_{lb}$ for big $n$ and small $i$. The reason has been given in Theorems \ref{th33} and \ref{th30}. Thus, if $i$ is small while $n$ is very big, the frequency of edges computed with the frequency $K_i$s is not rational for separating $OHC$ edges from some ordinary edges. The frequencies of edges computed with the frequency $K_i$s containing more vertices are better to characterize and identify the $OHC$ edges.

As $i$ rises, $p_{ohc}$ and $p_i(e)$ increases quickly according to $i$ whereas $p_i(g)$ decreases sharply, and it becomes very small as $i$ is  big. Moreover, $p_i(g)$ decreases according to $n$ at $i=n$. It means that more and more ordinary edges are excluded from the $OP^i$s containing more vertices. Since the $OHC$ edges have the bigger value of $p_i(e)$ according to $i$, these ordinary edges are largely replaced by the $OHC$ edges to construct the $OP^i$s. As $i=n$, only a portion of the ordinary edges are contained in the $OP^n$s, and the percentage of the ordinary edges contained in the $OP^n$s tends to zero and it becomes smaller according to $n$. For the ordinary edges, Theorem \ref{th33} is verified with the experimental results. The experimental results demonstrated that the frequency changes or probability changes for the $OHC$ edges and ordinary edges are different according to $i$, which discloses that they have different structure properties with respect to $OHC$ and $OP^i$s. 
 %Even if $i$=4, $p_i(g) < \frac{1}{2}$ holds for these small $TSP$ instances. As $i$ rises, $p_i(g)$ decreases sharply and tends to zero as $i$ is close to $n$. 

If $i$ is small, there will be some ordinary edges having the probability (or frequency) bigger than the smallest probability (or frequency) of the $OHC$ edges. The smallest probability $p_{min}(e)$ of the $OHC$ edges and the maximum probability $p_{max}(g)$ of the ordinary edges are given according to $i$ in Table \ref{chgf}. $p_{min}(e)$ is usually smaller than $p_{max}(g)$ as $i$ is small for these small $TSP$ instances. However, $p_{min}$ becomes bigger than $p_{max}(g)$ as $i$ is relatively big. Moreover, although $p_{min}(e)$ decreases from $i$ to $i+1$ as $i$ is small, it will increase from certain $i$ to $n$. On the other hand, $p_{max}(g)$ always decreases according to $i\in[4,n]$.

In the next, the change of $p_{min}(e)$ will be analyzed according to $i$. $p_{min}(e)$ does not always rise according to $i$, and sometimes it decreases from certain $i$ to $i+1$. One sees the probability decrement is relatively big according to the small $i$s and it becomes smaller as $i$ rises. Moreover, the decrement is smaller than $\frac{2p_{min}(e)}{i(i-1)}$ from $i$ to $i+1$ in all the stages for these small $TSP$ instances except from $i$=4 to 5 for $K_9$. Even if for the smallest probabilities of the $OHC$ edges for the small $TSP$ instances, Theorem \ref{th3} still works well to predict their changes according to $i$, i.e., $p_{i+1}(e)\geq \left[1-\frac{2}{i(i-1)}\right]p_i(e)$ holds for all $OHC$ edges if $p_i(e)$ decreases. Although $p_{min}(e)$ decreases at some steps, the number of such steps is very small comparing to $n$. Moreover, $p_{min}(e)$ begins increasing according to the following $i$s after the smallest value until $i=n$. The experiments demonstrated that $p_{min}(e)$ reaches the smallest value at certain $i < 2i_d$ and it begins increasing after that. As $p_{min}(e)$ increases from $i$ to $i+1$, the probability increment is bigger than $\frac{2p_{min}(e)}{i(i-1)}$ in most cases for these small $TSP$ instances, for example from $i=9$ to 12 for $K_{13}$ and from $i=9$ to 13 for $K_{14}$. Theorem \ref{th3} demonstrates that  the smaller $p_i(e)$ of an $e\in OHC$ will increase faster according to $i$ if $p_i(e)$ increases from $i$ to $i+1$. It implies that the probability increment for the $OHC$ edges is generally bigger than the average probability increment computed for all edges contained in the $OP^n$s. On the other hand, the probability increment for the ordinary edges will be smaller than $\frac{2p_i(g)}{i(i-1)}$ in most cases if $p_i(g)$ increases from $i$ to $i+1$. The experiments illustrate that $p_{min}(e)$ changes steadily from $i$=4 to $n$. As $i< P_0=\frac{n}{2} + 2$ for even $n$ and $P_0 = \frac{n+1}{2} + 1$ for odd $n$, the small decrement of $p_{min}(e)$ from $i$ to $i+1$ does not prevent the frequency from increasing. Thus, the frequency of the related $OHC$ edge will increase according to $i$ until it reaches the peak value at $P_0$. For the other $OHC$ edges, their probability is bigger than $p_{min}(e)$ according to $i$. The probabilities of most $OHC$ edges will increase according to $i$, and these $OHC$ edges will have the even bigger peak frequency at $P_0$. %Since $p_{min}(e)$ decreases As $p_{min}(e)$ decreases from $i$ to $i+1$, $i$ isand it becomes bigger than $\frac{1}{2}$

Based on Theorems \ref{th3} and \ref{th33}, it is known that the $p_i(e)$ increases according to $i$ whereas $p_i(g)$ decreases according to $i$ in the average case. For the smallest frequency and probability of the $OHC$ edges, the changes will not be the same as that of $p_i(e)$. Firstly, as $n$ is small, such as $n$=9 and 10, $p_{min}(e)<\frac{1}{2}$ occurs and sometimes it is even smaller than $\frac{7}{18}$ for some particular $TSP$ instances(for example, the $TSP$ instances contain many equal-weight edges.). It says that the lower frequency bound $f_{lb} = \frac{1}{2}{{i}\choose{2}}$ does not work well for some $OHC$ edges for the very small $TSP$ instances. It is known that the (average) frequency of an $OHC$ edge rises according to $n$ for given $i$. As $n$ becomes big, such as $n\geq 11$, $p_{min}(e) > \frac{7}{18}$ exists and it is bigger than $\frac{1}{2}$ in most cases according to $i$ even if for these small $TSP$ instances. Thus, $f_{lb}= \frac{1}{2}{{i}\choose{2}}$ will work well in real-world applications as $n$ is big or most $K_i$s contain ${{i}\choose{2}}$ $OP^i$s. It mentions that the constraints of $OP^i$s are not considered for deriving the lower frequency bound for $OHC$ edges. If the constraints of $OP^i$s are taken into account, the lower frequency bound for $OHC$ edges will be improved \cite{DBLP:journals/Wang25}.

Thirdly, the $p_{max}(g)$ change will be analyzed according to $i$. In view of the experimental results of $p_{max}(g)$, $p_{max}(g) > \frac{1}{2}$ occurs as $i$ is small. It indicates that the ordinary edge is contained in more than half $OP^i$s according to the $K_i$s containing it. However, it always decreases according to $i\geq 4$ and never rises from $i$ to $i+1$ at any one step. Based on formula (\ref{F5}), the smallest $i_d$ = 4 is computed for $n = $9, 10 and 11, $i_d = 5$ is computed for $n = 12$, 13 and 14. All the changes of $p_{max}(g)$ conform to Theorem \ref{th30}. $p_{max}(g)$ becomes smaller than $\frac{1}{2}$ before the $2i_d$  and much smaller than $\frac{1}{2}$ at $i=n$ for these small $TSP$ instances. At first, $p_{max}(g)$ decreases slowly according to $i$. For example, the probability decrement is smaller than $\frac{2p_{max}(g)}{i(i-1)}$ from $i$=4 to 5, or from $i$=5 to 6. However, the decrement becomes bigger than $\frac{2p_{max}(g)}{i(i-1)}$ from $i$=5 for $n=9$, 10, and from $i=6$ for $n=$11, 12, 13 and 14. Moreover, the probability decrement increases according to $i$ for these small $TSP$ instances until $p_{max}(g)$ is near the smallest value. For example, $p_{max}(g\in OP^{i+1}) < \frac{i}{i+1}p_{max}(g\in OP^i)$ happens if $i\geq 6$ for $n=$9 and 10, $i\geq 7$ for $n=$11 and 12, $i\geq 8$ for $n=$13 and 14, respectively. Once sees $p_{max}(g\in OP^{i+1}) > \frac{i}{i+1}p_{max}(g\in OP^i)$ exists according to $i$ within a small number of steps. However, $p_{max}(g\in OP^{i+1}) < \frac{i}{i+1}p_{max}(g\in OP^i)$ exists according to $i$ in most cases. Once $p_{max}(g)$ decreases, $p_{max}(g\in OP^{i+1}) < \frac{i}{i+1}p_{max}(g\in OP^i)$ will appear after a small number of steps. Thus, the probability decrement for ordinary edges will be bigger than  $\frac{p_i(g)}{i+1}$ from $i$ to $i+1$ even if $i$ is not big. Once the probability decreases from certain $i$, the probability decrement becomes bigger according to the following $i$s, which accelerates the probability to become small quickly. This is the difference (i.e., different probability decrements) between $OHC$ edges and ordinary edges.  In addition, $p_{max}(g)$ is smaller than $p_i(e)$ according to $i$ for these small instances except $K_{14}$ at $i=4$. The error between $p_i(e)$ and $p_{max}(g)$ becomes bigger according to $i$. It indicates that $p_{max}(g)$ will be smaller than $p_i(e)$ in most cases even if $i$ is not big. Since $p_i(e)$ is bigger than $\frac{i^2-4i+7}{i(i-1)}$, %$p_{max}(g)$ must be smaller than $p_i(e)$,
$p_{max}(g) < \frac{i^2 - 4i + 7}{i(i-1)}$ must hold as $i$ is relatively big. Theorems \ref{th33} and \ref{th30} are verified by the experimental datum. For each $K_n$, it mentions that the $OHC$ edges and ordinary edges having the $p_{min}(e)$ and $p_{max}(g)$ will change according to $i$. Thus, the related edges with $p_{min}(e)$ or $p_{max}(g)$ are not the same edges in these $K_n$s. %Theorem \ref{th30} is verified by the experiments of small $TSP$ instances. 

As $n$ is big, the probability of some ordinary edges may increase according to $i$ if $i$ is small (The probability increment is generally smaller than $\frac{2p_i(g)}{i(i-1)}$ from $i$ to $i+1$ for $g$ based on Theorem \ref{th30}). Once $i > i_d$  computed with formula (\ref{F5}), the probability will decrease, and the decrement will be bigger than $\frac{2p_i(g)}{i(i-1)}$ from $i$ to $i+1$ and $\frac{p_i(g)}{i+1}$ in most cases. Thus, the decrement of $p_i(g)$ can be taken as one condition to disclose ordinary edges. If $p_{i+1}(g) < \left[1 - \frac{2}{i(i-1)}\right]p_i(g)$ or $ < \frac{i}{i+1}p_i(g)$ occurs from certain $i$ to $i+1$, $g$ will be an ordinary edge. Theorems \ref{th33}, \ref{th3} are verified by the experiments. 
 
\begin{table}
	\begin{center}
		\caption{The frequency and probability changes for $OHC$ edges and ordinary edges according to $i$ for $n\in[9,14]$ in Table \ref{Examp}.}
		{\footnotesize \begin{tabular}{ p{0.1cm}  p{0.7cm}  p{0.4cm}  p{0.4cm}  p{0.6cm}  p{0.6cm}  p{0.6cm}  p{0.6cm}  p{0.6cm}  p{0.6cm}  p{0.5cm}  p{0.5cm}  p{0.4cm} }
		\hline
		$n$	&  &  &  &  &  & $i$ &  &  &  &  &  &  \\
		&  & 4 & 5 & 6 & 7 & 8 & 9 & 10 & 11 & 12 & 13 & 14   \\				
		\hline
		9 & $F_{tot}$ & 2268 & 5040 & 6300 & 4536 & 1765 & 288 &  &  &  &  &   \\	
		& $F_{ohc}$ & 791 & 2248 & 3421 & 2918 & 1315 & 244 &  &  &  &  &   \\	
		& $p_{ohc}(\%)$ & 34.87 & 44.60 & 54.30 & 64.33 & 74.50 & 84.72 &  &  &  &  & \\ 
		& $p_i(e)$ & 0.698 & 0.714 & 0.724 & 0.735 & 0.745 & 0.753 &  &  &  &  & \\ 
		& $p_i(g)$ & 0.434 & 0.295 & 0.203 & 0.136 & 0.085 & 0.045 &  &  &  &  & \\ 		
		& $p_{min}(e)$ & 0.484 & 0.397 & 0.365 & 0.383 & 0.423 & 0.444 &  &  &  &  & \\ 
		& $p_{max}(g)$ & 0.659 & 0.614 & 0.533 & 0.426 & 0.301 & 0.278 &  &  &  &  & \\ 	
		10 & $F_{tot}$ & 3780 & 10080 & 15750 & 15120 & 8820 & 2880 & 405 &  &  &  &   \\	
		& $F_{ohc}$ & 1188 & 4066 & 7742 & 8796 & 5944 & 2206 & 347 &  &  &  &   \\	
		& $p_{ohc}(\%)$ & 31.43 & 40.34 & 49.16 & 58.17 & 67.39 & 76.60 & 85.68 &  &  &  & \\ 
		& $p_i(e)$ & 0.707 & 0.726 & 0.737 & 0.748 & 0.758 & 0.766 & 0.771 &  &  &  & \\ 
		& $p_i(g)$ & 0.441 & 0.307 & 0.218 & 0.154 & 0.105 & 0.067 & 0.037 &  &  &  & \\ 		
		& $p_{min}(e)$ & 0.488 & 0.405 & 0.376 & 0.387 & 0.412 & 0.434 & 0.444 &  &  &  & \\ 
		& $p_{max}(g)$ & 0.667 & 0.629 & 0.558 & 0.464 & 0.360 & 0.260 & 0.267 &  &  &  & \\ 		
		11 & $F_{tot}$ & 5940 & 18480 & 34650 & 41580 & 32340 & 15840 & 4455 & 550 &  &  &   \\	
		& $F_{ohc}$ & 1688 & 6770 & 15497 & 21949 & 19699 & 10944 & 3443 & 470 &  &  &   \\	
		& $p_{ohc}(\%)$ & 28.42 & 36.63 & 44.72 & 52.79 & 60.91 & 69.09 & 77.28 & 85.45 &  &  & \\ 
		& $p_i(e)$ & 0.710 & 0.733 & 0.745 & 0.754 & 0.761 & 0.768 & 0.773 & 0.777 &  &  & \\ 
		& $p_i(g)$ & 0.447 & 0.317 & 0.230 & 0.169 & 0.122 & 0.086 & 0.057 & 0.033 &  &  & \\ 		
		& $p_{min}(e)$ & 0.519 & 0.446 & 0.422 & 0.423 & 0.431 & 0.441 & 0.444 & 0.436 &  &  & \\ 
		& $p_{max}(g)$ & 0.685 & 0.657 & 0.599 & 0.519 & 0.428 & 0.329 & 0.249 & 0.254 &  &  & \\ 		
		12 & $F_{tot}$ & 8910 & 31680 & 69300 & 99792 & 97020 & 63360 & 26730 & 6600 & 726 &  &   \\	
		& $F_{ohc}$ & 2318 & 10624 & 28413 & 48355 & 54267 & 40230 & 19018 & 5205 & 629 &  &   \\	
		& $p_{ohc}$ & 26.02 & 33.54 & 41.0 & 48.46 & 55.93 & 63.49 & 71.15 & 78.86 & 86.64 &  & \\ 
		& $p_i(e)$ & 0.715 & 0.738 & 0.752 & 0.761 & 0.769 & 0.776 & 0.783 & 0.789 & 0.784 &  & \\ 
		& $p_i(g)$ & 0.452 & 0.325 & 0.240 & 0.180 & 0.135 & 0.099 & 0.071 & 0.047 & 0.027 &  & \\ 		
		& $p_{min}(e)$ & 0.537 & 0.473 & 0.445 & 0.439 & 0.444 & 0.456 & 0.471 & 0.48 & 0.485 &  & \\ 
		& $p_{max}(g)$ & 0.7 & 0.681 & 0.630 & 0.560 & 0.479 & 0.390 & 0.294 & 0.238 & 0.242 &  & \\ 	
		13 & $F_{tot}$ & 12870 & 51480 & 128700 & 216216 & 252252 & 205920 & 115830 & 42900 & 9438 & 936 &   \\	
		& $F_{ohc}$ & 3077 & 15852 & 48417 & 96189 & 129616 & 120087 & 75648 & 31052 & 7506 & 812 &   \\	
		& $p_{ohc}$ & 23.91 & 30.79 & 37.62 & 44.49 & 51.38 & 58.32 & 65.31 & 72.38 & 79.53 & 86.75 & \\ 
		& $p_i(e)$ & 0.717 & 0.739 & 0.752 & 0.763 & 0.771 & 0.778 & 0.784 & 0.790 & 0.795 & 0.801 & \\ 
		& $p_i(g)$ & 0.457 & 0.332 & 0.250 & 0.190 & 0.146 & 0.111 & 0.083 & 0.060 & 0.041 & 0.024 & \\ 		
		& $p_{min}(e)$ & 0.561 & 0.504 & 0.475 & 0.464 & 0.463 & 0.471 & 0.486 & 0.502 & 0.515 & 0.526 & \\ 
		& $p_{max}(g)$ & 0.706 & 0.705 & 0.662 & 0.602 & 0.532 & 0.456 & 0.378 & 0.296 & 0.227 & 0.231 & \\ 
		14 & $F_{tot}$ & 18018 & 80080 & 220173 & 432432 & 588588 & 576576 & 405405 & 200200 & 66066 & 13104 &  1183 \\	
		& $F_{ohc}$ & 3960 & 22620 & 77606 & 176040 & 276770 & 308078 & 243098 & 133401 & 48502 & 10521 & 1031 \\	
		& $p_{ohc}(\%)$ & 21.98 & 28.25 & 35.25 & 40.71 & 47.02 & 53.43 & 59.96 & 66.63 & 73.41 & 80.29 & 87.15 \\ 
		& $p_i(e)$ & 0.714 & 0.734 & 0.747 & 0.756 & 0.764 & 0.772 & 0.780 & 0.787 & 0.795 & 0.803 & 0.809 \\ 
		& $p_i(g)$ & 0.461 & 0.339 & 0.249 & 0.200 & 0.157 & 0.122 & 0.0946 & 0.0717 & 0.0524 & 0.0358 & 0.0217 \\ 		
		& $p_{min}(e)$ & 0.571 & 0.522 & 0.492 & 0.474 & 0.470 & 0.476 & 0.491 & 0.512 & 0.529 & 0.546 & 0.560 \\ 
		& $p_{max}(g)$ & 0.722 & 0.717 & 0.680 & 0.630 & 0.570 &  0.493 & 0.412 & 0.331 & 0.268 & 0.241 & 0.231 \\ 															
		%\bottomrule
		\hline
		\end{tabular}}
		\label{chgf}
	\end{center}
\end{table}

\subsubsection{The probability change for two adjacent $OHC$ edges according to $i$}
 In this section, we shall show the probability change for two adjacent $OHC$ edges for $K_{14}$ in Table \ref{Examp}. In Table \ref{chgf}, one sees that the probability $p_{min}(e)$ is generally smaller than $\frac{1}{2}$, and it sometimes decreases according to $i\in [4,4]$. The changes of the probability $p_i(e_1) + p_i(e_2)$ for two adjacent $OHC$ edges $e_1$ and $e_2$ will be different based on Theorem \ref{th3}. In other words, for two $OHC$ edges $e_1$ and $e_2$ containing a vertex $v$, $p_i(e_1) + p_i(e_2) > 1$ holds and it will increases according to $i$. The  $p_i(e_1) + p_i(e_2)$ for each vertex $v\in [0,13]$ in the $K_{14}$ is computed according to $i\in [4,14]$ and illustrated in Table \ref{chgtwof}.
 
 Firstly,  $p_i(e_1) + p_i(e_2) > 1$ holds for every vertex $v$ according to $i$. The minimum value of $p_i(e_1) + p_i(e_2)$ is 1.174 in the experiments. Thus, the average value of $p_i(e_1)$ and $p_i(e_2)$ is obviously bigger than $\frac{1}{2}$. Although some $OHC$ edges $e_1$ has the probability $p_i(e_1) < \frac{1}{2}$, the probability condition  $p_i(e_1) + p_i(e_2) > 1$ indicates that the adjacent $OHC$ edges $e_2$ of $e_1$ will have an even bigger probability 
$p_i(e_2) > \frac{1}{2}$ based on the frequency $K_i$s.  

 Moreover, $p_i(e_1) + p_i(e_2)$ increases according to $i$ for most vertices, such as 0, 1, 3, 6-8, 11-13. The probability increment is relatively big as $i$ is small, and it becomes smaller as $i$ big. In general, the probability $p_i(e_1)$ and $p_2(e_2)$ increases simultaneously according to $i$ for these vertices. For the other vertices, $p_i(e_1) + p_i(e_2)$ decreases from $i$ to $i+1$ as $i$ is small, and it increases according to $i$ after the smallest value. As $p_i(e_1) + p_i(e_2)$ decreases from $i$ to $i+1$, only $p_i(e_1)$ or $p_i(e_2)$ becomes smaller and the other one rises for most vertices. Among the 14 vertex, there is only one vertex 10 whose $p_i(e_1)$ and $p_i(e_2)$ decrease from $i=4$ to 8 at the same time. It means that the number of such vertex is very small, especially for medium and big $TSP$. Theorem \ref{th3} is verified by the experiments. 

 \begin{table}
 	\begin{center}
 		\caption{The probability change for two adjacent $OHC$ edges according to $i$ for $K_{14}$.}
 		{\footnotesize \begin{tabular}{ p{0.1cm}  p{0.7cm}  p{0.5cm}  p{0.5cm}  p{0.5cm}  p{0.5cm}  p{0.5cm}  p{0.5cm}  p{0.5cm}  p{0.5cm}  p{0.5cm}  p{0.5cm}   }
 		\hline
 		 &  &  &  & $p_i(e_1)$ &  & + &  & $p_i(e_2)$ &  &  &   \\
 		$v$ & $i = $ 4 & 5 & 6 & 7 & 8 & 9 & 10 & 11 & 12 & 13 & 14   \\				
 		\hline
0	&	1.571 	&	1.676 	&	1.722 	&	1.746 	&	1.760 	&	1.768 	&	1.772 	&	1.776 	&	1.779 	&	1.785 	&	1.791 	\\ 
1	&	1.596 	&	1.707 	&	1.755 	&	1.779 	&	1.793 	&	1.802 	&	1.810 	&	1.816 	&	1.822 	&	1.828 	&	1.835 	\\ 
2	&	1.343 	&	1.343 	&	1.332 	&	1.322 	&	1.317 	&	1.319 	&	1.330 	&	1.347 	&	1.366 	&	1.387 	&	1.407 	\\ 
3	&	1.364 	&	1.385 	&	1.408 	&	1.432 	&	1.452 	&	1.467 	&	1.476 	&	1.481 	&	1.482 	&	1.482 	&	1.484 	\\ 
4	&	1.273 	&	1.248 	&	1.292 	&	1.356 	&	1.416 	&	1.463 	&	1.497 	&	1.522 	&	1.541 	&	1.556 	&	1.560 	\\ 
5	&	1.288 	&	1.270 	&	1.316 	&	1.381 	&	1.441 	&	1.492 	&	1.531 	&	1.560 	&	1.578 	&	1.587 	&	1.582 	\\ 
6	&	1.389 	&	1.410 	&	1.412 	&	1.412 	&	1.415 	&	1.424 	&	1.440 	&	1.460 	&	1.479 	&	1.495 	&	1.505 	\\ 
7	&	1.576 	&	1.676 	&	1.722 	&	1.747 	&	1.760 	&	1.768 	&	1.774 	&	1.778 	&	1.781 	&	1.782 	&	1.780 	\\ 
8	&	1.490 	&	1.555 	&	1.596 	&	1.624 	&	1.642 	&	1.652 	&	1.658 	&	1.666 	&	1.677 	&	1.697 	&	1.725 	\\ 
9	&	1.293 	&	1.275 	&	1.268 	&	1.269 	&	1.279 	&	1.295 	&	1.318 	&	1.349 	&	1.386 	&	1.429 	&	1.473 	\\ 
10	&	1.293 	&	1.271 	&	1.231 	&	1.193 	&	1.174 	&	1.177 	&	1.203 	&	1.243 	&	1.291 	&	1.339 	&	1.374 	\\ 
11	&	1.384 	&	1.404 	&	1.432 	&	1.461 	&	1.484 	&	1.497 	&	1.504 	&	1.508 	&	1.509 	&	1.511 	&	1.516 	\\ 
12	&	1.596 	&	1.706 	&	1.757 	&	1.786 	&	1.806 	&	1.820 	&	1.830 	&	1.837 	&	1.842 	&	1.845 	&	1.846 	\\ 
13	&	1.545 	&	1.637 	&	1.660 	&	1.659 	&	1.658 	&	1.666 	&	1.684 	&	1.709 	&	1.735 	&	1.760 	&	1.780 	\\ 
			
 				%\bottomrule
 				\hline
 		\end{tabular}}
 		\label{chgtwof}
 	\end{center}
 \end{table}
 
 The probabilities for all edges are computed with the frequency $K_i$s where $i\in[4,14]$ for the small $TSP$ instance $K_{14}$ in Table \ref{Examp} to show the probability  change according to $i$. For a given $i$, the probabilities of all edges are ordered from small to big values according to $k\in[1,91]$ and a probability sequence $\left(p_1, p_2,...,p_k,...,p_{{{n}\choose{2}}}\right)$ is obtained where $p_k$ denotes the $k^{th}$ probability. Then the $p_k$s are lined according to $k$ and each $i$, respectively. The probability change curves according to $i=4\sim 14$ are shown in Figure \ref{fcalle14}. In the picture, the lower probability bound $Lbp = \frac{1}{2}$ for $OHC$ edges is denoted with one black line. Moreover, the smallest probability of the $OHC$ edges is denoted with one square dot in each curve for $i=$4, 5, 6, 7, 8 and 9. If $i \geq 10$, the $i$ biggest probabilities in each curve belong to the $OHC$ edges so the smallest probability of the $OHC$ edges is not highlighted. 
 
 According to each probability change curve computed based on the frequency $K_i$s, there are more than $\frac{2}{3}{{n}\choose{2}}$ edges with the probability smaller than the smallest probability of the $OHC$ edges even if $i=4$. It says that at least two thirds of all edges can be neglected according to their probabilities or frequencies as $TSP$ is resolved. Moreover, the number of ordinary edges with the probability smaller than the smallest probability of the $OHC$ edges increases according to $i$ even if the smallest probability becomes smaller. For example, as $i=4$, there are 63 ordinary edges with the probability smaller than the smallest probability of the $OHC$ edges. Moreover, the number of such ordinary edges rises to 75 as $i=9$. In addition, the smallest probability of the $OHC$ edges is either bigger than or close to the lower probability bound $Lbp=\frac{1}{2}$. One also sees that each probability curve rises quickly from $Lbp$ to the biggest probability. It means that the number of the edges with the probability bigger than $Lbp$ is small, and the edges with the probabilities close to $Lbp$ is very small. Thus, most $OHC$ edges have the probabilities much bigger than $Lbp$, and there will be a small number of $OHC$ edges with the probability near $Lbp$. As $i$ increases, the difference between the probabilities of the $OHC$ edges and those of the ordinary edges becomes bigger. As $i$ is close to 14, such as $i\geq 10$, the smallest probability of the $OHC$ edges is bigger than that of each ordinary edge. 
 
 Based on Theorem \ref{th33}, the probabilities of two ordinary edges will have the fixed relationship as $i$ is big. Except the $OHC$ edges, it implies that the probability of each ordinary edge has the fixed order in the probability sequences for big $i$s. Given the order number $k$ for an ordinary edge, the probability becomes smaller according to $i$ in Figure \ref{fcalle14}. It means that it is contained in the smaller percentage of $OP^i$s according to $i$. Although some ordinary edges have a big probability as $i$ is small, the order number $k$ in the probability sequence still decreases as $i$ is big. On the other hand, the probability of an $OHC$ edge becomes bigger or keeps the nearly equal value according to $i$. Thus, the order number of the smallest probability of the $OHC$ edges increases according to $i$. The probability or frequency change related to all edges for the other general $TSP$ instances is similar to that for the small $TSP$ instance.  Theorems \ref{th3} and \ref{th33} are testified with the numerical results. 
 
  \begin{figure}
 	\centering
 	\includegraphics[width=3in,bb=0 0 450 260]{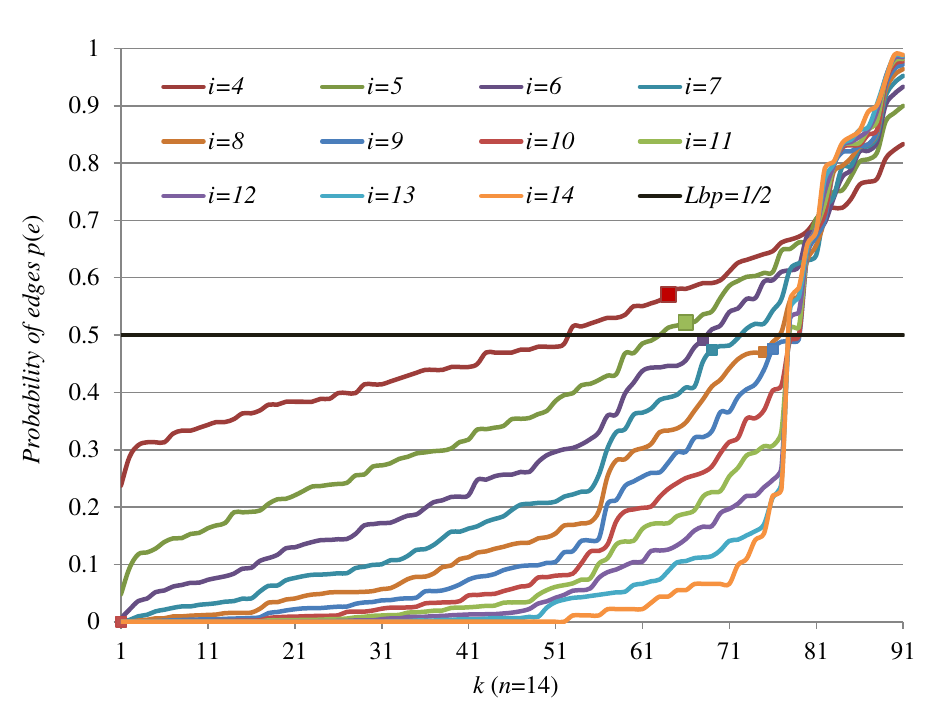}
 	\caption{The changes of the ordered probabilities related to all edges according to $k\in[1,91]$ computed based on frequency $K_i$s in $K_{14}$ where $i\in[4,14]$.}
 	\label{fcalle14}
 \end{figure}

 \subsubsection{The frequency and probability changes for the same edges according to $i$ for small $TSP$ instances}
 The $p_{min}(e)$s and $p_{max}(g)$s according to $i$ in Table \ref{chgf} may not for the same edges. For the same edge, $p_{min}(e)$ will change more steadily whereas $p_{max}(g)$ will decrease more sharply from $i$ to $i+1$ in some steps. Thus, the results of $p_{min}(e)$ and $p_{max}(g)$ in Table \ref{chgf} are not accurate enough to demonstrate the frequency and probability changes for the $OHC$ edges and ordinary edges. To illustrate the frequency and probability changes for the same $OHC$ edges and ordinary edges, three $OHC$ edges with the smallest frequencies among all $OHC$ edges and three ordinary edges with the biggest frequencies among all ordinary edges at $i=4$ are chosen from the $K_{13}$ and $K_{14}$ in Table \ref{Examp}, respectively. The frequency and probability changes according to $i\in[4,n]$ for these edges are illustrated in Figure \ref{fpcs13} and Figure \ref{fpcs14}, respectively. 

 In each picture in Figures \ref{fpcs13} and \ref{fpcs14}, the frequency and probability changes related to $OHC$ edges are denoted with the solid lines while those for the ordinary edges are denoted with the dashed lines. In Figure \ref{fpcs13} where $n=13$, (3,5), (9,10) and (2,6) are the three $OHC$ edges having the smallest frequencies and probabilities at $i$ = 4, and (0,2), (10,11) and (6,8) are the three ordinary edges having the biggest frequencies and probabilities at $i=4$. In the first and second picture in Figure \ref{fpcs13}, the frequencies and probabilities of each edge computed with the frequency $K_i$s are lined according to $i=4 \sim 13$, respectively. Since the three $OHC$ edges have the smallest frequencies while the three ordinary edges have the biggest frequencies at $i=4$, the frequencies of the $OHC$ edges are smaller than those of the ordinary edges and they do not have much difference as $i$ is small, such as $i$=4 and 5. Moreover, the frequency of each of the $OHC$ edges and ordinary edges rises according to $i \leq 8$, and each edge reaches its peak frequency at $i$=8, respectively. Then, the frequency of each edge begins decreasing according to $i>8$, and the frequencies of the ordinary edges decrease faster than those of the $OHC$ edges. Before $i=8$, the frequencies of the three ordinary edges are always bigger than the smallest frequency of the $OHC$ edge (2,6). After $i=8$, the frequencies of the three ordinary edges become smaller than that of (2,6) according to $i$. It indicates that some ordinary edges will maintain the big frequency according to $i < P_0$ until $i$ is close to $P_0$. As $n$ becomes big, the frequencies of the $OHC$ edges will rise according to $n$, and the frequencies of most ordinary edges will be smaller than those of $OHC$ edges at certain $i < P_0$, see Theorem \ref{th30} and \ref{th5}. Even if this $i$ is far from $P_0$, sometimes it may be very big, such as $i=2i_d$ computed with formula (\ref{F5}). In this case, it will be time-consuming to compute the frequencies of edges with the frequency $K_i$s for separating $OHC$ edges from ordinary edges (according to their frequencies). 
 
   \begin{figure}
 	\centering
 	\includegraphics[width=3in,bb=0 0 500 400]{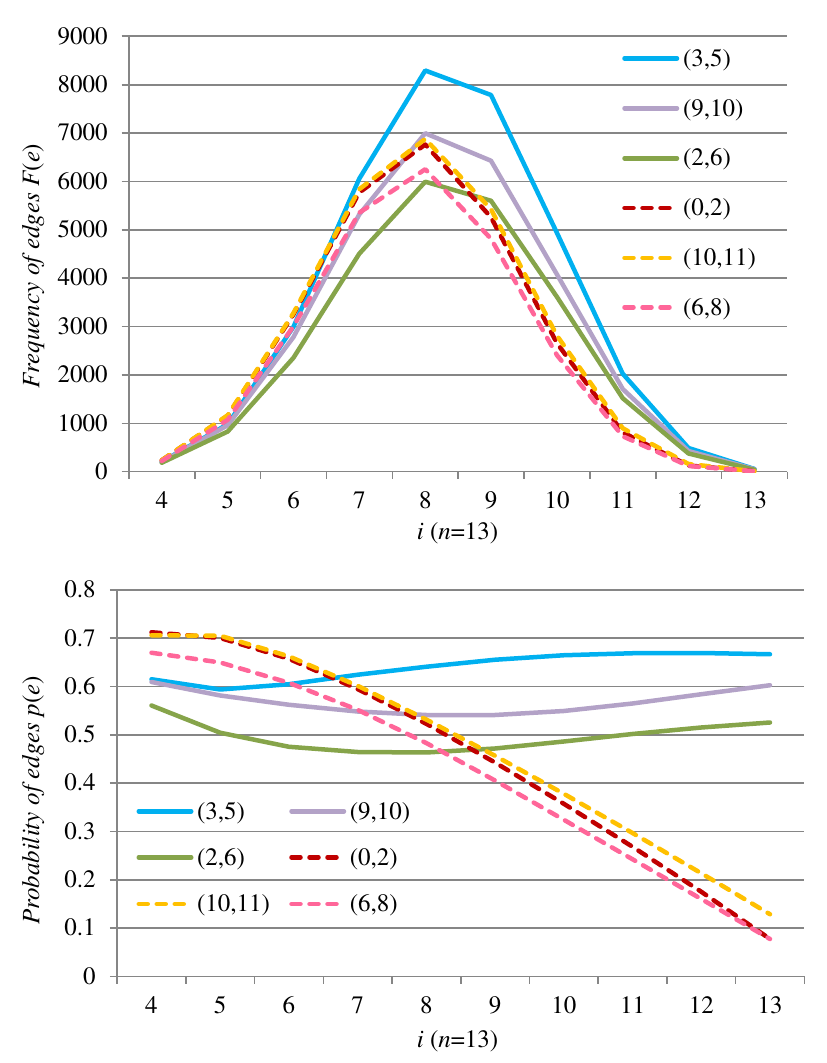}
 	\caption{The frequency and probability changes for three $OHC$ edges with the smallest frequencies and three ordinary edges with the biggest frequencies at $i=4$ for $K_{13}$.}
 	\label{fpcs13}
 \end{figure}

 In the second picture in Figure \ref{fpcs13}, the probabilities of the three $OHC$ edges decrease from $i$ to $i+1$ as $i$ is small, and then they increase from $i$ to $i+1$ as $i$ is big until $i = n-1$. Moreover, the probability of (3,5) decreases or increases slightly from $i$ to $i+1$, i.e., $p_{i+1}(e)\in \left[1-\frac{2}{i(i-1)}, 1+\frac{2}{i(i-1)}\right]p_i(e)$ holds for (3,5). For (9,10) and (2,6), $p_{i+1}(e) > \left[1-\frac{2}{i(i-1)}\right]p_i(e)$ holds if $p_i(e)$ decreases from $i$ to $i+1$, and $p_{i+1}(e) > \left[1+\frac{2}{i(i-1)}\right]p_i(e)$ exists from $i=10$ to 12 for (9,10), from $i=9$ to 12 for (2,6) as $p_i(e)$ increases. Once sees the smallest probabilities of the $OHC$ edges generally decrease according to $i$ as $i$ is small. As they increase from the smallest values, the smaller $p_i(e)$ increases faster than the bigger $p_i(e)$ according to $i$. Since $\frac{(i+1)(n-i)}{(i-1)^2}\times \left[1-\frac{2}{i(i-1)}\right] > 1$ exists as $i < P_0$, the frequency $F(e)$ always increases from $i$ to $i+1$ for the three $OHC$ edges even if $p_i(e)$ decreases slightly from $i$ to $i+1$. Thus, the three $OHC$ edges reach their peak frequencies at $i=P_0=8$, respectively. As $i > 8$, $p_i(e)$ of the three $OHC$ edges increases according to $i$ until $i=n$. It indicates that $p_i(e)$s of the three $OHC$ edges increase before $2i_d = 10$ for $n=13$. Moreover, $p_{i+1}(e) > \left[1+\frac{2}{i(i-1)}\right]p_i(e)$ occurs to (9,10) and (2,6) from $i$ to $i+1$ in most steps. It implies that $p_{i+1}(e) > \left[1+\frac{2}{i(i-1)}\right]p_i(e)$ will happen if $p_i(e)$ increases from the smallest value. Thus, they are contained in more percentage of $OP^i$s according to $i$. The maximum probability is reached at $i=n$ for (3,5) and (9,10). For (2,6), $p_i(e)$ is also rising after $i=8$. However,  the smallest probability at $i=8$ is small and the number of steps from $i = 8$ to 12 is only 4. Thus, the maximum probability is obtained at $i=4$ rather than 13. Although (2,6) does not obtain the maximum probability at $n$,  $p_n(e)=0.525641 > 0.5$ still holds at $n$. Because $\frac{(i+1)(n-i)}{(i-1)^2} < 1$ decreases sharply from $i$ to $i+1$ as $i > P_0$ while $\frac{p_{i+1}(e)}{p_i(e)}$ increases very slow,  $\frac{(i+1)(n-i)}{(i-1)^2}\times \frac{p_{i+1}(e)}{p_i(e)} < 1$ appears and $F(e)$ decreases according to $i > 8$ until $i = n$. 
 
 For the three ordinary edges, the probability $p_i(g)$ always decreases according to $i\in[4,13]$, see the change of the three dashed lines. Moreover, $p_i(g)$ becomes smaller in the nearly linear way according to $i\geq 6$. In Figure \ref{pdi}, the probability decrement increases at first, and then it decreases according to $i$ after the biggest value. In the experiments, the probability decrement always increases according to $i$ as $i$ is small, and it keeps the nearly equal big value from $i$ to $i+1$ until $p_i(g)$ is near the smallest value. The probability change implies that each of the ordinary edges is contained in the smaller percentage of $OP^i$s containing more vertices based on all $K_i$s containing them, respectively. Moreover, the probability decrement change implies that an ordinary edge is also contained in the smaller percentage of the $OP^i$s according to $i$ based on the ${{n-i}\choose{i-2}}-K$ $K_i$s containing it. As $i=4$ and 5, the probabilities decrease slightly as those of the $OHC$ edges, and the frequencies $F(g)$ increase quickly from $i$ to $i+1$. However, $\frac{p_{i+1}(g)}{p_i(g)} < 1- \frac{2}{i(i-1)}$ appears for the three ordinary edges from $i$=6 and $\frac{p_{i+1}(g)}{p_i(g)} <  \frac{i}{i+1}$ exists from $i = 8$.  The probability decrement becomes bigger and keeps the nearly equal big value according to $i$. As $i > 6$, $p_i(g)$ decreases quickly according to $i$, which prevents the frequency $F(g)$ from increasing fast. One can observe and compare the increasing slopes of the  $F(e)$ and $F(g)$ curves from $i$=7 to 8. It is clearly that the $F(e)$ curves related to the three $OHC$ edges have the bigger increasing slopes than the $F(g)$ curves related to the three ordinary edges. It implies that the frequencies of the ordinary edges increase slower than those of the $OHC$ edges after $i = 6$. Because $\frac{(i+1)(n-i)}{(i-1)^2}\times \frac{p_{i+1}(g)}{p_i(g)} > 1$  exists for the three ordinary edges if $i < 8$, the frequencies are still increasing and the three ordinary edges also reach their peak frequencies at $i=8$, respectively. Since $\frac{p_{i+1}(g)}{p_i(g)} < \frac{i}{i+1}< 1 < \frac{p_{i+1}(e)}{p_i(e)}$ as $i \geq 8$, $F(g)$s for the three ordinary edges decrease more sharply than those of the $OHC$ edges, see the decreasing slopes of the $F(e)$ curves and $F(g)$ curves. At last, $p_i(g) < \frac{1}{2}$ appears from $i=9$ and $F(g)$s of the three ordinary edges become smaller than (the smallest) $F(e)$ of edge (2,6). $p_i(g) > 0$ exists for the three ordinary edges at $i=13$. It means that the three ordinary edges are contained in some $OP^{13}$s. Thus, the edges with a frequency much bigger than $f_{lb}$ at the small $i$s, such as $i=$ 4, 5 and 6, are either $OHC$ edges or contained in the $OP^n$s. Theorem \ref{th33} is verified by these experimental results. 
 
 Comparing with the frequency changes for the $OHC$ edges and ordinary edges, some ordinary edges will have the big frequencies as the $OHC$ edges before $P_0$, and the frequencies of the ordinary edges will also increase according to $i < P_0$. The ordinary edges having the  big frequencies at the small $i$s are generally contained in $OP^n$s. Thus, the frequencies and probabilities of these ordinary edges are not equal to zero at $n$. For the other ordinary edges excluding from any one $OP^n$, they will have a small frequency and probability as $i$ is small based on Theorem \ref{th33}. The frequency and probability decrease more quickly according to $i$ than those of the ordinary edges contained in $OP^n$s. In addition, it discloses that two edges (either two $OHC$ edges or two ordinary edges) generally preserves the fixed frequency and probability relationship according to $i$. For example the three $OHC$ edges (3,5), (9,10) and (2,6), the frequency of (3,5) is always bigger than that of (9,10), and the frequency of (9,10) is always bigger that of (2,6) according to $i$. In the second picture, the probabilities of the three $OHC$ edges conform to the same relationships. For the three ordinary edges, the frequency of (0,2) is the biggest at $i=4$, and it becomes smaller than that of (10,11) from $i=5$. After that, the frequency of (0,2) is smaller than that of (10,11) according to $i$. It indicates that the frequency of (0,2) increases slower than that of (10,11) according to $i<P_0=8$ and it decreases faster than that of (10,11) after $P_0 = 8$. Moreover, the probability of (0,2) is always smaller than that of (10,11) until $i=13$. One sees the frequency and probability relationship for two edges will change as $i$ is small. As $i$ is big, the frequency and probability relationship between them will keep intact according to $i$. Theorem \ref{th33} is verified with the experimental results. 
 
 Since the average frequencies of some ordinary edges are bigger than the frequency bound $f_{lb}=\frac{1}{2}{{i}\choose{2}}$ and they may increase before $P_0$, it will be time-consuming to compute the frequency $K_i$s containing many vertices for separating $OHC$ edges and ordinary edges in $K_n$. Comparing the frequency change  with the probability change for the ordinary edges, the probability decreases much earlier before the frequency does. That is to say, although the frequency is still increasing at certain $i \ll P_0$, the probability begins decreasing before $i = i_d \ll P_0$ for the ordinary edges. Theorem \ref{th30} is verified with these experimental results. At first, the probability decrement  will be smaller than $\frac{2p_i(g)}{i(i-1)}$ from $i$ to $i+1$ as $i$ is small. For example, this case occurs for the three ordinary edges as $i\leq 5$. As $i$ is increased to certain big number, the probability decrement will be bigger than $\frac{2p_i(g)}{i(i-1)}$ and $\frac{p_i(g)}{i+1}$ from $i$ to $i+1$. For example, this case happens to the three ordinary edges as $i\geq 6$ in Figure \ref{fpcs13}. Moreover,  $\frac{p_{i+1}(g)}{p_i(g)}$ decreases quickly within two steps and it becomes smaller than $\frac{i}{i+1}$ from $i$ to $i+1$ if $i\geq 8$ even if the three ordinary edges are contained in the $OP^n$s. It illustrates that $p_{i+1}(g) < \frac{i}{i+1}p_i(g)$ will exist for ordinary edges from $i$ to $i+1$ in most cases. Once $p_i(g)$ becomes smaller, $p_i(g)$ will decrease quickly according to $i$ until it approaches the smallest value. This case never happens to the three $OHC$ edges in Figure \ref{fpcs13}. Since this case does not happen to $OHC$ edges, the edges with such probability decrement must be ordinary edges. Thus, according to the necessary condition of probability decrement for $OHC$ edges and the probability change of ordinary edges, it will consume less time to separate more ordinary edges from $OHC$ edges. %In other words,  probability of ordinary edges 
 %保持较大频率，被保留在frequencyKn内
 
 %If $p_i(g)$ for the ordinary edges decreases, it never increases as that of the $OHC$ edges. although $F(g) > F(e)$ exists. 
  \begin{figure}
 	\centering
 	\includegraphics[width=3in,bb=0 0 500 400]{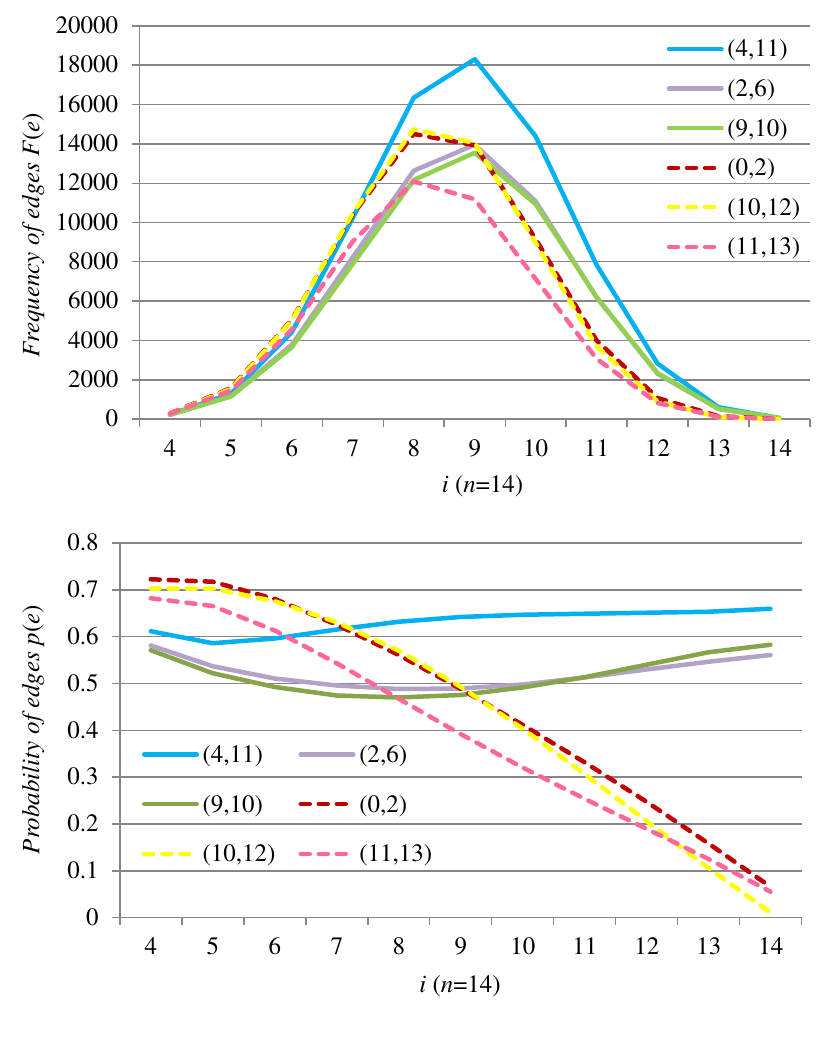}
 	\caption{The frequency and probability changes for three $OHC$ edges with the smallest frequencies and three ordinary edges with the biggest frequencies at $i=4$ for $K_{14}$.}
 	\label{fpcs14}
 \end{figure}

 For the even number $n=14$, the frequency and probability changes for the three $OHC$ edges with the smallest frequencies and three ordinary edges with the biggest frequencies at $i=4$ are illustrated in Figure \ref{fpcs14}. The frequency changes for the $OHC$ edges and ordinary edges are similar to those for the $OHC$ edges and ordinary edges in Figure \ref{fpcs13}, respectively. The difference is that each of the $OHC$ edges reaches the peak frequency at $i=9$ (i.e., $\frac{n}{2}+2$) whereas the three ordinary edges reach their peak frequencies at $i=8$ (i.e., $\frac{n}{2}+1$). It indicates that the frequencies of the ordinary edges decrease before those of the $OHC$ edges for even $n$. Thus, the frequency of an $OHC$ edge begins decreasing from $P_0$ whereas the frequency of an ordinary edge decreases before $P_0$ for even $n$. The probability changes for the $OHC$ edges and ordinary edges are also similar to those for the $OHC$ edges and ordinary edges in Figure \ref{fpcs13}, respectively. The probability of each $OHC$ edge decreases slightly from $i$ to $i+1$ as $i$ is small because these $OHC$ edges have the smallest probabilities at $i=4$.  As $i$ is relatively big, i.e., at certain $i\leq2i_d = 10$, the probability of each $OHC$ edge begins increasing according to $i$. If $p_i(e)$ decreases from $i$ to $i+1$, $p_{i+1}(e) > \left[1-\frac{2}{i(i-1)}\right]p_i(e)$ holds for each of the three $OHC$ edges. As $p_i(e)$ increases from $i$ to $i+1$, $p_{i+1}(e) < \left[1+\frac{2}{i(i-1)}\right]p_i(e)$ exists for (4,11) from $i=8$ to 13. $p_{i+1}(e) > \left[1+\frac{2}{i(i-1)}\right]p_i(e)$ appears for (2,6) from $i=10$ to 13, and for (9,10) from $i=9$ to 13, respectively. 
 
 The three ordinary edges are also contained in some $OP^n$s. For the two ordinary edges (0,2) and (11,13), the probability always decreases according to $i$, and the probability decrement deviates from $\frac{2p_i(g)}{i(i-1)}$ further and further. For example, the probability decrement is smaller than $\frac{2p_i(g)}{i(i-1)}$ as $i \leq 5$, and it becomes bigger than $\frac{2p_i(g)}{i(i-1)}$ as $i\geq 6$, and it is bigger than $\frac{p_i(g)}{i+1}$ as $i\geq 8$ and $i\geq 7$ for (0,2) and (11,13), respectively. For the ordinary edge (10,12), the probability increases slightly from $i=4$ to 5, and the probability increment is 0.000707 which is much smaller than $\frac{2p_i(g)}{i(i-1)} = 0.117003$ ($p_4(g) = 0.70202$ and $p_5(g) = 0.702727$ for (10,12)). Then, it decreases according to $i$ as $i\geq 6$. The probability decrement is smaller than $\frac{2p_i(g)}{i(i-1)}$ if $i\leq 6$, and it is bigger than $\frac{2p_i(g)}{i(i-1)}$ if $i\geq 7$, and it is bigger than $\frac{p_i(g)}{i+1}$ if $i\geq 8$. One sees $p_i(g)$ for some ordinary edges may increase as $i$ is small. However, the probability increment is generally smaller than $\frac{2p_i(g)}{i(i-1)}$ from $i$ to $i+1$. Since the three ordinary edges have the biggest probabilities at $i=4$, the probability decreases slower than those of the other ordinary edges. Although the probability decrement is small at first, it will become bigger than $\frac{2p_i(g)}{i(i-1)}$ and $\frac{p_i(g)}{i+1}$ within a small number of steps and keeps the big value according to $i$ until $p_i(g)$ tends to the smallest value. The experiments illustrate that $p_{i+1}(g) < \frac{i}{i+1}p_i(g)$ holds from $i$ to $i+1$ in most cases. Moreover, the experiments demonstrate that the probability decrement does not become smaller once $p_i(g)$ decrease. Thus, $p_i(g)$ will decrease quickly according to $i$ until it approaches the smallest value.

If the frequencies of the $OHC$ edges or ordinary edges are obviously different at $i=4$, the frequency and probability changes according to $i$ will also be  different, respectively. The three general $OHC$ edges and ordinary edges are chosen from $K_{13}$ and $K_{14}$, respectively to show the frequency and probability changes according to $i$. The frequencies of the three $OHC$ edges are obviously different at $i=4$ as well as their probabilities. In similarity, the frequencies or probabilities of the three ordinary edges are not close at $i=4$. They represent different types of ordinary edges which are contained in the $OP^n$s or not. It says that each of them is contained in different number of $OP^4$s although each of them is contained in the same number of $K_4$s in the $K_{13}$ or $K_{14}$, respectively. The frequency and probability changes for the three $OHC$ edges and ordinary edges in $K_{13}$ are illustrated in Figure \ref{fpcs13or}, and those for the three $OHC$ edges and ordinary edges in $K_{14}$ are shown in Figure \ref{fpcs14or}. In the same manner, the frequency and probability changes for each $OHC$ edge are denoted with the solid lines and those for each ordinary edge are denoted with the dashed lines. In the two Figures, the lower frequency bound $Lb F = \frac{1}{2}{{i}\choose{2}}{{n-2}\choose{i-2}}$ for $OHC$ edges is also shown and denoted with the solid black lines as well as the lower probability bound $Lb p= \frac{1}{2}$ for $OHC$ edges.

In the first picture in Figure \ref{fpcs13or}, the frequencies of the three $OHC$ edges (0,1), (6,7) and (1,2) are much bigger than the lower frequency bound $LbF$, and $LbF$ is much bigger than the frequencies of the three ordinary edges (6,9), (1,5) and (3,8) according to $i$. It indicates that the frequencies of the $OHC$ edges are generally much bigger than those of the ordinary edges as the frequency of each edge is computed with the frequency $K_i$s. Moreover, the frequencies of the three $OHC$ edges do not have much difference as $i$ is big, and they increase faster than $LbF$ according to $i<P_0$. It implies that most $OHC$ edges have the frequency much bigger than $LbF$. If $i > P_0$, the frequencies of the $OHC$ edges decreases quickly according to $i$ yet they are still much bigger than $LbF$. In fact, they decrease faster than $LbF$ and those of the ordinary edges according to $i > P_0$. Comparing the three smallest frequencies of the $OHC$ edges in Figure \ref{fpcs13}, it finds that only a small number $OHC$ edges having the frequency close to $LbF$, and the frequencies of most $OHC$ edges are much bigger than $LbF$ and the smallest frequency of the $OHC$ edges. 

On the other hand, the frequencies of the three ordinary edges increases much slower than $LbF$ according to $i$, and they are much smaller than $LbF$. For the three $OHC$ edges, they reach their peak frequencies at $i=8$. For the ordinary edges (6,9), (1,5) and (3,8), they reaches their peak frequencies at $i=8$, 7 and 6, respectively. For most ordinary edges, the peak frequency is computed before $P_0$, such as (1,5) and (3,8). The frequencies of these ordinary edges decrease sharply after the corresponding peak value, and become zero at certain number $i \leq n$. These ordinary edges are not contained in the $OP^n$s. For example, the frequency of (1,5) becomes zero at $i=13$ and that of (3,8) becomes zero at $i=10$.  Only a small number of the ordinary edges have the peak frequency at $\frac{n+1}{2} + 1$ for odd $n$. Most of them are  contained in the $OP^n$s, such as (6,9) in Figure \ref{fpcs13or}, and (0,2), (10,11) and (6,8) in Figure \ref{fpcs13}.  

In the second picture in Figure \ref{fpcs13or}, the probabilities of the three general $OHC$ edges are always bigger than the lower probability bound $Lbp = \frac{1}{2}$  according to $i$, and $Lbp$ is not bigger than those of some ordinary edges as $i$ is small, such as $i=4$. It indicates that most $OHC$ edges are contained in more than half $OP^i$s in each $K_i$ containing them on average. Moreover, the probabilities of the three $OHC$ edges always increases according to $i$. It indicates that they are contained in more percentage of the $OP^i$s according to $i$ whether the frequency increases or decreases. If the $p_i(e)$  for an $OHC$ edge is not close to the lower probability bound $Lbp = \frac{1}{2}$ as $i = 4$, it will always increase according to $i$. In addition, $p_{i+1}(e) < \left[1 + \frac{2}{i(i-1)}\right]p_i(e)$ exists from $i$ to $i+1$ for the three general $OHC$ edges since $p_i(e) > \frac{3}{4}$ holds according to $i$. For the $OHC$ edges with $p_i(e)$ close to 1, the probability increment will be smaller than $\frac{2p_i(e)}{i(i-1)}$ if $p_i(e)$ increases from $i$ to $i+1$. Otherwise, $p_{i+1}(e)$ will be bigger than 1. It is one contradiction. 

%As $i\geq 6$, the probabilities of the three ordinary edges becomes smaller than $Lbp$. 
%It is known that the average probability of all ordinary edges is smaller than $\frac{1}{2}$ according to $i$. The experimental results illustrate that the probabilities of most ordinary edges are smaller than $\frac{1}{2}$ even if $i$ is close to 4. Since the average probability of all ordinary edges is monotone decreasing according to $i$, it will deviate from $\frac{1}{2}$ further and further, see the experimental results in Table \ref{chgf}. %In addition, the three $OHC$ edges have the obviously different probabilities at $i=4$ as well as those of the three ordinary edges. It implies that each of them is contained in different number of $OP^4$s although each of them is contained in the same number of $K_4$s in the $K_{13}$. 

It is known that the average probability of all ordinary edges is smaller than $\frac{1}{2}$ according to $i$. The experimental results illustrate that the probabilities of most ordinary edges are smaller than $\frac{1}{2}$ even if $i$ is close to 4. Since the average probability of all ordinary edges is monotone decreasing according to $i$, it will deviate from $\frac{1}{2}$ further and further, see the experimental results in Table \ref{chgf}. For the three ordinary edges, the probability of each of them is always decreasing according to $i$, and the probability is smaller than $Lbp$ if $i \geq 6$. It means that each of the ordinary edges is contained in the smaller and smaller percentage of $OP^i$s according to $i$, and they are usually contained in less than half $OP^i$s in each $K_i$ containing them in the average case, respectively, especially as $i$ is relatively big. Moreover, the $p_i(g)$s for the three general ordinary edges decrease in the polynomial or exponential way according to $i$, and they decrease faster than those for the three ordinary edges in Figure \ref{fpcs13}. It indicates that the $p_i(g)$s of most ordinary edges not in $OP^n$s will become smaller than $\frac{1}{2}$ as $i$ is small. For the edges (1,5) and (3,8), the probability decreases more sharply than (6,9) according to $i$. Thus, the frequency of (1,5) and (3,8) increases much slower than (6,9) according to $i<P_0$. Moreover, the frequency of (1,5) and (3,8) decreases from $i=$ 7 and 6, respectively although the number of $K_i$s containing them is still rising at this time. (6,9) is contained in twelve $OP^{13}$s whereas (1,5) and (3,8) are excluding from any one $OP^{13}$ in the frequency $K_{13}$. The edges contained in the $OP^n$s have the similar frequency and probability changes as those for (6,9), and the edges excluding from any $OP^n$  have the similar frequency and probability changes as those for (1,5) and (3,8). In addition, the probability decrement is bigger than $\frac{2p_i(g)}{i(i-1)}$ from $i=5$ for (6,9), from $i=4$ for (1,5) and (3,8). For (6,9) that is contained in some $OP^{13}$s, $p_{i+1}(g) < \frac{i}{i+1}p_i(g)$ appears if $i\geq 6$. For (1,5) and (3,8) excluding from $OP^{13}$s, $p_{i+1}(g) < \frac{i}{i+1}p_i(g)$ exists from $i = 4$. Thus, the probability decrement of most ordinary edges will be bigger than $\frac{2p_i(g)}{i(i-1)}$ and $\frac{p_i(g)}{i+1}$ at certain  number $i \ll P_0$. This case will happen to most ordinary edges before $i=2i_d$. One can use this condition to filter out most ordinary edges while the $OHC$ edges will be preserved. 

 \begin{figure}
	\centering
	\includegraphics[width=3in,bb=0 0 500 340]{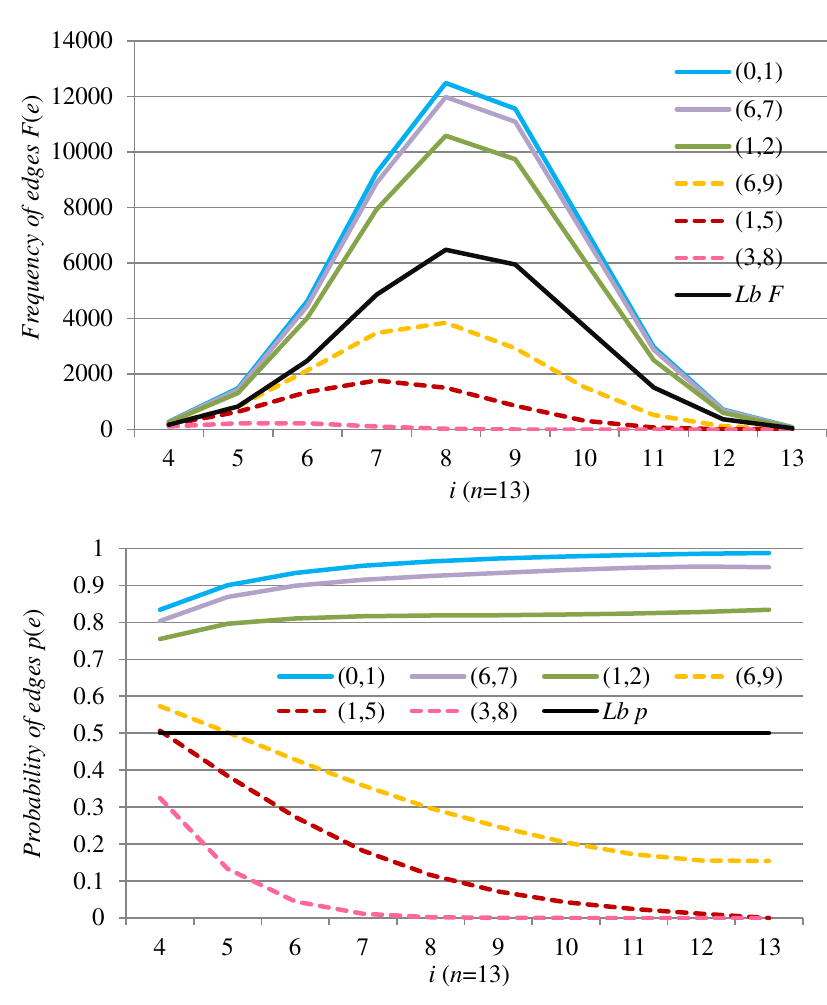}
	\caption{The frequency and probability changes for three general $OHC$ edges and three general ordinary edges for $K_{13}$.}
	\label{fpcs13or}
\end{figure}

In Figure \ref{fpcs14or}, the frequency and probability changes for the three general $OHC$ edges and ordinary edges are illustrated according to $i$ for even $n=14$. In the first picture, the frequency changes of the three $OHC$ edges and three ordinary edges are similar to those for the $OHC$ edges and ordinary edges in Figure \ref{fpcs13or}, respectively. The difference is that the $OHC$ edges reaches their peak frequencies at $P_0=9$ whereas the frequency of each of the ordinary edges arrives at the peak value before $P_0$. Moreover, only a small number of ordinary edges reach the peak frequency at $i=8$, such as (8,13), and most ordinary edges obtain the peak frequency at an even smaller number $i$. For example, (0,11) obtains the peak frequency at $i=7$ and (2,10) obtains the peak frequency at $i=6$. The ordinary edges having the peak frequency at certain $i < P_0$ are usually not contained in any one $OP^n$. For example, the frequency of (0,11) and (2,10) becomes zero at $i=14$ and 12, respectively. In this example, although (8,13) obtains the peak frequency at $i=8$ which is near 9, the frequency becomes zero at $i=14$. Thus, most ordinary edges having the frequency or probability close to and below the frequency bound or probability bound at $i=4$ are usually not contained in any one $OP^n$, and they can be identified as $i$ is small. 

\begin{figure}
	\centering
	\includegraphics[width=3in,bb=0 0 500 320]{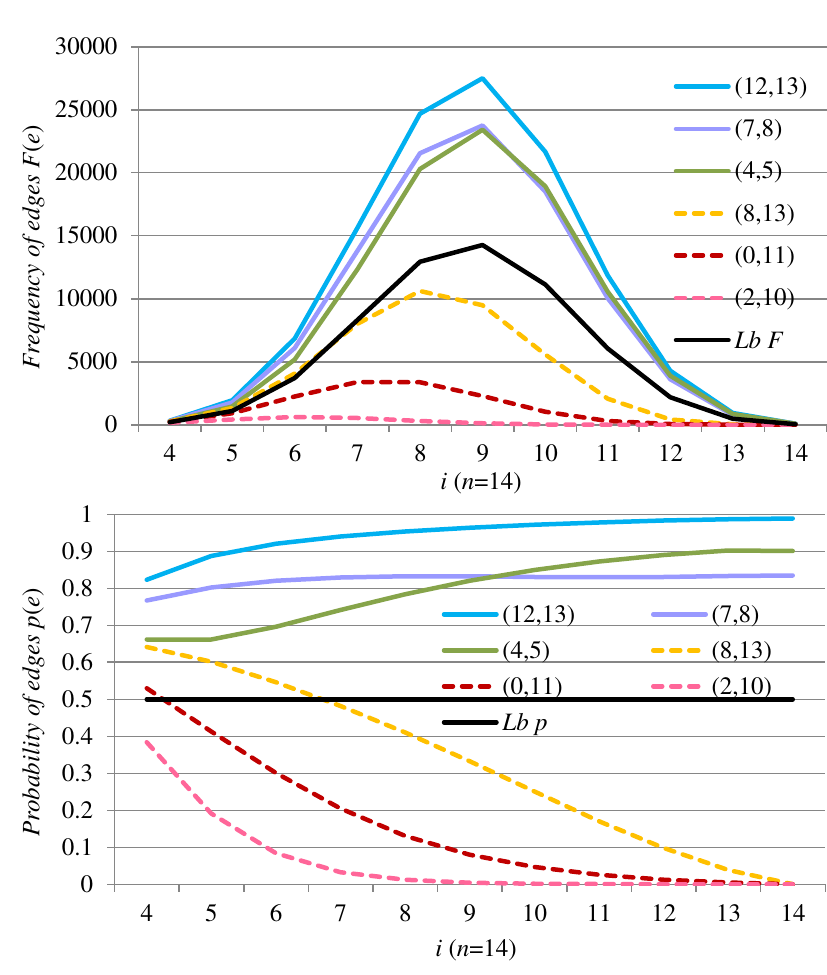}
	\caption{The frequency and probability changes for three general $OHC$ edges and three general ordinary edges for $K_{14}$.}
	\label{fpcs14or}
\end{figure}

In addition, the frequency change curves for these edges are not symmetrical with respect to the position of the peak frequency. In the left part of the curves for the $OHC$ edges before $i=P_0$, the frequencies increase in the relatively slow manner at first, and they increase very fast with the sharp slopes in the middle stage. At last, the increasing slopes becomes smaller from $P_0-1$ to $P_0$. After $i=P_0$, the frequencies decrease quickly according to $i$ at first, and the decreasing slopes become smaller from $i$ to $i+1$ as $i$ is close to $n$. As $n$ is odd, for example in the first picture in Figure \ref{fpcs13or}, the areas in the left part below the frequency change curves are smaller than those below the right part of the frequency change curves, respectively. If $n$ is even, such as the first picture in Figure \ref{fpcs14or}, the areas under the left frequency change curves are bigger than those under the right frequency change curves, respectively. For the ordinary edges, the areas under the left frequency change curves are generally smaller than those under the right frequency change curves, respectively. In total, the frequency changes for the $OHC$ edges and those for most ordinary edges have much difference. One can identify many ordinary edges according to their frequency changes according to $i$. However, some ordinary edges with the big frequency cannot be neglected according to the frequency bound $LbF$, especially as $i$ is small.  %increases  and For the $OHC$ edges, the frequency 
%not symmetrical 

%In view of the second picture, the probabilities of the three $OHC$ edges are much bigger than $Lbp=\frac{1}{2}$, and they are always increasing according to $i$. For the three ordinary edges, the probabilities decrease in the polynomial or exponential way according to $i$. They become smaller than $Lbp$ as $i$ is small. Moreover, they decrease quickly according to $i$ until it tends to the smallest value or zero. 
In the second picture in Figure \ref{fpcs14or}, the probabilities of the six edges at $i=4$ are obviously different. It indicates that they are contained in different number of $OP^4$s. In similarity, the probability changes for the $OHC$ edges have much difference from those for the ordinary edges. Firstly, the probability of an $OHC$ edge is much bigger than the lower probability bound $Lbp$ according to $i$, and $Lbp$ is bigger than that of the ordinary edge (8,13) as $i\geq 7$, (0,11) as $i\geq5$, and (2,10) as $i\geq 4$, respectively. As the probabilities of the $OHC$ edges are much bigger than $Lbp$ from $i = 4$, the probabilities of the $OHC$ edges never decrease and they always increase according to $i$. However, those of the ordinary edges decrease in the polynomial or exponential way simultaneously. The probabilities of the ordinary edges become smaller than $Lbp$ as $i$ is small. Moreover, they decrease quickly according to $i$ until they tend to the smallest value or zero. For the two $OHC$ edges (12,13) and (7,8), $p_{i+1}(e) < \left[1 + \frac{2}{i(i-1)}\right]p_i(e)$ exists from $i$ to $i+1$ since the $p_i(e) > 0.75$ which is close to 1. For the $OHC$ edge (4,5), $p_{i+1}(e) > \left[1 + \frac{2}{i(i-1)}\right]p_i(e)$ occurs from $i=7$ to 11 because these $p_i(e)$s are not close 1 until $p_{12}(e) = 0.890725$ appears.  From $i=12$ to 13, the probability increment becomes smaller than $\frac{2p_i(e)}{i(i-1)}$ since these $p_i(e)$s are close to 1. For the $OHC$ edges with the relatively small $p_i(e)$ (such as $<0.75$), the probability increment is generally bigger than $\frac{2p_i(e)}{i(i-1)}$ if $p_i(e)$ increases from $i$ to $i+1$. On the other hand, the probabilities of the three ordinary edges decrease quickly according to $i$. The probability decrement becomes smaller than $\frac{2p_i(g)}{i(i-1)}$ from $i=6$ for (6,9), from $i=4$ for (1,5) and (3,8). Moreover, $p_{i+1}(g) < \frac{i}{i+1}p_i(g)$ appears for (3,8) if $i \geq 8$, for (1,5) and (6,9) if $i\geq 4$. Such cases never happen to the $OHC$ edges. At the small  $i$s,  the probability decrement for most ordinary edges begins violating the necessary condition for $OHC$ edges. Thus, most ordinary edges can be identified at the small $i$s according to the probability decrement.

%In addition, we see the number $N_f$ of edges with frequency above zero is much smaller than the total number $N_{tot}$ of edges for big $n$, such as $n > 10$ in Table \ref{Examp}. If we neglect the edges with frequency 1, the number of edges is less than $3n$. Because $p(e\in OHC)\to 1$ as $n$ is big enough, the number of edgeswith frequency above $n-1$ will approach $n$. For big frequency graph $K_n$, the frequencies of the $OHC$ edges are close to $\frac{(n+1)(n-2)}{2}$ and the  frequencies of the other edges will tend to 1 or zero. 
\subsubsection{The number of ordinary edges with the probability decrement bigger than $\frac{2p_i(g)}{i(i-1)}$ for small $TSP$ instances}

Based on Theorem \ref{th3}, the probability of an $OHC$ edge is allowed to have the slight decrement smaller than $\frac{2p_i(e)}{i(i-1)}$ from $i$ to $i+1$. According to Theorem \ref{th33} and the experiments in above subsections, most ordinary edges will have the probability decrement $pd_i(g) > \frac{2p_i(g)}{i(i-1)}$ or $\frac{p_i(g)}{i+1}$ from $i$ to $i+1$ even if $i$ is small. Thus, these ordinary edges can be found based on the frequency $K_i$s containing a small number of vertices. The experiments for the six small $TSP$ instances in Table \ref{Examp} where  $n\in[9,14]$ are executed, and the number of ordinary edges with  $pd_i(g) > \frac{2p_i(g)}{i(i-1)}$ from $i$ to $i+1\in[5,n]$ is given in Table \ref{noerr} for each of the $TSP$ instances. 

\begin{table}
	\begin{center}
		\caption{The number of ordinary edges with the probability decrement bigger than $\frac{2p_i(g)}{i(i-1)}$ from $i$ to $i+1$ for the six small $TSP$ instances in Table \ref{Examp}.}
		{\footnotesize \begin{tabular}{ p{0.5cm}  p{1cm}  p{1cm}  p{1.5cm}  p{1cm}  p{1cm}  p{0.5cm} }
		\hline
		$n$	&  &  & $i\to i+1$ &  & &  $N_{ord}$ \\
		& $4\to 5$ & $5\to 6$ & $6\to 7$ & $7\to 8$  & $8\to 9$ &   \\				
		\hline
		9 & 23 & 27 &  &  &  &  27 \\		
		10 & 27 & 35 &  &  & &  35 \\	
		11 & 33 & 42 & 44 &  &  &  44 \\		
		12 & 31 & 49 & 54 &  &  &  54 \\	
		13 & 48 & 58 & 65 &  &  &  65 \\	
		14 & 58 & 64 & 76 & 77 &   &  77  \\					
		%\bottomrule
		\hline
		\end{tabular}}
		\label{noerr}
	\end{center}
\end{table}

One sees that the number of ordinary edges with $pd_4(g) > \frac{2p_4(g)}{12}$ from $i=4$ to 5 is near the total number of ordinary edges $N_{ord}$. It indicates that many ordinary edges can be separated from $OHC$ edges at $i=5$ with respect to the probability decrement. From $i=5$ to 6, all ordinary edges are found for $n=9$ and 10 with respect to the probability decrement. From $i=6$ to 7, all ordinary edges are identified for $n$=11, 12 and 13, and all ordinary edges are found from $i=7$ to 8 for $n=14$.  However, there are quite a few ordinary edges having the frequency or probability bigger than the smallest frequency or probability of the $OHC$ edges at $i=5$ for these small $TSP$ instances, see Table \ref{naminf}. Comparing with the lower frequency bound or probability bound for $OHC$ edges, the probability decrement is more useful to neglect the ordinary edges as $TSP$ is resolved.  In addition, to separate all ordinary edges from $OHC$ edges, the frequency $K_i$s containing more vertices will be used for the big scale of $TSP$. For large $TSP$, it will be time-consuming to compute the $OP^i$s and frequency $K_i$s containing many vertices. On the other hand, one can compute the sparse graphs containing a small number of ordinary edges with respect to the probability decrement for edges. Although not all ordinary edges are filtered out as $i$ is small, the number of the preserved ordinary edges will be much smaller than $\frac{n(n-3)}{2}$, and the search space of $OHC$ will be greatly reduced.

As the probability decrement is used to filter out the ordinary edges for the six small $TSP$ instances in Table \ref{Examp}, the sparse graphs containing the preserved ordinary edges and $OHC$ are illustrated in Figures \ref{sg912} and \ref{sg1314}, respectively. In the two Figures, the blue dashed lines represent the $OHC$ edges in each $K_n$, and the solid pink lines are the preserved ordinary edges in case that $pd_i(g) \leq \frac{2p_i(g)}{i(i-1)}$ exists from $i$ to $i+1$. For the $OHC$ edges, $pd_i(e) \leq \frac{2p_i(e)}{i(i-1)}$ exists, and $pd_i(e)  < 0$ happens in most cases. One sees that the number of the preserved ordinary edges is much smaller than $\frac{n(n-3)}{2}$ for these small $TSP$ instances even if from $i$=4 to 5. As $i$ rises, the number of the preserved ordinary edges decreases sharply. If $i$ is small, a few ordinary edges are preserved for these small $TSP$ instances. Moreover, the number of the preserved ordinary edges increases according to $n$ for given $i$, such as $n=10\sim 14$ for $i=4$. It says that the probabilities or frequencies of edges computed with the frequency $K_i$s containing more vertices are more useful to filter out the ordinary edges. 

In addition, the preserved ordinary edges are not always shorter than the neglected ordinary edges. If the edges are neglected according to the distances of edges, more ordinary edges will be preserved or some $OHC$ edges with the big distances will be eliminated. It implies that the distances of edges generally do not illustrate the structure properties of the $OHC$ edges so they are  not helpful to separate $OHC$  edges from ordinary edges. According to the probability decrement, all $OHC$ edges are preserved for these $TSP$ instances whereas most ordinary edges are filtered out even if $i$ is small. It indicates that the probabilities or frequencies of edges computed based on the frequency $K_i$s demonstrate the structure properties of the $OHC$ edges so they are useful to distinguish $OHC$ edges from ordinary edges. Moreover, the probability decrement is better than the lower frequency bound and probability bound for separating more ordinary edges from $OHC$ edges. 

Furthermore, most of the preserved ordinary edges are distributed along the $OHC$ in local topological structures. In general, the two vertices (or endpoints) contained in a preserved ordinary edge are the endpoints of one $OHC$ path only containing two or three $OHC$ edges. It indicates that the ordinary edges close to the (short) $OHC$ paths containing a small number of edges generally have the relatively big frequencies or probabilities, and they have the relatively small $pd_i(g)$ from $i$ to $i+1$ as $i$ is not big. For most of the other ordinary edges whose endpoints are those of the $OHC$ paths containing many edges, they generally have the small frequencies and probabilities. $pd_i(g) > \frac{2p_i(g)}{i(i-1)}$ will exist  for these ordinary edges from $i$ to $i+1$, and they can be identified even if $i$ is small. %As $i$ is small, one sees tht   they are ditributed along the $OHC$ edges

In addition, it is observed that after the $OHC$ edges in the $K_n$ are replaced by the new $OHC$ edges in the $K_{n+1}$, they are usually preserved to the sparse graphs $G_{n+1}$ corresponding to the $K_{n+1}$ as $i$ is small. For an $OHC$ edge in the $K_n$, it has the big frequency in each frequency $K_i$ containing it on average. As it is replaced by the other $OHC$ edge and becomes one ordinary edge in the $K_{n+1}$, it still maintains the big frequency in the ${{n-2}\choose{i-2}}$ frequency $K_i$s contained in the $K_n$. In each of the other ${{n-2}\choose{i-3}}$ frequency $K_i$s containing the edge in the $K_{n+1}$, it will have the small frequency $(<i-3)$. Since ${{n-2}\choose{i-2}}$ is much bigger than ${{n-2}\choose{i-3}}$ as $i$ is small, the average frequency of the edge will not decrease much according to the ${{n-1}\choose{i-2}}$ frequency $K_i$s containing it in the $K_{n+1}$. Thus, the replaced $OHC$ edges in the $K_{n+1}$ will keep the nearly equal average frequency as $i$ is small. The experiments illustrated that $pd_i(g) < \frac{2p_i(g)}{i(i-1)}$ exists for such ordinary edges from $i$ to $i+1$ as $i$ is small. Thus, they are preserved to the sparse graphs $G_{n+1}$ as $i$ is small. In fact, the preserved ordinary edges are mostly the $OHC$ edges in some $K_n$s contained in the $K_{n+1}$. Once $i$ becomes big, the probabilities of such ordinary edges will decrease quickly, and $pd_i(g) > \frac{2p_i(g)}{i(i-1)}$ will happen from $i$ to $i+1$. In this case, these ordinary edges are eliminated from the sparse graphs.

\begin{figure}
	\centering
	\includegraphics[width=2.5in,bb=0 0 450 400]{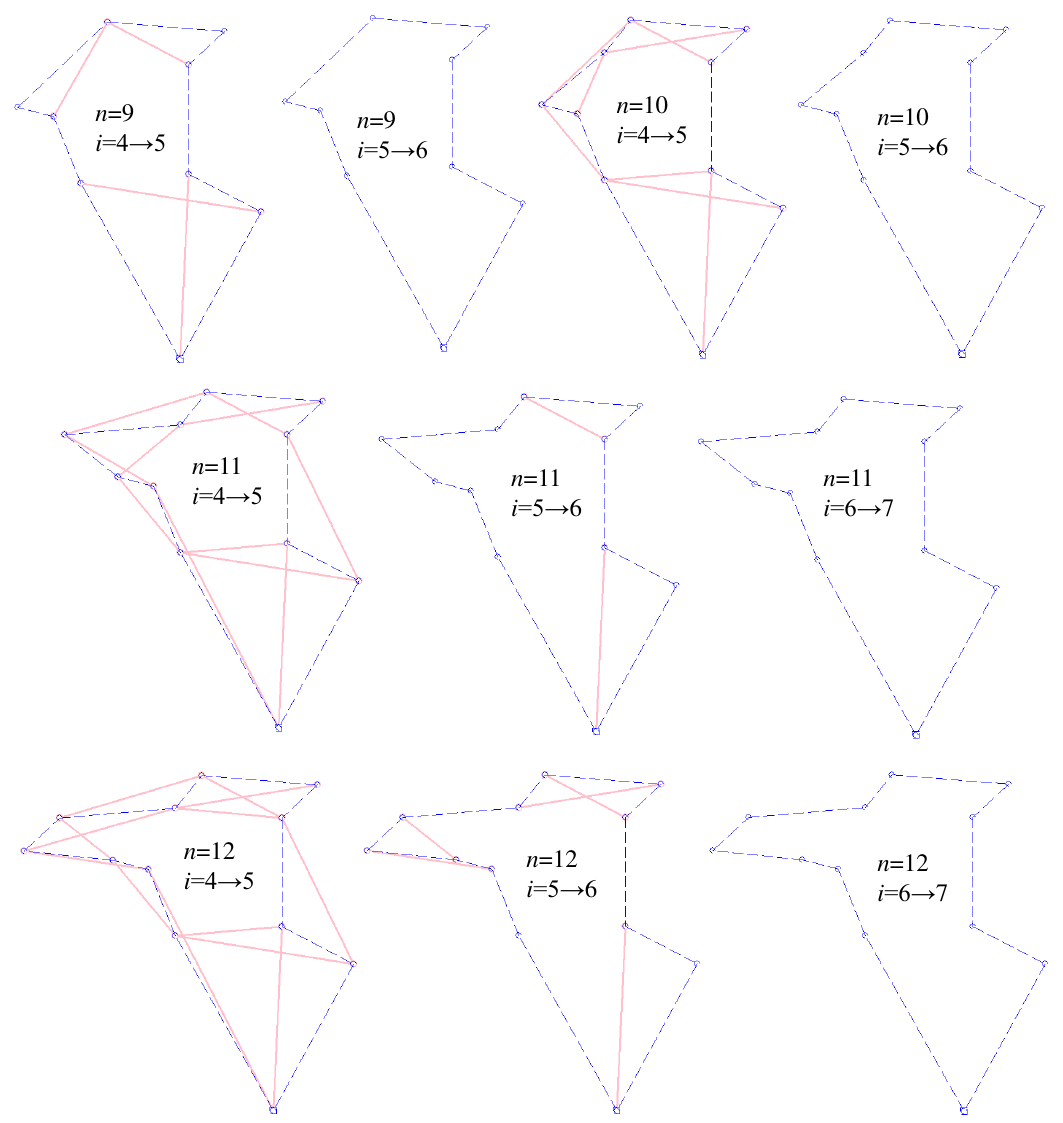}
	\caption{The sparse graphs computed with respect to the probability decrement for the small $TSP$ instances where $n=$9, 10, 11 and 12.}
	\label{sg912}
\end{figure}

\begin{figure}
	\centering
	\includegraphics[width=2.5in,bb=0 0 450 350]{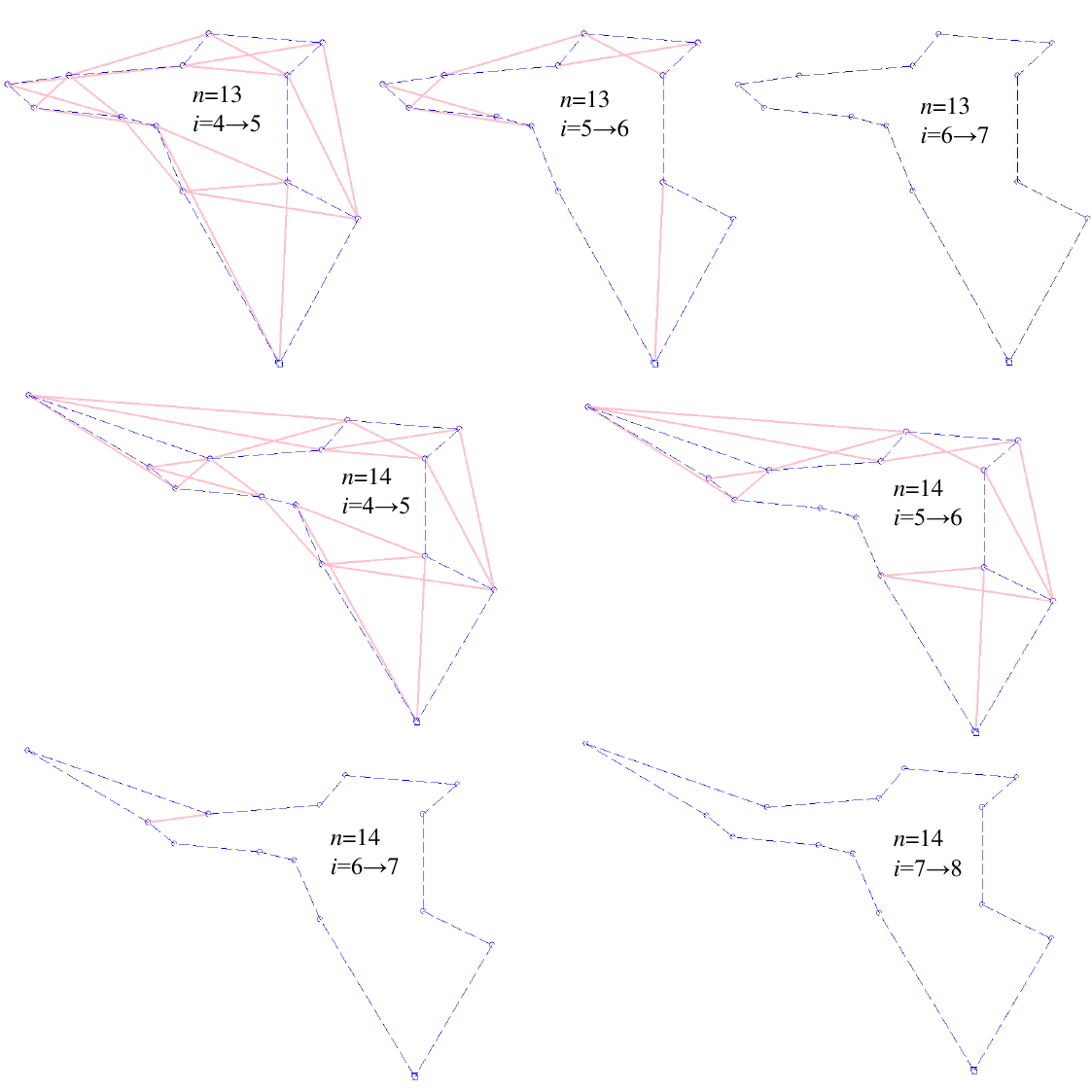}
	\caption{The sparse graphs computed with respect to the probability decrement for the small $TSP$ instances where $n=$13 and 14.}
	\label{sg1314}
\end{figure}

\subsection{The experiments for the real-world $TSP$ instances}
In this section, the experiments are executed for the real-world $TSP$ instances in $TSPLIB$ \cite{DBLP:journals/Reinelt95}. They include four types of $TSP$ which are Euclidean, ATT, GEO and Matrix. The methods to compute the distances of edges are different for the four types of $TSP$. For general $TSP$ having one $OHC$ and ${{n}\choose{2}}$ $OP^n$s in $K_n$, and each sub-graph $K_i$ contains ${{i}\choose{2}}$ $OP^i$s and one $OHC$, the theorems in this paper will work well for characterizing the $OHC$ edges and ordinary edges based on frequency $K_i$s. We will also meet some $TSP$ instances having a lot of $OHC$s or equal-weight edges. There will be more than ${{i}\choose{2}}$ $OP^i$s in many $K_i$s. In this case, it is difficult to choose the proper $OP^i$s to compute the high frequency for unknown $OHC$ edges. If the $OP^i$s containing a few $OHC$ edges are seldom used to compute the frequencies of edges, the theorems will lose powers to separate ordinary edges from some $OHC$ edges. Because $n$ is big in the experiments, it is impossible to compute the $OP^n$s for these $TSP$ instances using dynamic programming. To illustrate the frequency properties for $OHC$ edges and ordinary edges, we compute the frequency of each edge with the frequency $K_i$s as $i\ll n$, such as $i=$4, 5, 6, etc. Moreover, it is not necessary to compute the frequency of an edge with the ${{n-2}\choose{i-2}}$ frequency $K_i$s containing it. Given one $TSP$ instance, the $OP^i$s in each $K_i$ are determined according to the distances on edges. Each edge will have one specified frequency in every frequency $K_i$ containing it. Thus, the frequency distribution for each edge is fixed according to the frequency $K_i$s containing them, respectively. In this case, we can choose $N$ random frequency $K_i$s for an edge to compute the average frequency and probability. Based on the law of large numbers, the average frequency and probability of each edge will converge to the corresponding expected frequency and probability if the value of $N$ is rational. As the frequency of each edge is computed with $N$ frequency $K_i$s, the computation time to compute the frequencies of all edges is $O(Ni^42^in^2)$. In the experiments, $N=1000$ and $i\in[4,8]$ much smaller than $n$ are considered for the real-world $TSP$ instances. It mentions that one $OHC$ of each $TSP$ instance has been computed with Concorde online \cite{DBLP:journals/Mittelmann22}. They are used to verify the special structure properties denoted by the frequencies and probabilities of the $OHC$ edges and ordinary edges. According to $i\geq 4$, the frequency and probability changes for $OHC$ edges and ordinary edges will also be demonstrated. 

\subsubsection{The frequencies of the $OHC$ edges computed based on frequency $K_i$s}
For one given $TSP$ instance, the average frequency of each $OHC$ edge is computed with the frequency $K_i$s where $i\in [4,8]$. The smallest frequency and average frequency related to all $OHC$ edges will be illustrated and compared with the lower frequency bound $f_{lb}$ and theoretical average frequency $f_{oavg}$, respectively. For each of the $TSP$ instances, the first eight smallest frequencies of the $OHC$ edges are selected and shown in Tables \ref{ohcfs}$\sim$\ref{ohcfbl}, respectively. In each of the four tables, the first column denotes the name of $TSP$ in which the digit is the number of vertices $n$ for each $TSP$ instance. The second column $i$ denotes the number of vertices in the frequency $K_i$s which are used to compute the average frequency and probability of each $OHC$ edge. From the third to the tenth column, it gives the eight smallest frequencies of the $OHC$ edges which are ordered from small to big values. For example, the third column gives the first smallest frequency of the $OHC$ edges  and the tenth column presents the $8^{th}$ smallest frequency of the $OHC$ edges for each $TSP$ instance. The average frequency of all $OHC$ edges is also computed and denoted by $f_{avg}$. In the last two columns, the lower frequency bound $f_{lb}$ and theoretical average frequency $f_{oavg}$ for $OHC$ edges are given for comparisons. 

We first compare the first smallest frequency with $f_{lb}$ for these $TSP$ instances. As $n$ is small, the first smallest frequency is smaller than $f_{lb}$ for several $TSP$ instances, such as att48 when $i=6$, 7, pr76 when $i=6$, rat99 when $i=$5, 6 and 7, si535 when $i=$7 and 8, and si1032. However, the first smallest frequency is bigger than the lower bound $\frac{7}{18}{{i}\choose{2}}$ in the worst case for all the $TSP$ instances. For the other $TSP$ instances, the first smallest frequency is bigger than $f_{lb}$ according to $i\in [4,8]$.  Moreover, the difference between the first smallest frequency and $f_{lb}$ becomes bigger according to $n$. For example in Table \ref{ohcfbl}, the first smallest frequency of the $OHC$ edges is much bigger than $f_{lb}$. It implies that the average frequency of an $OHC$ edge becomes bigger according to $n$ for given $i$. Thus, the $OHC$ edges are contained in more percentage of the $OP^i$s according to $n$. 

The second smallest frequency is totally bigger than $f_{lb}$ for all the $TSP$ instances according to $i$ except si1032. It says the average frequency of almost every $OHC$ edges is bigger than $f_{lb}$ for each $TSP$ instance. For si1032, the third smallest frequency becomes bigger than $f_{lb}$ and only two $OHC$ edges have the average frequency smaller than $f_{lb}$. Moreover, the fourth smallest frequency is much bigger than the first three smallest frequencies. It mentions that si535 and si1032 are Matrix $TSP$, and they contain many equal-weight edges. In the two $K_n$s, there are a lot of $K_i$s which contain more than ${{i}\choose{2}}$ $OP^i$s. It is difficult to choose the right $OP^i$s for computing the high frequency for some $OHC$ edges in the known $OHC$. Nevertheless, the second and third smallest frequency of the $OHC$ edges become bigger than $f_{lb}$ for si535 and si1032, respectively. For the other $TSP$ instances, Theorem \ref{th3} works well to estimate the lower frequency bound for the $OHC$ edges computed based on frequency $K_i$s. 

From the first to the $8^{th}$ smallest frequency, it is observed that the frequencies of the $OHC$ edges increase quickly. In view of the average frequency $f_{avg}$ of all $OHC$ edges, it is bigger than  $f_{oavg}$ for whether the small or big $TSP$ instances. It indicates that the number of $OHC$ edges with the average frequency close to $f_{lb}$ is very small. Most of the $OHC$ edges will have an average frequency much bigger than $f_{lb}$. Based on the datum in Tables \ref{ohcfs}$\sim$\ref{ohcfbl}, one can compute the probability for the $OHC$ edge with the corresponding smallest frequency and all $OHC$ edges, it finds that the probability becomes bigger according to $n$ for these $OHC$ edges, and it also increases quickly according to $i\in[4,8]$ for given $n$. It indicates that an $OHC$ edge is not only contained in more percentage of $OP^i$s according to $n$ for given $i$, and it is also contained in more percentage of the $OP^i$s according to $i$ for given $n$. On the other hand, an ordinary edge must be contained in the smaller percentage of $OP^i$s according to $n$ and $i$ on average. Theorem \ref{th3}  are verified with these experimental results. 
%first, second and third smallest frequency and average frequency according to $i$

\begin{table}
	\begin{center}
		\caption{The eight smallest frequencies of the $OHC$ edges computed with frequency $K_i$s ($i\in[4,8]$) for small real-world $TSP$ instances.}
		{\footnotesize \begin{tabular}{ p{0.5cm}  p{0.1cm}  p{0.5cm}  p{0.5cm}  p{0.5cm}  p{0.5cm}  p{0.5cm}  p{0.5cm}  p{0.5cm}  p{0.5cm}  p{0.5cm}  p{0.3cm}  p{0.4cm} }
				\hline
		$TSP$ & $i$ &  & The &  &  smallest &  &  & freq. &  & $f_{avg}$ & $f_{lb}$ &  $f_{oavg}$  \\				
		$/n$ &  & $1^{st}$ & $2^{nd}$ & $3^{rd}$ & $4^{th}$ & $5^{th}$ & $6^{th}$ & $7^{th}$ & $8^{th}$ &  &  &  \\	
		\hline
	att48	&	4	&	3.410 	&	3.462 	&	3.510 	&	3.708 	&	3.768 	&	3.772 	&	3.786 	&	3.792 	&	4.305 	&	3	&	3.5	\\
	&	5	&	5.439 	&	5.483 	&	5.558 	&	5.821 	&	5.860 	&	6.095 	&	6.132 	&	6.214 	&	7.370 	&	5	&	6	\\
	48	&	6	&	7.436 	&	7.674 	&	8.355 	&	8.694 	&	8.914 	&	8.941 	&	9.118 	&	9.152 	&	11.216 	&	7.5	&	9.5	\\
	&	7	&	10.434 	&	10.581 	&	11.689 	&	11.986 	&	12.092 	&	12.097 	&	12.604 	&	12.842 	&	15.835 	&	10.5	&	14	\\
	&	8	&	14.535 	&	14.629 	&	15.634 	&	16.171 	&	16.259 	&	16.416 	&	16.484 	&	18.094 	&	21.319 	&	14	&	19.5	\\	
	
	eil51	&	4	&	3.668 	&	3.760 	&	3.788 	&	3.804 	&	3.818 	&	3.846 	&	3.872 	&	3.908 	&	4.424 	&	3	&	3.5	\\
	&	5	&	5.928 	&	5.939 	&	6.043 	&	6.094 	&	6.139 	&	6.163 	&	6.327 	&	6.457 	&	7.688 	&	5	&	6	\\
	51	&	6	&	8.795 	&	8.869 	&	9.002 	&	9.121 	&	9.294 	&	9.385 	&	9.413 	&	9.940 	&	11.772 	&	7.5	&	9.5	\\
	&	7	&	11.934 	&	13.197 	&	13.209 	&	13.320 	&	13.486 	&	13.707 	&	13.720 	&	13.755 	&	16.727 	&	10.5	&	14	\\
	&	8	&	15.360 	&	17.430 	&	17.453 	&	18.099 	&	18.122 	&	18.739 	&	18.801 	&	18.881 	&	22.466 	&	14	&	19.5	\\
	
	{\tiny berlin52}	&	4	&	3.696 	&	3.774 	&	3.802 	&	3.824 	&	3.842 	&	3.876 	&	3.948 	&	3.948 	&	4.413 	&	3	&	3.5	\\
	&	5	&	5.894 	&	6.040 	&	6.243 	&	6.364 	&	6.451 	&	6.534 	&	6.673 	&	6.704 	&	7.680 	&	5	&	6	\\
	52	&	6	&	9.063 	&	9.143 	&	9.302 	&	9.744 	&	10.058 	&	10.084 	&	10.113 	&	10.408 	&	11.801 	&	7.5	&	9.5	\\
	&	7	&	12.662 	&	12.899 	&	13.466 	&	13.684 	&	13.896 	&	14.294 	&	14.567 	&	14.570 	&	16.760 	&	10.5	&	14	\\
	&	8	&	17.224 	&	17.372 	&	18.257 	&	18.503 	&	18.642 	&	19.275 	&	19.457 	&	20.101 	&	22.650 	&	14	&	19.5	\\
	
	pr76	&	4	&	3.418 	&	3.532 	&	3.608 	&	3.794 	&	3.950 	&	4.032 	&	4.092 	&	4.100 	&	4.515 	&	3	&	3.5	\\
	&	5	&	5.139 	&	5.448 	&	5.692 	&	5.877 	&	6.548 	&	6.712 	&	6.748 	&	6.931 	&	7.918 	&	5	&	6	\\
	76	&	6	&	6.933 	&	8.108 	&	8.426 	&	9.070 	&	10.068 	&	10.337 	&	10.599 	&	10.649 	&	12.225 	&	7.5	&	9.5	\\
	&	7	&	8.942 	&	11.211 	&	12.347 	&	13.359 	&	14.270 	&	14.437 	&	14.786 	&	14.868 	&	17.434 	&	10.5	&	14	\\
	&	8	&	10.903 	&	15.511 	&	17.586 	&	18.539 	&	18.679 	&	18.884 	&	19.559 	&	19.598 	&	23.487 	&	14	&	19.5	\\
	rat99	&	4	&	3.524 	&	3.556 	&	3.558 	&	3.652 	&	3.656 	&	3.656 	&	3.708 	&	3.730 	&	4.426 	&	3	&	3.5	\\
	&	5	&	4.919 	&	5.269 	&	5.306 	&	5.710 	&	5.747 	&	5.894 	&	5.898 	&	6.106 	&	7.667 	&	5	&	6	\\
	99	&	6	&	7.311 	&	7.767 	&	8.176 	&	8.686 	&	8.752 	&	8.866 	&	8.921 	&	9.083 	&	11.833 	&	7.5	&	9.5	\\
	&	7	&	10.428 	&	11.507 	&	12.414 	&	12.671 	&	12.760 	&	12.780 	&	12.905 	&	13.113 	&	16.912 	&	10.5	&	14	\\
	&	8	&	14.536 	&	15.424 	&	16.703 	&	17.477 	&	17.511 	&	17.798 	&	18.065 	&	18.305 	&	22.937 	&	14	&	19.5	\\

   gr120	&	4	&	3.294 	&	3.648 	&	3.666 	&	3.720 	&	3.784 	&	3.800 	&	3.818 	&	3.848 	&	4.495 	&	3	&	3.5	\\
   &	5	&	4.749 	&	5.804 	&	5.897 	&	6.070 	&	6.143 	&	6.156 	&	6.204 	&	6.220 	&	7.843 	&	5	&	6	\\
   120	&	6	&	6.998 	&	8.365 	&	9.151 	&	9.226 	&	9.226 	&	9.249 	&	9.302 	&	9.308 	&	12.052 	&	7.5	&	9.5	\\
   &	7	&	9.733 	&	11.611 	&	12.775 	&	12.894 	&	13.142 	&	13.216 	&	13.273 	&	13.524 	&	17.123 	&	10.5	&	14	\\
   &	8	&	13.667 	&	15.336 	&	17.304 	&	17.564 	&	18.106 	&	18.314 	&	18.472 	&	18.684 	&	23.105 	&	14	&	19.5	\\
   
   gr137	&	4	&	3.564 	&	3.696 	&	3.748 	&	3.762 	&	3.766 	&	3.766 	&	3.844 	&	3.860 	&	4.462 	&	3	&	3.5	\\
   &	5	&	5.485 	&	5.799 	&	5.903 	&	6.101 	&	6.184 	&	6.290 	&	6.326 	&	6.343 	&	7.785 	&	5	&	6	\\
   137	&	6	&	7.849 	&	8.351 	&	8.639 	&	9.148 	&	9.278 	&	9.490 	&	9.546 	&	9.780 	&	12.041 	&	7.5	&	9.5	\\
   &	7	&	11.514 	&	11.656 	&	12.274 	&	12.431 	&	12.917 	&	13.727 	&	13.729 	&	13.950 	&	17.260 	&	10.5	&	14	\\
   &	8	&	15.652 	&	16.367 	&	17.063 	&	17.131 	&	17.210 	&	17.313 	&	18.422 	&	19.364 	&	23.407 	&	14	&	19.5	\\
   
	{\tiny kroA200}	&	4	&	3.562 	&	3.704 	&	3.738 	&	3.754 	&	3.780 	&	3.856 	&	3.876 	&	3.896 	&	4.598 	&	3	&	3.5	\\
	&	5	&	5.500 	&	5.605 	&	5.974 	&	6.149 	&	6.244 	&	6.303 	&	6.309 	&	6.316 	&	8.103 	&	5	&	6	\\
	200	&	6	&	8.398 	&	8.645 	&	9.077 	&	9.352 	&	9.651 	&	9.712 	&	9.813 	&	9.946 	&	12.575 	&	7.5	&	9.5	\\
	&	7	&	12.208 	&	12.856 	&	12.904 	&	13.202 	&	13.348 	&	13.371 	&	14.258 	&	14.263 	&	17.933 	&	10.5	&	14	\\
	&	8	&	17.695 	&	17.703 	&	18.092 	&	18.709 	&	18.877 	&	19.197 	&	19.945 	&	19.950 	&	24.362 	&	14	&	19.5	\\
		
   gr202	&	4	&	3.980 	&	4.030 	&	4.030 	&	4.032 	&	4.032 	&	4.086 	&	4.102 	&	4.108 	&	4.593 	&	3	&	3.5	\\
   &	5	&	6.616 	&	6.675 	&	6.822 	&	6.845 	&	6.920 	&	6.954 	&	6.993 	&	7.101 	&	8.114 	&	5	&	6	\\
   202	&	6	&	10.292 	&	10.397 	&	10.410 	&	10.450 	&	10.901 	&	11.035 	&	11.102 	&	11.130 	&	12.606 	&	7.5	&	9.5	\\
   &	7	&	14.622 	&	15.108 	&	15.216 	&	15.373 	&	15.437 	&	15.934 	&	16.011 	&	16.063 	&	18.041 	&	10.5	&	14	\\
   &	8	&	19.990 	&	20.210 	&	20.841 	&	20.980 	&	21.075 	&	21.387 	&	21.553 	&	21.611 	&	24.405 	&	14	&	19.5	\\
   a280	&	4	&	3.630 	&	3.946 	&	4.104 	&	4.118 	&	4.120 	&	4.148 	&	4.156 	&	4.168 	&	4.645 	&	3	&	3.5	\\
   &	5	&	5.731 	&	6.599 	&	6.831 	&	7.045 	&	7.115 	&	7.129 	&	7.133 	&	7.156 	&	8.229 	&	5	&	6	\\
   280	&	6	&	8.761 	&	10.366 	&	10.755 	&	11.183 	&	11.252 	&	11.281 	&	11.323 	&	11.352 	&	12.801 	&	7.5	&	9.5	\\
   &	7	&	12.878 	&	13.997 	&	15.729 	&	16.115 	&	16.118 	&	16.191 	&	16.200 	&	16.312 	&	18.349 	&	10.5	&	14	\\
   &	8	&	18.158 	&	21.383 	&	21.760 	&	21.908 	&	22.083 	&	22.426 	&	22.483 	&	22.495 	&	24.855 	&	14	&	19.5	\\

   lin318	&	4	&	3.788 	&	3.908	&	3.926	&	3.95	&	3.954	&	3.964	&	3.964	&	3.976	&	4.691	&	3	&	3.5	\\
&	5	&	5.882	&	6.384 	&	6.505	&	6.608	&	6.621	&	6.685	&	6.698 	&	6.707	&	8.345	&	5	&	6	\\
318	&	6	&	9.111	&	9.65	&	9.926	&	10.002 	&	10.08	&	10.158	&	10.207	&	10.242	&	12.962 	&	7.5	&	9.5	\\
&	7	&	13.049	&	14.226 	&	14.378	&	14.478	&	14.581	&	14.682	&	14.756 	&	14.788	&	18.522	&	10.5	&	14	\\
&	8	&	17.726	&	19.439 	&	19.845	&	19.884	&	20.015	&	20.041	&	20.148 	&	20.157	&	25.024	&	14	&	19.5	\\
				%\bottomrule
				\hline
		\end{tabular}}
		\label{ohcfs}
	\end{center}
\end{table}

\begin{table}
	\begin{center}
		\caption{The eight smallest frequencies of the $OHC$ edges computed with frequency $K_i$s ($i\in[4,8]$) for medium real-world $TSP$ instances.}
		{\footnotesize \begin{tabular}{ p{0.6cm}  p{0.1cm}  p{0.5cm}  p{0.5cm}  p{0.5cm}  p{0.5cm}  p{0.5cm}  p{0.5cm}  p{0.5cm}  p{0.5cm}  p{0.5cm}  p{0.3cm}  p{0.4cm} }
		\hline
	$TSP$ & $i$ &  & The &  &  smallest &  &  & freq. &  & $f_{avg}$ & $f_{lb}$ &  $f_{oavg}$  \\				
	$/n$ &  & $1^{st}$ & $2^{nd}$ & $3^{rd}$ & $4^{th}$ & $5^{th}$ & $6^{th}$ & $7^{th}$ & $8^{th}$ &  &  &  \\	
	\hline
   rd400	&	4	&	3.840 	&	3.928 	&	4.042 	&	4.044 	&	4.064 	&	4.066 	&	4.066 	&	4.082 	&	4.691 	&	3	&	3.5	\\
   &	5	&	6.587 	&	6.600 	&	6.752 	&	6.796 	&	6.856 	&	6.905 	&	6.914 	&	6.978 	&	8.356 	&	5	&	6	\\
   400	&	6	&	9.878 	&	10.311 	&	10.604 	&	10.705 	&	10.811 	&	10.822 	&	10.851 	&	11.098 	&	12.982 	&	7.5	&	9.5	\\
   &	7	&	14.567 	&	15.423 	&	15.498 	&	15.707 	&	15.724 	&	15.758 	&	15.874 	&	16.199 	&	18.568 	&	10.5	&	14	\\
   &	8	&	20.674 	&	21.086 	&	21.202 	&	21.433 	&	21.887 	&	21.940 	&	21.977 	&	22.037 	&	25.096 	&	14	&	19.5	\\
   
   %pr439	&	4	&	3.874 	&	3.884 	&	3.992 	&	4.000 	&	4.018 	&	4.052 	&	4.056 	&	4.098 	&	4.706 	&	3	&	3.5	\\
   %&	5	&	6.516 	&	6.586 	&	6.707 	&	6.721 	&	6.748 	&	6.888 	&	6.905 	&	6.969 	&	8.395 	&	5	&	6	\\
   %439	&	6	&	10.113 	&	10.392 	&	10.449 	&	10.629 	&	10.674 	&	10.786 	&	10.908 	&	10.916 	&	13.059 	&	7.5	&	9.5	\\
   %&	7	&	14.309 	&	15.059 	&	15.291 	&	15.484 	&	15.689 	&	15.720 	&	15.810 	&	16.044 	&	18.690 	&	10.5	&	14	\\
   %&	8	&	18.266 	&	20.729 	&	21.034 	&	21.426 	&	21.552 	&	21.653 	&	21.711 	&	21.782 	&	25.278 	&	14	&	19.5	\\
   
   pcb442	&	4	&	3.868 	&	4.008 	&	4.010 	&	4.164 	&	4.179 	&	4.192 	&	4.198 	&	4.212 	&	4.700 	&	3	&	3.5	\\
   &	5	&	6.315 	&	6.804 	&	7.020 	&	7.131 	&	7.150 	&	7.208 	&	7.245 	&	7.256 	&	8.371 	&	5	&	6	\\
   442	&	6	&	9.707 	&	10.453 	&	10.973 	&	10.993 	&	11.023 	&	11.327 	&	11.456 	&	11.526 	&	13.021 	&	7.5	&	9.5	\\
   &	7	&	14.627 	&	15.033 	&	15.612 	&	15.774 	&	15.863 	&	16.135 	&	16.411 	&	16.583 	&	18.638 	&	10.5	&	14	\\
   &	8	&	20.319 	&	20.331 	&	20.888 	&	21.157 	&	21.433 	&	21.905 	&	22.326 	&	22.413 	&	25.200 	&	14	&	19.5	\\
   
   att532	&	4	&	4.032 	&	4.038 	&	4.046 	&	4.060 	&	4.064 	&	4.064 	&	4.072 	&	4.080 	&	4.652 	&	3	&	3.5	\\
   &	5	&	6.446 	&	6.811 	&	6.846 	&	6.875 	&	6.915 	&	6.932 	&	6.936 	&	6.950 	&	8.262 	&	5	&	6	\\
   532	&	6	&	9.846 	&	10.398 	&	10.560 	&	10.679 	&	10.793 	&	10.798 	&	10.803 	&	10.860 	&	12.861 	&	7.5	&	9.5	\\
   &	7	&	14.076 	&	14.248 	&	15.299 	&	15.460 	&	15.562 	&	15.611 	&	15.668 	&	15.680 	&	18.442 	&	10.5	&	14	\\
   &	8	&	18.111 	&	19.901 	&	19.922 	&	19.999 	&	20.778 	&	21.083 	&	21.170 	&	21.432 	&	24.974 	&	14	&	19.5	\\
   
   si535	&	4	&	3.576 	&	3.580 	&	3.806 	&	4.054 	&	4.058 	&	4.078 	&	4.372 	&	4.444 	&	4.944 	&	3	&	3.5	\\
   &	5	&	5.388 	&	5.510 	&	6.184 	&	6.714 	&	6.890 	&	6.956 	&	7.449 	&	7.612 	&	8.872 	&	5	&	6	\\
   535	&	6	&	7.817 	&	7.964 	&	9.188 	&	10.171 	&	10.251 	&	10.574 	&	11.136 	&	11.631 	&	13.770 	&	7.5	&	9.5	\\
   &	7	&	9.823 	&	11.197 	&	13.184 	&	14.526 	&	14.730 	&	15.517 	&	16.053 	&	16.071 	&	19.634 	&	10.5	&	14	\\
   &	8	&	12.858 	&	14.712 	&	17.861 	&	18.645 	&	19.294 	&	21.235 	&	21.424 	&	21.436 	&	26.456 	&	14	&	19.5	\\
   
   u574	&	4	&	3.984 	&	3.988 	&	4.002 	&	4.014 	&	4.016 	&	4.016 	&	4.028 	&	4.054 	&	4.687 	&	3	&	3.5	\\
   &	5	&	6.504 	&	6.567 	&	6.661 	&	6.722 	&	6.744 	&	6.766 	&	6.809 	&	6.826 	&	8.350 	&	5	&	6	\\
   574	&	6	&	9.947 	&	10.118 	&	10.362 	&	10.454 	&	10.510 	&	10.607 	&	10.680 	&	10.718 	&	13.014 	&	7.5	&	9.5	\\
   &	7	&	14.405 	&	14.603 	&	15.449 	&	15.554 	&	15.680 	&	15.725 	&	15.753 	&	15.788 	&	18.662 	&	10.5	&	14	\\
   &	8	&	19.360 	&	19.874 	&	21.637 	&	21.676 	&	21.720 	&	21.895 	&	21.907 	&	21.926 	&	25.292 	&	14	&	19.5	\\
   
   rat575	&	4	&	3.846 	&	3.906 	&	3.914 	&	3.982 	&	4.008 	&	4.012 	&	4.014 	&	4.016 	&	4.636 	&	3	&	3.5	\\
   &	5	&	6.366 	&	6.689 	&	6.755 	&	6.831 	&	6.845 	&	6.903 	&	6.907 	&	6.923 	&	8.235 	&	5	&	6	\\
   575	&	6	&	10.060 	&	10.286 	&	10.762 	&	10.778 	&	10.798 	&	10.913 	&	10.932 	&	10.975 	&	12.834 	&	7.5	&	9.5	\\
   &	7	&	15.386 	&	15.466 	&	15.914 	&	15.918 	&	15.964 	&	16.037 	&	16.126 	&	16.192 	&	18.430 	&	10.5	&	14	\\
   &	8	&	21.843 	&	21.887 	&	21.890 	&	21.924 	&	22.267 	&	22.295 	&	22.305 	&	22.363 	&	25.023 	&	14	&	19.5	\\
   
   p654	&	4	&	3.552 	&	3.556 	&	3.558 	&	3.656 	&	3.724 	&	3.784 	&	3.838 	&	3.846 	&	4.729 	&	3	&	3.5	\\
   &	5	&	5.195 	&	5.514 	&	5.753 	&	5.877 	&	6.120 	&	6.296 	&	6.414 	&	6.451 	&	8.449 	&	5	&	6	\\
   654	&	6	&	7.043 	&	7.784 	&	9.214 	&	9.223 	&	9.301 	&	9.591 	&	10.072 	&	10.175 	&	13.035 	&	7.5	&	9.5	\\
   &	7	&	9.182 	&	11.295 	&	13.434 	&	13.523 	&	13.564 	&	14.044 	&	14.441 	&	14.529 	&	18.449 	&	10.5	&	14	\\
   &	8	&	11.343 	&	15.776 	&	18.137 	&	18.423 	&	19.254 	&	19.321 	&	19.507 	&	19.973 	&	24.704 	&	14	&	19.5	\\
   
   d657	&	4	&	3.842 	&	4.048 	&	4.050 	&	4.124 	&	4.130 	&	4.150 	&	4.198 	&	4.206 	&	4.728 	&	3	&	3.5	\\
   &	5	&	6.314 	&	6.764 	&	6.969 	&	7.012 	&	7.013 	&	7.048 	&	7.154 	&	7.231 	&	8.445 	&	5	&	6	\\
   657	&	6	&	9.915 	&	10.922 	&	10.964 	&	11.026 	&	11.187 	&	11.314 	&	11.374 	&	11.424 	&	13.152 	&	7.5	&	9.5	\\
   &	7	&	14.456 	&	16.160 	&	16.309 	&	16.350 	&	16.391 	&	16.479 	&	16.538 	&	16.542 	&	18.830 	&	10.5	&	14	\\
   &	8	&	20.280 	&	21.945 	&	22.051 	&	22.220 	&	22.442 	&	22.641 	&	22.646 	&	22.691 	&	25.470 	&	14	&	19.5	\\
   gr666	&	4	&	4.124 	&	4.178 	&	4.216 	&	4.244 	&	4.262 	&	4.274 	&	4.288 	&	4.302 	&	4.702 	&	3	&	3.5	\\
   &	5	&	6.785 	&	7.086 	&	7.292 	&	7.325 	&	7.331 	&	7.344 	&	7.367 	&	7.527 	&	8.413 	&	5	&	6	\\
   666	&	6	&	10.444 	&	11.277 	&	11.363 	&	11.435 	&	11.464 	&	11.600 	&	11.618 	&	11.674 	&	13.153 	&	7.5	&	9.5	\\
   &	7	&	14.789 	&	16.097 	&	16.327 	&	16.568 	&	16.599 	&	16.678 	&	16.690 	&	16.758 	&	18.854 	&	10.5	&	14	\\
   &	8	&	19.823 	&	22.003 	&	22.325 	&	22.594 	&	22.812 	&	22.982 	&	23.028 	&	23.074 	&	25.527 	&	14	&	19.5	\\
   u724	&	4	&	3.922 	&	4.020 	&	4.022 	&	4.050 	&	4.054 	&	4.058 	&	4.060 	&	4.066 	&	4.701 	&	3	&	3.5	\\
   &	5	&	6.655 	&	6.757 	&	6.822 	&	6.884 	&	6.893 	&	6.955 	&	6.994 	&	6.998 	&	8.382 	&	5	&	6	\\
   724	&	6	&	10.452 	&	10.789 	&	10.803 	&	10.818 	&	11.043 	&	11.097 	&	11.112 	&	11.165 	&	13.060 	&	7.5	&	9.5	\\
   &	7	&	15.280 	&	15.993 	&	16.195 	&	16.245 	&	16.403 	&	16.427 	&	16.572 	&	16.584 	&	18.736 	&	10.5	&	14	\\
   &	8	&	21.110 	&	22.316 	&	22.436 	&	22.555 	&	22.604 	&	22.682 	&	22.702 	&	22.710 	&	25.386 	&	14	&	19.5	\\
   rat783	&	4	&	3.984 	&	3.988 	&	3.990 	&	3.994 	&	4.016 	&	4.082 	&	4.092 	&	4.097 	&	4.679 	&	3	&	3.5	\\
   &	5	&	6.701 	&	6.717 	&	6.739 	&	6.766 	&	6.769 	&	6.790 	&	6.886 	&	6.900 	&	8.328 	&	5	&	6	\\
   783	&	6	&	10.717 	&	10.754 	&	10.795 	&	10.860 	&	10.905 	&	10.925 	&	11.069 	&	11.105 	&	12.996 	&	7.5	&	9.5	\\
   &	7	&	15.900 	&	15.922 	&	16.008 	&	16.031 	&	16.251 	&	16.252 	&	16.253 	&	16.302 	&	18.667 	&	10.5	&	14	\\
   &	8	&	22.019 	&	22.095 	&	22.409 	&	22.490 	&	22.581 	&	22.587 	&	22.587 	&	22.629 	&	25.323 	&	14	&	19.5	\\	
		%\bottomrule
		\hline
		\end{tabular}}
		\label{ohcfm}
	\end{center}
\end{table}

\begin{table}
	\begin{center}
		\caption{The eight smallest frequencies of the $OHC$ edges computed with frequency $K_i$s ($i\in[4,8]$) for big  real-world $TSP$ instances.}
		{\footnotesize \begin{tabular}{ p{0.6cm}  p{0.1cm}  p{0.5cm}  p{0.5cm}  p{0.5cm}  p{0.5cm}  p{0.5cm}  p{0.5cm}  p{0.5cm}  p{0.5cm}  p{0.5cm}  p{0.3cm}  p{0.4cm} }
		\hline
	$TSP$ & $i$ &  & The &  &  smallest &  &  & freq. &  & $f_{avg}$ & $f_{lb}$ &  $f_{oavg}$  \\				
	$/n$ &  & $1^{st}$ & $2^{nd}$ & $3^{rd}$ & $4^{th}$ & $5^{th}$ & $6^{th}$ & $7^{th}$ & $8^{th}$ &  &  &  \\	
		\hline
	pr1002	&	4	&	3.840 	&	4.034	&	4.050 	&	4.066	&	4.116 	&	4.12	&	4.138 	&	4.142	&	4.760 	&	3	&	3.5	\\
	&	5	&	6.322	&	6.814	&	7.003	&	7.023	&	7.036	&	7.045	&	7.149	&	7.196	&	8.519 	&	5	&	6	\\
	1002	&	6	&	9.948	&	10.815	&	10.85	&	10.877	&	11.115	&	11.129	&	11.184	&	11.251	&	13.273 	&	7.5	&	9.5	\\
	&	7	&	14.434	&	15.744	&	15.775	&	15.926	&	16.098	&	16.220	&	16.282	&	16.378	&	19.007 	&	10.5	&	14	\\
	&	8	&	20.608	&	21.302	&	22.030	&	22.142	&	22.208	&	22.25	&	22.294	&	22.294	&	25.705 	&	14	&	19.5	\\
	
	si1032	&	4	&	2.740 	&	3.270 	&	3.572 	&	4.550 	&	4.914 	&	4.918 	&	4.924 	&	4.926 	&	4.987 	&	3	&	3.5	\\
	&	5	&	3.700 	&	5.003 	&	5.436 	&	8.054 	&	8.807 	&	8.814 	&	8.818 	&	8.831 	&	8.970 	&	5	&	6	\\
	1032	&	6	&	4.850 	&	7.446 	&	7.713 	&	12.403 	&	13.636 	&	13.683 	&	13.686 	&	13.696 	&	13.946 	&	7.5	&	9.5	\\
	&	7	&	6.106 	&	10.376 	&	10.674 	&	17.714 	&	19.437 	&	19.464 	&	19.521 	&	19.547 	&	19.915 	&	10.5	&	14	\\
	&	8	&	8.157 	&	13.602 	&	14.036 	&	23.467 	&	26.276 	&	26.313 	&	26.331 	&	26.363 	&	26.875 	&	14	&	19.5	\\
	
	{\tiny vm1084}	&	4	&	3.894 	&	3.912 	&	3.916 	&	3.964 	&	4.036 	&	4.072 	&	4.098 	&	4.112 	&	4.755 	&	3	&	3.5	\\
	&	5	&	6.307 	&	6.393 	&	6.443 	&	6.783 	&	6.879 	&	6.978 	&	6.994 	&	7.065 	&	8.505 	&	5	&	6	\\
	1084	&	6	&	9.841 	&	10.059 	&	10.200 	&	10.643 	&	10.753 	&	10.762 	&	10.874 	&	10.913 	&	13.243 	&	7.5	&	9.5	\\
	&	7	&	14.847 	&	14.983 	&	15.119 	&	15.266 	&	15.469 	&	15.590 	&	15.709 	&	15.717 	&	18.954 	&	10.5	&	14	\\
	&	8	&	20.338 	&	20.352 	&	20.777 	&	20.846 	&	21.170 	&	21.419 	&	21.652 	&	21.661 	&	25.640 	&	14	&	19.5	\\
	
	d1291	&	4	&	3.908 	&	4.002 	&	4.088 	&	4.124 	&	4.142 	&	4.162 	&	4.232 	&	4.238 	&	4.800 	&	3	&	3.5	\\
	&	5	&	6.604 	&	6.716 	&	7.046 	&	7.060 	&	7.151 	&	7.358 	&	7.373 	&	7.437 	&	8.619 	&	5	&	6	\\
	1291	&	6	&	10.350 	&	10.551 	&	10.992 	&	11.194 	&	11.345 	&	11.660 	&	11.669 	&	11.679 	&	13.438 	&	7.5	&	9.5	\\
	&	7	&	15.270 	&	15.393 	&	15.991 	&	16.378 	&	16.397 	&	16.536 	&	16.606 	&	16.722 	&	19.237 	&	10.5	&	14	\\
	&	8	&	20.649 	&	21.645 	&	21.693 	&	22.065 	&	22.605 	&	22.627 	&	22.647 	&	22.798 	&	26.018 	&	14	&	19.5	\\
	
	rl1323	&	4	&	4.056 	&	4.102	&	4.104 	&	4.154	&	4.162 	&	4.172	&	4.176 	&	4.198	&	4.800 	&	3	&	3.5	\\
	&	5	&	7.038 	&	7.043	&	7.182 	&	7.187	&	7.232 	&	7.251	&	7.252 	&	7.263	&	8.613 	&	5	&	6	\\
	1323	&	6	&	10.895	&	10.98	&	11.002	&	11.017	&	11.054	&	11.081	&	11.118	&	11.333	&	13.433 	&	7.5	&	9.5	\\
	&	7	&	15.872	&	15.97	&	16.042	&	16.18	&	16.218	&	16.342	&	16.344	&	16.390	&	19.243 	&	10.5	&	14	\\
	&	8	&	21.175	&	21.790	&	21.945	&	22.05	&	22.190	&	22.215	&	22.30	&	22.505	&	26.034 	&	14	&	19.5	\\
	
    u1432	&	4	&	4.056	&	4.07	&	4.192	&	4.196	&	4.232	&	4.234	&	4.24	&	4.252	&	4.758 	&	3	&	3.5	\\
    &	5	&	6.894	&	7.129	&	7.368	&	7.409	&	7.418	&	7.44	&	7.442	&	7.49	&	8.523 	&	5	 &	6	\\
    1432 &	6	&	10.925	&	11.559	&	11.632	&	11.769	&	11.775	&	11.794	&	11.798	&	11.813	&	13.282 	&	7.5	&	9.5	\\
    &	7	&	16.172	&	17.034	&	17.254	&	17.281	&	17.351	&	17.370	&	17.388	&	17.499	&	19.024 	&	10.5	&	14	\\
    &	8	&	23.025	&	23.35	&	23.465	&	23.635	&	23.805	&	23.820	&	23.855	&	23.905	&	25.745 	&	14	&	19.5	\\
	
	fl1577	&	4	&	3.774 	&	3.800 	&	3.928 	&	4.002 	&	4.148 	&	4.202 	&	4.296 	&	4.324 	&	4.832 	&	3	&	3.5	\\
	&	5	&	6.030 	&	6.058 	&	6.423 	&	6.554 	&	7.174 	&	7.192 	&	7.502 	&	7.549 	&	8.689 	&	5	&	6	\\
	1577	&	6	&	8.961 	&	9.159 	&	9.451 	&	10.243 	&	10.416 	&	11.217 	&	11.307 	&	11.493 	&	13.554 	&	7.5	&	9.5	\\
	&	7	&	12.788 	&	12.935 	&	13.364 	&	14.932 	&	15.003 	&	15.448 	&	15.537 	&	15.887 	&	19.400 	&	10.5	&	14	\\
	&	8	&	16.471 	&	17.715 	&	17.843 	&	19.295 	&	20.249 	&	20.292 	&	20.716 	&	21.186 	&	26.215 	&	14	&	19.5	\\
	
	d1655	&	4	&	3.626	&	3.788	&	3.826	&	4.312	&	4.336	&	4.34	&	4.342	&	4.344	&	4.799 	&	3	&	3.5	\\
	&	5	&	5.7	&	6.219	&	6.54	&	7.42	&	7.594	&	7.672	&	7.722	&	7.74	&	8.615 	&	5	&	6	\\
	1655	&	6	&	8.573	&	9.582	&	10.331	&	11.343	&	11.89	&	11.916	&	11.961	&	12.038	&	13.428 	&	7.5	&	9.5	\\
	&	7	&	12.241	&	13.967	&	15.412	&	15.918	&	16.85	&	17.012	&	17.111	&	17.405	&	19.225 	&	10.5	&	14	\\
	&	8	&	17.235	&	19.725	&	21.195	&	21.670	&	21.7	&	22.395	&	23.365	&	23.385	&	26.000 	&	14	&	19.5	\\
	
	{\tiny vm1748}	&	4	&	4.106 	&	4.126 	&	4.134 	&	4.144 	&	4.150 	&	4.156 	&	4.168 	&	4.222 	&	4.796 	&	3	&	3.5	\\
	&	5	&	7.010 	&	7.141 	&	7.171 	&	7.200 	&	7.200 	&	7.276 	&	7.344 	&	7.379 	&	8.603 	&	5	&	6	\\
	1748	&	6	&	11.111 	&	11.194 	&	11.407 	&	11.423 	&	11.512 	&	11.513 	&	11.534 	&	11.586 	&	13.400 	&	7.5	&	9.5	\\
	&	7	&	16.323 	&	16.403 	&	16.557 	&	16.560 	&	16.654 	&	16.813 	&	16.847 	&	16.856 	&	19.174 	&	10.5	&	14	\\
	&	8	&	21.969 	&	22.304 	&	22.341 	&	22.380 	&	22.429 	&	22.663 	&	22.847 	&	22.847 	&	25.917 	&	14	&	19.5	\\
	u1817	&	4	&	3.834 	&	4.070 	&	4.124 	&	4.138 	&	4.144 	&	4.158 	&	4.178 	&	4.192 	&	4.784 	&	3	&	3.5	\\
	&	5	&	5.783 	&	6.955 	&	7.021 	&	7.035 	&	7.094 	&	7.174 	&	7.223 	&	7.257 	&	8.582 	&	5	&	6	\\
	1817	&	6	&	9.326 	&	10.439 	&	11.138 	&	11.258 	&	11.271 	&	11.323 	&	11.383 	&	11.601 	&	13.395 	&	7.5	&	9.5	\\
	&	7	&	13.652 	&	14.561 	&	16.012 	&	16.579 	&	16.692 	&	16.703 	&	16.844 	&	16.848 	&	19.203 	&	10.5	&	14	\\
	&	8	&	19.230 	&	19.911 	&	21.950 	&	22.785 	&	23.105 	&	23.184 	&	23.204 	&	23.262 	&	26.003 	&	14	&	19.5	\\
	rl1889	&	4	&	3.820 	&	3.926 	&	4.010 	&	4.046 	&	4.100 	&	4.102 	&	4.128 	&	4.136 	&	4.821 	&	3	&	3.5	\\
	&	5	&	6.439 	&	6.729 	&	6.961 	&	7.121 	&	7.141 	&	7.165 	&	7.293 	&	7.314 	&	8.659 	&	5	&	6	\\
	1889	&	6	&	10.161 	&	10.643 	&	10.723 	&	11.017 	&	11.177 	&	11.387 	&	11.451 	&	11.528 	&	13.503 	&	7.5	&	9.5	\\
	&	7	&	15.224 	&	15.383 	&	15.598 	&	15.906 	&	16.557 	&	16.604 	&	16.784 	&	16.816 	&	19.342 	&	10.5	&	14	\\
	&	8	&	20.859 	&	21.491 	&	21.699 	&	22.191 	&	22.605 	&	22.803 	&	22.980 	&	23.013 	&	26.171 	&	14	&	19.5	\\

		%\bottomrule
		\hline
		\end{tabular}}
		\label{ohcfb}
	\end{center}
\end{table}

\begin{table}
	\begin{center}
		\caption{The eight smallest frequencies of the $OHC$ edges computed with frequency $K_i$s ($i\in[4,8]$) for big and large real-world $TSP$ instances.}
		{\footnotesize \begin{tabular}{ p{0.6cm}  p{0.1cm}  p{0.5cm}  p{0.5cm}  p{0.5cm}  p{0.5cm}  p{0.5cm}  p{0.5cm}  p{0.5cm}  p{0.5cm}  p{0.5cm}  p{0.3cm}  p{0.4cm} }
		\hline
	$TSP$ & $i$ &  & The &  &  smallest &  &  & freq. &  & $f_{avg}$ & $f_{lb}$ &  $f_{oavg}$  \\				
	$/n$ &  & $1^{st}$ & $2^{nd}$ & $3^{rd}$ & $4^{th}$ & $5^{th}$ & $6^{th}$ & $7^{th}$ & $8^{th}$ &  &  &  \\	
	\hline
   	d2103	&	4	&	4.384 	&	4.444 	&	4.458 	&	4.458 	&	4.458 	&	4.462 	&	4.468 	&	4.468 	&	4.819 	&	3	&	3.5	\\
   &	5	&	7.887 	&	7.920 	&	7.938 	&	8.018 	&	8.020 	&	8.032 	&	8.039 	&	8.052 	&	8.665 	&	5	&	6	\\
   2103	&	6	&	12.115 	&	12.642 	&	12.757 	&	12.797 	&	12.802 	&	12.804 	&	12.817 	&	12.872 	&	13.510 	&	7.5	&	9.5	\\
   &	7	&	16.744 	&	18.098 	&	18.312 	&	18.362 	&	18.565 	&	18.577 	&	18.613 	&	18.636 	&	19.339 	&	10.5	&	14	\\
   &	8	&	21.609 	&	24.511 	&	24.600 	&	24.706 	&	25.043 	&	25.058 	&	25.233 	&	25.251 	&	26.153 	&	14	&	19.5	\\
   
   u2152	&	4	&	4.072 	&	4.114 	&	4.142 	&	4.154 	&	4.160 	&	4.170 	&	4.184 	&	4.202 	&	4.788 	&	3	&	3.5	\\
   &	5	&	7.061 	&	7.137 	&	7.141 	&	7.158 	&	7.244 	&	7.245 	&	7.274 	&	7.316 	&	8.592 	&	5	&	6	\\
   2152	&	6	&	11.296 	&	11.299 	&	11.331 	&	11.351 	&	11.389 	&	11.392 	&	11.534 	&	11.566 	&	13.410 	&	7.5	&	9.5	\\
   &	7	&	16.479 	&	16.587 	&	16.597 	&	16.727 	&	16.746 	&	16.889 	&	16.958 	&	16.965 	&	19.233 	&	10.5	&	14	\\
   &	8	&	22.414 	&	22.908 	&	23.044 	&	23.093 	&	23.402 	&	23.445 	&	23.449 	&	23.533 	&	26.040 	&	14	&	19.5	\\
   
   u2319	&	4	&	4.264 	&	4.270 	&	4.342 	&	4.344 	&	4.350 	&	4.356 	&	4.356 	&	4.358 	&	4.760 	&	3	&	3.5	\\
   &	5	&	7.659 	&	7.722 	&	7.752 	&	7.755 	&	7.787 	&	7.788 	&	7.792 	&	7.792 	&	8.533 	&	5	&	6	\\
   2319	&	6	&	12.259 	&	12.274 	&	12.283 	&	12.381 	&	12.397 	&	12.426 	&	12.431 	&	12.435 	&	13.308 	&	7.5	&	9.5	\\
   &	7	&	18.055 	&	18.085 	&	18.101 	&	18.110 	&	18.121 	&	18.128 	&	18.174 	&	18.196 	&	19.069 	&	10.5	&	14	\\
   &	8	&	24.693 	&	24.805 	&	24.818 	&	24.835 	&	24.842 	&	24.868 	&	24.888 	&	24.893 	&	25.806 	&	14	&	19.5	\\
   
   pr2392	&	4	&	4.076 	&	4.134 	&	4.206 	&	4.212 	&	4.232 	&	4.260 	&	4.262 	&	4.286 	&	4.811 	&	3	&	3.5	\\
   &	5	&	7.116 	&	7.246 	&	7.259 	&	7.284 	&	7.350 	&	7.432 	&	7.442 	&	7.465 	&	8.642 	&	5	&	6	\\
   2392	&	6	&	11.504 	&	11.589 	&	11.625 	&	11.687 	&	11.701 	&	11.784 	&	11.804 	&	11.882 	&	13.477 	&	7.5	&	9.5	\\
   &	7	&	16.980 	&	17.101 	&	17.143 	&	17.149 	&	17.165 	&	17.175 	&	17.183 	&	17.200 	&	19.300 	&	10.5	&	14	\\
   &	8	&	22.820 	&	23.296 	&	23.305 	&	23.323 	&	23.331 	&	23.494 	&	23.635 	&	23.695 	&	26.105 	&	14	&	19.5	\\
   
   {\tiny pcb3038}	&	4	&	4.230 	&	4.250 	&	4.268 	&	4.280 	&	4.298 	&	4.306 	&	4.312 	&	4.314 	&	4.812 	&	3	&	3.5	\\
   &	5	&	7.375 	&	7.434 	&	7.451 	&	7.471 	&	7.541 	&	7.657 	&	7.664 	&	7.673 	&	8.645 	&	5	&	6	\\
   3038	&	6	&	11.836 	&	11.926 	&	12.019 	&	12.021 	&	12.112 	&	12.131 	&	12.152 	&	12.161 	&	13.482 	&	7.5	&	9.5	\\
   &	7	&	17.533 	&	17.613 	&	17.692 	&	17.729 	&	17.789 	&	17.800 	&	17.812 	&	17.853 	&	19.303 	&	10.5	&	14	\\
   &	8	&	24.048 	&	24.123 	&	24.145 	&	24.212 	&	24.319 	&	24.421 	&	24.440 	&	24.461 	&	26.110 	&	14	&	19.5	\\
   
   fl3795	&	4	&	3.612 	&	3.746 	&	3.938 	&	4.176 	&	4.220 	&	4.240 	&	4.316 	&	4.462 	&	4.878 	&	3	&	3.5	\\
   &	5	&	5.207 	&	6.073 	&	6.727 	&	7.277 	&	7.331 	&	7.364 	&	7.407 	&	7.580 	&	8.782 	&	5	&	6	\\
   3795	&	6	&	7.661 	&	9.307 	&	10.298 	&	10.913 	&	11.391 	&	11.436 	&	11.676 	&	11.772 	&	13.685 	&	7.5	&	9.5	\\
   &	7	&	9.810 	&	13.225 	&	15.086 	&	15.356 	&	16.447 	&	16.632 	&	16.721 	&	16.868 	&	19.561 	&	10.5	&	14	\\
   &	8	&	12.203 	&	17.621 	&	20.484 	&	20.634 	&	22.287 	&	22.381 	&	22.677 	&	22.887 	&	26.406 	&	14	&	19.5	\\
   
   {\tiny fnl4461}	&	4	&	4.318 	&	4.340 	&	4.350 	&	4.354 	&	4.370 	&	4.374 	&	4.388 	&	4.394 	&	4.806 	&	3	&	3.5	\\
   &	5	&	7.624 	&	7.729 	&	7.742 	&	7.768 	&	7.792 	&	7.795 	&	7.800 	&	7.802 	&	8.635 	&	5	&	6	\\
   4461	&	6	&	12.278 	&	12.338 	&	12.343 	&	12.360 	&	12.363 	&	12.395 	&	12.404 	&	12.409 	&	13.478 	&	7.5	&	9.5	\\
   &	7	&	17.772 	&	17.772 	&	17.839 	&	18.035 	&	18.179 	&	18.184 	&	18.196 	&	18.215 	&	19.314 	&	10.5	&	14	\\
   &	8	&	24.247 	&	24.620 	&	24.675 	&	24.695 	&	24.812 	&	24.822 	&	24.825 	&	24.828 	&	26.145 	&	14	&	19.5	\\
   
   rl5915	&	4	&	4.296 	&	4.342 	&	4.346 	&	4.356 	&	4.368 	&	4.372 	&	4.376 	&	4.380 	&	4.861 	&	3	&	3.5	\\
   &	5	&	7.498 	&	7.643 	&	7.693 	&	7.761 	&	7.783 	&	7.792 	&	7.797 	&	7.802 	&	8.750 	&	5	&	6	\\
   5915	&	6	&	12.182 	&	12.249 	&	12.275 	&	12.346 	&	12.386 	&	12.438 	&	12.472 	&	12.474 	&	13.650 	&	7.5	&	9.5	\\
   &	7	&	17.643 	&	17.754 	&	17.880 	&	17.890 	&	17.898 	&	17.984 	&	18.063 	&	18.125 	&	19.546 	&	10.5	&	14	\\
   &	8	&	23.229 	&	23.841 	&	24.131 	&	24.221 	&	24.316 	&	24.595 	&	24.676 	&	24.678 	&	26.437 	&	14	&	19.5	\\
   
   rl5934	&	4	&	4.298 	&	4.330 	&	4.352 	&	4.356 	&	4.358 	&	4.364 	&	4.368 	&	4.378 	&	4.858 	&	3	&	3.5	\\
   &	5	&	7.394 	&	7.672 	&	7.704 	&	7.743 	&	7.744 	&	7.745 	&	7.772 	&	7.785 	&	8.743 	&	5	&	6	\\
   5934	&	6	&	11.940 	&	12.188 	&	12.197 	&	12.367 	&	12.376 	&	12.381 	&	12.412 	&	12.418 	&	13.644 	&	7.5	&	9.5	\\
   &	7	&	17.514 	&	17.524 	&	17.802 	&	17.834 	&	18.034 	&	18.046 	&	18.088 	&	18.114 	&	19.544 	&	10.5	&	14	\\
   &	8	&	23.563 	&	24.153 	&	24.347 	&	24.400 	&	24.529 	&	24.533 	&	24.623 	&	24.656 	&	26.440 	&	14	&	19.5	\\
   
   {\tiny xsc6880}	&	4	&	4.228	&	4.304	&	4.332	&	4.346	&	4.358	&	4.362	&	4.382	&	4.384	&	4.858 	&	3	&	3.5	\\
   &	5	&	7.44	&	7.443	&	7.685	&	7.717	&	7.748	&	7.748	&	7.752	&	7.771	&	8.742 	&	5	&	6	\\
   6880	&	6	&	11.714	&	11.822	&	12.234	&	12.318	&	12.319	&	12.338	&	12.342	&	12.345	&	13.638 	&	7.5	&	9.5	\\
   &	7	&	17.486	&	17.572	&	17.780	&	17.798	&	17.848	&	17.888	&	17.976	&	18.1	&	19.530 	&	10.5	&	14	\\
   &	8	&	23.99	&	24.052	&	24.204	&	24.398	&	24.470	&	24.552	&	24.582	&	24.636	&	26.413 	&	14	&	19.5	\\
   
   {\tiny pla7397}	&	4	&	4.138 	&	4.164 	&	4.174 	&	4.178 	&	4.180 	&	4.186 	&	4.212 	&	4.232 	&	4.859 	&	3	&	3.5	\\
   &	5	&	7.128 	&	7.129 	&	7.237 	&	7.267 	&	7.322 	&	7.328 	&	7.349 	&	7.364 	&	8.738 	&	5	&	6	\\
   7397	&	6	&	11.191 	&	11.278 	&	11.481 	&	11.611 	&	11.662 	&	11.701 	&	11.725 	&	11.731 	&	13.592 	&	7.5	&	9.5	\\
   &	7	&	16.369 	&	16.769 	&	16.890 	&	16.993 	&	17.090 	&	17.118 	&	17.263 	&	17.291 	&	19.414 	&	10.5	&	14	\\
   &	8	&	22.769 	&	23.287 	&	23.428 	&	23.496 	&	23.595 	&	23.610 	&	23.711 	&	23.813 	&	26.210 	&	14	&	19.5	\\
				
		%\bottomrule
		\hline
		\end{tabular}}
		\label{ohcfbl}
	\end{center}
\end{table}

\subsubsection{The preserved ordinary edges according to the smallest frequencies of $OHC$ edges for the real-world $TSP$ instances}
In this section, we shall investigate the number of the preserved ordinary edges according to the smallest frequencies of the $OHC$ edges. Several smallest frequencies of the $OHC$ edges are taken as the thresholds to filter out the ordinary edges with the smaller frequencies, and the number of the preserved ordinary edges is recorded for each $TSP$ instance, respectively. As $i$ = 4, the average frequency of each edge is computed with all frequency $K_4$s containing them, respectively. The computation time is $O(n^4)$. Thus, the experiments are executed for some small and medium $TSP$ instances. The experimental results are given in Table \ref{nopoes}. Except for the column of $TSP$ name, each of the other columns illustrates the number of the preserved ordinary edges according to the $k^{th}$ smallest frequency of the $OHC$ edges. For example in the $5^{th}$ column, it gives the number of the preserved ordinary edges if the $5^{th}$ smallest frequency of the $OHC$ edges is taken as the frequency threshold to eliminate the edges with the smaller frequency. Five smallest frequencies, i.e., the $1^{st}$, $5^{th}$, $10^{th}$, $15^{th}$ and $20^{th}$ smallest frequency of the $OHC$ edges, are tried and the number of the preserved ordinary edges is computed in Table \ref{nopoes}, respectively. 
\begin{table}
	\begin{center}
		\caption{The number of preserved ordinary edges according to the $k^{th}$ smallest frequency of $OHC$ edges computed with frequency $K_4$s for the real-world $TSP$ instances.}
		{\footnotesize \begin{tabular}{ p{0.7cm}  p{0.5cm}  p{0.5cm}  p{0.5cm}  p{0.5cm}  p{0.5cm}  p{0.6cm}  p{0.6cm}  p{0.6cm}  p{0.6cm}  p{0.6cm}  p{0.6cm} }
				\hline
		$TSP$  &  The & $k^{th}$ & smallest &  & freq. & $TSP$ & The & $k^{th}$ & smallest &  & freq. \\		
		&  $1^{st}$ & $5^{th}$ & $10^{th}$ & $15^{th}$ & $20^{th}$ & $TSP$ & $1^{st}$ & $5^{th}$ & $10^{th}$ & $15^{th}$ & $20^{th}$ \\	
				\hline
		gr21	&	39	&	8	&		&		&		&	si175	&	1714	&	774	&	535	&	239	&	217	\\
		gr24	&	78	&	27	&		&		&		& 	tsp225	&	3298	&	2644	&	2228	&	1970	&	1856	\\
		bay29	&	90	&	37	&		&		&		&	a280	&	7641	&	3312	&	2969	&	2865	&	2746	\\
		bays29	&	102	&	41	&		&		&		&	rd400	&	11808	&	8368	&	6991	&	6480	&	5920	\\
		att48	&	243	&	151	&	121	&	87	&	49	&	gr431	&	17099	&	13331	&	10179	&	8162	&	7481	\\
		eil51	&	193	&	182	&	121	&	88	&	50	&	pcb442	&	14154	&	7644	&	6559	&	6064	&	5818	\\
		berlin52	&	195	&	137	&	104	&	62	&	50	&	att532	&	15761	&	13129	&	12309	&	11254	&	10397	\\
		st70	&	737	&	223	&	164	&	136	&	109	&	si535	&	28646	&	14182	&	6440	&	5031	&	4204	\\
		pr76	&	728	&	320	&	197	&	156	&	143	&	u574	&	19173	&	17605	&	14930	&	13948	&	12385	\\
		rat99	&	1318	&	877	&	597	&	462	&	348	&	d657	&	31601	&	19772	&	15324	&	14067	&	13462	\\
		rd100	&	897	&	667	&	551	&	326	&	244	&	gr666	&	18919	&	12962	&	9832	&	8427	&	7790	\\
		gr120	&	2063	&	1021	&	778	&	676	&	466	&	rat783	&	32185	&	29680	&	26021	&	23739	&	23080	\\
		gr137	&	1879	&	1244	&	1008	&	947	&	882	&	si1032	&	205509	&	2400	&	1946	&	1108	&	550	\\
		ch150	&	2309	&	1143	&	1051	&	967	&	776	&	u1060	&	68005	&	52238	&	48182	&	43069	&	41244	\\
		gr202	&	2008	&	1853	&	1565	&	1404	&	1258	&	rl1302	&	123823	&	71158	&	57125	&	48732	&	43871	\\
		ts225	&	2413	&	1687	&	1572	&	1483	&	1263	&	u1432	&	111651	&	73407	&	63901	&	52964	&	51309	\\		
				%\bottomrule
				\hline
		\end{tabular}}
		\label{nopoes}
	\end{center}
\end{table}

As the ordinary edges with the smaller frequencies are eliminated according to the $k^{th}$ smallest frequency of the $OHC$ edges, the $k$ $OHC$ edges are also neglected. For the first four small $TSP$ instances where $n\leq 29$, the edges are eliminated according to the first and fifth smallest frequency of the $OHC$ edges, respectively. As the first smallest frequency of the $OHC$ edges is taken as the frequency threshold, approximate 20\%$\sim$30\% ordinary edges are preserved for most small $TSP$ instances where $n\leq 150$. As $n$ becomes bigger than 200, the percentage of the preserved ordinary edges becomes smaller, i.e., $10\%\sim15\%$ ordinary edges are preserved for most of them. Moreover, the percentage of the preserved ordinary edges decreases according to $n$. It says that the smallest frequency of the $OHC$ edges is generally rising according to $n$. Thus, more percentage of the ordinary edges have the frequency smaller than the smallest frequency of the $OHC$ edges. However, the number of the preserved ordinary edges is still big. It indicates that the first smallest frequency of the $OHC$ edges is not big enough, and there are many ordinary edges having the bigger frequency computed with the frequency $K_4$s. 

As the $5^{th}$ smallest frequency is taken as the frequency threshold, the percentage of the preserved ordinary edges decreases quickly for most $TSP$ instances. The preserved ordinary edges generally occupy smaller than $15\%$ of the total ordinary edges for the small $TSP$ instances, and smaller than $10\%$ of the total ordinary edges for the medium $TSP$ instances. For some $TSP$ instances, the first smallest frequency of the $OHC$ edges is close to $f_{lb}$, such as gr21, gr24, bay29, bays29, att48, st70, si175, a280, pcb442, si535 and  si1032. For these $TSP$ instances, the percentage of the preserved ordinary edges has one big drop according to the fifth smallest frequency of the $OHC$ edges. It indicates that the frequencies of the $OHC$ edges increases quickly from the first to the fifth smallest frequency. Thus, much more percentage of the ordinary edges are eliminated according to the fifth smallest frequency. In total, the number of the preserved ordinary edges decreases quickly according to the $5^{th}$, $10^{th}$, $15^{th}$ and $20^{th}$ smallest frequency of the $OHC$ edges. It indicates that the frequencies of the $OHC$ edges increase much faster than those of the ordinary edges computed with the frequency $K_4$s.  Although a few $OHC$ edges are neglected, much more ordinary edges will be eliminated according to these frequency thresholds. 

As the frequency $K_i$s containing more vertices are used, the better results will be obtained. The experimental results for five small real-world $TSP$ instances are illustrated in Table \ref{npoes2}. In the table, the percent of the preserved ordinary edges are computed based on the $k^{th}$ smallest probability of the $OHC$ edges where $k\in[1,10]$. For example in the $2^{nd}$ column, it illustrates the percent of the preserved ordinary edges according to the second smallest probability of the $OHC$ edges when the probability of each edge is computed with the frequency $K_i$s for $i\in[4,8]$. Firstly, more than 70\% of the ordinary edges are filtered out according to the first smallest probability of the $OHC$ edges. It indicates that the probabilities of most ordinary edges are smaller than those of the $OHC$ edges based on the frequency $K_i$s. 

Secondly, the percent of the preserved ordinary edges becomes smaller according to $i\in[4,8]$ based on the $k^{th}$ smallest probability, such as according to the $1^{st}$ smallest probability of the $OHC$ edges. Based on each $k^{th}$ smallest probability of the $OHC$ edges, the percent of the preserved ordinary edges decreases by 6\%$\sim$10\% from $i=4$ to 8 for these $TSP$ instances. It indicates that the probabilities of the $OHC$ edges increase faster than those of most ordinary edges according to $i$. Based on Theorems \ref{th33} and \ref{th3}, the probability of an ordinary edge generally decreases according to $i$, and that of an $OHC$ edge generally rises as well. Thus, more percentage of ordinary edges will have the probabilities smaller than those of the $OHC$ edges according to $i$. Thus, the number of the preserved ordinary edges decreases according to $i$. 

Thirdly, the percent of the preserved ordinary edges decreases quickly according to $k$ from 1 to 10. Based on the ten smallest probabilities from $k=1$ to 10 for each $i$, the percent of the preserved ordinary edges decreases by 7\% $\sim$ 20\% in the experiments. In general, the percent of the preserved ordinary edges according to the first smallest probability of the $OHC$ edges is two or three times of that according to the $10^{th}$ smallest probability of the $OHC$ edges. It means that the first smallest probability of the $OHC$ edges is  relatively small comparing to those of many ordinary edges. Based on the $k^{th}$ smallest probability of the $OHC$ edges where $k>1$, much more number of ordinary edges are filtered out although only $k\in[2,10]$ $OHC$ edges are neglected. It also indicates that the probability of the $OHC$ edges increases quickly from the smallest value to the big values according to $k$. The experiments again illustrates %Since the smallest probability of the $OHC$ edges is generally close to the lower probability bound $\frac{1}{2}$, it means 
that the number of the $OHC$ edges with the probability close to $\frac{1}{2}$ is very small, even if for the small real-world $TSP$ instances. Most $OHC$ edges have the probability much bigger than $\frac{1}{2}$, especially for the medium and big $TSP$ instances. In this case, the percent of the preserved ordinary edges decreases quickly according to the $k^{th}$ ($k > 1$) smallest probability of the $OHC$ edges as $k$ is relatively big. 

\begin{table}
	\begin{center}
		\caption{The percent (\%) of preserved ordinary edges according to the $k^{th}$ smallest probability of $OHC$ edges based on frequency $K_i$s ($i\in[4,8]$) for five $TSP$ instances.}
		{\footnotesize \begin{tabular}{ p{1cm}  p{0.2cm}  p{0.5cm}  p{0.5cm}  p{0.5cm}  p{0.5cm}  p{0.5cm}  p{0.5cm}  p{0.5cm}  p{0.5cm}  p{0.5cm}  p{0.5cm} }
			\hline
		$TSP$  &  $i$ &  & the &  & $k^{th}$ &   & smallest &  & prob. &  &   \\
		&   & $1^{st}$ & $2^{nd}$ & $3^{rd}$ & $4^{th}$ & $5^{th}$ & $6^{th}$ & $7^{th}$ & $8^{th}$ & $9^{th}$ & $10^{th}$  \\	
			\hline
		kroA100	&	4	&	25.65 	&	23.92 	&	23.57 	&	18.43 	&	17.34 	&	16.66 	&	14.52 	&	13.79 	&	13.67 	&	12.78 	\\
		&	5	&	28.82 	&	24.89 	&	23.44 	&	20.08 	&	17.09 	&	14.49 	&	14.39 	&	14.27 	&	12.62 	&	12.52 	\\
		&	6	&	24.82 	&	23.84 	&	17.18 	&	15.48 	&	14.10 	&	11.88 	&	11.24 	&	9.51 	&	9.30 	&	8.80 	\\
		&	7	&	21.65 	&	19.86 	&	16.45 	&	12.76 	&	10.60 	&	10.37 	&	8.39 	&	8.25 	&	6.64 	&	5.92 	\\
		&	8	&	17.77 	&	17.75 	&	13.20 	&	12.82 	&	8.60 	&	7.90 	&	7.03 	&	6.35 	&	5.20 	&	5.05 	\\
		kroB100	&	4	&	22.08 	&	17.90 	&	14.72 	&	14.06 	&	14.02 	&	13.38 	&	13.11 	&	12.31 	&	11.73 	&	11.26 	\\
		&	5	&	18.80 	&	17.01 	&	13.53 	&	13.38 	&	12.85 	&	12.41 	&	12.41 	&	12.00 	&	11.94 	&	11.90 	\\
		&	6	&	18.12 	&	16.27 	&	10.60 	&	10.14 	&	9.86 	&	9.86 	&	9.44 	&	9.38 	&	9.09 	&	8.82 	\\
		&	7	&	15.67 	&	12.25 	&	8.54 	&	8.23 	&	7.51 	&	7.34 	&	6.97 	&	6.95 	&	6.74 	&	6.74 	\\
		&	8	&	13.98 	&	9.55 	&	8.19 	&	7.75 	&	7.46 	&	6.08 	&	6.10 	&	5.98 	&	5.59 	&	5.38 	\\
		kroC100	&	4	&	26.10 	&	25.11 	&	22.87 	&	18.62 	&	17.94 	&	13.86 	&	12.85 	&	12.43 	&	11.61 	&	10.52 	\\
		&	5	&	26.39 	&	25.20 	&	22.27 	&	17.75 	&	17.59 	&	13.79 	&	12.31 	&	12.23 	&	10.85 	&	9.86 	\\
		&	6	&	24.95 	&	23.96 	&	19.18 	&	14.72 	&	12.54 	&	10.39 	&	8.43 	&	7.98 	&	7.40 	&	7.38 	\\
		&	7	&	22.04 	&	18.99 	&	14.58 	&	10.95 	&	9.53 	&	9.36 	&	7.36 	&	6.37 	&	6.29 	&	5.71 	\\
		&	8	&	20.45 	&	16.47 	&	11.55 	&	9.84 	&	8.27 	&	6.93 	&	6.49 	&	6.41 	&	5.65 	&	4.93 	\\
		kroD100	&	4	&	29.59 	&	25.77 	&	23.32 	&	18.76 	&	18.43 	&	18.33 	&	11.59 	&	11.03 	&	10.89 	&	10.70 	\\
		&	5	&	27.73 	&	24.04 	&	21.26 	&	20.47 	&	16.85 	&	16.41 	&	11.96 	&	11.51 	&	10.87 	&	9.92 	\\
		&	6	&	27.86 	&	19.48 	&	17.20 	&	16.85 	&	13.38 	&	13.18 	&	12.02 	&	8.39 	&	8.25 	&	8.02 	\\
		&	7	&	21.57 	&	14.82 	&	14.45 	&	13.98 	&	12.56 	&	10.52 	&	9.77 	&	7.03 	&	6.29 	&	6.19 	\\
		&	8	&	18.60 	&	12.19 	&	12.16 	&	10.78 	&	10.10 	&	9.26 	&	8.06 	&	6.95 	&	6.80 	&	5.98 	\\
		rd100	&	4	&	19.05 	&	16.06 	&	15.26 	&	15.26 	&	13.34 	&	12.95 	&	12.91 	&	12.29 	&	11.98 	&	11.94 	\\
		&	5	&	18.21 	&	16.91 	&	13.38 	&	12.58 	&	11.71 	&	11.34 	&	11.36 	&	10.19 	&	9.98 	&	9.90 	\\
		&	6	&	16.23 	&	12.60 	&	11.59 	&	10.82 	&	10.02 	&	9.69 	&	9.32 	&	7.75 	&	7.73 	&	7.53 	\\
		&	7	&	13.84 	&	9.79 	&	9.53 	&	9.13 	&	8.33 	&	7.71 	&	7.34 	&	6.49 	&	5.73 	&	4.91 	\\
		&	8	&	12.74 	&	8.25 	&	6.62 	&	6.58 	&	6.41 	&	6.25 	&	6.16 	&	5.36 	&	5.07 	&	4.27 	\\
		%\bottomrule
			\hline
		\end{tabular}}
		\label{npoes2}
	\end{center}
\end{table}

\subsubsection{The probability decrements for the $OHC$ edges according to $i$ for the real-world $TSP$ instances}
Based on the experimental results in Tables \ref{ohcfs}$\sim$\ref{ohcfbl}, it is known that the $OHC$ edges are contained in more and more percentage of $OP^i$s according to $n$ for given $i$. Moreover, they are contained in more percentage of $OP^i$s according to $i$ for given $n$. Thus, the average frequencies and probabilities of most $OHC$ edges increase according to $n$ and $i$, respectively. On the other hand, the average frequencies and probabilities of most ordinary edges decrease according to $n$ and $i$, respectively. In theory, the probability of an ordinary edge decreases according to $i > i_d$ computed based on formula (\ref{F5}). The experimental results totally conform to Theorems \ref{th3} and \ref{th33}. In addition, if the probability $p_i(e)$ of an $OHC$ edge decreases from $i$ to $i+1$, the decrement $pd_i(e) < \frac{2p_i(e)}{i(i-1)}$ holds. Since $p_i(e)\leq 1$, $pd_i(e)$ will decrease quickly according to $i$. For most ordinary edges, the probability decrement $pd_i(g) > \frac{2p_i(g)}{i(i-1)}$ happens even if $i$ is small.  As $p_i(g) < 1$ decreases, the $pd_i(g)$ will becomes bigger according to $i$. Thus, the condition $pd_i(g) < \frac{2p_i(g)}{i(i-1)}$ becomes more strict (to filter out the ordinary edges) according to $i$. In applications, the condition $pd_i(g)  > \frac{2p_i(g)}{i(i-1)}$ will be better than the lower probability bound $\frac{1}{2}$ to identify more ordinary edges. We first investigate this condition for the $OHC$ edges with the $TSP$ instances.  %, and the results are illustrated in Table \ref{proberr}

\begin{table}
	\begin{center}
		\caption{The biggest error $err= max\left\{pd_i(e)-\frac{2p_i(e)}{i(i-1)}\right\}$ from $i$ to $i+1$ with respect to all $OHC$ edges for the real-world $TSP$ instances where $pd_i(e) = p_i(e)-p_{i+1}(e)$.}
		{\footnotesize \begin{tabular}{ p{0.7cm}  p{0.8cm}  p{0.8cm}  p{0.8cm}  p{0.8cm}  p{0.7cm}  p{0.8cm}  p{0.8cm}  p{0.8cm}  p{0.8cm}  }
				\hline
		$TSP$  & $i$  & $\to$ & $i+1$ &  & $TSP$ & $i$  & $\to$ & $i+1$ &  \\
		&  $4\to 5$ & $5\to 6$ & $6\to 7$ & $7\to 8$ &  & $4\to 5$ & $5\to 6$ & $6\to 7$ & $7\to 8$ \\	
			\hline
		att48	&	-0.017 	&	-0.005 	&	-0.006 	&	\textbf{0.005}/1 	&	u574	&	-0.093 	&	-0.075 	&	-0.042 	&	-0.024 	\\
		eil51	&	-0.069 	&	-0.040 	&	-0.015 	&	-0.009 	&	rt575	&	-0.094 	&	-0.075 	&	-0.050 	&	-0.026 	\\
		berlin52	&	-0.073 	&	-0.055 	&	-0.030 	&	-0.016 	&	p654	&	-0.026 	&	-0.002 	&	\textbf{0.001}/1 	&	\textbf{0.011}/1 	\\
		st70	&	-0.034 	&	-0.019 	&	-0.025 	&	-0.012 	&	d657	&	-0.098 	&	-0.079 	&	-0.049 	&	-0.033 	\\
		pr76	&	-0.027 	&	-0.013 	&	\textbf{0.007}/1 	&	\textbf{0.010}/1 	&	gr666	&	-0.106 	&	-0.086 	&	-0.049 	&	-0.030 	\\
		rat99	&	-0.012 	&	-0.043 	&	-0.028 	&	-0.011 	&	u724	&	-0.111 	&	-0.081 	&	-0.045 	&	-0.028 	\\
		kroA100	&	-0.021 	&	-0.013 	&	-0.004 	&	\textbf{0.001}/1 	&	rat783	&	-0.106 	&	-0.080 	&	-0.052 	&	-0.017 	\\
		kroB100	&	-0.051 	&	-0.039 	&	-0.019 	&	-0.007 	&	si1032	&	\text{0.011}/1 	&	\textbf{0.010}/1 	&	\textbf{0.011}/1 	&	-0.002 	\\
		rd100	&	-0.077 	&	-0.050 	&	-0.022 	&	-0.013 	&	vm1084	&	-0.090 	&	-0.076 	&	-0.028 	&	-0.014 	\\
		lin105	&	\textbf{0.001}/1 	&	\textbf{0.012}/1 	&	-0.026 	&	\textbf{0.002}/1 	&	d1291	&	-0.105 	&	-0.077 	&	-0.052 	&	-0.024 	\\
		gr120	&	-0.017 	&	-0.035 	&	-0.024 	&	-0.004 	&	fl1577	&	-0.075 	&	-0.049 	&	-0.025 	&	-0.008 	\\
		bier127	&	-0.094 &	-0.036 	&	-0.032 	&	-0.015 	&	vm1748	&	-0.131 	&	-0.074 	&	-0.050 	&	-0.022 	\\
		ch130	&	-0.075 	&	-0.047 	&	-0.036 	&	-0.014 	&	u1817	&	-0.046 	&	-0.070 	&	-0.044 	&	-0.032 	\\
		gr137	&	-0.042 	&	-0.028 	&	\textbf{0.004}/1 	&	\textbf{0.010}/1 	&	rl1889	&	-0.113 	&	-0.081 	&	-0.051 	&	-0.032 	\\
		kroA200	&	-0.036 	&	-0.060 	&	-0.037 	&	-0.019 	&	d2103	&	-0.160 	&	-0.085 	&	-0.044 	&	-0.012 	\\
		gr202	&	-0.102 	&	-0.072 	&	-0.036 	&	-0.022 	&	u2152	&	-0.135 	&	-0.085 	&	-0.048 	&	-0.036 	\\
		a280	&	-0.069 	&	-0.068 	&	-0.046 	&	-0.029 	&	u2319	&	-0.153 	&	-0.099 	&	-0.063 	&	-0.037 	\\
		lin318	&	-0.077 	&	-0.058 	&	-0.040 	&	-0.014 	&	pr2392	&	-0.133 	&	-0.094 	&	-0.060 	&	-0.034 	\\
		rd400	&	-0.102 	&	-0.065 	&	-0.038 	&	-0.015 	&	pcb3038	&	-0.142 	&	-0.093 	&	-0.057 	&	-0.034 	\\
		fl417	&	-0.023 	&	-0.008 	&	-0.002 	&	\textbf{0.030}/2 	&	fl3795	&	-0.019 	&	-0.042 	&	\textbf{0.010}/1 	&	\textbf{0.009}/1 	\\
		pr439	&	-0.114 	&	-0.048 	&	-0.031 	&	-0.016 	&	fnl4461	&	-0.155 	&	-0.096 	&	-0.061 	&	-0.038 	\\
		pcb442	&	-0.094 	&	-0.068 	&	-0.048 	&	-0.018 	&	rl5915	&	-0.143 	&	-0.094 	&	-0.057 	&	-0.029 	\\
		att532	&	-0.095 	&	-0.062 	&	\textbf{0.009}/1 	&	-0.024 	&	rl5934	&	-0.142 	&	-0.097 	&	-0.055 	&	-0.037 	\\
		si535	&	-0.042 	&	-0.035 	&	\textbf{0.019}/1 	&	\textbf{0.002}/1 	&	pla7397	&	-0.132 	&	-0.097 	&	-0.057 	&	-0.035 	\\
				%\bottomrule
				\hline
		\end{tabular}}
		\label{proberr}
	\end{center}
\end{table}

Firstly, the error between the probability decrement $pd_i(e)=p_i(e) - p_{i+1}(e)$ and $\frac{2p_i(e)}{i(i-1)}$ from $i$ to $i+1$ is computed for all $OHC$ edges for a given $TSP$ instance, and the biggest error denoted by $err$ is recorded for each $i$, respectively, where $i\in [4,7]$. In other words, from $i$ to $i+1$ for given $i$, $pd_i(e) = p_i(e)-p_{i+1}(e)$ is computed for every edge $e\in OHC$ for a given $TSP$ instance, and  $err = max\left\{pd_i(e) - \frac{2p_i(e)}{i(i+1)}\right\}$ is computed considering all $OHC$ edges. If $p_i(e) \leq p_{i+1}(e)$ exists for certain $e$, it means that $p_i(e)$ increases from $i$ to $i+1$, and  $pd_i(e) < 0$ must hold. If $p_i(e) > p_{i+1}(e)$ occurs for $e$, it indicates that $p_i(e)$ decreases from $i$ to $i+1$. For an edge $e\in OHC$,  $pd_i(e) \leq \frac{2p_i(e)}{i(i-1)}$ exists. It implies that $err\leq 0$ exists for an $OHC$ edge from $i$ to $i+1$. The experimental results are illustrated in Table \ref{proberr}.  Due to the selection of random frequency $K_i$s for each $e$ to compute $p_i(e)$ , some $OHC$ edges may have an $err > 0$. The positive $err$ is denoted by the boldface text, and the number of such $OHC$ edges is given after the positive $err$ for each $TSP$ instance and  the specific $i$. One can compare such number with the $TSP$ scale $n$. 

In view of the value of $err$ for each $TSP$ instance, one sees that $err < 0$ appears for nearly all the $OHC$ edges from $i$ to $i+1$. It says that Theorem \ref{th3} works well for the $OHC$ edges in different types and scales of $TSP$ instances. For one or two $OHC$ edges  related to some $TSP$ instances from certain $i$ to $i+1$, $err$ is bigger than zero, such as att48, fl417 and p654 from $i=$7 to 8, pr76, gr137, si535, p654, and fl3795 from $i=6$ to 7 and from $i=$7 to 8. However, the number of such $OHC$ edges is only 1 or at most 2 (only for fl417 from $i=$7 to 8) for each of these $TSP$ instance. The number of the $OHC$ edges with $err > 0$ is too small comparing with the scale $n$ of $TSP$. It indicates that the probability decrement for nearly all $OHC$ edges for different types of $TSP$ instances conforms to Theorem \ref{th3}. 

It mentions that some $TSP$ instances, such as pr76, fl417, si535, p654, si1032, fl3795, etc., contain many equal-weight edges. There will be more than ${{i}\choose{2}}$ $OP^i$s in many $K_i$s in the $K_n$s related to these $TSP$ instances. As the $OP^i$s are used to compute the probability of each edge, it is difficult to choose the right $OP^i$s for computing the high probability for some $OHC$ edges. In this case, the probability of these $OHC$ edges will be lowered if many  $K_i$s containing these $OP^i$s are selected and used. Thus, $err > 0$ will appear for these $OHC$ edges. 

Although the value of $err$ is smaller than zero for the $OHC$ edges, it approaches zero according to $i$. It is known that  $\frac{2p_i(e)}{i(i+1)}$ decreases according to $i$ because $p_i(e) \leq 1$. It indicates that $pd_i(e)=  \vert p_{i+1}(e) - p_i(e) \rvert < \frac{2p_i(e)}{i(i+1)}$ also becomes smaller according to $i$ whether $p_i(e)$ increases or decreases. Otherwise, the $err$ will deviate from zero according to $i$. Thus, $p_i(e)$ will have certain bigger decrement or increment from $i$ to $i+1$ as $i$ is small. On the other hand, it will have the relatively smaller increment or decrement as $i$ is big. The value of  $\frac{2p_i(e)}{i(i-1)}$ gives the restrict condition to $pd_i(e)$ for an $OHC$ edge from $i$ to $i+1$. As $i$ is big, $\frac{2p_i(e)}{i(i-1)}$ will tend to zero. Since $p_i(e) - p_{i+1}(e) \leq \frac{2p_i(e)}{i(i+1)}$ holds, $p_i(e) \leq p_{i+1}(e)$ is derived. It says that $p_i(e)$ will increase for $OHC$ edges as $i$ is big. The experimental results also approved the condition for the $OHC$ edges. 

Comparing with the $err$s for the small and big $TSP$ instances, it found that $err$ becomes smaller according to $n$ for given $i\to i+1$. It means that $p_i(e)$ increases according to $n$ for the $OHC$ edges. Because $n$ is not very small for these $TSP$ instances ($n\geq 48$ for the $TSP$ instances in Table \ref{proberr}), $p_i(e)$ and $p_{i+1}(e)$ do not have much difference for an $OHC$ edge from $i$ to $i+1$. Thus, the value of $err$ is mainly constrained by $\frac{2p_i(e)}{i(i-1)}$ rather than $p_i(e)-p_{i+1}(e)$. As $n$ rises, $p_i(e)$ will become bigger accordingly for general $TSP$ instances. In addition, $p_i(e)$ for $OHC$ edges usually increases according to $i$ for the big $TSP$ instances. Based on the formula $err=max\left\{p_i(e)-p_{i+1}(e)-\frac{2p_i(e)}{i(i-1)}\right\}$, the big scale of $TSP$ will have the smaller value of $err < 0$. 
%For each $TSP$ instance, the biggest value of $err$ from $i$ to $i+1$ is computed for the $OHC$ edges

%To prove this condition $p_i(e) - p_{i+1}(e) < \frac{2p_i(e)}{i(i-1){{n-2}\choose{i-2}}}$ rather than $p_i(e) - p_{i+1}(e) < \frac{2p_i(e)}{i(i-1)}$ in the worst case where $\delta = 1$ and ${{n-2}\choose{i-2}} = 1$ for small $TSP$. 

\subsubsection{The number of edges conforming to the probability decrement condition for the $OHC$ edges for the real-world $TSP$ instances}
Based on $p_i(e) > p_{i+1}(e)$ and $err > 0$, the number of the $OHC$ edges and ordinary edges meeting the conditions will be studied according to $i\in[4,7]$ for the five $TSP$ instances kroA100, kroB100, kroC100, kroD100 and rd100. The experimental results are illustrated in Table \ref{nooes}. In the third column and fourth column, the number of the $OHC$ edges with $p_i(e)>p_{i+1}(e)$ and $err>0$ is given, respectively. In the fifth column and sixth column, the number of such ordinary edges is given, respectively. In the seventh column, the number of edges with the  probability smaller than the probability bound $Lbp = \frac{1}{2}$ is given, and the total number of edges in $K_n$ is given in the last column. 

\begin{table}
	\begin{center}
		\caption{The number of the $OHC$ edges and ordinary edges with $p_i(e) > p_{i+1}(e)$ and $err > 0$ based on frequency $K_i$s ($i\in[4,8]$) for five $TSP$ instances.}
		{\footnotesize \begin{tabular}{ p{0.8cm}  p{1.2cm}  p{1.3cm}  p{1cm}  p{1.3cm}  p{1.2cm}  p{0.8cm}  p{0.5cm}}
				\hline
				$TSP$  &  $i\to i+1$ & $OHC$ & edges & Ordinary & edges  & $<Lbp$  & $N_{tot}$\\
				&   & $p_i>p_{i+1}$ & $err > 0$ & $p_i>p_{i+1}$ & $err > 0$  &  & \\
				\hline
				kroA100	&	$4\to5$	&	12	&	0	&	4478	&	3129	&	3536	&	4950	\\
				&	$5\to6$	&	9	&	0	&	4609	&	3699	&	3853	&	4950	\\
				&	$6\to7$	&	14	&	0	&	4676	&	4044	&	4029	&	4950	\\
				&	$7\to8$	&	13	&	1	&	4672	&	4218	&	4140	&	4950	\\
				kroB100	&	$4\to5$	&	17	&	0	&	4491	&	3092	&	3561	&	4950	\\
				&	$5\to6$	&	11	&	0	&	4624	&	3686	&	3858	&	4950	\\
				&	$6\to7$	&	13	&	0	&	4674	&	4019	&	4026	&	4950	\\
				&	$7\to8$	&	18	&	0	&	4678	&	4249	&	4140	&	4950	\\
				kroC100	&	$4\to5$	&	12	&	0	&	4470	&	3110	&	3540	&	4950	\\
				&	$5\to6$	&	11	&	0	&	4614	&	3692	&	3809	&	4950	\\
				&	$6\to7$	&	12	&	0	&	4666	&	4030	&	3994	&	4950	\\
				&	$7\to8$	&	22	&	0	&	4667	&	4237	&	4127	&	4950	\\
				kroD100	&	$4\to5$	&	15	&	0	&	4499	&	3094	&	3535	&	4950	\\
				&	$5\to6$	&	12	&	1	&	4599	&	3681	&	3842	&	4950	\\
				&	$6\to7$	&	16	&	0	&	4637	&	4044	&	4028	&	4950	\\
				&	$7\to8$	&	18	&	0	&	4682	&	4229	&	4152	&	4950	\\
				rd100	&	$4\to5$	&	7	&	0	&	4423	&	3098	&	3491	&	4950	\\
				&	$5\to6$	&	9	&	0	&	4610	&	3659	&	3794	&	4950	\\
				&	$6\to7$	&	9	&	0	&	4705	&	4011	&	3973	&	4950	\\
				&	$7\to8$	&	19	&	0	&	4739	&	4264	&	4092	&	4950	\\
				
				%\bottomrule
				\hline
		\end{tabular}}
		\label{nooes}
	\end{center}
\end{table}

For each $TSP$ instance, there are quite a few $OHC$ edges with $p_i(e) > p_{i+1}(e)$ from $i$ to $i+1$ where $i\in[4,7]$. However, the number of the $OHC$ edges with $err > 0$ is just zero or one from certain $i$ to $i+1$. It indicates that nearly all $OHC$ edges conform to the condition $p_{i+1}(e)\geq \left[1 - \frac{2}{i(i-1)}\right]p_i(e)$ if $p_i(e)$ decreases from $i$ to $i+1$. Even if the probabilities of some $OHC$ edges decrease from $i$ to $i+1$, they decrease very slow. To reduce the computation time, we did not compute the average frequency and probability  with all the frequency $K_i$s for each edge. Since 1000 frequency $K_i$s containing each edge are selected at random, the results in two or more experiments will have slight difference for some edges. If more number of frequency $K_i$s containing an edge are used, the experimental results will be better.

For the ordinary edges, one sees that there are a lot of them with $p_i(g) > p_{i+1}(g)$ from each $i$ to $i+1$ where $i\in[4,7]$, and the number of such ordinary edges increases according to $i$. It indicates that the probabilities of most ordinary edges decrease according to $i$. Moreover, most of the ordinary edges have the $err > 0$, and the number of such ordinary edges also increases according to $i$. It says that the probabilities of most ordinary edges begin decreasing according to the small  $i$s, and most of them decrease quickly according to $i$. In addition, the number of the ordinary edges with $p_i(g)<Lbp$ is increasing according to $i$. It says more and more ordinary edges are contained in the smaller percentage of $OP^i$s according to $i$. On the other hand, the $OHC$ edges and the other ordinary edges will be contained in more percentage of the $OP^i$s. Theorems \ref{th3} and \ref{th33} are approved by these experimental results. 

Comparing the number of ordinary edges with $p_i(e) > p_{i+1}(e)$ with that of ordinary edges with $err>0$, it is observed that the number of ordinary edges with $err>0$ increases faster according to $i$. It also rises faster than the number of ordinary edges with  
$p_i(g) < Lbp$. As $i$ is small, such as $i=4\to 5$ and $i=5\to 6$, the number of the ordinary edges with $err>0$ is obviously smaller than that of the ordinary edges with $p_i(g) < Lbp$. If $i=6\to 7$ and $i=7\to 8$, the number of the ordinary edges with $err>0$ approaches and becomes bigger than that of the ordinary edges with $p_i(g) < Lbp$. Thus, the condition $p_{i+1}(e) > \left[1 - \frac{2}{i(i-1)}\right]p_i(e)$ is better than the lower frequency bound and probability bound to separate $OHC$ edges from most ordinary edges in real-life applications.

The percents of the $OHC$ edges and ordinary edges with $err < 0$ from $i$ to $i+1$ are computed for the five $TSP$ instances, and the percent changes are lined according to $i + 1\in [5,8]$ and illustrated in Figure \ref{perchg}. The percent changes for the $OHC$ edges are denoted with the solid lines, and those for the ordinary edges are denoted with the dashed lines. Although the $OHC$s of the five $TSP$ instances have different structures, the percents of the $OHC$ edges with $err < 0$ for  the same pair of $i$ and $i+1$ are nearly equal where $i\in [4,7]$. For the ordinary edges related to the five $TSP$ instances, the percents of them with $err < 0$ are also very close from the same $i$ to $i+1$. It implies that nearly all $OHC$ edges of different $TSP$ instances conform to the condition $p_{i+1}(e) > \left[1-\frac{2}{i(i-1)}\right]p_i(e)$, whereas most ordinary edges comply with the opposite condition $p_{i+1}(g) < \left[1-\frac{2}{i(i-1)}\right]p_i(g)$ according to $i$. 

For different $TSP$ instances, the percent changes for the $OHC$ edges with $err<0$ have the similar trend according to $i$, and the percent changes for the ordinary edges with $err<0$ are also similar to each other. Moreover, the percents of the $OHC$ edges with $err < 0$ are close to 100\% for the five $TSP$ instances for given $i$, and they are nearly equal according to $i$. It implies that every $OHC$ edge is contained in the nearly equal percentage of $OP^i$s or more percentage of $OP^i$s according to $i$. 

On the other hand, the percents of the ordinary edges with $err < 0$ are below 37\% from $i=4$ to 5 for the five $TSP$ instances, and they decrease quickly according to $i$. From $i=7$ to 8, the percents of such ordinary edges decrease below 13\% for the five $TSP$ instances. It implies that more ordinary edges are contained in the smaller percentage of $OP^i$s according to $i$. Since $OHC$ edges will be contained in more percentage of $OP^i$s according to $n$, the percent of the ordinary edges with $err < 0$ will become smaller for the big scale of $TSP$. As $i$ is bigger than certain number, it can be predicted that all ordinary edges will conform to the condition $p_{i+1}(g) < \left[1-\frac{2}{i(i-1)}\right]p_i(g)$ from $i$ to $i+1$. Theorems \ref{th3} and \ref{th33} are verified by the experimental results. 

\begin{figure}
	\centering
	\includegraphics[width=3in,bb=0 0 400 200]{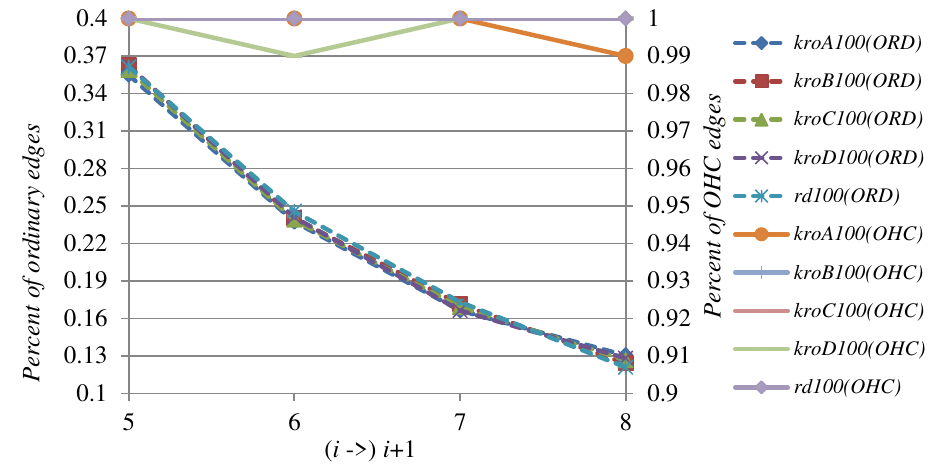}
	\caption{The percent changes for the ordinary edges ($ORD$) and $OHC$ edges with $err < 0$ according to $i+1$ for kroA100, kroB100, kroC100, kroD100 and rd100.}
	\label{perchg}
\end{figure}

 \subsubsection{The probability changes for all edges and $OHC$ edges according to $i$ for the real-world $TSP$ instances }
 The theorems and the above experimental results illustrated that an $OHC$ edge will be contained in more percentage of $OP^i$s than an ordinary edge in the average case. Moreover, the experiments illustrated that the number of the $OHC$ edges having the probability (or average frequency) close to the lower probability bound (or lower frequency bound) is very small for big scale of $TSP$. It says that the probabilities of the $OHC$ edges will increase faster than those of ordinary edges and all of edges according to $i$. In addition, the $OHC$ edges will be contained in more percentage of $OP^i$s according to $n$. Thus, the probabilities of $OHC$ edges will increase faster than those of ordinary edges according to $i$ and $n$, respectively. As a matter of fact, the probabilities of most ordinary edges will decrease according to $i$ even if $i$ is small. The experiments are executed to show the probability change for the $OHC$ edges and all edges according to $i$ for the five $TSP$ instances kroA100, kroB100, kroC100, kroD100 and rd100. For each $TSP$ instance, the probability of each edge is computed with the frequency $K_i$s for given $i$ where $i\in[4,8]$. Then the probabilities of the ${{100}\choose{2}}$ edges for each instance for given $i$ are ordered from small to big values, and one probability sequence $(p_1,p_2,...,p_k,...,p_{4950})$ is obtained, respectively,  where $k\in[1,4950]$. The ordered probabilities for each instance for given $i$ are lined according to $k\in[1,4950]$, respectively. For the $100$ $OHC$ edges for each instance, their probabilities are selected and lined according to their orders with respect to the probability sequence for given $i$, respectively. The probability changes for all edges and $OHC$ edges for the five $TSP$ instances are illustrated according to $i$ in Figures \ref{kroAB} and \ref{kroCDR}, respectively. 
 
 In the two pictures, the probability changes for the $OHC$ edges are denoted with the solid lines, and those for all edges are denoted with the dashed lines. In Figure \ref{kroAB}, the probability changes for the $OHC$ edges and all edges based on frequency $K_9$s are added for kroA100. It is observed that the probability changes related to all edges for the five $TSP$ instances are similar to each other according to given $i$, such as $i=4$, 5, etc. It implies that the probabilities of all edges conform to the nearly same distribution for the same size of $TSP$ instances. It says that the probabilities of edges computed with the frequency $K_i$s show the similar structure properties of edges for different $TSP$ instances. In addition, the number of edges with  $p_i(e) > \frac{1}{2}$ is much smaller than that of edges with  $p_i(e) < \frac{1}{2}$ for each of the $TSP$ instances. Moreover, the number of edges with $p_i(e) > \frac{1}{2}$ becomes smaller according to $i$. It says that more edges are contained in the smaller percentage of $OP^i$s according to $i$ so they have the smaller probabilities. One also sees that the small probabilities have obvious distinctions with respect to the same $k$ and different $i$s. For example according to a given $k$, the small probability with respect to certain $i+1$ is below that with respect to $i$. Based on Theorem \ref{th33}, two ordinary edges will keep their order relation with respect to their probabilities as $i$ is big. Most ordinary edges will keep the same order number $k$ with respect to the probability sequences according to $i$. In view of the probability changes for all edges, most ordinary edges are contained in the smaller percentage of $OP^i$s according to $i$. %Thus, the probabilities of most edges decreases according to $i$. 
 
 Moreover, the smallest probability of all edges also decreases according to $i$. As $i=4$, the smallest probability is close to 0.3 for the five $TSP$ instances. It becomes smaller than or near 0.1 as $i=5$. It continues decreasing according to $i$, and becomes zero as $i=7$ or 8. Meanwhile, the small probabilities of most edges decrease accordingly with this velocity. It says that the probabilities of most edges decrease quickly according to $i$. As $i$ is relatively big, the probabilities of most edges will approach zero, and only a small number of edges have the relatively big probabilities, see the probability changes for the small $TSP$ instances in Figure \ref{fcalle14}. In fact, most edges with the small probabilities are ordinary edges. The average probability of ordinary edges decreases according to $i$ based on Theorem \ref{th33}. The experiments illustrated that the probabilities of most ordinary edges decrease fast according to $i$. 
 
 If the probability is bigger than $\frac{1}{2}$, the difference between the probabilities according to different $i$s and the same $k$ is small for each of the $TSP$ instances. Moreover, the difference between the big probabilities at the same $k$ becomes smaller according to $i$ for these $TSP$ instances. For the edges with the big probabilities, the probability change conforms to the condition $p_{i+1}(e)\in \left[1-\frac{2}{i(i-1)}, 1+\frac{2}{i(i-1)}\right]p_i(e)$ from $i$ to $i+1$. Thus, the big probabilities at the same $k$ become closer according to $i$. 
 
 In view of the probability changes for the $OHC$ edges for the five $TSP$ instances for given $i$, they also have the similar trends according to $k$. For the same size of $TSP$ instances, the probabilities of the $OHC$ edges computed with the frequency $K_i$s have the nearly equal statistical parameters, such as the nearly equal minimum value, average value and maximum value, etc. Comparing the probability changes for the $OHC$ edges with those for all edges, the probabilities of the $OHC$ edges increase much faster according to $k$. For example for kroA100 as $i=4$, the smallest probability of the $OHC$ edges increases to the biggest probability according to $k$ within the narrow interval $[3606,4950]$. Moreover, the probabilities of the $OHC$ edges increase faster according to $i$ because the curve slopes become steeper. One sees that the biggest probability of the $OHC$ edges increases while the smallest probability increases or decreases according to $i$. Whether the smallest probability increases or decreases, the interval for $k$ changing becomes narrower and narrower according to $i$. Thus, the probabilities of the $OHC$ edges increase faster than those for all edges according to $i$. The probability changes imply that more percentage of the $OHC$ edges are used to build the $OP^i$s according to $i$. On the other hand, the smaller percentage of the ordinary edges are contained in the $OP^i$s according to $i$. %As $i$ is big enough, the probability of each $OHC$ edge will increase according to $i$. Thus, the frequency of each $OHC$ edge increases according to $i \leq P_0$, and the peak frequency is reached at $P_0$, respectively. 
 
 For kroA100, kroC100 and kroD100, the smallest probability of the $OHC$ edges is smaller than $\frac{1}{2}$ as $i \geq 5$ or $i\geq 6$. It implies that there are some $OHC$ edges having the probability $p_i(e)<\frac{1}{2}$ at some $i$s. Since the probabilities of the $OHC$ edges increase quickly according to $k$, the number of such $OHC$ edges is very small. We examined the experimental results and found there are one or two $OHC$ edges with $p_i(e)<\frac{1}{2}$ for all the five $TSP$ instances for given $i\in[5,8]$. For kroA100 at $i=9$, there are three $OHC$ edges with $p_i(e) <\frac{1}{2}$, and all the three smallest probabilities are bigger than 0.4. In addition, the smallest probability of the $OHC$ edges decreases according to $i$ as $i$ is small, such as $i=$4, 5 or 6 (for example, see the probability change for kroB100 and kroD100). It reaches the minimum value at certain number $i$. For example, the smallest probability of kroA100 reaches the minimum value as $i=7$. As $i$ becomes bigger, the smallest probability of the $OHC$ edges begins rising. For example, the smallest probability of the $OHC$ edges becomes bigger as $i=8$ and 9 for kroA100, $i=7$ and 8 for kroB100, $i=8$ for kroC100, $i=7$ and 8 for kroD100. For rd100, the smallest probabilities according to $i=7$ and 8 are nearly equal. 
 
 In the following, We did more experiments for rd100 to show the changes of the smallest probability of the $OHC$ edges according to $i\in[4,11]$. As $i$ becomes bigger, it is time-consuming to compute the average frequencies and probabilities for all edges. To reduce the computation time, the frequency and probability of each edge is computed with 100 frequency $K_i$s containing it, respectively. Because the number of sampled frequency $K_i$s is small for each edge, the probabilities of some edges will deviate from the corresponding expected values. Since the probabilities of most edges approach the expected values, the small number of inaccurate probabilities for some edges do not have much effect on the probability change for all edges and $OHC$ edges. %experiments, the probability of each edge is computed with 100 frequency $K_i$s. 
 The probability changes for all edges and $OHC$ edges for rd100 according to $i\in [4,11]$ and $k\in [1,4950]$ are illustrated in Figure \ref{rd10011}.

 It found that the probabilities of edges computed with the random frequency $K_i$s have slight difference in two or several experiments. For example, the smallest probabilities of the $OHC$ edges according to $i=7$ and 8 in Figure \ref{rd10011} are smaller than those in Figure \ref{kroCDR}, respectively. In general, the more accurate probability for an edge will be computed based on more number of frequency $K_i$s containing it. In theory, the probabilities shown in Figure \ref{kroCDR} for some edges will be more accurate than those illustrated in Figure \ref{rd10011} for these edges. For most other edges, the probabilities illustrated in the two Figures do not have much difference. Thus, the probability changes for all edges or $OHC$ edges in the two Figures are similar to each other according to $k$ and $i$. As $i$ becomes bigger, more ordinary edges have the even smaller probabilities, see the probability decrements from $i$ to $i+1$ at the same $k$. As the probability of an ordinary edge has one big decrement $pd_i(g) > \frac{2p_i(g)}{i(i-1)}$ at certain number $i$, the probability will decrease quickly according to $i$. From $i=9$ to 10 and 11,  more and more ordinary edges have the small probabilities tending to zero. On the other hand, the biggest probability of the $OHC$ edges increases accordingly. Moreover, the smallest probability of the $OHC$ edges becomes bigger as $i\geq 9$. It means that the $OHC$ edges with the smallest probability are contained in more percentage of $OP^i$s from $i\geq 9$. Although the probabilities of some $OHC$ edge decrease according to $i$ as $i$ is small, they will become bigger according to $i$ as $i$ is relatively big. Thus, the $OHC$ edges will have the bigger frequencies and probabilities than the ordinary edges as $i$ is big. 
\begin{figure}
	\centering
	\includegraphics[width=3in,bb=0 0 400 400]{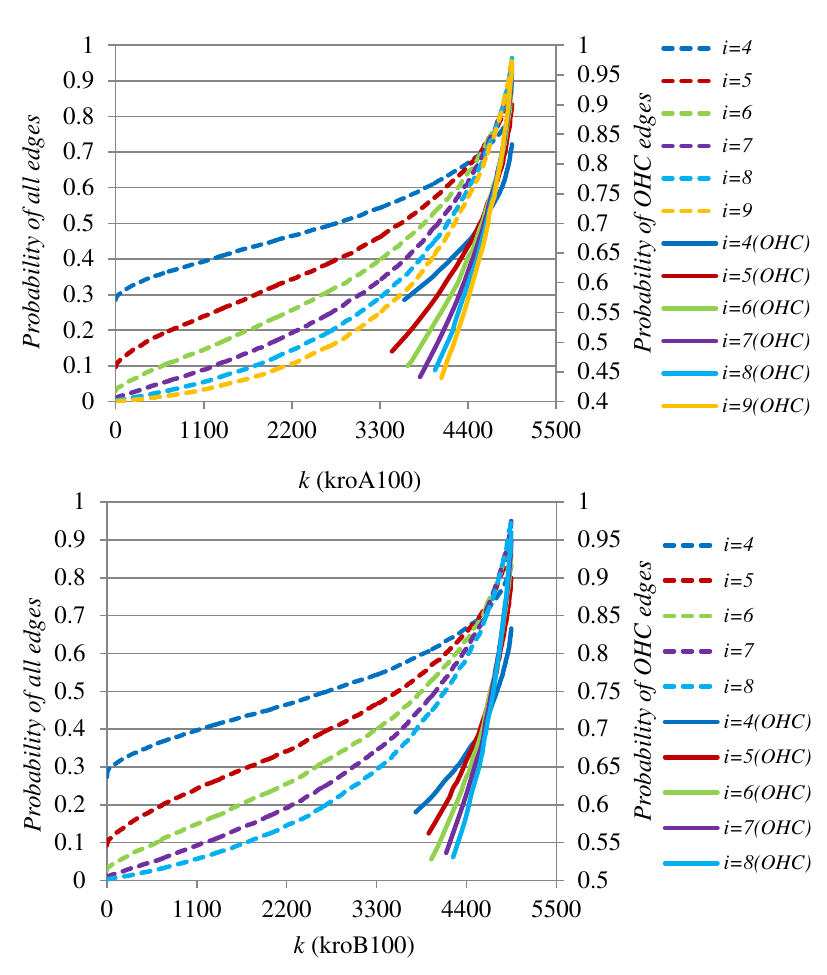}
	\caption{The probability changes for all edges and $OHC$ edges according to $k\in [1, 4950]$ and $i$ for kroA100 and kroB100.}
	\label{kroAB}
\end{figure}

\begin{figure}
	\centering
	\includegraphics[width=3in,bb=0 0 400 500]{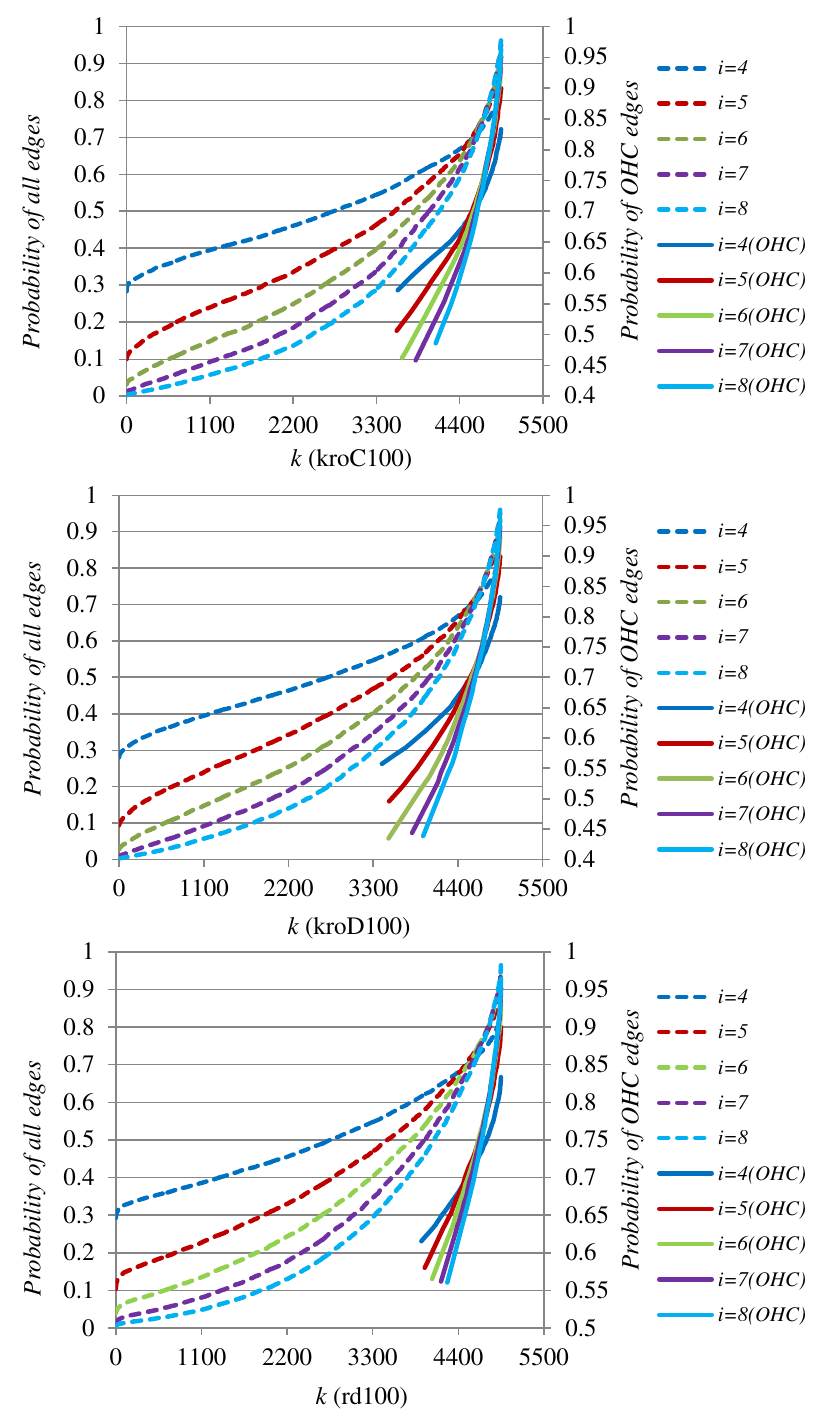}
	\caption{The probability changes for all edges and $OHC$ edges according to $k\in [1, 4950]$ and $i$ for kroC100, kroD100 and rd100.}
	\label{kroCDR}
\end{figure}

\begin{figure}
	\centering
	\includegraphics[width=3in,bb=0 0 400 200]{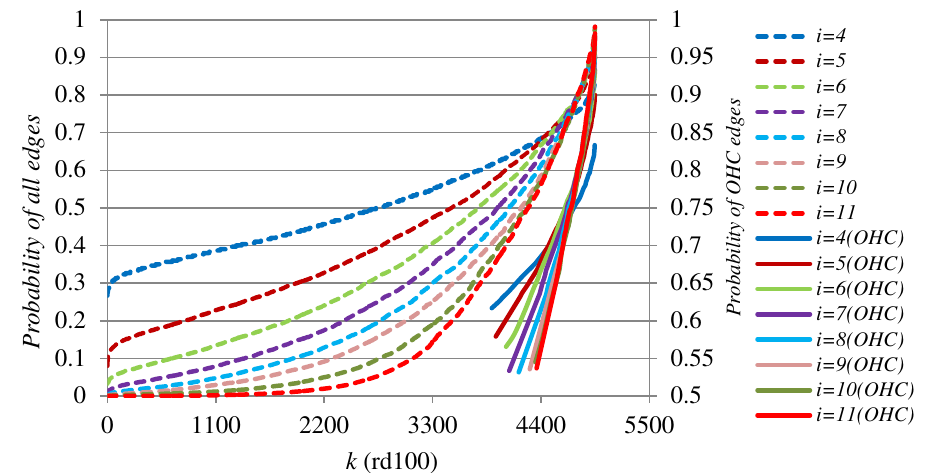}
	\caption{The probability changes for all edges and $OHC$ edges according to $k\in [1, 4950]$ and $i\in[4,11]$ for rd100.}
	\label{rd10011}
\end{figure}
%It means that every edge In other words, the frequency distribution for each edge is determined according to the frequency $K_i$s containing them, respecively. 
\subsubsection{The probability changes for the same $OHC$ edges and ordinary edges according to $i$ for rd100}
In this section, we shall do experiments to illustrate the probability changes for the same edges according to $i$ as the frequency of the edges are computed with the frequency $K_i$s. The probability changes for the $OHC$ edges with the smallest probabilities among all $OHC$ edges and the ordinary edges with the biggest probabilities among all ordinary edges based on frequency $K_4$s are considered. The ten $OHC$ edges with the smallest probabilities and ten ordinary edges with the biggest probabilities as $i=4$ are selected from rd100 for comparisons. The probabilities of these edges are computed with the frequency $K_i$s where $i$ changes from 4 to 11. To obtain the more accurate probabilities for the twenty edges, five experiments are executed to compute five probabilities for each edge according to each $i$, respectively. Moreover, 1000 frequency $K_i$s are randomly chosen for each edge to compute one probability in every experiment. Finally, the average value of the five probabilities  according to given $i$ is computed and taken as the probability of each edge at the given $i$, respectively. 

The probabilities of each edge are lined according to $i$ and shown in Figure \ref{prochrd100}. In Figure \ref{prochrd100}, the probability changes for the ten $OHC$ edges are shown in the first picture, and the probability changes for the ten ordinary edges are illustrated in the second picture. As more number of frequency $K_i$s are used and more experiments are executed for each edge, the probabilities of the edges will approach the corresponding expected probabilities nearer, respectively. One sees that the probabilities of the ten $OHC$ edges are bigger than 0.55 according to $i\in[4,11]$ although they have the smallest probabilities as $i=4$. It means that the probability of an $OHC$ edge will be bigger than $\frac{1}{2}$ according to $i$ for general $TSP$. 

For the ten $OHC$ edges with the smallest probabilities as $i=4$, their probability changes are different according to $i\in [4,11]$. The $p_i(e)$s of most of these $OHC$ edges decrease according to $i$ as $i$ is small, such as the edges (2,77), (10,51), (35,63), etc. The $p_i(e)$ of each of the edges reaches one smallest value at certain number $i$, and then they rise according to $i$ in the following stages, respectively. Moreover, the $p_i(e)$ of each edge rises quickly after the corresponding smallest value, respectively. For each of the $OHC$ edges, the $p_i(e)$ at $i=11$ is clearly much bigger than the corresponding smallest probability, respectively. It indicates that although the $p_i(e)$s of some $OHC$ edges decreases according to $i$ as $i$ is small, they will increase as $i$ is relatively big. This case only happens to the $OHC$ edges with the smallest  $p_i(e)$s as $i = 4$. For the $OHC$ edges with a big $p_i(e)$ at $i = 4$, $p_i(e)$ will always increase according to $i\in [4,n]$. Moreover, if the $p_i(e)$s of the $OHC$ edges decrease from $i$ to $i+1$, the decrement is relatively big according to the small number $i$, and it becomes smaller according to the bigger number $i$, see the probability change for edge (2,77) in Figure \ref{prochrd100}. It implies that the $p_i(e)$s of the $OHC$ edges conform to the condition $p_i(e) - p_{i+1}(e) \leq \frac{2p_i(e)}{i(i-1)}$ from $i$ to $i+1$ if $p_i(e) > p_{i+1}(e)$ occurs. Thus, $pd_i(e)$ will become smaller according to $i$ for $OHC$ edges, and it will tend to zero as $i$ is big. In addition, the step number for the probability decreasing is limited for these $OHC$ edges. For example, $p_i(e)$ of (2,77) decreases from $i=$4 to 8, and it has the maximum decreasing step  among those of the ten $OHC$ edges. It is still much smaller than the scale $n=100$ of rd100. The experiments illustrate that the step number of the probability decreasing for an $OHC$ edge is limited and small. The probability decreasing just occurs as $i$ is small. Once $i$ is relatively big, the $p_i(e)$s of $OHC$ edges will increase according to $i$. 
%holds bigger at small i, number of steps is limited, small enough 

\begin{figure}
	\centering
	\includegraphics[width=3.0in,bb=0 0 400 350]{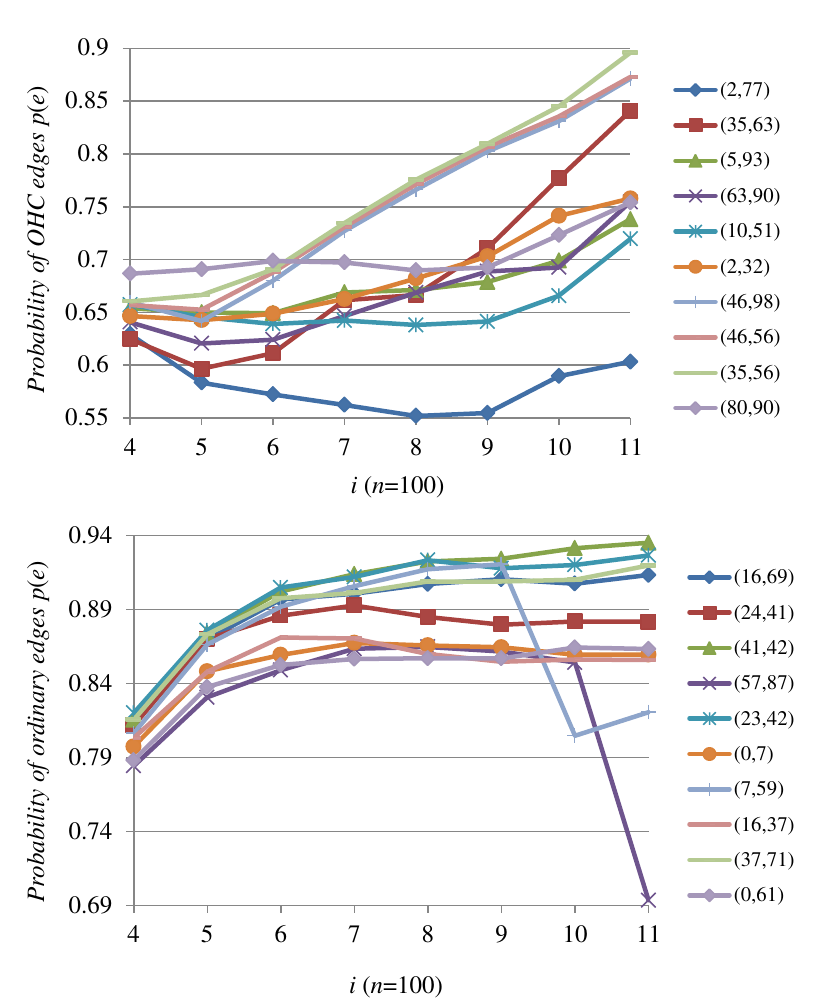}
	\caption{The probability changes for ten $OHC$ edges with the smallest probabilities and ten ordinary edges with the biggest probabilities at $i=4$  for rd100.}
	\label{prochrd100}
\end{figure}

For some $OHC$ edges, $p_i(e)$ always rises according to $i\in[4,11]$ although the $p_4(e)$ is not big, such as (35,56). This kind of probability change happens to most $OHC$ edges with the relatively big $p_i(g)$. It is interesting that the $p_i(e)$ of (35,56) increases quickly according to $i\in[6,11]$, and the rising rate does not become smaller according to $i$. It is different from the probability change for the $OHC$ edges with the big probability close to 1 as $i=4$, see Figures \ref{fpcs13or} and \ref{fpcs14or}. If the $OP^i$s are uniformly contained in the $OP^{i+1}$s from $i$ to $i+1$, the probabilities $p_{i+1}(e)$ and $p_i(e)$ will meet the condition $p_{i+1}(e) \leq \left[1 + \frac{2}{i(i-1)}\right]p_i(e)$ if $p_{i+1}(e) > p_i(e)$. We examined the probability increment for (35,56), and found that the probability increment  $p_{i+1}(e) - p_i(e) > \frac{2p_i(e)}{i(i-1)}$ exists from $i=7$. This case also happens to (46,56), (46,98) from $i=7$, and (35,63) from $i=8$. It indicates that the $OP^i$s containing more $OHC$ edges in the $K_{i+1}$s are contained in the $OP^{i+1}$s with the bigger probability, whereas the $OP^i$s containing more ordinary edges  are contained in the $OP^{i+1}$s with the smaller probability. Thus, the $p_i(e)$s of the $OHC$ edges increase faster than the average probability of all the edges contained in the $OP^n$s according to $i$, and the average probability of all edges contained in the $OP^n$s is generally bigger than those of most ordinary edges. If $p_i(e)$ increases,  $p_{i+1}(e) > \left[1+\frac{2}{i(i-1)}\right]p_i(e)$ will appear in most cases. Once $p_i(e)$ is close to 1, the probability increment will be smaller than $\frac{2p_i(e)}{i(i-1)}$ from $i$ to $i+1$. Theorems \ref{th1} and \ref{th3} are verified by the  numerical results. 

For the other $OHC$ edges, such as (2,32), (2,77), and (5,93), the probability increment is smaller than $\frac{2p_i(e)}{i(i-1)}$ from the smallest value to the bigger values in the first several steps. For example, the probability increment is smaller than $\frac{2p_i(e)}{i(i-1)}$ from $i=8$ to 9 for (2,77), from $i=5$ to 8 for (2,32), and from $i=6$ to 8 for (5,93). In the other steps where $p_i(e)$ increases from $i$ to $i+1$, the probability increment is bigger than $\frac{2p_i(e)}{i(i-1)}$ for these $OHC$ edges. %It is bigger than the probability increment (smaller than $\frac{2p_i(e)}{i(i-1)}$) computed for all edges contained in the $OP^n$s in the average case. 
In total, if the $p_i(e)$s of the $OHC$ edges increase from $i$ to $i+1$, the $p_i(e)$ much smaller than 1 (for example, smaller than 0.75) will have the probability increment bigger than $\frac{2p_i(e)}{i(i-1)}$ in most cases, whereas the $p_i(e)$ close to 1 (for example, bigger than 0.75) will have the probability increment smaller than  $\frac{2p_i(e)}{i(i-1)}$ as $i$ is small. In addition, $p_i(e)$ will increase quickly from $i$ to $i+1$ as $i$ is small, and it will maintain the nearly equal value as $i$ is big, see the probability changes for the $OHC$ edges as $i$ is close to or bigger than $P_0$ in Figures \ref{fpcs13or} and \ref{fpcs14or}. Theorems \ref{th1} and \ref{th3} are verified by the experimental results.

The probability changes for the ten selected ordinary edges according to $i\in[4,11]$ are illustrated in the second picture in Figure \ref{prochrd100}. It is observed that they have the big probabilities (from 0.7842 to 0.82 bigger than 0.75) as $i=4$. It implies that some ordinary edges have the big probability as $i$ is small. Moreover, the probabilities of a few ordinary edges rise  according to $i\in [4,11]$. It means that each of the ordinary edges is contained in more percentage of the $OP^i$s according to these $K_i$s. For the four edges (16,69), (23,42), (37,71) and (41,42), the probabilities are always rising from $i=4$ to 11. The probabilities increase quickly according to $i=4$ and 5. However, the probability increment becomes smaller if $i \geq6$, and the probability increment $p_{i+1}(g) - p_i(g)$ is much smaller than $\frac{2p_i(g)}{i(i-1)}$ from $i =6$ to 11. This case only happens to the $OHC$ edges with the probabilities close to 1. For the $OHC$ edges having the relatively smaller probabilities not close to 1, they will have the probability increment bigger than $\frac{2p_i(e)}{i(i-1)}$ in most cases. The experiments indicated that the probabilities of the ordinary edges increase slower than those of the $OHC$ edges according to $i$ in most cases, see Theorem \ref{th1}. If $p_i(g)$ of an ordinary edge increases from $i$ to $i+1$, $p_{i+1}(g) < [1+\frac{2}{i(i-1)}]p_i(g)$ will hold in most cases. 

%This is the difference between the $OHC$ edges and ordinary edges as their probabilties rise according to $i$.  
For the other ordinary edges with the big $p_i(g)$s as $i=4$, the $p_i(g)$s increase at first as $i$ is small. Each of them reaches the biggest value at the corresponding number $i$, then they decrease according to $i$ after the biggest values, respectively. Based on Theorem \ref{th30}, $p_i(g)$ for an ordinary edge will decrease from $i_d = 18$ for $n=$100. The $p_i(g)$s of most ordinary edges will have one big drop from $i$ to $i+1$ before $i = 18$. For example, the two ordinary edges (7,59) and (57,87), the $p_i(g)$s drop to the small values at $i=10$ and 11, respectively although the $p_4(g)$s are much bigger than $\frac{1}{2}$. Moreover, $p_{i+1}(g) < \frac{i}{i+1}p_i(g)$ appears for the two ordinary edges at $i=9$ and 10, respectively. The $p_i(g)$s of them will decrease quickly and never become bigger in the following stages based on Theorem \ref{th33}. Most ordinary edges with a big $p_i(g)$ at the small number $i$ have such probability changes, and they can be identified as $i$ is small. 

For the remainder ordinary edges with a big $p_i(g)$, the $p_i(g)$s decrease slowly according to $i \leq 11$ after the biggest values, respectively. However, the $p_i(e)$s of the $OHC$ edges are rising at this time, see the first picture in Figure \ref{prochrd100}. Thus, these ordinary edges can also be identified according to the decreasing probabilities. For example, if the $p_i(g)$ of an ordinary edge decreases according to $i$, $p_{i+1}(g) < \frac{i}{i+1}p_i(g)$ will happen as $i$ is relatively big. Plus the rising $p_i(e)$s of the $OHC$ edges at this time, the $p_i(g)$s of most ordinary edges will decrease quickly according to $i$. For example, although (7,59) and (57,87) has the big probabilities at $i=4$, the $p_i(g)$s decrease quickly after $i=10$ and 11, respectively. For the ordinary edges having the small $p_i(g)$s at the small number $i$, such as $i=4$, 5 or 6, their $p_i(g)$s will decrease quickly according to $i$, and most of the $p_i(g)$s will tend to zero as $i$ is relatively big. Figure \ref{rd10011} also illustrated that the $p_i(g)$s of more and more ordinary edges approach zero according to $i$. Such ordinary edges can be identified at the small number $i$. %big probabilities, ot the other decrease quickly,  the  

In addition, one sees some ordinary edges will maintain the big $p_i(g)$s according to $i$ as $i$ is not big. For the ordinary edges with the big $p_i(g)$s while the $p_i(g)$s increase according to $i$ as $i$ is small, such as (41,42), (23,42), (37,71) and (16,69), they are not easily separated from some $OHC$ edges based on the frequencies or probabilities computed with the frequency $K_i$s containing a small number of vertices. To disclose such ordinary edges, the frequency $K_i$s containing more vertices (such as $i_d$ vertices) will be used. However, it is time-consuming to compute the frequency $K_i$s containing many vertices using the exact methods, such as dynamic programming. In the experiments, the frequency $K_i$s containing one given edge are chosen at random from $K_n$ to compute the frequency and probability of the edge. Although the frequency and probability of each edge is illustrated, the restriction of the adjacent $OHC$ edges to reduce the frequency and probability of the ordinary edges is relaxed. Based on Theorems \ref{th3}, \ref{th33} and \ref{th30}, an ordinary edge in  the $K$ frequency $K_i$s containing one or two pairs of adjacent $OHC$ edges will have the small frequency. If we choose the frequency $K_i$s for three adjacent edges containing one vertex for computing the frequencies and probabilities of the three adjacent edges, the frequency and probability of the ordinary edge will be much smaller than those of the two $OHC$ edges according to $i$. In theory, it will consume $O(Ni^42^in^4)$ time to discovery all $OHC$ edges based on frequency $K_i$s where $N$ is the number of the selected frequency $K_i$s containing each pair of three adjacent edges containing a vertex and the $OP^i$s are computed with dynamic programming. % as $i$ is fixed. %ordinary only choose We guess such ordinary edges have the other special structure properties comparing with the $OHC$ edges based on frequency $K_i$s. For example, these ordinary edges will have some special frequencies or probabilities which are different from those of the adjacent $OHC$ edges. 
This work will be studied in the next paper. 
\section{Conclusion}
\label{sec7}
The frequency $K_i$s ($i\in [4,n]$) are studied for characterizing the special structure properties of the $OHC$ edges for symmetric $TSP$. $OHC$ edges illustrate the much higher frequencies and probabilities than those of ordinary edges computed with the frequency $K_i$s at certain $i\ll n$. Firstly, an $OHC$ edge related to a given $K_i$ has the frequency bigger than $\frac{1}{2}{{i}\choose{2}}$ in the  frequency $K_i$, whereas an ordinary edge has the frequency smaller than $2(i-3)$. On average, the expected frequency of an $OHC$ edge is bigger than $\frac{i^2-4i+7}{2}$ whereas an ordinary edge has the expected frequency smaller than 2 in each frequency $K_i$. Secondly, as the frequency of each edge in $K_n$ is computed with the frequency $K_i$s, an $OHC$ edge in $K_n$ has the average frequency bigger than $\frac{1}{2}{{i}\choose{2}}$. It indicates that an $OHC$ edge in $K_n$ is the $OHC$ edge in a $K_i$ containing it on average. Moreover, an $OHC$ edge will be contained in the nearly equal or more percentage of $OP^i$s according to $i\in[4,n]$. If the probability $p_i(e)$ of an $OHC$ edge decreases from $i$ to $i+1$, the probability decrement is smaller than $\frac{2p_i(e)}{i(i-1)}$. If $p_i(e)$ close to 1 increases from $i$ to $i+1$ as $i$ is small, the probability increment will be smaller than $\frac{2p_i(e)}{i(i-1)}$. Otherwise, the probability increment will be bigger than $\frac{2p_i(e)}{i(i-1)}$ in most cases. If $p_i(e)$ decreases according to the small $i$s, the decreasing step is limited comparing to $n$, and it will increase according to $i$ after the smallest value. Thus, the frequency of an $OHC$ edge will always increase according to $i$ until $i=\frac{n}{2}+2$ for even $n$ or $i=\frac{n+1}{2}+1$ for odd $n$. 

For ordinary edges, most of them are contained in a small percentage of the $OP^i$s, and the probability $p_i(g)$ that they are contained in the $OP^i$s becomes smaller according to $i$ in most cases. The probability decrement is generally bigger than $\frac{2p_i(g)}{i(i-1)}$ from $i$ to $i+1$, and it is bigger than $\frac{p_i(g)}{i+1}$ in most cases. Although $p_i(g)$ increases for some ordinary edges as $i$ is small, the probability increment will be smaller than $\frac{2p_i(g)}{i(i-1)}$ in most cases. In total, $p_i(g)$ for an ordinary edge decreases according to $i$ in the average case. Moreover, $p_i(g)$ definitely becomes smaller if $i \geq i_d$ where $i_d = O(n^\frac{4}{7})$. Based on the findings, a dynamic programming algorithm is presented to find all $OHC$ edges in $O(n^2i_d^42^{i_d})$ time. %If the probability change for each edge is evaluated according to $i$, less time will be used for finding the $OHC$. %Moreover, the average frequency of an ordinary edge will be smaller than $\frac{1}{2}{{i}\choose{2}}$ if $i \geq 2i_d$ where $i_d$ is the smallest number conforming to formula (\ref{F5}) and $i_d = O(n^{\frac{4}{7}})$. Based on the findings and dynamic programming, an algorithm is presented to find all $OHC$ edges in $O(n^2i_d^42^{2i_d})$ time. If the probability change for each edge is evaluated according to $i$, less time will be used for finding the $OHC$.  
%In applications, one can compute the probability that each edge is contained in the $OP^i$s with a number of $K_i$s in $K_n$. The edges out of $OHC$ can be identified according to the probability change according to $i$. That is to say, the probability will have a big drop before $\frac{n}{2}+2$ or $\frac{n+1}{2} +1$. 
%In addition, the graph containing all edges in $OP^n$s has more than $1.5n$ edges and it will have less than $\frac{n(n+2)}{4}$ edges in the average case. %For the relationships between the $OP^n$s and $OHC$, each $OP^n$ contains more than $\frac{3(n-1)}{4}$ $OHC$ edges. This result is also right for $OHP$ which is one $OP^n$. %and it contains at most $\frac{n}{6}$ edges 
%out of $OHC$ in the best average case.
%In applications, the probability that an edge is contained in $OP^i$s is useful to identify $OHC$ edges. %than the distances on edges. %According to $i$, the probability for an $OHC$ edge will increase or keep the nearly equal value. For most of the other edges, the probability will decrease or has a non-negligible decrement at some number $i<P_0$. 

The experiments are executed for various $TSP$ instances for verifying the findings. If most $K_i$s in $K_n$ contain one $OHC$ and ${{i}\choose{2}}$ $OP^i$s for general $TSP$, the frequencies and probabilities of the $OHC$ edges and ordinary edges computed with the frequency $K_i$s totally conform to the theorems. It indicates that the $OHC$ edges and ordinary edges have different structure properties which can be characterized by the frequency $K_i$s. Comparing with the distances of edges, the $OHC$ edges and ordinary edges are better characterized with the frequencies and probabilities computed based on frequency $K_i$s. 

Given a $TSP$ instance, if the probabilities of all ordinary edges are smaller than $\frac{1}{2}$ or decreases before certain constant number $c$, the $OHC$ will be found in $O(n^2c^42^{c})$ time. In the future, we will study if there are such special $TSP$, and the related algorithms based on the  frequency $K_i$s. In addition, the frequency and probability change for the ordinary edges with the big frequencies at the small numbers $i$ will be researched, especially for the $n$ special ordinary edges. %In addition, the lower frequency bound for $OHC$ edges will be explored under the constrains of $OP^i$s.

\section{Acknowledgements}
\label{sec7}
The author wants to thank the Fundamental Research funds for the Central Universities (No.2024JC006). The work benefits from the State Key Laboratory of Alterate Electrical Power System with Renewable Energy Sources (NCEPU).

\section{Declaration of competing interests}
\label{sec8}
The author declared that they had no affiliations with or involvement in any organization or entity with any financial interest in the subject matter or materials discussed in the manuscript. 
\section{Declaration of generative AI in scientific writing}
\label{sec9}
The author declared that there were no use of generative AI and AI-assisted technologies in scientific writing in the manuscript. 
\appendix
\section{Six frequency $K_4$s for $K_4$}
\label{app1}

The six frequency $K_4$s for any $K_4$ are computed as follows.
Given $K_4$ on four vertices $A$, $B$, $C$ and $D$, the layout of the four vertices is illustrated as that in Fig. \ref{quadrilaterals} (a). There are six edges $(A,B)$, $(A,C)$, $(A,D)$, $(B,C)$, $(B,D)$, and $(C,D)$. The distances of the six edges are denoted as $d(A,B)$, $d(A,C)$, $d(A,D)$, $d(B,C)$, $d(B,D)$ and $d(C,D)$, respectively. The six $OP^4$s in $K_4$ are up to the ordering of the three distance sums $d(A,B)+d(C,D)$, $d(A,C)+d(B,D)$ and $d(A,D)+d(B,C)$. For example, if $d(A,B)+d(C,D) < d(A,D)+d(B,C) < d(A,C)+d(B,D)$ is considered,
the six $OP^4$s are computed as $(A,D,C,B)$, $(A,B,D,C)$, $(A,B,C,D)$, $(B,A,D,C)$, $(B,A,C,D)$ and $(C,B,A,D)$. We enumerate the number of $OP^4$s containing an edge, and this number represents the frequency of the edge. For example, $(A,B)$ is contained in five of the $OP^4$s so it has the frequency of 5. As the frequency of each edge is computed with the six $OP^4$s, the frequency $K_4$ is illustrated in Figure  \ref{quadrilaterals} (a). The inequality related to the three distance sums is given below the frequency $K_4$. 
Since the three distance sums $d(A,B)+d(C,D)$, $d(A,C)+d(B,D)$ and $d(A,D)+d(B,C)$ have the other five orderings, one can derive the corresponding $OP^4$s and frequency $K_4$s. The other five frequency $K_4$s, and the inequalities related to the distance sums
are illustrated in Figure \ref{quadrilaterals} (b) $\thicksim$ (f), respectively.

\begin{figure}
	\centering
	\includegraphics[width=5.0in,bb=0 0 400 150]{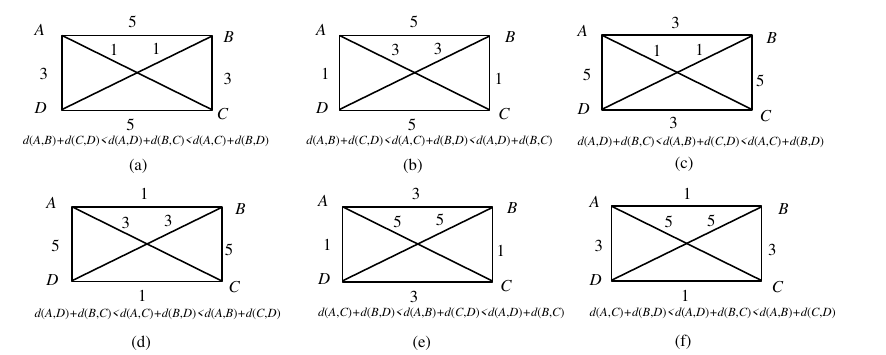}
	\caption{Six frequency $K_4$s for certain $K_4=ABCD$.}
	\label{quadrilaterals}
\end{figure}

In each frequency $K_4$, the frequency
of an edge is 1, 3 or 5. It implies that the edge is contained in 1, 3 or 5 $OP^4$s. Moreover, the frequency of $OHC$ edges in $K_4$ is either 3 or 5. For example, the $OHC$ edge $(A,B)$ has the frequency of 3 or 5 in (a), (b), (c), and (e). This means that the minimum frequency of $OHC$ edges in $K_4$ is 3. However, the expected frequency for an $OHC$ edge will be bigger than 3. Given an $OHC$ edge, such as $(A,B)$, it has the frequency of 3 in (c) and (e), and the frequency of 5 in (a) and (b). Thus, the expected frequency for $(A,B)$ is 4 with respect to the four frequency $K_4$s. Since there are six $OP^4$s in $ABCD$, the probability that $(A,B)$ is contained in an $OP^4$ in $ABCD$ is $\frac{2}{3}$. In addition, the two vertex-disjoint edges, such as $(A,B)$ vs $(C,D)$, $(A,C)$ vs $(B,D)$ and $(A,D)$ vs $(B,C)$, have the same frequency in frequency $K_4$. On the other hand, the three edges containing a vertex have three  different frequencies 1, 3 and 5, respectively. 

\section{The frequencies on the edges in the six $OP^4$s in a Frequency $K_4$}
\label{app2}

 $K_4$ contains six $OP^4$s. Each $OP^4$ contains three edges. One is in the middle, and the others are on both sides, respectively. Given an $OP^4$ in frequency $K_4$, the frequencies on the three edges are definite \cite{DBLP:journals/Wang24}. Without loss of generality, the six $OP^4$s in $ABCD$ in Fig.\ref{quadrilaterals} (b)
are shown in Fig.\ref{optimalpaths}, and the frequency on each edge is also illustrated. It is observed that two side edges in each $OP^4$ have the same frequency. Moreover, if the middle edge has a frequency of 1 or 3, the side edges must have the frequency of 5. As the middle edge has the frequency of 5, 
the side edges must have the frequency of 3 rather than 1. The side edges in an $OP^4$ never have the frequency of 1. Given an intermediate vertex in an $OP^4$, the frequency sum for the two adjacent edges containing the vertex is 8 or 6. According to the six $OP^4$s, frequency 6 appears four times and frequency 8 appears eight times. Thus, the expected frequency sum for two adjacent edges is $\frac{22}{3}$. %In frequency $K_4$, the three edges containing a vertex have the frequency sum 9. 
As $OHC$ in $K_n$ contains $n$ $OP^4$s, the frequencies on the edges in $OP^4$s can be taken as the constraints to derive the frequency bound for an $OHC$ edge or two adjacent $OHC$ edges. %Under the constrains of the $OP^4$s, the frequency of an $OHC$ edge in $K_n$ is bigger than that computed  
%As $e$ is denoted by vertices,
%such as $e=(A,B)$, $f(A,B)$ is equal to $f(e)$.
\begin{figure}
	\centering
	\includegraphics[width=3.5in,bb=0 0 450 100]{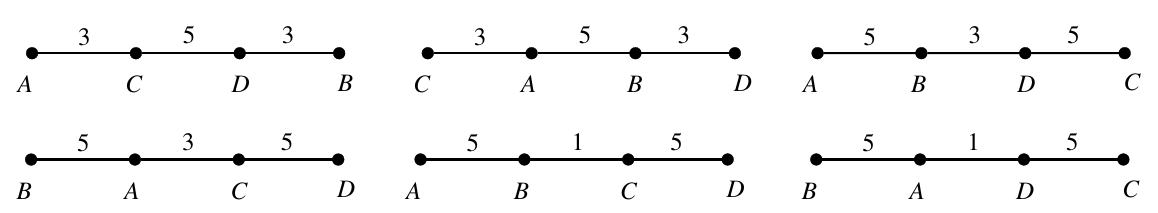}
	\caption{The frequency on each edge in the six $OP^4$s in Figure \ref{quadrilaterals} (b).}
	\label{optimalpaths}
\end{figure}
%%       \cite{<label>} ==> [1]

%%
%Example citation, See \cite{lamport94}.




\end{document}